\PassOptionsToPackage{unicode}{hyperref}
\PassOptionsToPackage{hypertexnames=false}{hyperref} 
\PassOptionsToPackage{hyphens}{url}
\documentclass[
  12pt,
  a4paper,
]{article}
\makeatletter
\let\mc@listI\@listI \let\mc@listi\@listi \let\mc@listii\@listii \let\mc@listiii\@listiii
\edef\mc@parskip{\the\parskip}\edef\mc@parindent{\the\parindent}
\edef\mc@parfillskip{\the\parfillskip}\edef\mc@partopsep{\the\partopsep}
\let\mc@fpsfigure\fps@figure
\let\mc@topfraction\topfraction \let\mc@bottomfraction\bottomfraction
\let\mc@textfraction\textfraction \let\mc@floatpagefraction\floatpagefraction
\edef\mc@emergencystretch{\the\emergencystretch}
\edef\mc@secnumdepth{\the\c@secnumdepth}
\let\mc@thefigure\thefigure
\makeatother

\ifx\pdfminorversion\undefined\else\pdfminorversion=7\fi
\usepackage{amsmath}
\usepackage{amsthm} 
\usepackage{setspace}
\usepackage{iftex}
\usepackage[T1]{fontenc}
\usepackage[utf8]{inputenc}
\usepackage{textcomp} %
\usepackage{newtxtext}   %
\usepackage{newtxmath}   %
\IfFileExists{upquote.sty}{\usepackage{upquote}}{}
\IfFileExists{microtype.sty}{%
  \usepackage[]{microtype}
  \UseMicrotypeSet[protrusion]{basicmath} %
}{}
\makeatletter
\@ifundefined{KOMAClassName}{%
  \IfFileExists{parskip.sty}{%
    \usepackage{parskip}
  }{%
    \setlength{\parindent}{0pt}
    \setlength{\parskip}{6pt plus 2pt minus 1pt}}
}{%
  \KOMAoptions{parskip=half}}
\makeatother
\usepackage{xcolor}
\usepackage[top=20mm,left=25mm,right=25mm,heightrounded]{geometry}
\usepackage{graphicx}
\makeatletter
\def\maxwidth{\ifdim\Gin@nat@width>\linewidth\linewidth\else\Gin@nat@width\fi}
\def\maxheight{\ifdim\Gin@nat@height>\textheight\textheight\else\Gin@nat@height\fi}
\makeatother
\makeatletter\let\mc@Gin@ewidth\Gin@ewidth \let\mc@Gin@eheight\Gin@eheight\makeatother
\setkeys{Gin}{width=\maxwidth,height=\maxheight,keepaspectratio}
\makeatletter
\def\fps@figure{htbp}
\makeatother
\renewcommand{\topfraction}{0.9}    %
\renewcommand{\bottomfraction}{0.9} %
\renewcommand{\textfraction}{0.1}   %
\renewcommand{\floatpagefraction}{0.8} %
\usepackage[sort&compress,super,numbers]{natbib}
\makeatletter
\renewcommand\@biblabel[1]{#1.}
\makeatother
\usepackage{tikz}
\usepackage{float}
\usepackage{pifont}
\IfFileExists{bookmark.sty}{\usepackage{bookmark}}{\usepackage{hyperref}}
\IfFileExists{xurl.sty}{\usepackage{xurl}}{} %
\hypersetup{
  pdftitle={Cooperation thrives when moral judgment matches moral conduct},
  pdfauthor={Jorge Peña, Julian Garcia, Toby Handfield},
  colorlinks=true,
  allcolors=blue,
  pdfcreator={LaTeX via pandoc}}

\usepackage{todonotes}
\usepackage{bm}

\newcommand{\Jorge}{red!40} %
\newcommand{\Julian}{blue!40} %
\newcommand{\Toby}{yellow!40} %
\newcommand{\Fernando}{green!40} %
\newcommand{\QDset}{\mathcal{Q}_{\mathcal{D}}}
\newcommand{\QDisc}{\mathcal{Q}_{\text{DISC}}}
\newcommand{\QCDisc}{\mathcal{Q}_{\text{CDISC}}}
\newcommand{\QPDisc}{\mathcal{Q}_{\text{PDISC}}}
\newcommand{\QODisc}{\mathcal{Q}_{\text{ODISC}}}

\usepackage{color, colortbl}
\definecolor{LRed}{rgb}{1,.44,.46}
\definecolor{LBlue}{rgb}{.7,.8,1}
\definecolor{LPurple}{rgb}{.8,.2,.8}

\usepackage[english]{babel}
\usepackage[autolanguage]{numprint} %

\newcommand{\Rr}{R}
\newcommand{\Ss}{S}
\newcommand{\Tt}{T}
\newcommand{\Pp}{P}

\newcommand{\af}{a_1} %
\newcommand{\ac}{a_2} %
\newcommand{\rf}{r_1} %
\newcommand{\rc}{r_2} %

\title{Cooperation thrives when moral judgment matches moral conduct}

\usepackage{authblk}

\author[1,$\ast$]{Jorge Peña}
\author[2]{Julian Garcia}
\author[3]{Toby Handfield}
\affil[1]{Department of Social and Behavioral Sciences, Toulouse School of Economics, Toulouse, France}
\affil[2]{Department of Data Science and AI, Faculty of Information Technology, Monash University, Melbourne, Australia}
\affil[3]{School of Philosophical, Historical and International Studies, Monash University, Melbourne, Australia}
\affil[$\ast$]{Corresponding author: \href{mailto:jorge.pena@tse-fr.eu}{jorge.pena@tse-fr.eu}}

\date{}

\makeatletter\let\mc@k@caption\caption \let\mc@k@@caption\@caption \let\mc@k@makecaption\@makecaption\makeatother
\usepackage{subcaption}
\definecolor{SubduedGray}{rgb}{0.8, 0.8, 0.8}
\newtheorem{theorem}{Theorem}
\newtheorem{lemma}{Lemma}
\newtheorem{proposition}{Proposition}
\newtheorem{corollary}{Corollary}
\newtheorem{definition}{Definition}
\newtheorem{remark}{Remark}
\makeatletter
\newcommand\mc@xref[2]{\@ifundefined{r@#2}{}{\expandafter\let\csname r@#1\expandafter\endcsname\csname r@#2\endcsname}}
\AtBeginDocument{%
\mc@xref{SI-subsec:labeling-conventions}{subsec:labeling-conventions}%
\mc@xref{SI-C=G=1,D=B=0}{C=G=1,D=B=0}%
\mc@xref{SI-subsec:payoff-function}{subsec:payoff-function}%
\mc@xref{SI-payoff-matrix}{payoff-matrix}%
\mc@xref{SI-subsec:cooperative-dilemmas}{subsec:cooperative-dilemmas}%
\mc@xref{SI-R>P}{R>P}%
\mc@xref{SI-R>S}{R>S}%
\mc@xref{SI-T>P}{T>P}%
\mc@xref{SI-2R>T+S}{2R>T+S}%
\mc@xref{SI-T>R}{T>R}%
\mc@xref{SI-P>S}{P>S}%
\mc@xref{SI-PD}{PD}%
\mc@xref{SI-SG}{SG}%
\mc@xref{SI-SH}{SH}%
\mc@xref{SI-T>S}{T>S}%
\mc@xref{SI-normalized-PD}{normalized-PD}%
\mc@xref{SI-normalized-SG}{normalized-SG}%
\mc@xref{SI-normalized-SH}{normalized-SH}%
\mc@xref{SI-subsec:shorthand}{subsec:shorthand}%
\mc@xref{SI-dr1r2-shorthand}{dr1r2-shorthand}%
\mc@xref{SI-tildedr1r2-shorthand}{tildedr1r2-shorthand}%
\mc@xref{SI-notdr1r2-shorthand}{notdr1r2-shorthand}%
\mc@xref{SI-subsec:norm-counting}{subsec:norm-counting}%
\mc@xref{SI-prop:norm-counting}{prop:norm-counting}%
\mc@xref{SI-Nk-recursion}{Nk-recursion}%
\mc@xref{SI-Ok-cardinality}{Ok-cardinality}%
\mc@xref{SI-subsec:resident-dynamics}{subsec:resident-dynamics}%
\mc@xref{SI-deltaij}{deltaij}%
\mc@xref{SI-gdot}{gdot}%
\mc@xref{SI-g*}{g*}%
\mc@xref{SI-abc}{abc}%
\mc@xref{SI-subsec:cooperation-index}{subsec:cooperation-index}%
\mc@xref{SI-coop-index}{coop-index}%
\mc@xref{SI-subsec:mutant-dynamics}{subsec:mutant-dynamics}%
\mc@xref{SI-tildedeltaij}{tildedeltaij}%
\mc@xref{SI-hdot}{hdot}%
\mc@xref{SI-h*}{h*}%
\mc@xref{SI-subsec:payoffs-stability}{subsec:payoffs-stability}%
\mc@xref{SI-omegaij}{omegaij}%
\mc@xref{SI-w}{w}%
\mc@xref{SI-tildeomega}{tildeomega}%
\mc@xref{SI-tildew}{tildew}%
\mc@xref{SI-evol-stab-cond}{evol-stab-cond}%
\mc@xref{SI-subsec:stability-indices}{subsec:stability-indices}%
\mc@xref{SI-avg-sigma}{avg-sigma}%
\mc@xref{SI-sigma-gamma}{sigma-gamma}%
\mc@xref{SI-shoelace}{shoelace}%
\mc@xref{SI-A-PD}{A-PD}%
\mc@xref{SI-A-SG}{A-SG}%
\mc@xref{SI-A-SH}{A-SH}%
\mc@xref{SI-A-SG-R+P>2S}{A-SG-R+P>2S}%
\mc@xref{SI-fig:lemma-regions}{fig:lemma-regions}%
\mc@xref{SI-lem:sigma-SG-R+P>2S}{lem:sigma-SG-R+P>2S}%
\mc@xref{SI-A-PD-2R>T+P}{A-PD-2R>T+P}%
\mc@xref{SI-lem:sigma-PD-2R>T+P}{lem:sigma-PD-2R>T+P}%
\mc@xref{SI-subsubsec:weighted-sigma}{subsubsec:weighted-sigma}%
\mc@xref{SI-subsec:def-conditional-discriminating}{subsec:def-conditional-discriminating}%
\mc@xref{SI-def:cond}{def:cond}%
\mc@xref{SI-def:disc}{def:disc}%
\mc@xref{SI-tab:strategies}{tab:strategies}%
\mc@xref{SI-def:contrite}{def:contrite}%
\mc@xref{SI-def:conditional}{def:conditional}%
\mc@xref{SI-def:discriminating}{def:discriminating}%
\mc@xref{SI-def:contrite-system}{def:contrite-system}%
\mc@xref{SI-subsec:def-self-approving-righteous}{subsec:def-self-approving-righteous}%
\mc@xref{SI-def:self-approving}{def:self-approving}%
\mc@xref{SI-self-approving-cond}{self-approving-cond}%
\mc@xref{SI-def:self-disapproving}{def:self-disapproving}%
\mc@xref{SI-self-disapproving-cond}{self-disapproving-cond}%
\mc@xref{SI-def:self-cooperative}{def:self-cooperative}%
\mc@xref{SI-def:righteous}{def:righteous}%
\mc@xref{SI-righteous-set}{righteous-set}%
\mc@xref{SI-def:universally-stable}{def:universally-stable}%
\mc@xref{SI-subsec:def-consistent}{subsec:def-consistent}%
\mc@xref{SI-def:consistent}{def:consistent}%
\mc@xref{SI-consistent-cond}{si-consistent-cond}%
\mc@xref{SI-consistent-cond-G}{si-consistent-cond-G}%
\mc@xref{SI-consistent-cond-B}{si-consistent-cond-B}%
\mc@xref{SI-rem:consistent-equivalence}{rem:consistent-equivalence}%
\mc@xref{SI-def:consistent-disc}{def:consistent-disc}%
\mc@xref{SI-QDisc-set}{QDisc-set}%
\mc@xref{SI-QCDisc-set}{QCDisc-set}%
\mc@xref{SI-QODisc-set}{QODisc-set}%
\mc@xref{SI-QPDisc-set}{QPDisc-set}%
\mc@xref{SI-subsec:leading8}{subsec:leading8}%
\mc@xref{SI-leading-eight-assessment}{leading-eight-assessment}%
\mc@xref{SI-dBGD=dGGD=B}{dBGD=dGGD=B}%
\mc@xref{SI-dBGC=dGGC=G}{dBGC=dGGC=G}%
\mc@xref{SI-dGBD=G}{dGBD=G}%
\mc@xref{SI-fig:consistent-discriminating}{fig:consistent-discriminating}%
\mc@xref{SI-LI-set}{LI-set}%
\mc@xref{SI-L2}{L2}%
\mc@xref{SI-LII-set}{LII-set}%
\mc@xref{SI-fig:leading-eight}{fig:leading-eight}%
\mc@xref{SI-L6}{L6}%
\mc@xref{SI-LIII-set}{LIII-set}%
\mc@xref{SI-leading-eight}{leading-eight}%
\mc@xref{SI-subsec:mirror-systems}{subsec:mirror-systems}%
\mc@xref{SI-def:mirror-system}{def:mirror-system}%
\mc@xref{SI-mirror-def}{mirror-def}%
\mc@xref{SI-mirror-def-p}{mirror-def-p}%
\mc@xref{SI-mirror-def-d}{mirror-def-d}%
\mc@xref{SI-lem:mirror-reputation}{lem:mirror-reputation}%
\mc@xref{SI-mirror-delta}{mirror-delta}%
\mc@xref{SI-mirror-g-h}{mirror-g-h}%
\mc@xref{SI-prop:mirror-invariance}{prop:mirror-invariance}%
\mc@xref{SI-cor:mirror-classes}{cor:mirror-classes}%
\mc@xref{SI-subsec:self-approving}{subsec:self-approving}%
\mc@xref{SI-prop:self-approving}{prop:self-approving}%
\mc@xref{SI-dBB+dGG=dBG+dGB}{dBB+dGG=dBG+dGB}%
\mc@xref{SI-limg_linear}{limg_linear}%
\mc@xref{SI-dBB+dGGneqdBG+dGB}{dBB+dGGneqdBG+dGB}%
\mc@xref{SI-limg_quadratic}{limg_quadratic}%
\mc@xref{SI-limg_quadratic_1}{limg_quadratic_1}%
\mc@xref{SI-limg_quadratic_2}{limg_quadratic_2}%
\mc@xref{SI-lem:partition-SA}{lem:partition-SA}%
\mc@xref{SI-cond-SA-1**1}{cond-SA-1**1}%
\mc@xref{SI-g-SA-1**1}{g-SA-1**1}%
\mc@xref{SI-tab:limg}{tab:limg}%
\mc@xref{SI-cond-SA-0111}{cond-SA-0111}%
\mc@xref{SI-g-SA-0111}{g-SA-0111}%
\mc@xref{SI-cond-SA-1111}{cond-SA-1111}%
\mc@xref{SI-g-SA-1111}{g-SA-1111}%
\mc@xref{SI-subsec:self-cooperative-righteous}{subsec:self-cooperative-righteous}%
\mc@xref{SI-lem:kappa-q}{lem:kappa-q}%
\mc@xref{SI-q}{q}%
\mc@xref{SI-kappa-q}{kappa-q}%
\mc@xref{SI-lim-kappa-q}{lim-kappa-q}%
\mc@xref{SI-tab:q}{tab:q}%
\mc@xref{SI-prop:self-cooperative}{prop:self-cooperative}%
\mc@xref{SI-gamma}{gamma}%
\mc@xref{SI-cor:righteous}{cor:righteous}%
\mc@xref{SI-righteous-set-simple}{righteous-set-simple}%
\mc@xref{SI-prop:cardinality-R}{prop:cardinality-R}%
\mc@xref{SI-subsec:consistent-systems}{subsec:consistent-systems}%
\mc@xref{SI-prop:cons->SA}{prop:cons->SA}%
\mc@xref{SI-prop:consistentdisc->righteous}{prop:consistentdisc->righteous}%
\mc@xref{SI-prop:consistent-rep-dominance}{prop:consistent-rep-dominance}%
\mc@xref{SI-tilde-delta-consistent}{tilde-delta-consistent}%
\mc@xref{SI-h-fixed-point}{h-fixed-point}%
\mc@xref{SI-alpha-beta}{alpha-beta}%
\mc@xref{SI-alpha-def}{alpha-def}%
\mc@xref{SI-beta-def}{beta-def}%
\mc@xref{SI-rho-identity}{rho-identity}%
\mc@xref{SI-grouping-identity}{grouping-identity}%
\mc@xref{SI-D-expansion}{D-expansion}%
\mc@xref{SI-prop:inconsistent-rep-nondominance}{prop:inconsistent-rep-nondominance}%
\mc@xref{SI-cor:rep-consistency}{cor:rep-consistency}%
\mc@xref{SI-subsec:leading-eight-righteous}{subsec:leading-eight-righteous}%
\mc@xref{SI-dBG=dGB=dGG=G}{dBG=dGB=dGG=G}%
\mc@xref{SI-prop:leading-eight-righteous}{prop:leading-eight-righteous}%
\mc@xref{SI-subsec:which-righteous-are-stable}{subsec:which-righteous-are-stable}%
\mc@xref{SI-subsubsec:preliminaries-stability}{subsubsec:preliminaries-stability}%
\mc@xref{SI-gchi-def}{gchi-def}%
\mc@xref{SI-hchi-def}{hchi-def}%
\mc@xref{SI-gchi-explicit}{gchi-explicit}%
\mc@xref{SI-fig:venn_analytic_result}{fig:venn_analytic_result}%
\mc@xref{SI-wchi-def}{wchi-def}%
\mc@xref{SI-tildewchi-def}{tildewchi-def}%
\mc@xref{SI-wchi-explicit}{wchi-explicit}%
\mc@xref{SI-tildewchi-explicit}{tildewchi-explicit}%
\mc@xref{SI-Deltaw-def}{Deltaw-def}%
\mc@xref{SI-pGG=C}{pGG=C}%
\mc@xref{SI-dGG=G}{dGG=G}%
\mc@xref{SI-subsubsec:invasion-ALLD}{subsubsec:invasion-ALLD}%
\mc@xref{SI-pm-ALLD}{pm-ALLD}%
\mc@xref{SI-prop:stability-against-ALLD-dGB=G}{prop:stability-against-ALLD-dGB=G}%
\mc@xref{SI-prop:stability-against-ALLD-dGB=B}{prop:stability-against-ALLD-dGB=B}%
\mc@xref{SI-limchi0w}{limchi0w}%
\mc@xref{SI-tab:limtildew-righteous}{tab:limtildew-righteous}%
\mc@xref{SI-tildewij-ALLD}{tildewij-ALLD}%
\mc@xref{SI-pGB=D}{pGB=D}%
\mc@xref{SI-limchi0tildew=wBGm}{limchi0tildew=wBGm}%
\mc@xref{SI-dGBm=dGB}{dGBm=dGB}%
\mc@xref{SI-dBGm=B}{dBGm=B}%
\mc@xref{SI-dGGm=B}{dGGm=B}%
\mc@xref{SI-cor:pBG=C}{cor:pBG=C}%
\mc@xref{SI-R'}{R'}%
\mc@xref{SI-prop:L-subset-R'}{prop:L-subset-R'}%
\mc@xref{SI-prop:QD-R'-membership}{prop:QD-R'-membership}%
\mc@xref{SI-prop:cardinality-R'}{prop:cardinality-R'}%
\mc@xref{SI-tab:cardinality-R'}{tab:cardinality-R'}%
\mc@xref{SI-subsubsec:invasion-ALLC}{subsubsec:invasion-ALLC}%
\mc@xref{SI-pm-ALLC}{pm-ALLC}%
\mc@xref{SI-prop:stability-against-ALLC-dGB=G}{prop:stability-against-ALLC-dGB=G}%
\mc@xref{SI-prop:stability-against-ALLC-dBG=B}{prop:stability-against-ALLC-dBG=B}%
\mc@xref{SI-prop:stability-against-ALLC-dBG=G}{prop:stability-against-ALLC-dBG=G}%
\mc@xref{SI-wBG-disc}{wBG-disc}%
\mc@xref{SI-dGGm=G-ALLC}{dGGm=G-ALLC}%
\mc@xref{SI-wchi-expansion-ALLC-AGXXG}{wchi-expansion-ALLC-AGXXG}%
\mc@xref{SI-tab:expansion-tildew-ALLC-AGXXG}{tab:expansion-tildew-ALLC-AGXXG}%
\mc@xref{SI-wBG+wGB-wBGm-wGBm}{wBG+wGB-wBGm-wGBm}%
\mc@xref{SI-wBG+wGB-R-wGBm}{wBG+wGB-R-wGBm}%
\mc@xref{SI-dBGm=B-ALLC}{dBGm=B-ALLC}%
\mc@xref{SI-tab:expansion-tildew-ALLC-A*GGG}{tab:expansion-tildew-ALLC-A*GGG}%
\mc@xref{SI-R''}{R''}%
\mc@xref{SI-prop:L-R''-membership}{prop:L-R''-membership}%
\mc@xref{SI-prop:QD-R''-membership}{prop:QD-R''-membership}%
\mc@xref{SI-prop:cardinality-R''}{prop:cardinality-R''}%
\mc@xref{SI-subsubsec:invasion-disc}{subsubsec:invasion-disc}%
\mc@xref{SI-pm-disc}{pm-disc}%
\mc@xref{SI-tab:cardinality-R''}{tab:cardinality-R''}%
\mc@xref{SI-prop:dBG=B-dBB=B-invaded}{prop:dBG=B-dBB=B-invaded}%
\mc@xref{SI-prop:stability-against-alternative-disc}{prop:stability-against-alternative-disc}%
\mc@xref{SI-wGGm-wGBm-discs}{wGGm-wGBm-discs}%
\mc@xref{SI-dGGm=G-discs}{dGGm=G-discs}%
\mc@xref{SI-dGBm=G-discs}{dGBm=G-discs}%
\mc@xref{SI-wchi-expansion-GBGG}{wchi-expansion-GBGG}%
\mc@xref{SI-tab:Deltaw-disc-ABGGG}{tab:Deltaw-disc-ABGGG}%
\mc@xref{SI-dBGm=B-discs-AGBGG}{dBGm=B-discs-AGBGG}%
\mc@xref{SI-wBG-wBGm-discs-AGBGG}{wBG-wBGm-discs-AGBGG}%
\mc@xref{SI-dBGm=G-discs-A*GGG}{dBGm=G-discs-A*GGG}%
\mc@xref{SI-wBGm-discs-A*GGG}{wBGm-discs-A*GGG}%
\mc@xref{SI-wchi-expansion-BGGG}{wchi-expansion-BGGG}%
\mc@xref{SI-tab:Deltaw-disc-AGGGG}{tab:Deltaw-disc-AGGGG}%
\mc@xref{SI-wchi-expansion-GGGG}{wchi-expansion-GGGG}%
\mc@xref{SI-Rset'''=QCDisc}{Rset'''=QCDisc}%
\mc@xref{SI-prop:L-R'''-membership}{prop:L-R'''-membership}%
\mc@xref{SI-prop:QD-R'''-membership}{prop:QD-R'''-membership}%
\mc@xref{SI-subsubsec:invasion-other-mutants}{subsubsec:invasion-other-mutants}%
\mc@xref{SI-thm:universal-stability}{thm:universal-stability}%
\mc@xref{SI-tab:stability-C-CDISC}{tab:stability-C-CDISC}%
\mc@xref{SI-subsec:stab-anal-disc}{subsec:stab-anal-disc}%
\mc@xref{SI-subsubsec:CDISC-asymptotics-stability}{subsubsec:CDISC-asymptotics-stability}%
\mc@xref{SI-prop:stability-QCDisc}{prop:stability-QCDisc}%
\mc@xref{SI-subsubsec:DISC-asymptotics-stability}{subsubsec:DISC-asymptotics-stability}%
\mc@xref{SI-prop:stability-QDisc}{prop:stability-QDisc}%
\mc@xref{SI-tab:stability-C-DISC}{tab:stability-C-DISC}%
\mc@xref{SI-subsubsec:ODISC-asymptotics-stability}{subsubsec:ODISC-asymptotics-stability}%
\mc@xref{SI-prop:stability-QODisc}{prop:stability-QODisc}%
\mc@xref{SI-2R>T+P}{2R>T+P}%
\mc@xref{SI-tab:stability-C-ODISC}{tab:stability-C-ODISC}%
\mc@xref{SI-subsubsec:PDISC-asymptotics-stability}{subsubsec:PDISC-asymptotics-stability}%
\mc@xref{SI-prop:stability-QPDisc}{prop:stability-QPDisc}%
\mc@xref{SI-tab:stability-C-PDISC}{tab:stability-C-PDISC}%
\mc@xref{SI-tab:consistent-disc-summary}{tab:consistent-disc-summary}%
\mc@xref{SI-subsec:stab-anal-leading-eight}{subsec:stab-anal-leading-eight}%
\mc@xref{SI-tab:stability-L1}{tab:stability-L1}%
\mc@xref{SI-prop:stability-L1}{prop:stability-L1}%
\mc@xref{SI-prop:stability-L2}{prop:stability-L2}%
\mc@xref{SI-prop:stability-L3}{prop:stability-L3}%
\mc@xref{SI-tab:stability-L3}{tab:stability-L3}%
\mc@xref{SI-prop:stability-L4}{prop:stability-L4}%
\mc@xref{SI-tab:stability-L4}{tab:stability-L4}%
\mc@xref{SI-tab:stability-L5}{tab:stability-L5}%
\mc@xref{SI-prop:stability-L5}{prop:stability-L5}%
\mc@xref{SI-prop:stability-L6}{prop:stability-L6}%
\mc@xref{SI-prop:stability-L7}{prop:stability-L7}%
\mc@xref{SI-tab:stability-L7}{tab:stability-L7}%
\mc@xref{SI-prop:stability-L8}{prop:stability-L8}%
\mc@xref{SI-tab:stability-L8}{tab:stability-L8}%
\mc@xref{SI-tab:leading-eight-summary}{tab:leading-eight-summary}%
\mc@xref{SI-sec:supp-numerical}{sec:supp-numerical}%
\mc@xref{SI-subsec:mirror-classes}{subsec:mirror-classes}%
\mc@xref{SI-prop:mirror-consistent}{prop:mirror-consistent}%
\mc@xref{SI-mirror-consistent}{mirror-consistent}%
\mc@xref{SI-mirror-consistent-B}{mirror-consistent-B}%
\mc@xref{SI-mirror-consistent-G}{mirror-consistent-G}%
\mc@xref{SI-prop:selfmirror-count}{prop:selfmirror-count}%
\mc@xref{SI-selfmirror-cond}{selfmirror-cond}%
\mc@xref{SI-selfmirror-cond-p}{selfmirror-cond-p}%
\mc@xref{SI-selfmirror-cond-d}{selfmirror-cond-d}%
\mc@xref{SI-prop:mirror-count}{prop:mirror-count}%
\mc@xref{SI-mirror-class-count}{mirror-class-count}%
\mc@xref{SI-def:mirror-representative}{def:mirror-representative}%
\mc@xref{SI-subsec:numerical-pipeline}{subsec:numerical-pipeline}%
\mc@xref{SI-subsec:alternative-errors}{subsec:alternative-errors}%
\mc@xref{SI-fig:alternative-errors}{fig:alternative-errors}%
\mc@xref{SI-fig:alternative_errors_zoomed_scatter}{fig:alternative_errors_zoomed_scatter}%
\mc@xref{SI-fig:nojitter-chi001-eps00}{fig:nojitter-chi001-eps00}%
\mc@xref{SI-fig:nojitter-chi001-eps01}{fig:nojitter-chi001-eps01}%
\mc@xref{SI-fig:nojitter-chi001-eps02}{fig:nojitter-chi001-eps02}%
\mc@xref{SI-fig:nojitter-chi001-eps1}{fig:nojitter-chi001-eps1}%
\mc@xref{SI-fig:nojitter-chi002-eps00}{fig:nojitter-chi002-eps00}%
\mc@xref{SI-fig:nojitter-chi002-eps01}{fig:nojitter-chi002-eps01}%
\mc@xref{SI-fig:nojitter-chi002-eps02}{fig:nojitter-chi002-eps02}%
\mc@xref{SI-fig:nojitter-chi002-eps1}{fig:nojitter-chi002-eps1}%
\mc@xref{SI-fig:nojitter-chi005-eps00}{fig:nojitter-chi005-eps00}%
\mc@xref{SI-fig:nojitter-chi005-eps01}{fig:nojitter-chi005-eps01}%
\mc@xref{SI-fig:nojitter-chi005-eps02}{fig:nojitter-chi005-eps02}%
\mc@xref{SI-fig:nojitter-chi005-eps1}{fig:nojitter-chi005-eps1}%
\mc@xref{SI-fig:nojitter-chi01-eps00}{fig:nojitter-chi01-eps00}%
\mc@xref{SI-fig:nojitter-chi01-eps01}{fig:nojitter-chi01-eps01}%
\mc@xref{SI-fig:nojitter-chi01-eps02}{fig:nojitter-chi01-eps02}%
\mc@xref{SI-fig:nojitter-chi01-eps1}{fig:nojitter-chi01-eps1}%
\mc@xref{SI-subsec:disaggregated-indices}{subsec:disaggregated-indices}%
\mc@xref{SI-fig:disaggregated_pd}{fig:disaggregated_pd}%
\mc@xref{SI-fig:disaggregated_sg}{fig:disaggregated_sg}%
\mc@xref{SI-fig:disaggregated_sh}{fig:disaggregated_sh}%
\mc@xref{SI-subsec:area-weighted}{subsec:area-weighted}%
\mc@xref{SI-fig:scatter_zoom_area}{fig:scatter_zoom_area}%
\mc@xref{SI-fig:scatter_panel_area}{fig:scatter_panel_area}%
\mc@xref{MT-fig:overview}{fig:overview}%
\mc@xref{MT-fig:scatter}{fig:scatter}%
\mc@xref{MT-fig:norms}{fig:norms}%
\mc@xref{MT-fig:stability_indices}{fig:stability_indices}%
\mc@xref{MT-methods}{methods}%
\mc@xref{MT-consistent-cond}{consistent-cond}%
\mc@xref{MT-consistent-cond-G}{consistent-cond-G}%
\mc@xref{MT-consistent-cond-B}{consistent-cond-B}%
\mc@xref{MT-stability-condition}{stability-condition}%
\mc@xref{MT-code-availability}{code-availability}%
\mc@xref{MT-fig:space}{fig:space}%
\mc@xref{MT-fig:payoff_space}{fig:payoff_space}%
\mc@xref{MT-fig:cdisc_invades}{fig:cdisc_invades}%
\mc@xref{MT-fig:disc_cannot_invade}{fig:disc_cannot_invade}%
}
\makeatother
\makeatletter
\let\mc@maketitle\maketitle \let\mc@@maketitle\@maketitle \let\mc@author\@author
\let\mc@title\title \let\mc@date\date \let\mc@and\and \let\mc@thanks\thanks
\let\mc@fnumfigure\fnum@figure
\makeatother
\begin{document}
\makeatletter\let\mc@c@caption\caption \let\mc@c@@caption\@caption \let\mc@c@makecaption\@makecaption
\let\caption\mc@k@caption \let\@caption\mc@k@@caption \let\@makecaption\mc@k@makecaption\makeatother 
\addtocontents{toc}{\protect\setcounter{tocdepth}{-10}}

\maketitle

\setstretch{1.4}
\frenchspacing

{\bfseries A leading explanation for why humans cooperate with strangers is indirect reciprocity: when conduct is observed, individuals earn good reputations by cooperating with the deserving\cite{alexander:book:1987,nowak:Nature:1998}.
Previous research has identified many `moral systems' that sustain cooperation through indirect reciprocity -- each pairing a strategy for how to act with a norm for how to judge others -- yet what makes them succeed, and whether that success generalizes beyond the donation game, has remained unclear\cite{ohtsuki:JTB:2004a,nowak:Nature:2005,Nakamura2014}.
Here we evaluate all $\numprint{1048576}$ moral systems whose norms condition on both players' reputations and actions, across every two-player social dilemma.
We find that every highly stable, cooperative system pairs a `discriminating' strategy -- one that cooperates with the good and defects against the bad when in good standing -- with a `consistent' norm that deems good exactly those who act as the strategy prescribes.
Such systems `practice what they preach' and preach what they practice.
In the limit of small errors, only one strategy paired with a consistent norm remains stable across all social dilemmas: the `contrite discriminator', which cooperates to make amends when its standing is bad.
Our results unify decades of research on indirect reciprocity, revealing moral consistency as the common ingredient of every stable, cooperative moral system.
}

Humans cooperate on a scale unmatched among other animals, routinely helping people they may never encounter again\cite{fehr_nature_2003,boyd_large-scale_2022}.
A central explanation is reciprocity.
Under direct reciprocity, individuals help those who have helped them\cite{trivers_evolution_1971}; under indirect reciprocity, they help those who have helped others\cite{boyd:SN:1989,kandori_social_1992,sugden1986economics}, so that cooperation is rewarded not by second parties but by third parties who have observed or heard about one's conduct\cite{alexander:book:1987}.
A reputation aggregates an individual's past conduct into a standing that others can act on without any direct experience of them\cite{nowak:Nature:2005}.
What turns such standings into cooperation is the `moral system' adopted by society at large: the rules of behavior and approval about ``what is right and wrong''\cite{alexander:book:1987}, typically formalized as a strategy paired with a social norm.

Prior work has uncovered an array of moral systems capable of sustaining cooperation through indirect reciprocity, the best known being the `leading eight' first identified by Ohtsuki and Iwasa\cite{ohtsuki:JTB:2004a} and studied extensively since\cite{ohtsuki:JTB:2006a,ohtsuki:JTB:2007,pacheco:PlosCB:2006,schmid_quantitative_2023,okada_review_2020}.
Yet this body of work has three important limitations.
First, most of what we know about successful moral systems comes from a single, highly simplified interaction: the donation game, in which one player decides whether to pay a personal cost to confer a larger benefit on the other\cite{nowak:Nature:2005}.
This is also the setting for almost all recent theoretical work, including studies of higher-order norms\cite{santos_social_2018}, the evolution of norms\cite{xu_cooperation_2019,kessinger_evolution_2023}, and indirect reciprocity under private reputations\cite{hilbe_indirect_2018,perret:SciRep:2021,fujimoto_who_2024,michel-mata:Nature:2024}.
Real-world human cooperation, by contrast, unfolds across a far wider range of situations.
These certainly include prisoner's dilemmas -- of which the donation game is a special case -- where cooperation is dominated. But they also include snowdrift (or chicken) games, where the best response to defection is to cooperate, and stag hunts (or assurance games), where each player is best off matching the other's action\cite{Pena2023,skyrms2004stag,doebeli_models_2005}.
Second, previous analyses have each searched only a small subset of the possible social norms, ignoring some of the information available to observers in a single interaction\cite{ohtsuki:JTB:2004a,Nakamura2014}.
Third, and most fundamentally, we lack a unifying account of why some moral systems succeed.
Existing characterizations are disparate and tied to particular settings, leaving no parsimonious explanation for what separates the moral systems that sustain cooperation from those that fail.
Even the study that covered the largest range of games concluded that its successful norms ``have different principles in different dilemmas for suppressing cheaters''\cite{Nakamura2014}.

Here we address all three limitations at once.
As in most models of indirect reciprocity\cite{ohtsuki:JTB:2004a,nowak:Nature:2005}, we take actions and reputations to be binary: an individual either cooperates or defects, and is judged either good or bad.
We study the complete space of social norms that condition on the two players' reputations and actions.
Combined with the $2^{4}=16$ possible strategies, this yields $2^{20}=\numprint{1048576}$ moral systems, which we evaluate across the entire set of symmetric two-player cooperative dilemmas where mutual cooperation is socially optimal\cite{Pena2023}.
Our analysis reveals a single organizing principle behind every successful moral system: whatever the dilemma, those systems that best sustain cooperation pair a `discriminating' strategy, which conditions cooperation on the partner's reputation, with a `consistent' social norm, which awards a good reputation to exactly those who act as the strategy dictates.
Among these consistent systems, those built on a `contrite' strategy -- one that cooperates to make amends when its own standing is bad -- stand out as the only moral systems that remain both stable and cooperative across every class of dilemma at once.

\section{Results}

\subsection{Mapping the full space of moral systems}

Modeling indirect reciprocity requires three ingredients (\hyperref[methods]{Methods}): a payoff matrix, which defines the rewards and costs of cooperating and defecting (Fig.~\ref{fig:overview}a); a strategy (or action rule), which specifies when to cooperate and when to defect (Fig.~\ref{fig:overview}b); and a social norm (or assessment rule), which tells observers when to assign a good or a bad reputation (Fig.~\ref{fig:overview}c).
Strategies and social norms can be classified by the amount and nature of the information they take as input.
For instance, early models considered a moral system where the strategy (`discriminator' or DISC; Fig.~\ref{fig:overview}b) prescribes cooperation if the co-player's reputation is good and defection if the co-player's reputation is bad, and the social norm (`Scoring'; Fig.~\ref{fig:overview}c) assigns a good reputation if the player cooperated and a bad reputation if the player defected\cite{nowak:Nature:1998}.
Both the strategy and the social norm of this moral system require one bit of information, and are thus said to be of first order.
Although this moral system can sustain cooperation against unconditional defectors, it is not robust to invasion by other strategies\cite{leimar:PRSB:2001}.

\begin{figure}[t]
\centering
\includegraphics[width=\textwidth]{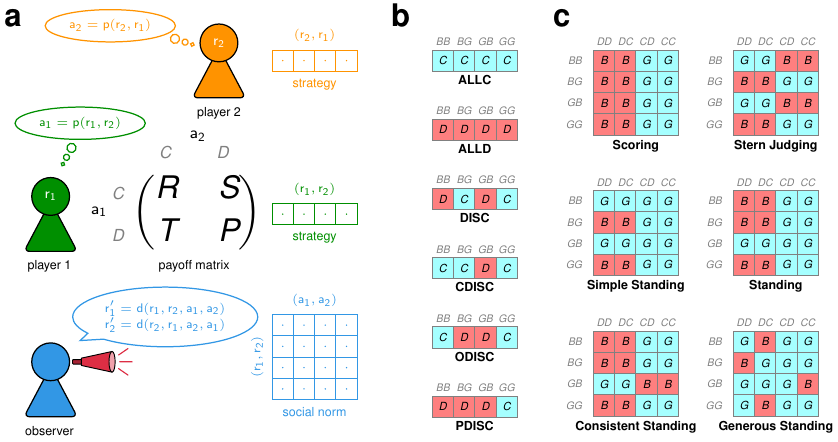}
\caption{{\bf Modeling indirect reciprocity.}
{\bf a}, Player 1 (\emph{green pawn}, reputation $r_1 \in \left\{G,B\right\}$, good or bad) and player 2 (\emph{orange pawn}, reputation $r_2$) play a game in which each can cooperate ($C$) or defect ($D$).
Player 1 chooses her action $a_1 = p(r_1,r_2)$ according to her strategy $p$ (\emph{green table}); player 2 likewise chooses $a_2 = p(r_2,r_1)$ (\emph{orange table}).
The players act simultaneously and obtain the payoffs shown in the matrix (each entry gives the row player's payoff), whose ordering defines the class of cooperative dilemma: prisoner's dilemma, snowdrift game, or stag hunt (\hyperref[methods]{Methods}; \ref{fig:payoff_space}).
An observer (\emph{blue pawn}) witnesses the interaction and assigns new reputations -- $r_1' = d(r_1,r_2,a_1,a_2)$ and $r_2' = d(r_2,r_1,a_2,a_1)$ -- according to the social norm $d$ (\emph{blue table}).
Assessment is public: reputations are shared by the whole population (\hyperref[methods]{Methods}).
\mbox{{\bf b}, A} strategy is a table whose entries give player 1's action for each combination of the players' reputations ($r_1$, $r_2$); shown are the unconditional cooperator (ALLC) and defector (ALLD), and the four discriminating strategies DISC, CDISC, ODISC, and PDISC.
\mbox{{\bf c}, A} social norm is a table whose entries give player 1's new reputation for each combination of reputations ($r_1$, $r_2$) and actions ($a_1$, $a_2$); shown are Scoring, Stern Judging, Simple Standing, Standing, Consistent Standing, and Generous Standing.
}
\label{fig:overview}
\end{figure}

Ohtsuki and Iwasa\cite{ohtsuki:JTB:2004a} were the first to systematically study all moral systems obtained by pairing each of the $2^{4}=16$ possible strategies of up to second order (based on two bits of information, namely the reputations of the focal player and the co-player) with each of the $2^{8}$ possible social norms of up to third order (based on three bits of information, namely the action of the player being assessed and the reputations of player and co-player).
Out of these $2^{12}=\numprint{4096}$ moral systems, only eight -- known as the `leading eight' -- were found both to sustain high levels of cooperation and to be evolutionarily stable when social interactions were modeled as a donation game.
Two examples of the leading eight are L6 and L1 (labels adopted from\cite{hilbe_indirect_2018}).
L6 employs DISC together with a second-order social norm (`Stern Judging'; Fig.~\ref{fig:overview}c) that calls an individual good only if they cooperated with a good partner or defected against a bad one.
In contrast, L1 uses a second-order strategy that takes the individual's own reputation into account, prescribing cooperation when the co-player is good or when one is oneself bad.
This prescription to cooperate when in bad standing can be read as a form of contrition that attempts to repair one's damaged reputation\cite{ohtsuki:JTB:2006a}; we therefore call
this strategy `contrite discriminator' (CDISC; Fig.~\ref{fig:overview}b).
L1 pairs CDISC with the third-order social norm `Standing' (Fig.~\ref{fig:overview}c), which also conditions on the reputation of the individual being judged: good individuals keep their good reputation unless they defect against a good partner, while bad individuals turn good if and only if they cooperate, whatever the partner's reputation.

Within indirect reciprocity, social dilemmas beyond the donation game have rarely been studied.
A notable exception is the study by Nakamura and Ohtsuki\cite{Nakamura2014}, who analyzed indirect reciprocity across prisoner's dilemmas, snowdrift games, and stag hunts, and identified which social norms stabilize cooperation in each game.
Building on their work, we ask whether a single moral system can sustain cooperation across the entire range of cooperative dilemmas.

Prior analyses differ in how much of the information available in a single interaction their social norms use (\ref{fig:space}).
The norms studied by Ohtsuki and Iwasa\cite{ohtsuki:JTB:2004a,ohtsuki:JTB:2006a} condition on the focal player's action and on both players' reputations, but never on the co-player's action; the norms of Nakamura and Ohtsuki\cite{Nakamura2014} condition on both players' actions and the co-player's reputation, but never on the focal player's own reputation.
Norms that use all four pieces of information available in a single interaction -- the two players' reputations and actions -- have so far been left unexplored\cite{Nakamura2014}. %
We study the complete fourth-order space of $2^{16}=\numprint{65536}$ social norms, which, paired with the $16$ strategies of up to second order, yields the $2^{20}=\numprint{1048576}$ moral systems we evaluate across the entire set of symmetric $2 \times 2$ cooperative dilemmas where mutual cooperation is socially optimal (\ref{fig:payoff_space}).
Throughout, we adopt a public-assessment model of reputation: a single observer witnesses each interaction and judges both players, with the reputations they assign shared by the whole community, so that everyone comes to hold the same view of any given individual\cite{ohtsuki:JTB:2004a,murase_indirect_2024}.

\subsection{Highly stable systems are consistent}

We characterize each moral system by two quantities: its cooperation index and its stability index (see \hyperref[methods]{Methods}).
The cooperation index ($\kappa$) is the probability that a randomly chosen resident cooperates at equilibrium.
Because this probability does not depend on the payoffs, it is the same for every game we study.
The stability index for a given dilemma class is the fraction of payoff space over which the resident strategy resists invasion by any alternative strategy, holding fixed the social norm.
We then take the simple average across the three dilemma classes to obtain the average stability index ($\sigma$).
We report results for low error rates in the execution of cooperation ($\epsilon=0.02$) and in the assignment of reputations ($\chi=0.02$), but our conclusions are robust to other small error rates and to alternative constructions of the stability index (SI Figs.~\ref{SI-fig:alternative-errors}, \ref{SI-fig:scatter_zoom_area}, \ref{SI-fig:scatter_panel_area}).

\begin{figure}[tp]
\centering
\includegraphics{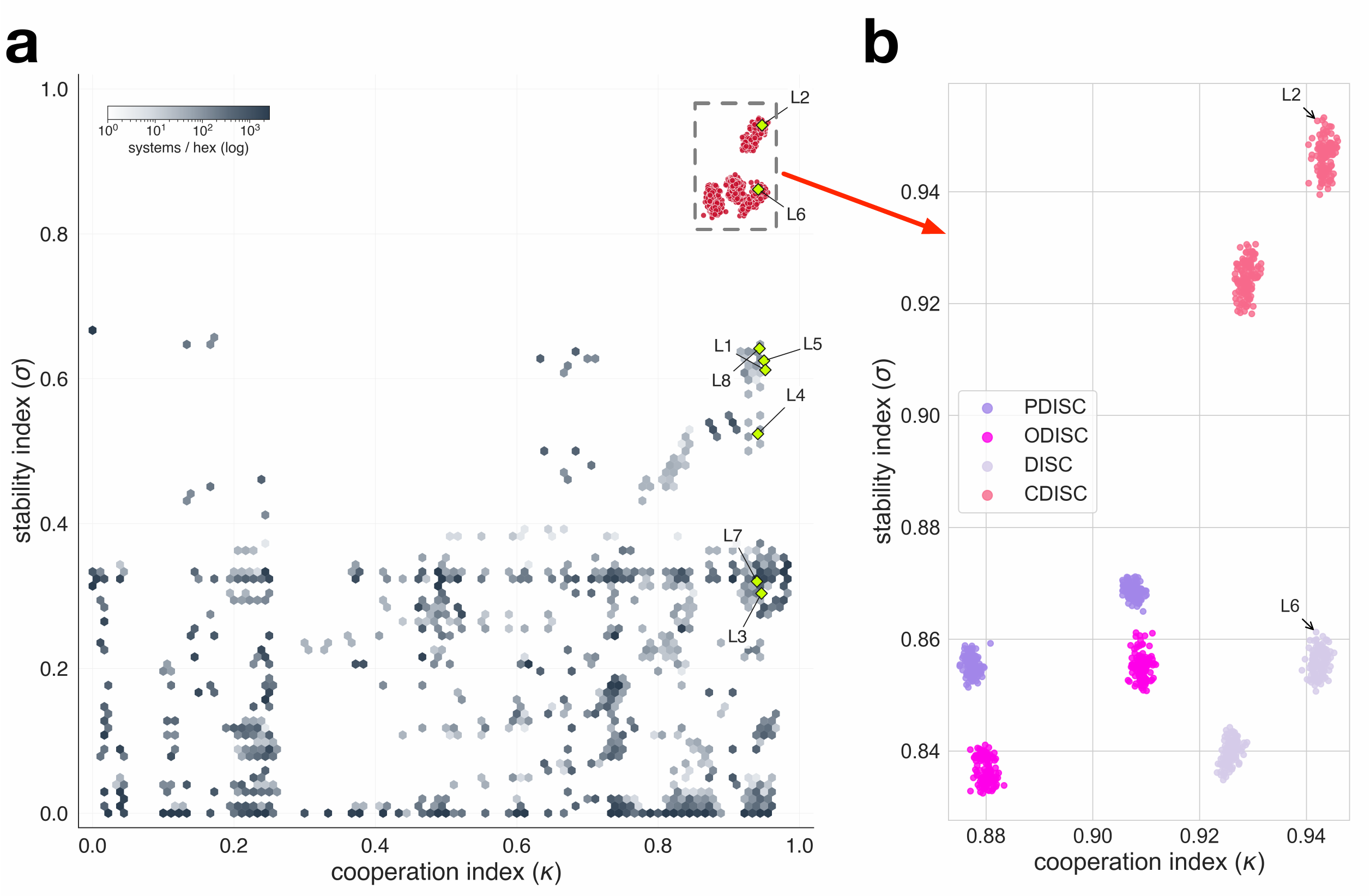}
\caption{\textbf{Stability and cooperation indices for all moral systems.} \textbf{a}, Each of the $\numprint{524800}$ moral systems that are distinct up to the relabeling of reputations (one representative per mirror class; SI Section~\ref{SI-subsec:mirror-classes}) is positioned by its resident cooperation index ($\kappa$, horizontal axis) and average stability index ($\sigma$, vertical axis); here $\chi = \epsilon = 0.02$.
The background hexbin layer shows the density of the $\numprint{523776}$ representatives that are not consistent discriminating systems, on a logarithmic color scale (inset).
The \numprint{1024} consistent discriminating systems (\emph{red}), shown as a jittered scatter plot overlaid on the density map, cluster tightly in the upper-right corner, separated from the bulk of moral systems by a clear gap in stability.
Of the leading-eight moral systems (L1–L8, \emph{green diamonds}), only L2 and L6 fall within this cluster.
\textbf{b}, Zoomed-in scatter plot of the consistent discriminating systems (jittered for clarity), colored by strategy; the $\numprint{256}$ systems using CDISC (the `consistent contrite discriminating' systems) achieve the highest stability.
See SI Figs.~\ref{SI-fig:alternative-errors}--\ref{SI-fig:nojitter-chi01-eps1} for results at alternative error values.}
\label{fig:scatter}
\end{figure}

Computing the cooperation index and the average stability index for all $2^{20}$ moral systems -- where systems that differ only by swapping the two reputation labels, and thus prescribe identical behavior, are evaluated only once (SI Section~\ref{SI-subsec:mirror-classes}) -- reveals our main result (Fig.~\ref{fig:scatter}): a single cluster of $2^{10}=\numprint{1024}$ moral systems stands out as being both highly stable and highly cooperative, separated from the rest by a clear gap in stability.
The moral systems in this cluster share two properties.
First, they each use a `discriminating' strategy: players in good standing cooperate with good co-players and defect against bad ones (the four discriminating strategies -- DISC, CDISC, ODISC, and PDISC -- are shown in Fig.~\ref{fig:overview}b; \hyperref[methods]{Methods}).
Second, they pair such a discriminating strategy with a `consistent' social norm: one that, conditional on the co-player playing the same strategy, assigns a good reputation to all and only those who behave as the system's own strategy prescribes (Fig.~\ref{fig:norms}a; \hyperref[methods]{Methods}).
Because consistency says nothing about cases where the co-player deviates from the system's strategy, it leaves eight bits unconstrained, so each strategy admits $2^8 = \numprint{256}$ consistent norms (for CDISC, these include, for instance, Consistent Standing and Generous Standing; Fig.~\ref{fig:overview}c).
The four discriminating strategies thus give rise to the $\numprint{1024}$ systems of the cluster, which we call `consistent discriminating'.
Intuitively, consistency can be likened to moral integrity and an absence of hypocrisy: the standard to which one holds others is exactly the standard by which one acts.
Interestingly, all leading-eight moral systems use discriminating strategies (DISC and CDISC), but only two are consistent: L2 and L6 (Fig.~\ref{fig:norms}b).

\begin{figure}[tp]
\centering
\includegraphics[width=\textwidth]{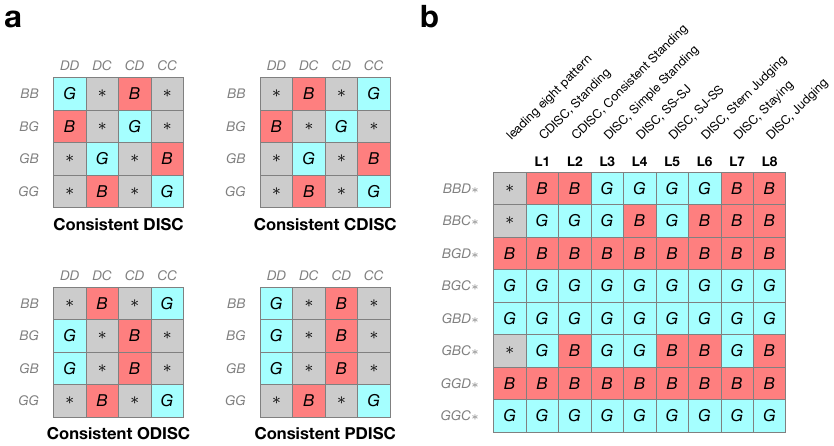}
\caption{\textbf{Social norms of the consistent discriminating systems and of the leading eight.} \textbf{a}, The social norm classes consistent with each of the four discriminating strategies (Fig.~\ref{fig:overview}b); because consistency constrains a norm only when the co-player behaves in accordance with the resident strategy, eight of the sixteen bits are left unconstrained ($\ast$), so each strategy admits $2^8 = \numprint{256}$ consistent norms. \textbf{b}, The social norms of the leading-eight moral systems. Rows are labeled by the reputations and action $(r_1 r_2 a_1)$ of the judged player, with $\ast$ standing for the co-player's action $a_2$, on which these third-order norms do not depend; the shared pattern column shows the entries common to all eight, with $\ast$ marking the three rows on which they differ. Only two of the leading eight feature a consistent social norm: L2 (CDISC, Consistent Standing) and L6 (DISC, Stern Judging).}
\label{fig:norms}
\end{figure}

The success of consistent discriminating systems follows from the basic function of a moral system in indirect reciprocity: namely, to sort the population so that reciprocators can direct cooperation toward individuals with the same strategy, and defect against all others.
Indeed, a consistent social norm is the one best suited to this task, as it assigns reputations so that those who follow the resident strategy are labeled `good', and those who deviate are labeled `bad'.
An inconsistent norm, by contrast, will assign good reputations to some mutant strategies at least as often as to residents, blurring the very distinction on which discrimination depends, thus opening the door to invasion.
We can make this intuition precise in the absence of execution errors: a moral system is consistent if and only if, at equilibrium, residents are judged good strictly more often than any mutant that deviates from the resident strategy (SI Propositions~\ref{SI-prop:consistent-rep-dominance}, \ref{SI-prop:inconsistent-rep-nondominance}).

\subsection{Contrition makes stability universal}

Within the cluster of consistent discriminating systems, a further subdivision emerges according to whether the strategy is contrite (Fig.~\ref{fig:scatter}b).
A `contrite' strategy prescribes cooperation whenever one's own standing is bad (\hyperref[methods]{Methods});
among the discriminating strategies, CDISC is the only one with this property. The \numprint{256} consistent systems associated with this contrite discriminator achieve the highest stability of any moral system for all but the largest assignment error rate we examine (Fig.~\ref{fig:scatter}b; SI Fig.~\ref{SI-fig:alternative_errors_zoomed_scatter}), remaining highly stable in every class of dilemma (SI Section~\ref{SI-subsec:disaggregated-indices}).
We further examine this finding analytically in a tractable regime: no execution errors ($\epsilon=0$) and the limit of vanishing assignment errors ($\chi \to 0$).
Of the $2^{20}$ moral systems we investigate, $\numprint{114688}$ are `righteous': those in which residents both cooperate at near-maximal rates and almost all attain good standing (\hyperref[methods]{Methods}; SI Definition~\ref{SI-def:righteous}, Proposition~\ref{SI-prop:cardinality-R}).
In structural terms, a righteous strategy is reputation-sensitive and prescribes cooperation with a good partner when in good standing.
A righteous norm, assessing a resident whose partner also follows the strategy, judges that resident good in at least three of the four combinations of the two players' standings -- always when both are good, and in at least two of the remaining three (SI Eq.~(\ref{SI-righteous-set-simple})).
The $\numprint{1024}$ consistent discriminating systems and all leading-eight systems are righteous (SI Propositions~\ref{SI-prop:consistentdisc->righteous},~\ref{SI-prop:leading-eight-righteous}).
We ask which righteous systems are `universally stable', that is, resistant to invasion across every cooperative dilemma we consider.
We show that a righteous moral system is universally stable if and only if its strategy is CDISC and its social norm is consistent (SI Theorem~\ref{SI-thm:universal-stability}).

We identify the universally stable systems by discarding, in turn, the righteous systems invadable by three important classes of strategies (SI Section~\ref{SI-subsec:which-righteous-are-stable}; SI Fig.~\ref{SI-fig:venn_analytic_result}).
First, resistance to unconditional defectors (ALLD) requires two conditions. The strategy must prescribe defection against a bad partner when in good standing; otherwise residents behave indistinguishably from unconditional cooperators and ALLD invades directly in prisoner's dilemmas and snowdrift games. All discriminating strategies meet this condition. The social norm must also judge (mutual) defection against a good partner, when in bad standing, as bad (SI Propositions~\ref{SI-prop:stability-against-ALLD-dGB=B},~\ref{SI-prop:stability-against-ALLD-dGB=G}). Among the consistent discriminating systems, those with DISC and CDISC as strategies satisfy this condition, whereas those with ODISC and PDISC do not (SI Proposition~\ref{SI-prop:QD-R'-membership}). All leading-eight systems satisfy this condition as well (SI Proposition~\ref{SI-prop:L-subset-R'}). Overall, $\numprint{12800}$ righteous systems resist unconditional defectors across every dilemma.

Second, resistance to unconditional cooperators (ALLC) is the converse problem. Since a discriminator in good standing defects only against bad partners, and almost all residents hold good reputations, an ALLC mutant differs from a resident mainly by cooperating with the rare bad ones. The social norm must therefore treat this unwarranted cooperation as a moral failure (SI Propositions~\ref{SI-prop:stability-against-ALLC-dGB=G},~\ref{SI-prop:stability-against-ALLC-dBG=B},~\ref{SI-prop:stability-against-ALLC-dBG=G}). The consistent systems built on DISC and CDISC satisfy this condition, but only four of the leading eight -- L2, L5, L6, and L8 -- do; L2 and L6 are among those consistent systems, while L5 and L8 are not. The other four leading-eight systems (L1, L3, L4, and L7) praise cooperation with a bad partner and are invaded by ALLC in stag hunts (SI Propositions~\ref{SI-prop:QD-R''-membership},~\ref{SI-prop:L-R''-membership}; Fig.~\ref{fig:stability_indices}b). Overall, only $\numprint{4608}$ righteous systems resist both unconditional defectors and unconditional cooperators across every dilemma.

Finally, resistance to mutant discriminating strategies singles out CDISC under a consistent norm -- that is, the consistent contrite discriminating systems.
The crucial asymmetry is that CDISC can invade DISC under a consistent norm, but not vice versa.
The two strategies act differently only when both players are bad: DISC defects, whereas CDISC shows contrition and cooperates.
A consistent DISC norm judges this contrition bad, so a CDISC mutant needs two rounds to regain good standing, collecting $2S$ where a DISC resident collects $P+R$ (\ref{fig:cdisc_invades}); whenever $2S>R+P$, as is possible in snowdrift games, CDISC invades.
Among CDISC residents under a consistent norm, a DISC mutant defecting in the rare state where both players are bad earns $T$ but is judged bad, and collects $T+S$ over the two rounds it takes to regain good standing; a resident in the same state cooperates, is immediately judged good, and collects $2R$ (\ref{fig:disc_cannot_invade}). Because mutual cooperation is efficient ($2R>T+S$), the mutant has no margin to invade.

CDISC under a consistent norm resists not only rival discriminators but also every alternative mutant strategy (SI Theorem~\ref{SI-thm:universal-stability}).
This is our central analytical result: with no execution errors and in the limit of vanishing assignment errors, a righteous moral system is evolutionarily stable across all cooperative dilemmas where mutual cooperation is efficient if and only if its strategy is discriminating and contrite, and its social norm is consistent with it -- the $\numprint{256}$ consistent contrite discriminating systems.
Only one member of the leading eight belongs to this set: L2 (CDISC, Consistent Standing).
In this sense, we can call L2 the `leading one'.
L2 is also the simplest consistent contrite discriminating system, as its third-order social norm ignores the co-player's action.

\subsection{Uncontrite and inconsistent systems fail}

Knowing which consistent discriminating systems are universally stable, we can also ask precisely what goes wrong for the rest (Fig.~\ref{fig:stability_indices}).
DISC under a consistent norm is stable across all prisoner's dilemmas and stag hunts, yet, as already mentioned, is invaded by CDISC in snowdrift games with $2S>R+P$ (yielding $\sigma=8/9$).
The remaining two discriminators, ODISC and PDISC, exhibit the complementary vulnerability: under consistent norms, they are stable across all snowdrift games and stag hunts, but invadable in prisoner's dilemmas with $2R<T+P$ (again $\sigma=8/9$).

\begin{figure}[tp]
\centering
\includegraphics[width=11cm]{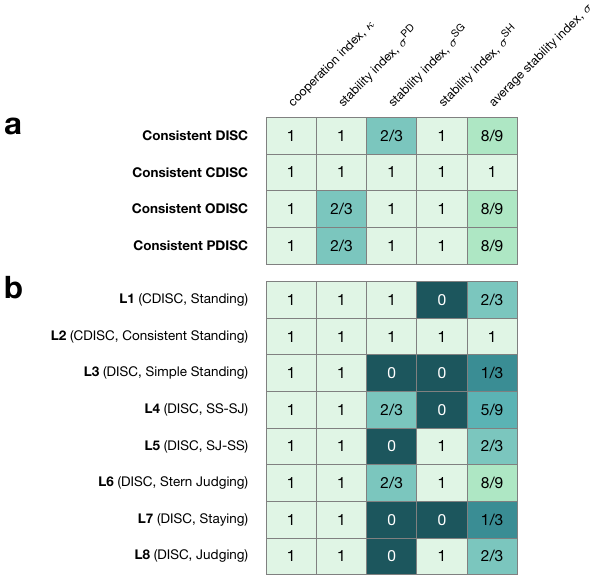}
\caption{\textbf{Cooperation and stability of consistent discriminating systems and of the leading eight.} The cooperation index and the stability indices across the three classes of cooperative dilemma (prisoner's dilemma, PD; snowdrift game, SG; and stag hunt, SH), with no execution errors ($\epsilon=0$) and in the limit of vanishing assignment errors ($\chi \to 0$); cells are shaded by the value of the index, darker for lower values. \textbf{a}, The four classes of consistent discriminating systems (their norms are shown in Fig.~\ref{fig:norms}a); in this limit, all $\numprint{256}$ members of a class share the same indices. \textbf{b}, The leading-eight moral systems, each labeled by its strategy and social norm (their norms are shown in Fig.~\ref{fig:norms}b).}
\label{fig:stability_indices}
\end{figure}

For the leading eight, the analytical results reveal which specific mutant strategies can invade and in which game classes (SI Section~\ref{SI-subsec:stab-anal-leading-eight}).
We find that all leading-eight systems are stable across the entire class of prisoner's dilemmas, and not merely the donation game assumed by the original analysis\cite{ohtsuki:JTB:2004a}.
However, when we expand the analysis to the full space of cooperative dilemmas, L2 and L6 are the only ones that fare particularly well, because they are the only leading-eight moral systems that feature a consistent social norm.
The remaining members of the leading eight all fail to achieve high levels of stability in snowdrift games, stag hunts, or both (Fig.~\ref{fig:stability_indices}b).

\section{Discussion}

The moral systems that succeed under indirect reciprocity have been described in many ways -- as those that distinguish `justified' from `unjustified' defection\cite{ohtsuki:JTB:2004a,Nakamura2014}, or that rely on social norms that are `simple and unambiguous' and `in which forgiveness and repent are welcome, while punishment is implacable'\cite{Chalub:2006tc}.
We have shown that what these successful systems share is moral consistency: the social norm deems good exactly those who act as the strategy prescribes.
For all two-player cooperative dilemmas with efficient mutual cooperation, and across a wide range of error rates,
consistent discriminating systems are both highly stable and highly cooperative.
The underlying mechanism is easy to state: a discriminating system whose social norm is inconsistent will award good reputations to some mutant at least as often as to its own residents, misdirecting cooperation and becoming vulnerable to invasion.
In the limit of small errors, consistent contrite discriminating systems are the only moral systems that are universally stable.
Contrition -- the readiness to cooperate when in bad standing -- is what secures stability in the snowdrift games where a consistent but uncontrite discriminator can be invaded.

Why did earlier literature not identify consistency and contrition as key ingredients of a successful moral system? Part of the answer lies in the limited scope of earlier analyses, which considered a restricted set of invaders, a restricted range of social dilemmas, or a restricted space of social norms. Because each dilemma poses its own profile of invasion risks, an analysis restricted to a subset of social dilemmas subjects moral systems to a weaker stability test: an inconsistent system that resists invasion in one dilemma can still be invaded in another.
Nakamura and Ohtsuki's analysis\cite{Nakamura2014} comes closest to ours: they surveyed the full range of dilemmas and allowed strategies to condition on the focal player's own reputation, the contrite discriminator among them. Yet none of their moral systems was both stable and cooperative across dilemmas.
The limitation lay in the set of norms they explored. Because a contrite discriminator cooperates with a bad partner only when it is itself in bad standing, a consistent norm must reward that act in a bad player while condemning it in a good one, a distinction it can draw only by reading the focal player's own reputation. Nakamura and Ohtsuki's norms conditioned on both players' actions and the co-player's reputation but not the focal player's own reputation, and so could not draw this distinction.

Which moral systems appear successful depends on the range of dilemmas used to test them. Consequently, the catalogues of `good' norms inherited from earlier work are, to some degree, artifacts of the kinds of social dilemma used to model social interactions. Our results thus call for a reappraisal of the best known of these, the leading eight\cite{ohtsuki:JTB:2004a,ohtsuki:JTB:2006a}.
Once we move to the full range of cooperative dilemmas, this set proves at once too inclusive and too exclusive. It is too inclusive because only two of its members -- L2 and L6 -- are consistent; the remaining six can be invaded once we consider social interactions beyond the prisoner's dilemma. It is too exclusive because it admits only the DISC and CDISC strategies, leaving out ODISC and PDISC.
However, consistent systems built on these excluded strategies are stable across all snowdrift games and stag hunts (Fig.~\ref{fig:stability_indices}a), and as a result achieve a higher average stability index than any inconsistent member of the leading eight -- even though leading-eight systems remain more stable within the class of prisoner's dilemmas.

The consistent ODISC and PDISC systems deserve closer attention, because the behavior they prescribe resonates with patterns seen in human societies.
In particular, PDISC encodes a social logic that evokes the phenomenon of honor cultures -- societies in which retaliation for insult is strongly socially endorsed, and in which failure to answer a slight can itself attract contempt\cite{boehm1984blood, nisbett1996culture,thrasher_honor_2018}.
Honor cultures are typically organized around a mixture of direct and indirect reciprocity: people retaliate against the particular individual who wronged them, often under public conditions in which status is at stake.
By contrast, our model is a pure case of indirect reciprocity and does not permit targeting a specific offender.
Even so, a consistent PDISC system can be read as a formal echo of honor culture. When its standing is compromised, a PDISC individual defects regardless of the partner's reputation -- and under its consistent norm this very defection is judged good, so that standing is restored by adversarial display rather than by apology. Viewed through this lens, honor cultures may not be pathologies of reciprocity but one among several consistent equilibria in which a moral system can settle and sustain cooperation.

More broadly, our results speak to a substantial body of empirical work on human indirect reciprocity. Controlled experiments show that people cooperate more with recipients of good reputation, and give more generously when their own conduct is observed and may be reported to future partners\cite{milinski:PRSB:2001,milinski_m._donors_2002,wedekind:CB:2002};
likewise, reputation shapes cooperation across a wide range of natural social settings\cite{romano_reputation_2021}. A recurring question in this literature is which social norm people actually apply when judging others\cite{milinski:PRSB:2001,seinen_social_2006,swakman_reputation-based_2016,okada_experimental_2018,gong_cooperation_2019}. Our theory sharpens that question. It predicts that a population should sustain a norm consistent with the strategy it plays, so that the standard people hold others to coincides with the behavior they display themselves. This alignment of judgment and conduct -- and how it might break down when information becomes private or noisy -- offers a target for empirical tests of when to expect reputation-based cooperation.

Our analysis assumes that individuals apply the same moral system whatever game they happen to be in. An alternative is to make behavior contingent on the game itself. Colnaghi et al.\cite{colnaghi_adaptations_2023} show that a costly ability to infer the structure of a social interaction and adapt to it can evolve and promote cooperation when the mix of games is sufficiently varied. Transposed to indirect reciprocity, this would require an entire society to detect which game it is playing and switch its moral system in concert. For instance, a population could adopt L3 in prisoner's dilemmas, L1 in snowdrift games, and L5 in stag hunts (each stable within its own class, for an average stability index of one). Yet every moral system in such a repertoire is inconsistent. Our results point to a complementary possibility. Any consistent contrite discriminating system is stable across the entire range of cooperative dilemmas, and thus needs no ability to identify the game and incurs no cost of switching. The whole work of sustaining cooperation can thus fall to a single moral system.

We have worked throughout with a public model of reputation, in which a single shared assessment of each individual circulates through the population.
This is a substantive assumption, relaxed by a growing body of work asking what becomes of cooperation when observers form their own private and sometimes conflicting opinions\cite{hilbe_indirect_2018,radzvilavicius_evolution_2019,perret:SciRep:2021,fujimoto_who_2024,michel-mata:Nature:2024,murase_indirect_2024}. We adopt the public case not because private assessment is unimportant, but because it is the cleanest setting in which to study how social norms make reputation a reliable signal on which to condition behavior.
Real populations surely draw on both shared and private information\cite{schmid_evolution_2021}, and tracing how moral consistency fares as assessment grows more private is a natural and important extension of the present work.

We have shown that the moral systems that sustain cooperation across most social dilemmas are discriminating and consistent, and that those that prevail across every dilemma are contrite as well. Consistency is a close cousin of the golden rule, the principle -- found in some form across many moral traditions -- that we should treat others as we would have them treat us. But where the golden rule governs conduct, consistency governs judgment: it asks us to hold others to the standard we keep ourselves.
For the contrite discriminator, consistency pairs a rule of conduct -- cooperate with everyone, but punish a bad partner when in good standing -- with a matching rule of judgment applied to others: that same punishment is deemed good in a partner, and those who cooperate while in bad standing regain their good name.
Across the entire range of cooperative dilemmas, selection favors this alignment of judgment and conduct -- moral systems that practice what they preach, and that preach what they practice.

\section{Methods} \label{methods}

\subsection{Model}

\paragraph{Cooperative dilemmas.}

We model social interactions as symmetric $2 \times 2$ games (Fig.~\ref{fig:overview}a).
The action of a focal player ($\af$) and that of their co-player ($\ac$) can be either `cooperate' ($C$) or `defect' ($D$).
We denote the payoffs for mutual cooperation, unilateral cooperation, unilateral defection, and mutual defection respectively by $\Rr$, $\Ss$, $\Tt$, and $\Pp$.

We focus on games that constitute social dilemmas of cooperation, or `cooperative dilemmas', in the following sense\cite{Pena2023}.
First, actions $C$ and $D$ qualify as cooperation and defection: (i) mutual cooperation is preferred over mutual defection ($\Rr>\Pp$), and (ii) each player prefers their co-player to cooperate rather than defect, whatever their own action ($\Rr>\Ss$ and $\Tt>\Pp$).
Second, the game is a genuine social dilemma: it has a (possibly mixed) evolutionarily stable strategy that does not coincide with the social optimum -- the profile maximizing total payoff -- which is the case if and only if $\Tt>\Rr$ or $\Pp>\Ss$.
Hence, depending on the ranking of payoff values, there are three classes of cooperative dilemma:
(i) prisoner's dilemmas (PD) if $\Tt > \Rr > \Pp > \Ss$;
(ii) snowdrift games (SG, also known as chicken games) if $\Tt > \Rr > \Ss > \Pp$;
and (iii) stag hunts (SH, also known as assurance games) if $\Rr > \Tt > \Pp > \Ss$.
The donation game used in most previous studies is a special case of the prisoner's dilemma, with $\Rr=b-c$, $\Ss=-c$, $\Tt=b$, $\Pp=0$, where the benefit of cooperation $b$ and the cost of cooperation $c$ satisfy $b>c>0$.
We further restrict attention to cooperative dilemmas where mutual cooperation is efficient, that is, where it is itself the social optimum; a sufficient condition for this, which we assume throughout, is $2 \Rr > \Tt + \Ss$\cite{Pena2023}.

Because our analyses are invariant under positive affine transformations of the payoffs -- adding a common constant and multiplying by a positive one -- we may normalize them by the min-max transformation
\begin{equation}
 \Pi' = \frac{\Pi - \min \left\{R,S,T,P\right\}}{\max \left\{R,S,T,P\right\} - \min \left\{R,S,T,P\right\}}, \ \ \text{ for all } \Pi \in \left\{R,S,T,P\right\} ,
\end{equation}
which maps the largest payoff to one and the smallest to zero.
This fixes two of the four payoffs, leaving each cooperative dilemma class $\Gamma \in \left\{\text{PD},\text{SG},\text{SH}\right\}$ with two free parameters; each class can then be represented as a polygon whose points are the individual dilemmas of that class (\ref{fig:payoff_space}), a representation we exploit when computing stability indices (see below).
Our analysis uses the following normalized dilemmas, for which the efficiency condition $2\Rr > \Tt+\Ss$ reduces to the constraint indicated:
(i) $1 > R > P > 0$ with $R > 1/2$ (PD);
(ii) $1 > R > S > 0$ with $2 R > 1+S$ (SG);
and (iii) $1 > T > P > 0$ with no further constraint (SH, where it holds automatically).
The normalized donation game is the special case of the normalized prisoner's dilemma with $P=1-R$.

\paragraph{Strategies and execution errors.}

Strategies are used by individuals when deciding how to play -- when to cooperate with, and when to defect against, a co-player.
The action of an individual in each interaction depends on their own reputation ($\rf$) and on that of the co-player ($\rc$).
With binary reputations (`good'$=G=1$ or `bad'$=B=0$), the strategy used by each individual is a Boolean function $p: \left\{B,G\right\}^2 \to \left\{D,C\right\}$
prescribing an action (`cooperate'$=C=1$ or `defect'$=D=0$) to each combination of these two reputations.
We denote each strategy by a string $p=p_0 p_1 p_2 p_3$, where the bit $p_i \in \left\{D,C\right\}$ gives the action prescribed for one combination of reputations, ordered $\rf$ then $\rc$ with $B$ before $G$: $p_0$ corresponds to $\rf=\rc=B$, $p_1$ to $\rf=B$ and $\rc=G$, $p_2$ to $\rf=G$ and $\rc=B$, and $p_3$ to $\rf=\rc=G$.
Such strategies are exactly the Boolean functions of two arguments, of which there are $2^{2^2}=2^4=16$.

We model execution errors, which capture an individual's occasional inability to carry out an intended action.
As in previous work\cite{ohtsuki:JTB:2004a,Nakamura2014}, these errors affect only intended cooperation: an individual who intends to cooperate defects instead with probability $\epsilon \in [0,1/2)$ (owing, for instance, to a lack of ability, time, or energy), whereas intended defection is always carried out.

\paragraph{Social norms, assignment errors, and observation model.}

Social norms are used by observers when deciding how to assign reputations to individuals after a social interaction.
We model the social norm as a Boolean function $d: \left\{B,G\right\}^2 \times \left\{D,C\right\}^2 \to \left\{B,G\right\}$
prescribing a new reputation (`good'$=G=1$ or `bad'$=B=0$) to each combination of the reputations and actions of the individual being judged and their co-player.
We denote each social norm by a string $d= d_0 d_1 \ldots d_{15}$, where the bit $d_i \in \left\{B,G\right\}$ gives the new reputation assigned to the judged individual for one combination of the current reputations $\rf$, $\rc$ and actions $\af$, $\ac$ of that individual and their co-player, ordered $\rf$, $\rc$, $\af$, $\ac$ with $B$ before $G$ and $D$ before $C$.
For example, $d_0$ corresponds to $\rf=\rc=\af=\ac=0$, $d_5$ to $\rf=\af=0$ and $\rc=\ac=1$, and $d_{15}$ to $\rf=\rc=\af=\ac=1$.
Each social norm can equivalently be displayed as a table, as in Figs~\ref{fig:overview}c and~\ref{fig:norms}a, whose rows list the reputation pairs $(\rf,\rc) = BB, BG, GB, GG$ from top to bottom and whose columns list the action pairs $(\af,\ac) = DD, DC, CD, CC$ from left to right.
The bitstring $d_0 d_1 \ldots d_{15}$ is simply this table read row by row: $d_0$ is the top-left entry (reputations $BB$, actions $DD$) and $d_{15}$ the bottom-right one (reputations $GG$, actions $CC$).
Such social norms are exactly the Boolean functions of four arguments, of which there are $2^{2^4}=2^{16}=\numprint{65536}$, spanning different orders and classes (see~\ref{fig:space}).

We also model assignment errors: with probability $\chi \in (0,1/2)$, the observer assigns the reputation opposite to the one prescribed by the social norm.
We adopt a `public-assessment model'\cite{santos_social_2018,ohtsuki:JTB:2004a,ohtsuki:JTB:2006a,ohtsuki:JTB:2007}: once an individual's reputation is updated, it is transmitted to the whole population (for instance, through gossip or a central authority), so that everyone agrees on the reputation of others.

\paragraph{Moral systems.}

A moral system $(p,d)$ pairs a strategy $p$ with a social norm $d$; a population implements it if, at all times, individuals behave according to $p$ and observers judge according to $d$.
Combining every strategy with every social norm yields $2^{4} \times 2^{16} = 2^{20} = \numprint{1048576}$ moral systems.

\subsection{Classes of moral systems}

\paragraph{Discriminating strategies and systems.}

Among the sixteen strategies, those that cooperate with a good co-player and defect against a bad one when the focal player is in good standing ($\rf=G$) are central to our analysis.
Formally, a strategy $p$ is `discriminating' if $p(G,B)=D$ and $p(G,G)=C$; as a string, $p={*}{*}DC$.
These two fixed bits leave behavior in bad standing unspecified, yielding the four discriminating strategies (Fig.~\ref{fig:overview}b), which differ only in what they prescribe in that case: the `(standard) discriminator' (DISC, $p=DCDC$) discriminates by the co-player's reputation just as it does when in good standing; the `contrite discriminator' (CDISC, $p=CCDC$) cooperates regardless of the co-player's reputation; the `obfuscated discriminator' (ODISC, $p=CDDC$) reverses the discriminating logic; and the `petulant discriminator' (PDISC, $p=DDDC$) defects regardless of the co-player's reputation.
A moral system $(p,d)$ is discriminating if its strategy $p$ is discriminating.

\paragraph{Contrite strategies and systems.}

A strategy $p$ is `contrite' if it cooperates whenever the focal player is in bad standing ($\rf=B$), that is, if $p(B,B)=C$ and $p(B,G)=C$; as a string, $p=CC{*}{*}$.
Contrition thus captures the readiness to make amends -- cooperating even with a bad co-player -- when one's own reputation is bad.
A moral system $(p,d)$ is contrite if its strategy $p$ is contrite; among the four discriminating strategies, CDISC is the only contrite one, so a discriminating system is contrite precisely when its strategy is CDISC.

\paragraph{Consistent systems.}

A moral system $(p,d)$ is `consistent' if the social norm $d$ judges individuals in strict accordance with the strategy $p$, whenever the co-player behaves as $p$ prescribes.
 Formally, $(p,d)$ is consistent if, for all $r_1, r_2 \in \{B,G\}$:
 \begin{subequations} \label{consistent-cond}
 \begin{align}
  d(r_1, r_2, p(r_1,r_2), p(r_2,r_1)) &= G, \label{consistent-cond-G} \\
  d(r_1, r_2, \neg p(r_1,r_2), p(r_2,r_1)) &= B, \label{consistent-cond-B}
 \end{align}
 \end{subequations}
 where $\neg C = D$ and $\neg D = C$.
 Consistency is silent about what reputation to assign when the co-player itself deviates from $p$: conditions~\eqref{consistent-cond} fix eight of the social norm's sixteen bits, and the other eight -- covering exactly these deviation cases -- remain free.
 Each strategy $p$ therefore gives rise to exactly $2^8 = \numprint{256}$ consistent social norms.
 In the absence of execution errors ($\epsilon = 0$), co-players -- who are residents under the rare-mutant assumption -- always execute the prescribed action, so these unconstrained bits are never invoked.
 All $\numprint{256}$ consistent social norms for a given strategy therefore produce identical reputation dynamics, payoffs, and hence stability and cooperation indices in this case (see SI Remark~\ref{SI-rem:consistent-equivalence}).
 A moral system that is both discriminating and consistent is a `consistent discriminating system'; the four discriminating strategies therefore give $4 \times \numprint{256} = \numprint{1024}$ such systems (their social norms are shown in Fig.~\ref{fig:norms}a).

\subsection{Evolutionary analysis}

\paragraph{Reputation dynamics and equilibrium.}

We consider an infinite, well-mixed population whose members are repeatedly and randomly paired to play the game, with a single observer assessing both players after each interaction (public assessment).
In a resident population implementing moral system $(p,d)$, let $g_{p,d}$ denote the fraction of individuals in good standing.
An individual with reputation $\rf$ who meets a co-player with reputation $\rc$ is then assigned a good reputation with probability $\delta_{\rf\rc}$, a function of the social norm $d$ and the error rates $\epsilon$ and $\chi$ (SI Eq.~(\ref{SI-deltaij})).
Averaging over the reputations of randomly matched co-players yields a one-dimensional dynamic for $g_{p,d}$, with a unique, asymptotically stable equilibrium $g^*_{p,d}$ (SI Eqs.~(\ref{SI-gdot}),~(\ref{SI-g*})).
All resident quantities below are evaluated at this reputation equilibrium.

\paragraph{Cooperation index.}

Each moral system $(p,d)$ is characterized by a cooperation index $\kappa_{p,d}$, the probability that, at equilibrium, a randomly chosen resident cooperates in a social interaction.
It thus measures how cooperative the moral system is, and is independent of the game used to model social interactions, though it depends on the error rates $\epsilon$ and $\chi$.
Concretely, $\kappa_{p,d}$ is a resident's probability of cooperating, averaged over the reputations of two randomly matched residents and evaluated at the reputation equilibrium $g^*_{p,d}$ (SI Eq.~(\ref{SI-coop-index})).

\paragraph{Expected payoffs and stability.}

Fix a cooperative dilemma (determined by a class
$\Gamma \in \{\mathrm{PD},\mathrm{SG},\mathrm{SH}\}$ together with specific
values of its two payoff parameters), error rates $(\epsilon,\chi)$, and a social norm $d$.
Strategy $p$ is stable under $d$ if, in a resident population implementing $(p,d)$, its expected payoff $w_{p,d}$ strictly exceeds that of every rare mutant strategy:
\begin{equation} \label{stability-condition}
w_{p,d} > \tilde w_{\tilde p,p,d}
\quad
\text{for all }
\tilde p \neq p,
\end{equation}
where $\tilde w_{\tilde p,p,d}$ denotes the expected payoff of a mutant using $\tilde p$.
If condition \eqref{stability-condition} holds, then $p$ constitutes a strict Nash equilibrium of the induced
population game under social norm $d$.
The payoffs $w_{p,d}$ and $\tilde w_{\tilde p,p,d}$ are computed at the joint reputation equilibrium.
Because mutants are vanishingly rare, they interact only with residents and leave the resident reputation equilibrium $g^*_{p,d}$ unperturbed, while the mutant sub-population settles at its own unique, asymptotically stable equilibrium fraction of good reputations $h^*_{\tilde p,p,d}$ against that resident background (SI Eqs.~(\ref{SI-hdot}),~(\ref{SI-h*})).
Each payoff is then the expected game payoff averaged over the reputations of the interacting players -- resident--resident encounters for $w_{p,d}$, and mutant--resident encounters for $\tilde w_{\tilde p,p,d}$ -- evaluated at $(g^*_{p,d},h^*_{\tilde p,p,d})$ and adjusted for execution errors (SI Eqs.~(\ref{SI-w}),~(\ref{SI-tildew})).
We determine stability conditions with a combination of analytical
and numerical techniques that extend methods previously developed
for the donation game and social norms up to
third order\cite{ohtsuki:JTB:2004a,Nakamura2014}.

\paragraph{Stability index.}

Let $\mathcal G_\Gamma$ denote the normalized two-dimensional payoff
region corresponding to game class
$\Gamma \in \{\mathrm{PD},\mathrm{SG},\mathrm{SH}\}$.
For fixed error rates $(\epsilon,\chi)$ and a moral system $(p,d)$,
define
\begin{equation}
\mathcal S_{p,d}^\Gamma(\epsilon,\chi)
:=
\{\, \bm{\pi} \in \mathcal G_\Gamma \; : \;
p \text{ is stable under } d \text{ at payoff } \bm{\pi} \,\},
\end{equation}
where stability at a given payoff configuration is understood in the
sense defined above.
We define the stability index
\begin{equation}
\sigma_{p,d}^\Gamma(\epsilon,\chi)
:=
\frac{\lambda_2\big(\mathcal S_{p,d}^\Gamma(\epsilon,\chi)\big)}
{\lambda_2\big(\mathcal G_\Gamma\big)},
\end{equation}
where $\lambda_2$ denotes the two-dimensional Lebesgue measure on the
payoff parameter space.
Thus $\sigma_{p,d}^\Gamma \in [0,1]$ represents the normalized area of
the payoff region within class $\Gamma$ for which $p$ is stable under
$d$. Equivalently, $\sigma_{p,d}^\Gamma$ is the
probability that $p$ is stable under $d$ when a payoff configuration
is drawn uniformly at random from $\mathcal G_\Gamma$.

To compare moral systems across the three classes of
cooperative dilemma, we also report the average stability index
\begin{equation}
\sigma_{p,d}(\epsilon,\chi)
:=
\frac{1}{3}\sigma^{\mathrm{PD}}_{p,d}(\epsilon,\chi)
+
\frac{1}{3}\sigma^{\mathrm{SG}}_{p,d}(\epsilon,\chi)
+
\frac{1}{3}\sigma^{\mathrm{SH}}_{p,d}(\epsilon,\chi).
\end{equation}
This average weights the three dilemma classes equally, regardless of the differing areas of their payoff regions.
As a robustness check, we further compute an alternative aggregate
stability index in which each class $\Gamma$ is weighted in proportion
to the area $\lambda_2(\mathcal G_\Gamma)$ of its normalized payoff
region (\ref{fig:payoff_space}); see
SI Fig.~\ref{SI-fig:scatter_panel_area}.

\paragraph{Numerical evaluation.}

For every moral system, we test the stability condition~\eqref{stability-condition} against all $15$ alternative strategies.
With no execution errors and in the limit of vanishing assignment errors ($\epsilon=0$, $\chi \to 0$), the stability region $\mathcal S_{p,d}^\Gamma$ is a polygon whose vertices we obtain analytically and whose area we evaluate in closed form with the shoelace formula (SI Section~\ref{SI-subsec:stability-indices}).
This yields the exact stability indices reported in Fig.~\ref{fig:stability_indices}a,b (computed in SI Sections~\ref{SI-subsec:stab-anal-disc} and~\ref{SI-subsec:stab-anal-leading-eight}, respectively).
At the finite error rates used for our main results ($\epsilon=\chi=0.02$; Fig.~\ref{fig:scatter}), we instead estimate each stability index by Monte Carlo integration, drawing $\numprint{1000}$ payoff configurations uniformly at random from the normalized region $\mathcal G_\Gamma$ of each dilemma class and recording the fraction of configurations at which $p$ is stable under $d$.
Because a moral system and its `mirror' -- the system obtained by swapping the two reputation labels throughout -- are behaviorally equivalent, we evaluate one representative per mirror class, $\numprint{524800}$ in total (SI Sections~\ref{SI-subsec:mirror-classes} and~\ref{SI-subsec:numerical-pipeline}).

\paragraph{Righteous and universally stable systems.}

Our central analytical results concern several further properties of moral systems, all evaluated in a tractable regime: no execution errors ($\epsilon=0$) and the limit of vanishing assignment errors ($\chi \to 0$).
A moral system is `self-approving' if almost all of its residents attain good standing, that is, if the resident reputation equilibrium approaches one, $\lim_{\chi\to0}(g^*_{p,d})_{\epsilon=0}=1$ (SI Definition~\ref{SI-def:self-approving}).
It is `self-cooperative' if its residents cooperate at near-maximal rates, that is, if the cooperation index approaches one, $\lim_{\chi\to0}(\kappa_{p,d})_{\epsilon=0}=1$ (SI Definition~\ref{SI-def:self-cooperative}).
It is `righteous' if it is both self-approving and self-cooperative (SI Definition~\ref{SI-def:righteous}).
Finally, it is `universally stable' if it is stable across every cooperative dilemma we consider, that is, if its stability index approaches one in each of the three classes, $\lim_{\chi\to0}(\sigma^\Gamma_{p,d})_{\epsilon=0}=1$ for all $\Gamma\in\{\mathrm{PD},\mathrm{SG},\mathrm{SH}\}$ (SI Definition~\ref{SI-def:universally-stable}).

\subsection{Data availability}

The archived payoff samples used for the Monte Carlo stability estimates are distributed with the code (see \hyperref[code-availability]{Code availability}). The computed cooperation and stability indices for all moral systems and error-rate combinations, from which the figures are generated (SI Section~\ref{SI-subsec:numerical-pipeline}), are deposited at \url{https://doi.org/10.26180/33068042}.

\subsection{Code availability}\label{code-availability}

The Julia package \texttt{MoralConsistency}, which implements the reputation dynamics, fitness computations, and stability tests, together with the Python scripts that generate the figures, is openly available at \url{https://github.com/ghostleopold/moral-consistency-reproducibility} and permanently archived at \url{https://doi.org/10.5281/zenodo.21519395}.

\bibliographystyle{naturemag}
\bibliography{refs}

\section*{Acknowledgements}
We are grateful to Fernando P. Santos for valuable discussions during the early stages of this project, and to Jean-Fran\c{c}ois Bonnefon and Catherine Molho for feedback on an earlier draft.
J.P. is grateful to the Institute for Advanced Study (IAS) of the University of Amsterdam for hosting the sabbatical research fellowship during which this project began.
J.P. acknowledges financial support from the French National Research Agency (ANR-22-CE26-0014-01 and ANR-17-EURE-0010).
T.H.'s work on this project was supported by the Australian Research Council, FT180100067.

\section*{Author contributions}
J.P., J.G., and T.H. conceived and designed the study.
J.P. derived the analytical results and wrote the Supplementary Information.
J.G. developed the Julia package and performed the numerical analyses.
T.H. produced the figures and numerical robustness analyses.
J.P. and T.H. wrote the first draft of the manuscript.
All authors revised the manuscript.

\section*{Competing interests}
The authors declare no competing interests.

\section*{Additional information}
Correspondence and requests for materials should be addressed to Jorge Peña (\href{mailto:jorge.pena@tse-fr.eu}{jorge.pena@tse-fr.eu}).

\newpage
\setcounter{figure}{0}
\renewcommand{\thefigure}{Extended Data Fig.~\arabic{figure}}
\renewcommand{\theHfigure}{ED.\arabic{figure}} %
\makeatletter
\renewcommand{\fnum@figure}{\thefigure}
\makeatother

\begin{figure}[p]
\centering
\includegraphics[width=\textwidth]{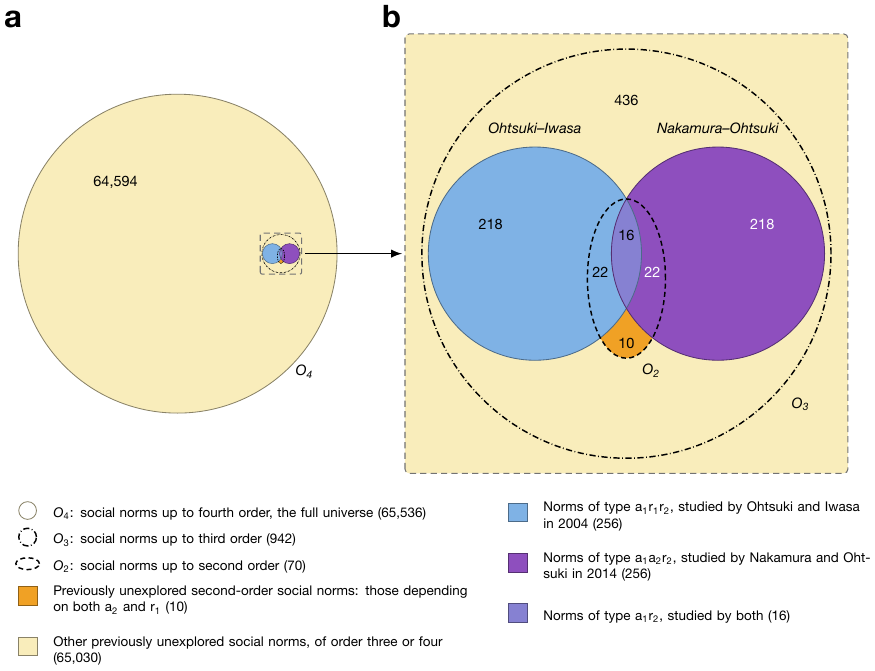}
\caption{\textbf{The space of social norms and its previously explored subsets.}
\textbf{a}, The universe $O_4$ of the $2^{16} = \numprint{65536}$ fourth-order social norms studied in this paper, which depend on the two players' actions and reputations; $\af$ and $\rf$ denote the action and the reputation of the judged player, $\ac$ and $\rc$ those of the co-player.
\textbf{b}, Magnification of the dashed rectangle in \textbf{a}.
Most previous work, following the study by Ohtsuki and Iwasa\cite{ohtsuki:JTB:2004a}, focused on the 256 third-order norms of type $\af\rf\rc$ (\emph{blue}); Nakamura and Ohtsuki\cite{Nakamura2014} considered the 256 third-order norms of type $\af\ac\rc$ (\emph{purple}); the two spaces share the 16 norms of type $\af\rc$ (\emph{violet}).
Although one existing paper studies fourth-order norms\citep{santos_social_2018}, it goes beyond the simple setting of previous research by allowing reputation to depend on information from past encounters.
The dashed and dash-dotted boundaries mark the sets $O_2$ and $O_3$ of norms of order at most two and three; \emph{yellow} marks the previously unexplored second-order norms.
Within each panel, the area of every region is proportional to the number of social norms it contains.
All counts are derived in SI Section~\ref{SI-subsec:norm-counting}.
}
\label{fig:space}
\end{figure}

\begin{figure}[p]
\centering
\includegraphics[width=\textwidth]{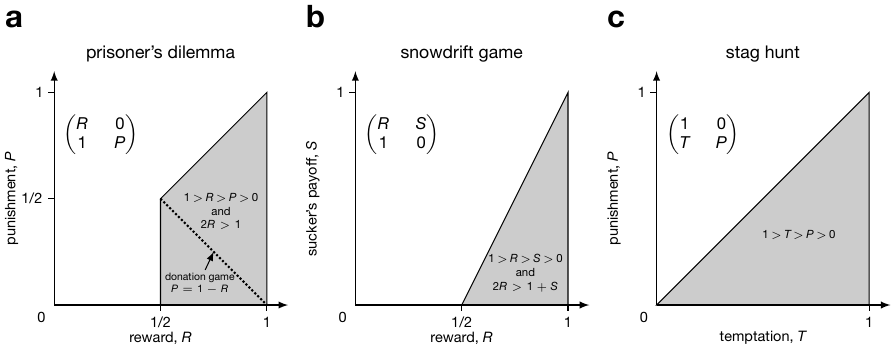}
\caption{{\bf The space of $2 \times 2$ cooperative dilemmas.}
We consider the three classes of $2 \times 2$ cooperative dilemma, shown after the min-max normalization of payoffs to $[0,1]$ (\hyperref[methods]{Methods}); throughout, we impose $2R > T+S$ so that mutual cooperation is socially optimal.
\textbf{a}, Prisoner's dilemmas (PD), with payoff ranking $1 > R > P > 0$; the efficiency condition adds the constraint $2R > 1$. The donation game, used in most previous studies of indirect reciprocity, is the special case with $P = 1-R$ (\emph{dotted line}).
\textbf{b}, Snowdrift games (SG), with payoff ranking $1 > R > S > 0$; the efficiency condition adds the constraint $2R > 1+S$.
\textbf{c}, Stag hunts (SH), with payoff ranking $1 > T > P > 0$; the efficiency condition holds automatically.}
\label{fig:payoff_space}
\end{figure}

\begin{figure}[p]
  \centering
  \includegraphics[width=0.98\textwidth]{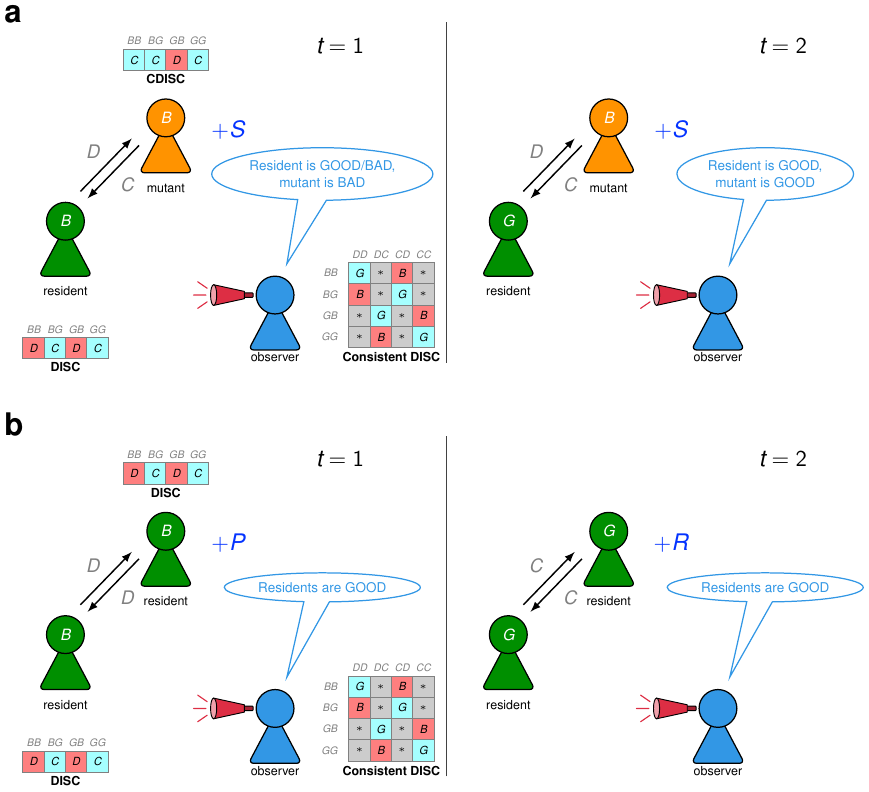}
  \caption{\textbf{How CDISC can invade a consistent DISC system.}
  CDISC and DISC differ only in the rare, error-induced state where both players have bad reputations.
  Pawns carry their reputations ($G$ or $B$) and follow the strategies shown beside them (Fig.~\ref{fig:overview}b); grey letters give actions, and blue annotations the focal player's payoffs.
  Each round, an observer (\emph{blue pawn}) judges both players by a DISC-consistent social norm ($\ast$ marks entries unconstrained by consistency; Fig.~\ref{fig:norms}a) and broadcasts the verdicts.
  \textbf{a}, Trajectory of a mutant CDISC (\emph{orange}) starting in bad standing.
  At $t=1$, the mutant meets a resident also in bad standing and contritely cooperates, earning $S$; because cooperating with a bad partner deviates from DISC's prescription, it is judged bad (the resident's verdict is unconstrained, hence `GOOD/BAD').
  At $t=2$, it meets another resident, almost surely in good standing, and cooperates, matching DISC's prescription; it is judged good but its co-player defects, so it earns $S$ again.
  Its reputation is restored only after two rounds, for a total payoff of $2S$.
  \textbf{b}, Trajectory of a resident DISC in the analogous situation.
  At $t=1$, it defects, as DISC prescribes, earns $P$, and is judged good; at $t=2$, in good standing like almost all residents, it cooperates and earns $R$, for a total of $P+R$.
  CDISC thus gains a net advantage, and can invade, whenever $2S > R+P$, which only snowdrift games satisfy.}
  \label{fig:cdisc_invades}
\end{figure}

\begin{figure}[p]
  \centering
  \includegraphics[width=0.98\textwidth]{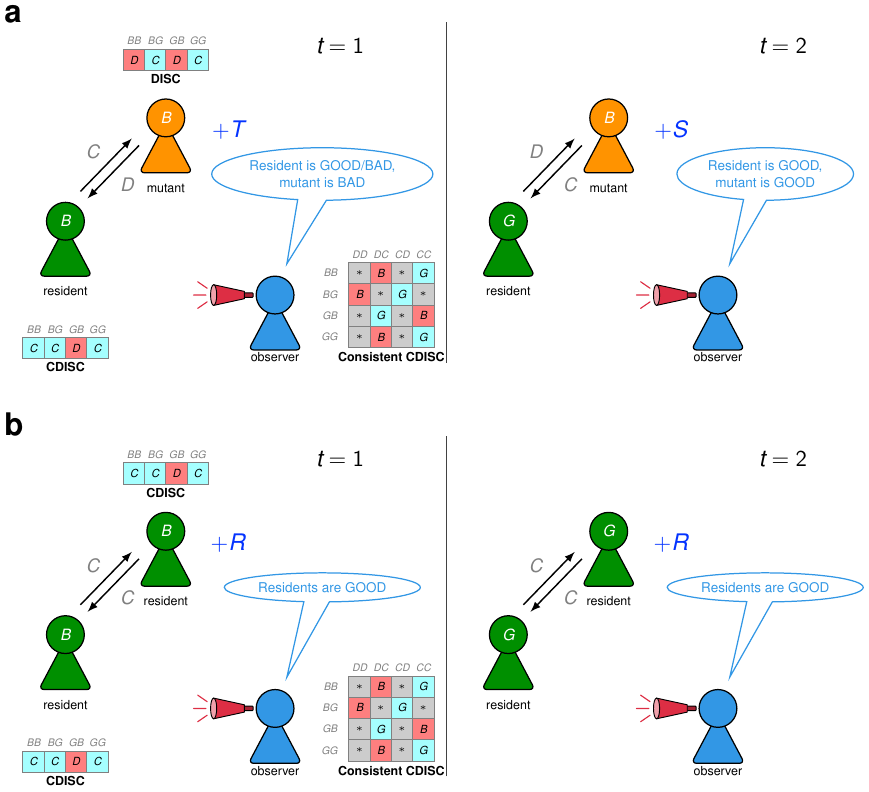}
  \caption{\textbf{Why DISC cannot invade a consistent CDISC system.}
  The converse of \ref{fig:cdisc_invades}: a rare DISC mutant among CDISC residents, judged by a social norm consistent with CDISC; notation as in \ref{fig:cdisc_invades}.
  \textbf{a}, Trajectory of a mutant DISC (\emph{orange}) starting from a bad reputation.
  At $t=1$, the mutant meets a resident who is also in bad standing and defects on its contritely cooperating partner; it earns the temptation payoff $T$ but, because defecting on a bad partner deviates from what CDISC prescribes, is judged bad (the resident's verdict again falls on an unconstrained entry, hence `GOOD/BAD').
  At $t=2$, the mutant meets another resident, almost surely in good standing, and cooperates, as both strategies prescribe; it is judged good, but its co-player defects and it earns only $S$.
  Its reputation is restored after two rounds, over which it accumulates $T+S$.
  \textbf{b}, Trajectory of a resident CDISC in the analogous situation.
  At $t=1$, it cooperates with its bad partner, as CDISC prescribes, earns $R$, and is immediately judged good; at $t=2$ it cooperates again and earns $R$, for a total of $2R$.
  Because we require mutual cooperation to be efficient, $2R > T+S$ holds in every dilemma we study: DISC never gains an advantage and cannot invade a consistent CDISC system.}
  \label{fig:disc_cannot_invade}
\end{figure}

\clearpage
\makeatletter
\setstretch{1}
\everydisplay{\@displaytrue}
\setlength\parskip{\mc@parskip}\setlength\parindent{\mc@parindent}
\setlength\parfillskip{\mc@parfillskip}\setlength\partopsep{\mc@partopsep}
\let\@listI\mc@listI \let\@listi\mc@listi \let\@listii\mc@listii \let\@listiii\mc@listiii
\@listi
{\@listi \global\bibsep\itemsep \global\advance\bibsep by\parsep}
\patchcmd\@starttoc
  {\begingroup \makeatletter
      \skip@\lastskip
      \advance\skip@\parskip
      \vskip-\lastskip
      \vskip\skip@
      \parskip\parskip@tocskip}
  {\begingroup \makeatletter}{}{\PackageError{merge}{Could not unpatch starttoc}{}}
\patchcmd\@startsection
  {\advance\@tempskipa-\parskip\addvspace\@tempskipa}{\addvspace\@tempskipa}
  {}{\PackageError{merge}{Could not unpatch startsection}{}}
\patchcmd\@xsect
  {\advance\@tempskipa-\parskip\vskip\@tempskipa}{\vskip\@tempskipa}
  {}{\PackageError{merge}{Could not unpatch xsect}{}}
\let\fps@figure\mc@fpsfigure
\let\topfraction\mc@topfraction \let\bottomfraction\mc@bottomfraction
\let\textfraction\mc@textfraction \let\floatpagefraction\mc@floatpagefraction
\setlength\emergencystretch{\mc@emergencystretch}
\setcounter{secnumdepth}{\mc@secnumdepth}
\let\Gin@ewidth\mc@Gin@ewidth \let\Gin@eheight\mc@Gin@eheight \Gin@isofalse
\setcounter{page}{1}\setcounter{section}{0}\setcounter{figure}{0}\setcounter{table}{0}
\setcounter{equation}{0}\setcounter{footnote}{0}
\let\thefigure\mc@thefigure \let\fnum@figure\mc@fnumfigure 
\let\maketitle\mc@maketitle \let\@maketitle\mc@@maketitle \let\@author\mc@author
\let\title\mc@title \let\date\mc@date \let\and\mc@and \let\thanks\mc@thanks
\let\caption\mc@c@caption \let\@caption\mc@c@@caption \let\@makecaption\mc@c@makecaption
\let\mc@newcommand\newcommand
\def\newcommand#1{\let#1\relax\mc@newcommand#1}
\makeatother
\newcommand{\Jorge}{red!40} %
\newcommand{\Julian}{blue!40} %
\newcommand{\Toby}{yellow!40} %
\newcommand{\Fernando}{green!40} %
\newcommand{\cmark}{\textcolor{LBlue}{\ding{51}}} %
\newcommand{\xmark}{\textcolor{LRed}{\ding{55}}}  %
\newcommand{\Rr}{R}
\newcommand{\Ss}{S}
\newcommand{\Tt}{T}
\newcommand{\Pp}{P}
\newcommand{\af}{a_1} %
\newcommand{\ac}{a_2} %
\newcommand{\rf}{r_1} %
\newcommand{\rc}{r_2} %
\newcommand{\dd}{\mathrm{d}} %
\newcommand{\C}{\texttt{C}}
\newcommand{\D}{\texttt{D}}
\newcommand{\G}{\texttt{G}}
\newcommand{\B}{\texttt{B}}
\newcommand{\wildcard}{\texttt{*}}
\newcommand{\vDD}{\mathsf{v}_{\D\D}}
\newcommand{\vDC}{\mathsf{v}_{\D\C}}
\newcommand{\vCD}{\mathsf{v}_{\C\D}}
\newcommand{\vCC}{\mathsf{v}_{\C\C}}
\newcommand{\prfrc}{\mathsf{p}_{\rf\rc}}
\newcommand{\prcrf}{\mathsf{p}_{\rc\rf}}
\newcommand{\drfrc}{\mathsf{d}_{\rf\rc}}
\newcommand{\drfrcm}{\tilde{\mathsf{d}}_{\rf\rc}}
\newcommand{\notdrfrc}{\bar{\mathsf{d}}_{\rf\rc}}
\newcommand{\pBB}{\mathsf{p}_{\B\B}}
\newcommand{\pBG}{\mathsf{p}_{\B\G}}
\newcommand{\pGB}{\mathsf{p}_{\G\B}}
\newcommand{\pGG}{\mathsf{p}_{\G\G}}
\newcommand{\prfrcm}{\tilde{\mathsf{p}}_{\rf \rc}}
\newcommand{\pBBm}{\tilde{\mathsf{p}}_{\B\B}}
\newcommand{\pBGm}{\tilde{\mathsf{p}}_{\B\G}}
\newcommand{\pGBm}{\tilde{\mathsf{p}}_{\G\B}}
\newcommand{\pGGm}{\tilde{\mathsf{p}}_{\G\G}}
\newcommand{\dBB}{\mathsf{d}_{\B\B}}
\newcommand{\dBG}{\mathsf{d}_{\B\G}}
\newcommand{\dGB}{\mathsf{d}_{\G\B}}
\newcommand{\dGG}{\mathsf{d}_{\G\G}}
\newcommand{\dBBm}{\tilde{\mathsf{d}}_{\B\B}}
\newcommand{\dBGm}{\tilde{\mathsf{d}}_{\B\G}}
\newcommand{\dGBm}{\tilde{\mathsf{d}}_{\G\B}}
\newcommand{\dGGm}{\tilde{\mathsf{d}}_{\G\G}}
\newcommand{\notdBB}{\bar{\mathsf{d}}_{\B\B}}
\newcommand{\notdBG}{\bar{\mathsf{d}}_{\B\G}}
\newcommand{\notdGB}{\bar{\mathsf{d}}_{\G\B}}
\newcommand{\notdGG}{\bar{\mathsf{d}}_{\G\G}}
\newcommand{\dBBDD}{\mathsf{d}_{\B\B\D\D}}
\newcommand{\dBBDC}{\mathsf{d}_{\B\B\D\C}}
\newcommand{\dBBD}{\mathsf{d}_{\B\B\D*}}
\newcommand{\dBBCD}{\mathsf{d}_{\B\B\C\D}}
\newcommand{\dBBCC}{\mathsf{d}_{\B\B\C\C}}
\newcommand{\dBBC}{\mathsf{d}_{\B\B\C*}}
\newcommand{\dBGDD}{\mathsf{d}_{\B\G\D\D}}
\newcommand{\dBGDC}{\mathsf{d}_{\B\G\D\C}}
\newcommand{\dBGD}{\mathsf{d}_{\B\G\D*}}
\newcommand{\dBGCD}{\mathsf{d}_{\B\G\C\D}}
\newcommand{\dBGCC}{\mathsf{d}_{\B\G\C\C}}
\newcommand{\dBGC}{\mathsf{d}_{\B\G\C*}}
\newcommand{\dGBDD}{\mathsf{d}_{\G\B\D\D}}
\newcommand{\dGBDC}{\mathsf{d}_{\G\B\D\C}}
\newcommand{\dGBD}{\mathsf{d}_{\G\B\D*}}
\newcommand{\dGBCD}{\mathsf{d}_{\G\B\C\D}}
\newcommand{\dGBCC}{\mathsf{d}_{\G\B\C\C}}
\newcommand{\dGBC}{\mathsf{d}_{\G\B\C*}}
\newcommand{\dGGDD}{\mathsf{d}_{\G\G\D\D}}
\newcommand{\dGGDC}{\mathsf{d}_{\G\G\D\C}}
\newcommand{\dGGD}{\mathsf{d}_{\G\G\D*}}
\newcommand{\dGGCD}{\mathsf{d}_{\G\G\C\D}}
\newcommand{\dGGCC}{\mathsf{d}_{\G\G\C\C}}
\newcommand{\dGGC}{\mathsf{d}_{\G\G\C*}}
\newcommand{\DDDD}{\D\D\D\D} %
\newcommand{\CDDD}{\C\D\D\D} %
\newcommand{\DCDD}{\D\C\D\D} %
\newcommand{\CCDD}{\C\C\D\D} %
\newcommand{\DDCD}{\D\D\C\D} %
\newcommand{\CDCD}{\C\D\C\D} %
\newcommand{\DCCD}{\D\C\C\D} %
\newcommand{\CCCD}{\C\C\C\D} %
\newcommand{\DDDC}{\D\D\D\C} %
\newcommand{\CDDC}{\C\D\D\C} %
\newcommand{\DCDC}{\D\C\D\C} %
\newcommand{\CCDC}{\C\C\D\C} %
\newcommand{\DDCC}{\D\D\C\C} %
\newcommand{\CDCC}{\C\D\C\C} %
\newcommand{\DCCC}{\D\C\C\C} %
\newcommand{\CCCC}{\C\C\C\C} %
\newcommand{\ALLD}{\mathrm{ALLD}}
\newcommand{\ALLC}{\mathrm{ALLC}}
\newcommand{\DISC}{\mathrm{DISC}}
\newcommand{\CDISC}{\mathrm{CDISC}}
\newcommand{\ODISC}{\mathrm{ODISC}}
\newcommand{\PDISC}{\mathrm{PDISC}}
\newcommand{\Cset}{\mathcal{C}}
\newcommand{\Dset}{\mathcal{D}}
\newcommand{\Aset}{\mathcal{A}}
\newcommand{\AsetGXXG}{\mathcal{A}_{\G\wildcard\wildcard\G}}
\newcommand{\AsetGBGG}{\mathcal{A}_{\G\B\G\G}}
\newcommand{\AsetGGBG}{\mathcal{A}_{\G\G\B\G}}
\newcommand{\AsetBGGG}{\mathcal{A}_{\B\G\G\G}}
\newcommand{\AsetGGGG}{\mathcal{A}_{\G\G\G\G}}
\newcommand{\Zset}{\mathcal{Z}}
\newcommand{\maxAset}{\mathcal{A}_*}
\newcommand{\Rset}{\mathcal{R}}
\newcommand{\maxRset}{\mathcal{R}_*}
\newcommand{\Kset}{\mathcal{K}}
\newcommand{\Qset}{\mathcal{Q}}
\newcommand{\QDset}{\mathcal{Q}_{\mathcal{D}}}
\newcommand{\QDisc}{\mathcal{Q}_{\DISC}}
\newcommand{\QCDisc}{\mathcal{Q}_{\CDISC}}
\newcommand{\QPDisc}{\mathcal{Q}_{\PDISC}}
\newcommand{\QODisc}{\mathcal{Q}_{\ODISC}}
\newcommand{\Lset}{\mathcal{L}}
\newcommand{\LI}{\mathcal{L}_{\text{I}}}
\newcommand{\LII}{\mathcal{L}_{\text{II}}}
\newcommand{\LIII}{\mathcal{L}_{\text{III}}}
\newcommand{\wBB}{\left(\omega_{\B\B}\right)_{\epsilon=0}}
\newcommand{\wBG}{\left(\omega_{\B\G}\right)_{\epsilon=0}}
\newcommand{\wGB}{\left(\omega_{\G\B}\right)_{\epsilon=0}}
\newcommand{\wGG}{\left(\omega_{\G\G}\right)_{\epsilon=0}}
\newcommand{\wBBm}{\left(\tilde{\omega}_{\B\B}\right)_{\epsilon=0}}
\newcommand{\wBGm}{\left(\tilde{\omega}_{\B\G}\right)_{\epsilon=0}}
\newcommand{\wGBm}{\left(\tilde{\omega}_{\G\B}\right)_{\epsilon=0}}
\newcommand{\wGGm}{\left(\tilde{\omega}_{\G\G}\right)_{\epsilon=0}}
\newcommand{\sigmaPD}{\sigma^{\mathrm{PD}}}
\newcommand{\sigmaSG}{\sigma^{\mathrm{SG}}}
\newcommand{\sigmaSH}{\sigma^{\mathrm{SH}}}
\newcommand{\sigmabar}{\sigma}
\newcommand{\sigmabarw}{\sigmabar^{\mathsf{w}}}
\newcommand{\BigO}{O}
\makeatletter\let\newcommand\mc@newcommand\makeatother
\addtocontents{toc}{\protect\setcounter{tocdepth}{3}}
\title{Cooperation thrives when moral judgment matches moral conduct: Supplementary Information}
\date{}

\frenchspacing
\maketitle

\tableofcontents

\clearpage

\section{Model}

This section provides a detailed description of the model used in the main text.
We consider an infinite, well-mixed population whose members interact in pairwise encounters, each modeled as a symmetric two-player game.
Each game belongs to one of three classes of cooperative dilemmas (prisoner's dilemmas, snowdrift games, or stag hunts); throughout we restrict attention to those for which mutual cooperation is the social optimum, an efficiency condition on the payoffs that we make precise below.
Individuals choose their actions according to a strategy (or action rule) that conditions on the binary reputations of the focal player and their co-player, while an observer updates both players' reputations according to an assessment rule (or social norm) that conditions additionally on the two actions taken.
Sections~\ref{subsec:labeling-conventions}--\ref{subsec:cooperative-dilemmas} establish labeling conventions and define the payoff function and the three classes of cooperative dilemmas.
Section~\ref{subsec:shorthand} introduces notation for strategies and assessment rules, and Section~\ref{subsec:norm-counting} counts assessment rules by the arguments they use, situating the space of social norms we explore relative to those studied in previous work.
Section~\ref{subsec:resident-dynamics} describes the reputation dynamics of a resident population, from which Section~\ref{subsec:cooperation-index} derives the cooperation index; Section~\ref{subsec:mutant-dynamics} extends the reputation dynamics to rare mutants, and Section~\ref{subsec:payoffs-stability} derives the equilibrium payoffs and defines evolutionary stability.
Finally, Section~\ref{subsec:stability-indices} defines the stability indices we use to assess the evolutionary robustness of moral systems across the full space of cooperative dilemmas.

\subsection{Labeling conventions for actions and reputations}\label{subsec:labeling-conventions}

Actions and reputations are binary.
In each interaction, an individual either `cooperates' ($\C$) or `defects' ($\D$); after the interaction, an observer may assign the individual a `good' ($\G$) or a `bad' ($\B$) reputation.
For notational convenience, we identify cooperation and a good reputation with $1$, and defection and a bad reputation with $0$:
\begin{subequations} \label{C=G=1,D=B=0}
\begin{align}
 \C & = \G = 1, \ \text{ and} \\
 \D & = \B = 0.
\end{align}
\end{subequations}
We also use the wildcard symbol $*$ to denote a placeholder that may represent either $0$ ($\D$ or $\B$) or $1$ ($\C$ or $\G$).

\subsection{Payoff function}\label{subsec:payoff-function}

Social interactions are modeled as symmetric $2 \times 2$ games.
A focal player choosing action $\af \in \left\{\D,\C\right\}$ against a co-player choosing action $\ac \in \left\{\D,\C\right\}$ receives a payoff $v(a_1,a_2)$, where $v: \left\{\D,\C\right\}^2 \to \mathbb{R}$.
Writing the focal player's action as the row and the co-player's as the column, this payoff function takes the matrix form
\begin{equation} \label{payoff-matrix}
 \bordermatrix{~ & \C & \D \cr
              \C & v(\C,\C) & v(\C,\D) \cr
              \D & v(\D,\C) & v(\D,\D) \cr} .
\end{equation}
Identifying
\begin{align*}
 v(\C,\C) & = R, \ v(\C,\D) = S, \ v(\D,\C) = T, \ \text{ and } \ v(\D,\D) = P,
\end{align*}
we recover the formulation given in the main text (Fig.~\ref{MT-fig:overview}a; see also \hyperref[MT-methods]{Methods}).

Throughout, we also use the shorthand notation:
\begin{equation*}
 \mathsf{v}_{\af\ac} = v(\af,\ac) .
\end{equation*}

\subsection{Cooperative dilemmas and payoff spaces}\label{subsec:cooperative-dilemmas}

We focus on the generic $2 \times 2$ games that constitute cooperative dilemmas in which mutual cooperation is the social optimum\citep{si-Pena2023}.
This requires that all of the following inequalities hold:
\begin{align}
 R & > P, \label{R>P}\\
 R & > S, \label{R>S}\\
 T & > P, \label{T>P} \\
 2R & > T+S, \label{2R>T+S}
\end{align}
and that either
\begin{align}
 T & > R \label{T>R}
\end{align}
or
\begin{align}
 P & > S \label{P>S}
\end{align}
also holds.
Inequality \eqref{R>P} means that mutual cooperation is preferred over mutual defection.
Inequalities \eqref{R>S} and \eqref{T>P} imply that each player prefers the co-player to cooperate rather than defect.
Inequality \eqref{2R>T+S} ensures that two cooperators earn a higher total payoff ($2R$) than a pair in which one cooperates and the other defects ($T+S$); together with \eqref{R>P}, this makes mutual cooperation the social optimum.
Inequalities \eqref{T>R} and \eqref{P>S} indicate that there is some temptation for players to defect; this is a necessary condition for the game to be a social dilemma\citep{si-Pena2023}.

Three classes of games fulfill these requirements\citep{si-Pena2023}: prisoner's dilemmas (PD), characterized by
\begin{equation} \label{PD}
 T>R>P>S \ \text{ and } \ 2R>T+S;
\end{equation}
snowdrift games (SG), characterized by
\begin{equation} \label{SG}
 T>R>S>P \ \text{ and } \ 2R>T+S;
\end{equation}
and stag hunts (SH), characterized by
\begin{equation} \label{SH}
 R>T>P>S,
\end{equation}
which trivially satisfy~\eqref{2R>T+S}.
Note that all three classes satisfy
\begin{equation} \label{T>S}
 T>S .
\end{equation}

Normalizing the payoffs by a min-max transformation (\hyperref[MT-methods]{Methods}), so that the largest payoff equals one and the smallest equals zero, we obtain
\begin{equation} \label{normalized-PD}
 1>R>P>0 \ \text{ and } \ 2R>1
\end{equation}
for prisoner's dilemmas,
\begin{equation} \label{normalized-SG}
 1>R>S>0 \ \text{ and } \ 2R>1+S
\end{equation}
for snowdrift games,
and
\begin{equation} \label{normalized-SH}
 1>T>P>0
\end{equation}
for stag hunts. This fixes two of the four payoffs, leaving each class with two free parameters (see \ref{MT-fig:payoff_space} of the main text).

\subsection{Strategies, assessment rules, and shorthand notations}\label{subsec:shorthand}

Consider a population where observers judge players according to assessment rule $d$, resident players adopt strategy $p$, and a small fraction of mutant players adopt strategy $\tilde{p}$.
As specified in the \hyperref[MT-methods]{Methods} of the main text, the strategies $p$ and $\tilde{p}$ are Boolean functions $p,\tilde{p}: \left\{\B,\G\right\}^{2} \to \left\{\D,\C\right\}$ prescribing an action for each combination of the reputations of the focal player and their co-player.
Likewise, the assessment rule is a Boolean function
$d: \left\{\B,\G\right\}^2 \times \left\{\D,\C\right\}^2 \to \left\{\B,\G\right\}$ prescribing a reputation
for each combination of the reputations and actions of the focal player and their co-player.

Consider a focal player who interacts with a co-player and is judged by an observer.
Let $r_1 \in \left\{\B,\G\right\}$ denote the reputation of the focal player and $r_2 \in \left\{\B,\G\right\}$ that of the co-player.
The focal player's strategy prescribes action $p(r_1,r_2)$ if they are a resident and $\tilde{p}(r_1,r_2)$ if they are a mutant.
Once the focal player has taken an action $a_1 \in \left\{\D,\C\right\}$ against a co-player who has played $a_2 \in \left\{\D,\C\right\}$, the assessment rule prescribes a new reputation $d(r_1,r_2,a_1,a_2)$ for the focal player.

It is convenient to introduce shorthand notation for the actions and reputations prescribed in different situations.
We use the following conventions throughout:
\begin{subequations}
\begin{align}
 \prfrc & := p(\rf,\rc), \\
 \prfrcm & := \tilde{p}(\rf,\rc), \\
 \mathsf{d}_{\rf\rc\af\ac} & := d(\rf,\rc,\af,\ac), \\
 \drfrc & := d(\rf,\rc,\prfrc,\prcrf), \label{dr1r2-shorthand} \\
 \drfrcm & := d(\rf,\rc,\prfrcm,\prcrf), \label{tildedr1r2-shorthand} \\
 \notdrfrc & := d(\rf,\rc,\neg\prfrc,\prcrf), \label{notdr1r2-shorthand}
\end{align}
\end{subequations}
where $\neg$ stands for the logical negation (with our labeling convention \eqref{C=G=1,D=B=0}, $\neg \C = \D$ and $\neg \D = \C$).
In expressions \eqref{dr1r2-shorthand}--\eqref{notdr1r2-shorthand} the co-player is a resident who behaves as prescribed, playing $p(r_2,r_1)$; the three differ only in the focal player's action.
Expression \eqref{dr1r2-shorthand} is the new reputation prescribed to a resident focal player, who plays $p(r_1,r_2)$, and \eqref{tildedr1r2-shorthand} that prescribed to a mutant focal player, who plays $\tilde{p}(r_1,r_2)$ -- both in the absence of execution errors.
Expression \eqref{notdr1r2-shorthand} is the new reputation prescribed to a focal player (resident or mutant) whose action deviates from the resident prescription, playing $\neg p(r_1,r_2)$.

\subsection{Counting assessment rules by the arguments they use}\label{subsec:norm-counting}

\ref{MT-fig:space} of the main text situates the space of assessment rules we study relative to the spaces explored in previous work.
Here we define the sets appearing there and derive their cardinalities.

An assessment rule is a Boolean function of the four binary arguments in $\mathcal{V} = \left\{\rf, \rc, \af, \ac\right\}$.
For a subset $\mathcal{U} \subseteq \mathcal{V}$, let $\mathcal{T}_{\mathcal{U}}$ denote the set of assessment rules of type $\mathcal{U}$: those that can be written as functions of the arguments in $\mathcal{U}$ alone.
There are $|\mathcal{T}_{\mathcal{U}}| = 2^{2^{|\mathcal{U}|}}$ rules of type $\mathcal{U}$, and $\mathcal{T}_{\mathcal{U}} \cap \mathcal{T}_{\mathcal{U}'} = \mathcal{T}_{\mathcal{U} \cap \mathcal{U}'}$, since a rule that ignores the arguments outside $\mathcal{U}$ and the arguments outside $\mathcal{U}'$ ignores all arguments outside $\mathcal{U} \cap \mathcal{U}'$.
For $k = 0, 1, \ldots, 4$, let
\begin{equation*}
O_k := \bigcup_{\mathcal{U} \subseteq \mathcal{V},\, |\mathcal{U}| = k} \mathcal{T}_{\mathcal{U}}
\end{equation*}
denote the set of assessment rules of order at most $k$; in particular, $O_4$ is the universe of $2^{2^4} = 2^{16} = \numprint{65536}$ assessment rules studied in this paper.

A rule is said to depend \emph{essentially} on the subset $\mathcal{U}$ if it is of type $\mathcal{U}$ but not of type $\mathcal{U}'$ for any proper subset $\mathcal{U}' \subset \mathcal{U}$.
The following proposition counts assessment rules by their essential arguments and, with them, the sets $O_k$.

\begin{proposition}[Number of assessment rules of each order] \label{prop:norm-counting}
For $k = 0, 1, \ldots, 4$, the number of assessment rules that depend essentially on a given subset $\mathcal{U} \subseteq \mathcal{V}$ with $k$ elements depends only on $k$.
Denoting this number by $N_k$, it satisfies the recursion
\begin{equation} \label{Nk-recursion}
 N_k = 2^{2^{k}} - \sum_{j<k} \binom{k}{j} N_j ,
\end{equation}
giving $N_0 = 2$, $N_1 = 2$, $N_2 = 10$, $N_3 = 218$, and $N_4 = \numprint{64594}$.
Moreover, the sets of assessment rules of order at most $k$ have cardinalities
\begin{equation} \label{Ok-cardinality}
 |O_k| = \sum_{j \leq k} \binom{4}{j} N_j ,
\end{equation}
giving $|O_0| = 2$, $|O_1| = 10$, $|O_2| = 70$, $|O_3| = 942$, and $|O_4| = 2^{16} = \numprint{65536}$.
\end{proposition}

\begin{proof}
Since type classes are closed under intersection, each rule is of a unique minimal type, that is, it depends essentially on exactly one subset of $\mathcal{V}$.
Consequently, the $2^{2^{|\mathcal{U}|}}$ rules of type $\mathcal{U}$ partition according to the subset $\mathcal{U}' \subseteq \mathcal{U}$ on which they depend essentially, so the number $E(\mathcal{U})$ of rules depending essentially on $\mathcal{U}$ satisfies $E(\mathcal{U}) = 2^{2^{|\mathcal{U}|}} - \sum_{\mathcal{U}' \subsetneq\, \mathcal{U}} E(\mathcal{U}')$.
By induction on $|\mathcal{U}|$, starting from $E(\emptyset) = 2$ (the two constant rules, which assign $\G$ to everybody or $\B$ to everybody), $E(\mathcal{U})$ depends only on $k = |\mathcal{U}|$, and grouping the proper subsets of $\mathcal{U}$ by their cardinalities yields the recursion~\eqref{Nk-recursion}.
Finally, a rule belongs to $O_k$ if and only if the subset on which it depends essentially has at most $k$ elements; summing $N_j$ over the $\binom{4}{j}$ subsets of $\mathcal{V}$ with $j \leq k$ elements yields Eq.~\eqref{Ok-cardinality}.
\end{proof}

The spaces explored by previous literature are $\mathcal{T}_{\left\{\af,\rf,\rc\right\}}$, the $2^{2^3} = 256$ third-order rules of type $\af\rf\rc$\citep{si-ohtsuki:JTB:2004a}, and $\mathcal{T}_{\left\{\af,\ac,\rc\right\}}$, the $256$ third-order rules of type $\af\ac\rc$\citep{si-Nakamura2014}.
By the intersection property above, a rule belongs to both spaces if and only if it is of type $\af\rc$.
The two spaces therefore share $|\mathcal{T}_{\left\{\af,\rc\right\}}| = 2^{2^2} = 16$ rules, and their union contains $256 + 256 - 16 = 496$ rules.

Consider first the $|O_2| = 70$ rules of order at most two (Proposition~\ref{prop:norm-counting}).
Such a rule lies in $\mathcal{T}_{\left\{\af,\rf,\rc\right\}}$ whenever its essential arguments avoid $\ac$, and in $\mathcal{T}_{\left\{\af,\ac,\rc\right\}}$ whenever they avoid $\rf$.
The rules of order at most two lying outside both previously studied spaces are therefore exactly the rules that depend essentially on the pair $\ac\rf$, which number $N_2 = 10$.
The rules of order at most two belonging to $\mathcal{T}_{\left\{\af,\rf,\rc\right\}}$ are those whose essential arguments form one of the $1 + 3 + 3$ subsets of $\left\{\af,\rf,\rc\right\}$ with at most two elements; they number $N_0 + 3N_1 + 3N_2 = 38$.
Of these, the $16$ rules of type $\af\rc$ are shared with $\mathcal{T}_{\left\{\af,\ac,\rc\right\}}$, leaving $38 - 16 = 22$ exclusive ones; by symmetry, the same counts hold for $\mathcal{T}_{\left\{\af,\ac,\rc\right\}}$.

Consider next the rules of essential order three.
Those depending essentially on all of $\af$, $\rf$, $\rc$ number $N_3 = 218$ and complete the count $38 + 218 = 256$ of $\mathcal{T}_{\left\{\af,\rf,\rc\right\}}$; symmetrically for $\mathcal{T}_{\left\{\af,\ac,\rc\right\}}$.
The two remaining triples define the type classes $\mathcal{T}_{\left\{\af,\ac,\rf\right\}}$ and $\mathcal{T}_{\left\{\ac,\rf,\rc\right\}}$, which contain the pair $\ac\rf$ and were not previously explored.
Their $2 N_3 = 436$ essentially third-order rules, together with the $10$ rules that depend essentially on the pair $\ac\rf$, account for all $|O_3| - 496 = 446$ rules of order at most three lying outside the previously studied spaces.
The remaining $N_4 = \numprint{64594}$ rules of the universe $O_4$ are of essential order four.

All of these counts were additionally verified by exhaustive enumeration of the $2^{16}$ assessment rules.
In \ref{MT-fig:space}, the area of every region is proportional, within each panel, to the number of assessment rules it contains.

\subsection{Resident reputation dynamics and equilibrium}\label{subsec:resident-dynamics}
Throughout, we adopt a public-assessment model of reputation\citep{si-ohtsuki:JTB:2004a,si-ohtsuki:JTB:2006a}: a single observer witnesses each interaction and assigns each of the two players a new reputation, and these assessments are then shared with the whole population (for instance, through gossip or a central authority), so that every individual holds the same view of any given player.
This shared, population-wide agreement is what allows the population's state to be summarized by a single variable -- the fraction of individuals in good standing -- whose dynamics we now derive.

Consider a resident population implementing moral system $(p,d)$: all players behave according to strategy $p$ and observers judge according to assessment rule $d$.
Within this population, take a focal player with reputation $\rf \in \left\{\B,\G\right\}$ playing a game against a co-player with reputation $\rc \in \left\{\B,\G\right\}$; the resident strategy then prescribes that the focal player play $\prfrc$ and the co-player play $\prcrf$.
Given our labeling convention \eqref{C=G=1,D=B=0}, the probability that the focal player is assigned a good reputation can be written as
\begin{align} \label{deltaij}
 \delta_{\rf\rc} & = \chi + (1-2\chi) \left\{ (1-\epsilon)^2 \mathsf{d}_{\rf\rc} \right. \nonumber \\
 & \left. + \epsilon (1-\epsilon) \left[ d(\rf,\rc,\D,\prcrf) + d(\rf,\rc,\prfrc,\D) \right] \right. \nonumber \\
 & \left. + \epsilon^2 d(\rf,\rc,\D,\D) \right\},
\end{align}
where $\chi \in (0,1/2)$ is the probability that the observer makes an assignment error\citep{si-ohtsuki:JTB:2004a}, and $\epsilon \in [0,1/2)$ is the probability that a player makes an execution error\citep{si-ohtsuki:JTB:2004a}.

Let $g_{p,d}$ denote the proportion of individuals with a good reputation in the resident population.
Following the same assumptions and arguments as in previous studies\citep{si-ohtsuki:JTB:2004a,si-Nakamura2014}, $g_{p,d}$ changes according to the continuous-time dynamics
\begin{align} \label{gdot}
 \dot{g}_{p,d} & = \delta_{\G\G} g_{p,d}^2 + \left(\delta_{\B\G}+\delta_{\G\B}\right) g_{p,d}(1-g_{p,d}) + \delta_{\B\B} (1-g_{p,d})^2 - g_{p,d} .
\end{align}
Setting the right-hand side to zero and solving for $g_{p,d}$ yields the (unique) resident reputation equilibrium (ref.~\citenum{si-Nakamura2014}, Eq.~3)
\begin{equation} \label{g*}
 g_{p,d}^* =
\left\{
	\begin{array}{ll}
		\frac{-b-\sqrt{b^2-4 a c}}{2 a}  & \mbox{if } a \neq 0 \\
		-\frac{c}{b} & \mbox{if } a = 0
	\end{array}
\right.,
\end{equation}
where
\begin{subequations} \label{abc}
 \begin{align}
  a & = \delta_{\B\B} + \delta_{\G\G} - \delta_{\B\G} - \delta_{\G\B}, \\
  b & = \delta_{\B\G} + \delta_{\G\B} - 2 \delta_{\B\B} - 1, \\
  c & = \delta_{\B\B}.
 \end{align}
\end{subequations}
This equilibrium is asymptotically stable\citep{si-ohtsuki:JTB:2004a,si-Nakamura2014}.

\subsection{Cooperation index}
\label{subsec:cooperation-index}

We define the cooperation index $\kappa_{p,d}$ of moral system $(p,d)$ as the probability that an individual in a population implementing $(p,d)$ cooperates at equilibrium.
This probability depends on the execution and assignment error rates $\epsilon$ and $\chi$, and is given by
\begin{align} \label{coop-index}
 \kappa_{p,d} & = (1-\epsilon)\left[g^2 \pGG + g(1-g) \left(\pGB + \pBG \right) + (1-g)^2 \pBB \right]_{g=g^*_{p,d}} .
\end{align}

\subsection{Mutant reputation dynamics and equilibrium}\label{subsec:mutant-dynamics}

Suppose the resident population has settled at the reputation equilibrium $g_{p,d}^*$ \eqref{g*}, and that a vanishingly small fraction of individuals (`mutants') now switch to a strategy $\tilde{p} \neq p$.
Since mutants are rare, they interact only with residents.
A focal mutant with reputation $\rf$ interacting with a resident co-player with reputation $\rc$ is assigned a good reputation with probability
\begin{align} \label{tildedeltaij}
 \tilde{\delta}_{\rf\rc} & = \chi + (1-2\chi) \left\{ (1-\epsilon)^2 \tilde{\mathsf{d}}_{\rf\rc} + \epsilon (1-\epsilon) \left[ d(\rf,\rc,\D,\prcrf) + d(\rf,\rc,\prfrcm,\D) \right] \right. \nonumber \\
 & \left. + \epsilon^2 d(\rf,\rc,\D,\D) \right\} .
\end{align}
The proportion $h_{\tilde{p},p,d}$ of mutants with a good reputation then changes according to the continuous-time dynamics
\begin{align} \label{hdot}
 \dot{h}_{\tilde{p},p,d} & = h_{\tilde{p},p,d} g^*_{p,d} \tilde{\delta}_{\G\G} + h_{\tilde{p},p,d} (1-g^*_{p,d}) \tilde{\delta}_{\G\B} \nonumber \\
 & + (1-h_{\tilde{p},p,d}) g^*_{p,d} \tilde{\delta}_{\B\G} + (1-h_{\tilde{p},p,d})(1-g^*_{p,d}) \tilde{\delta}_{\B\B} - h_{\tilde{p},p,d} .
\end{align}
Solving for the (unique) mutant reputation equilibrium (cf. ref.~\citenum{si-Nakamura2014}, Eq.~8) yields
\begin{align} \label{h*}
 h^*_{\tilde{p},p,d} & = \frac{\tilde{\delta}_{\B\B} + \left(\tilde{\delta}_{\B\G} - \tilde{\delta}_{\B\B}\right) g^*_{p,d}}{1+\tilde{\delta}_{\B\B}-\tilde{\delta}_{\G\B}+\left(\tilde{\delta}_{\B\G}+\tilde{\delta}_{\G\B}-\tilde{\delta}_{\B\B}-\tilde{\delta}_{\G\G}\right)g^*_{p,d}},
\end{align}
which is asymptotically stable\citep{si-ohtsuki:JTB:2004a,si-Nakamura2014}.

\subsection{Equilibrium payoffs and evolutionary stability}\label{subsec:payoffs-stability}

Consider a polymorphic population composed of residents playing strategy $p$ and vanishingly rare mutants playing strategy $\tilde{p}$, settled at the reputation equilibrium $(g^*_{p,d},h^*_{\tilde{p},p,d})$;
that is, the proportion of good residents is $g^*_{p,d}$ and the proportion of good mutants is $h^*_{\tilde{p},p,d}$.

Since mutants are rare, the expected payoff to a resident depends only on their interactions with other residents.
Taking into account execution errors occurring with probability $\epsilon$, the expected payoff of a resident with reputation $\rf$ playing against a resident with reputation $\rc$ is given by
\begin{align} \label{omegaij}
 \omega_{\rf\rc} & = (1-\epsilon)^2 v(\prfrc,\prcrf) + \epsilon (1-\epsilon) \left[v(\prfrc,\D) + v(\D,\prcrf)\right] + \epsilon^2 \vDD .
\end{align}
Using this expression, the expected payoff to a resident at the resident reputation equilibrium $g^*_{p,d}$ can then be written as
\begin{align} \label{w}
 w_{p,d} & = \left[(1-g)^2 \omega_{\B\B} + g(1-g)\left(\omega_{\B\G} + \omega_{\G\B} \right)+ g^2 \omega_{\G\G}\right]_{g=g^*_{p,d}} .
\end{align}

Again, since mutants are rare, the expected payoff to a mutant depends only on their interactions with residents.
The expected payoff of a mutant with reputation $\rf$ playing against a resident with reputation $\rc$ is thus given by
\begin{align} \label{tildeomega}
 \tilde{\omega}_{\rf\rc} & = (1-\epsilon)^2 v(\prfrcm,\prcrf) + \epsilon (1-\epsilon) \left[v(\prfrcm,\D) + v(\D,\prcrf)\right] + \epsilon^2 \vDD .
\end{align}
Using this expression, the expected payoff to a mutant at the reputation equilibrium $(g^*_{p,d},h^*_{\tilde{p},p,d})$ can be written as
\begin{align} \label{tildew}
 \tilde{w}_{\tilde{p},p,d} & = \left[ (1-h)(1-g) \tilde{\omega}_{\B\B} + (1-h)g \tilde{\omega}_{\B\G} + h(1-g) \tilde{\omega}_{\G\B} + h g \tilde{\omega}_{\G\G} \right]_{g=g^*_{p,d},h=h^*_{\tilde{p},p,d}} .
\end{align}

The evolutionary stability of strategy $p$ under assessment rule $d$ can be determined once the resident payoff and the payoffs of all possible mutant strategies have been calculated from \eqref{w} and \eqref{tildew}.
Under our solution concept (strict Nash equilibrium), $p$ is stable under $d$ if and only if
\begin{align} \label{evol-stab-cond}
w_{p,d} > \tilde w_{\tilde p,p,d}
\quad
\text{for all }
\tilde p \in \mathcal P \setminus \{p\}
\end{align}
holds, where $\mathcal{P}$ is the set of all strategies we consider (see also Eq.~(\ref{MT-stability-condition}) in the main text).
For brevity, we henceforth say that $(p,d)$ is stable to mean that $p$ is stable under $d$.

\subsection{Stability indices}
\label{subsec:stability-indices}

Let $\mathcal{G}_\Gamma$ denote the normalized two-dimensional payoff region corresponding to game class $\Gamma \in \{\mathrm{PD},\mathrm{SG},\mathrm{SH}\}$ (see Eqs.~\eqref{normalized-PD}--\eqref{normalized-SH}), each point $\bm{\pi} \in \mathcal{G}_\Gamma$ representing a normalized payoff configuration.
For fixed error rates $(\epsilon,\chi)$ and a moral system $(p,d)$, define
\[
 \mathcal{S}_{p,d}^\Gamma(\epsilon,\chi) := \left\{\, \bm{\pi} \in \mathcal{G}_\Gamma \; : \; p \text{ is stable under } d \text{ at payoff } \bm{\pi} \,\right\},
\]
where stability at a given payoff configuration is understood in the sense of condition~\eqref{evol-stab-cond}.

The stability index of $p$ under $d$ in class $\Gamma$ is
\[
 \sigma_{p,d}^\Gamma(\epsilon,\chi) := \frac{\lambda_2\!\left(\mathcal{S}_{p,d}^\Gamma(\epsilon,\chi)\right)}{\lambda_2\!\left(\mathcal{G}_\Gamma\right)},
\]
where $\lambda_2$ denotes the two-dimensional Lebesgue measure on the payoff parameter space.
Thus $\sigma_{p,d}^\Gamma(\epsilon,\chi) \in [0,1]$ is the normalized area of the payoff region within class $\Gamma$ for which $p$ is stable under $d$ -- equivalently, the probability that $p$ is stable under $d$ when a payoff configuration is drawn uniformly at random from~$\mathcal{G}_\Gamma$.
For instance, $\sigmaSG_{p,d}=1$ indicates that $p$ is stable under $d$ for every snowdrift game in the admissible parameter region, whereas $\sigmaPD_{p,d}=1/2$ indicates stability on exactly one half of the normalized prisoner's dilemma payoff space.

To facilitate comparison of moral systems across the three classes of cooperative dilemmas, we also report the average stability index
\begin{equation} \label{avg-sigma}
\sigmabar_{p,d} =  \frac{1}{3} \sigmaPD_{p,d} + \frac{1}{3} \sigmaSG_{p,d} + \frac{1}{3} \sigmaSH_{p,d} .
\end{equation}
This quantity assigns equal weight to each of the three dilemma classes, regardless of the differing areas of their payoff regions.

Below, we report $\sigma^{\Gamma}_{p,d}$ for each $\Gamma \in \left\{\mathrm{PD},\mathrm{SG},\mathrm{SH}\right\}$ and the average stability index~\eqref{avg-sigma} for certain moral systems $(p,d)$ in the absence of execution errors ($\epsilon=0$) and in the limit of vanishing assignment errors ($\chi \to 0$).
Writing $A_{\Gamma,p,d} := \lambda_2\!\left(\mathcal{S}_{p,d}^\Gamma\right)$ and $A_\Gamma := \lambda_2\!\left(\mathcal{G}_\Gamma\right)$ for the corresponding areas, these stability indices can be expressed as
\begin{equation} \label{sigma-gamma}
 \lim_{\chi \to 0} \left(\sigma^\Gamma_{p,d}\right)_{\epsilon=0} = \frac{\lim_{\chi \to 0} \left(A_{\Gamma,p,d}\right)_{\epsilon=0}}{A_{\Gamma}}.
\end{equation}

In the cases we study, the stability regions $\mathcal{S}_{p,d}^\Gamma$ are polygons whose vertices can be determined analytically; their areas are then computed with the shoelace formula:
\begin{equation} \label{shoelace}
 A = \frac{1}{2} \sum_{i=0}^{n-1} \left(x_{i}y_{i+1}-x_{i+1}y_i\right),
\end{equation}
where $(x_0,y_0), (x_1,y_1), \ldots, (x_n,y_n)$ are the vertices listed in counter-clockwise order, with $x_n=x_0$ and $y_n=y_0$.

This polygonal structure of $\mathcal{S}_{p,d}^\Gamma$, on which the analytical stability indices reported below rely, holds in the absence of execution errors and in the limit of vanishing assignment errors ($\epsilon=0$, $\chi \to 0$).
At the finite error rates used for our numerical results ($\epsilon=\chi=0.02$ in the main text, Fig.~\ref{MT-fig:scatter}; and the error-rate grid of Section~\ref{sec:supp-numerical}), the stability regions $\mathcal{S}_{p,d}^\Gamma$ are generally no longer polygons.
There, we estimate each stability index by Monte Carlo integration, as detailed in Section~\ref{subsec:numerical-pipeline}.

\subsubsection{Areas of the payoff regions}

We compute the total area $A_\Gamma = \lambda_2(\mathcal{G}_\Gamma)$ for each game class.
In each case, $\mathcal{G}_\Gamma$ is a convex polygon determined by the normalization constraints (Eqs.~\eqref{normalized-PD}--\eqref{normalized-SH}).

\paragraph{Prisoner's dilemmas.}

The two free parameters are $R$ and $P$, with constraints given by~\eqref{normalized-PD}.
The region $\mathcal{G}_{\mathrm{PD}}$ (\ref{MT-fig:payoff_space}a) is a quadrilateral whose vertices, in $(R,P)$ coordinates, are
\begin{equation*}
  \left(\frac{1}{2},0\right), \left(1,0\right), \left(1,1\right), \left(\frac{1}{2},\frac{1}{2}\right).
\end{equation*}
Applying the shoelace formula~\eqref{shoelace} gives
\begin{equation} \label{A-PD}
  A_{\mathrm{PD}} = \frac{3}{8} .
\end{equation}

\paragraph{Snowdrift games.}

The two free parameters are $R$ and $S$, with constraints given by~\eqref{normalized-SG}.
The region $\mathcal{G}_{\mathrm{SG}}$ (\ref{MT-fig:payoff_space}b) is a triangle whose vertices, in $(R,S)$ coordinates, are
\begin{equation*}
  \left(\frac{1}{2},0\right), \left(1,0\right), \left(1,1\right).
\end{equation*}
Applying the shoelace formula~\eqref{shoelace} gives
\begin{equation} \label{A-SG}
  A_{\mathrm{SG}} = \frac{1}{4} .
\end{equation}

\paragraph{Stag hunts.}

The two free parameters are $T$ and $P$, with constraints given by~\eqref{normalized-SH}.
The region $\mathcal{G}_{\mathrm{SH}}$ (\ref{MT-fig:payoff_space}c) is a triangle whose vertices, in $(T,P)$ coordinates, are
\begin{equation*}
 \left(0,0\right), \left(1,0\right), \left(1,1\right).
\end{equation*}
Applying the shoelace formula~\eqref{shoelace} gives
\begin{equation} \label{A-SH}
  A_{\mathrm{SH}} = \frac{1}{2} .
\end{equation}

\subsubsection{Stability indices for two recurring binding conditions}

Only two non-trivial stability regions arise in the stability analyses that follow (Sections~\ref{subsec:stab-anal-disc} and~\ref{subsec:stab-anal-leading-eight}); we compute their stability indices here.

\paragraph{Snowdrift games with $R+P>2S$.}

Consider a snowdrift game~\eqref{SG} and a moral system $(p,d)$ such that $p$ is stable under $d$ for $\epsilon=0$ and $\chi \to 0$ if and only if $R+P>2S$ holds.
Under the normalization~\eqref{normalized-SG}, this condition becomes $R>2S$, so the stability region is
\[
 \mathcal{S}_{p,d}^{\mathrm{SG}} = \left\{\, \bm{\pi} \in \mathcal{G}_{\mathrm{SG}} \; : \; R > 2S \,\right\},
\]
a quadrilateral (Fig.~\ref{fig:lemma-regions}a) whose vertices, in $(R,S)$ coordinates, are
\begin{equation*}
  \left(\frac{1}{2},0\right), \left(1,0\right), \left(1,\frac{1}{2}\right), \left(\frac{2}{3},\frac{1}{3}\right).
\end{equation*}
Applying the shoelace formula~\eqref{shoelace} gives
\begin{equation} \label{A-SG-R+P>2S}
 \lim_{\chi \to 0} \left( A_{\mathrm{SG},p,d} \right)_{\epsilon=0} = \frac{1}{6}.
\end{equation}
Substituting~\eqref{A-SG-R+P>2S} and~\eqref{A-SG} into~\eqref{sigma-gamma} yields the following lemma.

\begin{lemma}[Stability index across snowdrift games with binding condition $R+P>2S$] \label{lem:sigma-SG-R+P>2S}
Let $\epsilon=0$ and $\chi \to 0$.
If $p$ is stable under $d$ across snowdrift games if and only if $R+P>2S$, then $\lim_{\chi \to 0} \left( \sigmaSG_{p,d} \right)_{\epsilon=0} = 2/3$.
\end{lemma}

\paragraph{Prisoner's dilemmas with $2R>T+P$.}

Consider a prisoner's dilemma~\eqref{PD} and a moral system $(p,d)$ such that $p$ is stable under $d$ for $\epsilon=0$ and $\chi \to 0$ if and only if $2R>T+P$ holds.
Under the normalization~\eqref{normalized-PD}, this condition becomes $P<2R-1$, so the stability region is
\[
 \mathcal{S}_{p,d}^{\mathrm{PD}} = \left\{\, \bm{\pi} \in \mathcal{G}_{\mathrm{PD}} \; : \; P < 2R - 1 \,\right\},
\]
a triangle (Fig.~\ref{fig:lemma-regions}b) whose vertices, in $(R,P)$ coordinates, are
\begin{equation*}
  \left(\frac{1}{2},0\right), \left(1,0\right), \left(1,1\right).
\end{equation*}
Applying the shoelace formula~\eqref{shoelace} gives
\begin{equation} \label{A-PD-2R>T+P}
 \lim_{\chi \to 0} \left( A_{\mathrm{PD},p,d} \right)_{\epsilon=0} = \frac{1}{4} .
\end{equation}
Substituting~\eqref{A-PD-2R>T+P} and~\eqref{A-PD} into~\eqref{sigma-gamma} yields the following lemma.

\begin{lemma}[Stability index across prisoner's dilemmas when $2R>T+P$ binds] \label{lem:sigma-PD-2R>T+P}
Let $\epsilon=0$ and $\chi \to 0$.
If $p$ is stable under $d$ across prisoner's dilemmas if and only if $2R>T+P$, then $\lim_{\chi \to 0} \left( \sigmaPD_{p,d} \right)_{\epsilon=0} = 2/3$.
\end{lemma}

\begin{figure}[tp]
    \centering
    \includegraphics{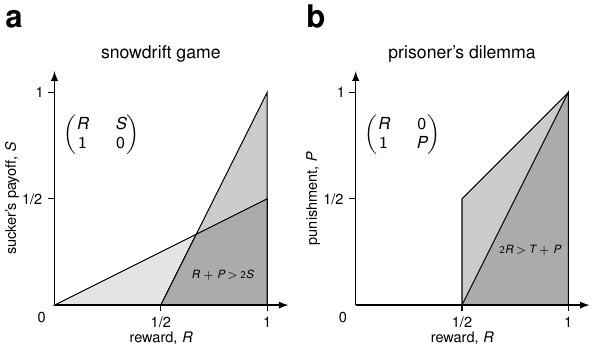}
    \caption{Illustration of Lemmas~\ref{lem:sigma-SG-R+P>2S} and~\ref{lem:sigma-PD-2R>T+P}. \textbf{a}, The region where $R+P>2S$ is superimposed on the region of snowdrift games, in normalized payoff space (Lemma~\ref{lem:sigma-SG-R+P>2S}). \textbf{b}, The region where $2R>T+P$ is superimposed on the region of prisoner's dilemmas in normalized payoff space (Lemma~\ref{lem:sigma-PD-2R>T+P}).\label{fig:lemma-regions}}
\end{figure}

\subsubsection{Weighted average stability indices}
\label{subsubsec:weighted-sigma}

More generally, we can define a weighted average stability index
\begin{equation*}
 \sigmabarw_{p,d} = \mathsf{w}_\mathrm{PD} \sigmaPD_{p,d} + \mathsf{w}_\mathrm{SG} \sigmaSG_{p,d} + \mathsf{w}_\mathrm{SH} \sigmaSH_{p,d}
\end{equation*}
for weights $\mathsf{w}_\mathrm{PD},\mathsf{w}_\mathrm{SG},\mathsf{w}_\mathrm{SH} > 0$ with $\mathsf{w}_\mathrm{PD}+\mathsf{w}_\mathrm{SG}+\mathsf{w}_\mathrm{SH}=1$.
The average stability index~\eqref{avg-sigma} is the particular case of equal weights, $\mathsf{w}_\mathrm{PD}=\mathsf{w}_\mathrm{SG}=\mathsf{w}_\mathrm{SH}=1/3$.
An alternative weights each class in proportion to the area of its payoff region,
\begin{equation*}
 \mathsf{w}_\Gamma = \frac{A_{\mathrm{\Gamma}}}{\sum_{\Gamma}A_{\mathrm{\Gamma}}} \ \text{ for } \Gamma \in \left\{\mathrm{PD},\mathrm{SG},\mathrm{SH}\right\},
\end{equation*}
which, using the areas $A_\mathrm{PD}=3/8$, $A_\mathrm{SG}=1/4$, $A_\mathrm{SH}=1/2$ computed above, gives
\begin{equation*}
 \mathsf{w}_{\mathrm{PD}} = \frac{1}{3}, \quad \mathsf{w}_{\mathrm{SG}} = \frac{2}{9}, \quad \text{and} \quad \mathsf{w}_{\mathrm{SH}} = \frac{4}{9}.
\end{equation*}
The prisoner's dilemma thus receives the same weight, $1/3$, under both the equal-weight and area-weighted schemes; the two differ only in how they split the remaining weight between snowdrift games and stag hunts.
We adopt this area-weighted index as a robustness check on our main results (Figs.~\ref{fig:scatter_zoom_area} and~\ref{fig:scatter_panel_area}).

\section{Definitions and preliminary results}

Here we introduce terminology for the subsets of strategies and moral systems that recur in our analysis.
These definitions provide the conceptual scaffolding for the theoretical results of subsequent sections.

In Section~\ref{subsec:def-conditional-discriminating}, we distinguish strategies according to whether they respond to reputation information.
Unconditional strategies ($\ALLC$ and $\ALLD$) prescribe the same action regardless of the two players' reputations, whereas conditional strategies allow behavior to depend on this information.
Among conditional strategies, we single out the discriminating strategies, which cooperate with co-players regarded as good and defect against those regarded as bad (provided the focal player is in good standing), and the contrite strategies, which cooperate whenever the focal player is in bad standing.
We extend these notions to moral systems $(p,d)$, defining conditional, discriminating, and contrite systems accordingly.

In Section~\ref{subsec:def-self-approving-righteous}, we characterize moral systems by their equilibrium properties at the population level.
A moral system is self-approving if a resident population comes to regard itself as good, and self-cooperative if it sustains mutual cooperation, both in the absence of execution errors and the limit of vanishing assignment errors.
Combining these, a system is righteous if it is both self-approving and self-cooperative.
Finally, a system is universally stable if it resists invasion by any alternative strategy across all three classes of cooperative dilemma.

In Section~\ref{subsec:def-consistent}, we introduce consistent systems, in which the assessment rule endorses as good exactly those who follow the action rule and penalizes those who deviate from it.
Consistent discriminating systems are central to our analysis; we describe their four subsets, one for each discriminating strategy.

In Section~\ref{subsec:leading8}, we define the leading-eight moral systems of Ohtsuki and Iwasa\citep{si-ohtsuki:JTB:2004a}, which serve as reference points throughout our analysis.
Finally, in Section~\ref{subsec:mirror-systems}, we formalize the mirror operation -- the swapping of the two reputation labels -- and prove the preliminary results used in later sections: a moral system and its mirror are behaviorally equivalent, whereas label-dependent properties such as self-approval hold for at most one member of each mirror pair.

\subsection{Conditional, discriminating, and contrite strategies and systems}\label{subsec:def-conditional-discriminating}

We begin by defining unconditional and conditional strategies.

\begin{definition}[Unconditional strategies]
A strategy is said to be \emph{unconditional} if $\pBB=\pBG=\pGB=\pGG$ holds.
\end{definition}

There are exactly two unconditional strategies in our strategy space: the `unconditional defector' ($\ALLD$), with $\pBB = \pBG = \pGB = \pGG = \D$, and the `unconditional cooperator' ($\ALLC$), with $\pBB = \pBG = \pGB = \pGG = \C$.

\begin{definition}[Conditional strategies] \label{def:cond}
A strategy is said to be \emph{conditional} if it is not unconditional.
\end{definition}

Among conditional strategies, a particularly important subset consists of those that, when the focal player is in good standing ($r_1=\G$), prescribe cooperation toward co-players regarded as good ($r_2=\G$) and defection against those regarded as bad ($r_2=\B$).
We define these strategies formally as follows.

\begin{definition}[Discriminating strategies] \label{def:disc}
A strategy $p$ is said to be \emph{discriminating} if both
\begin{equation*}
 \pGB = \D \ \text{ and } \ \pGG = \C
\end{equation*}
hold.
\end{definition}

As a string, discriminating strategies take the form $p = \wildcard\wildcard\D\C$.
Depending on what they prescribe when the focal player is in bad standing, we distinguish four discriminating strategies:
\begin{enumerate}
 \item `(Standard) discriminator' ($\DISC$, $p=\D\C\D\C$), which, when in bad standing, discriminates by the co-player's reputation just as in good standing.
 \item `Contrite discriminator' ($\CDISC$, $p=\C\C\D\C$), which, when in bad standing, cooperates regardless of the co-player's reputation.
 \item `Obfuscated discriminator' ($\ODISC$, $p=\C\D\D\C$), which, when in bad standing, reverses the discriminating logic, cooperating with bad co-players and defecting against good ones.
 \item `Petulant discriminator' ($\PDISC$, $p=\D\D\D\C$), which, when in bad standing, defects regardless of the co-player's reputation.
\end{enumerate}

A further property central to our analysis concerns strategies that prescribe cooperation whenever the focal player is in bad standing ($r_1 = \B$), a readiness to make amends that we call contrition.

\begin{definition}[Contrite strategies] \label{def:contrite}
A strategy $p$ is said to be \emph{contrite} if both
\begin{equation*}
 \pBB = \C \ \text{ and } \ \pBG = \C
\end{equation*}
hold.
\end{definition}

As a string, contrite strategies take the form $p = \C\C\wildcard\wildcard$.
Among the four discriminating strategies, $\CDISC$ ($p=\CCDC$) is the only one that is also contrite; it is thus the unique contrite discriminator.

All $2^{4}=16$ strategies, along with their conditional, discriminating, and contrite status, are shown in Table~\ref{tab:strategies}.

\begin{table}
\centering
\begin{tabular}{c|c|c|c|c|c}
ID & Bitstring & Name & Conditional? & Discriminating? & Contrite? \\
\hline
0  & $\DDDD$ & $\ALLD$ & \xmark & \xmark & \xmark \\
1  & $\CDDD$ &       & \cmark & \xmark & \xmark \\
2  & $\DCDD$ &       & \cmark & \xmark & \xmark \\
3  & $\CCDD$ &       & \cmark & \xmark & \cmark \\
4  & $\DDCD$ &       & \cmark & \xmark & \xmark \\
5  & $\CDCD$ &       & \cmark & \xmark & \xmark \\
6  & $\DCCD$ &       & \cmark & \xmark & \xmark \\
7  & $\CCCD$ &       & \cmark & \xmark & \cmark \\
8  & $\DDDC$ & $\PDISC$ & \cmark & \cmark & \xmark \\
9  & $\CDDC$ & $\ODISC$ & \cmark & \cmark & \xmark \\
10 & $\DCDC$ & $\DISC$ & \cmark & \cmark & \xmark \\
11 & $\CCDC$ & $\CDISC$ & \cmark & \cmark & \cmark \\
12 & $\DDCC$ &       & \cmark & \xmark & \xmark \\
13 & $\CDCC$ &       & \cmark & \xmark & \xmark \\
14 & $\DCCC$ &       & \cmark & \xmark & \xmark \\
15 & $\CCCC$ & $\ALLC$ & \xmark & \xmark & \cmark \\
\end{tabular}
\caption{The 16 strategies (action rules). Each strategy is a Boolean function written as a bitstring $p = \pBB\pBG\pGB\pGG$ in the order $\B\B$, $\B\G$, $\G\B$, $\G\G$; its ID is this bitstring read as a little-endian binary number ($\C=\G=1$, $\D=\B=0$, with $\B\B$ the least significant bit). The last three columns indicate whether the strategy is conditional, discriminating, and contrite (\cmark) or not (\xmark).}
\label{tab:strategies}
\end{table}

Our general goal is to identify moral systems $(p,d)$ that sustain high levels of cooperation (as measured by the cooperation index) and resist invasion by every alternative strategy $\tilde{p} \neq p$.
Neither unconditional strategy yields such a system, whatever the assessment rule $d$: $\ALLD$ ($p=\DDDD$) sustains no cooperation, since individuals following it never cooperate, while $\ALLC$ ($p=\CCCC$) is invaded by $\ALLD$ in prisoner's dilemmas and snowdrift games\citep{si-Pena2023}.
For these reasons, we focus on moral systems implementing a conditional strategy.

\begin{definition}[Conditional systems] \label{def:conditional}
A moral system $(p,d)$ is said to be \emph{conditional} if its action rule $p$ is conditional.
\end{definition}

We denote the set of conditional systems by $\Cset$.
Explicitly,
\begin{equation*}
 \Cset = \left\{(p,d) \left| p \notin \left\{\ALLC,\ALLD \right\} \right. \right\} .
\end{equation*}

Since conditional systems implementing a discriminating strategy play a central role in our analysis, we single them out with the following definition.

\begin{definition}[Discriminating systems] \label{def:discriminating}
A moral system $(p,d)$ is said to be \emph{discriminating} if its action rule $p$ is discriminating.
\end{definition}

We denote the set of discriminating systems by $\Dset$.
Explicitly,
\begin{equation*}
 \Dset = \left\{(p,d) \left| p \in \left\{ \DISC,\CDISC,\ODISC,\PDISC \right\} \right. \right\} .
\end{equation*}

Analogously, we single out systems whose action rule is contrite.

\begin{definition}[Contrite systems] \label{def:contrite-system}
A moral system $(p,d)$ is said to be \emph{contrite} if its action rule $p$ is contrite.
\end{definition}

Since $\CDISC$ is the unique contrite discriminating strategy, a discriminating system is contrite if and only if its action rule is $\CDISC$.

\subsection{Self-approving, self-cooperative, righteous, and universally stable systems}
\label{subsec:def-self-approving-righteous}

We begin by defining properties of moral systems related to their reputation equilibria.

\begin{definition}[Self-approving systems] \label{def:self-approving}
A conditional system $(p,d)$ is said to be \emph{self-approving} if
its resident reputation equilibrium in the absence of execution errors approaches one in the limit of vanishing assignment errors,
that is, if
\begin{equation} \label{self-approving-cond}
 \lim_{\chi \to 0} \left(g^*_{p,d}\right)_{\epsilon=0} = 1
\end{equation}
holds.
If $(p,d)$ is self-approving, we also say that $p$ is self-approving under $d$.
\end{definition}

We denote the set of self-approving systems by $\Aset$.
Thus,
\begin{equation*}
 \Aset = \left\{(p,d) \left| (p,d) \in \Cset \text{ and } \lim_{\chi \to 0} \left(g^*_{p,d}\right)_{\epsilon=0} = 1 \right. \right\} .
\end{equation*}

We also define the converse property, in which a resident population comes to regard itself as bad at equilibrium.

\begin{definition}[Self-disapproving systems] \label{def:self-disapproving}
A conditional system $(p,d)$ is said to be \emph{self-disapproving} if
its resident reputation equilibrium in the absence of execution errors approaches zero in the limit of vanishing assignment errors,
that is, if
\begin{equation} \label{self-disapproving-cond}
 \lim_{\chi \to 0} \left(g^*_{p,d}\right)_{\epsilon=0} = 0
\end{equation}
holds.
If $(p,d)$ is self-disapproving, we also say that $p$ is self-disapproving under $d$.
\end{definition}

We denote the set of self-disapproving systems by $\Zset$.
Thus,
\begin{equation*}
 \Zset = \left\{(p,d) \left| (p,d) \in \Cset \text{ and } \lim_{\chi \to 0} \left(g^*_{p,d}\right)_{\epsilon=0} = 0 \right. \right\} .
\end{equation*}

We now define an analogous property concerning cooperation levels rather than reputation equilibria, using the cooperation index (Section~\ref{subsec:cooperation-index}), which measures how frequently individuals cooperate.

\begin{definition}[Self-cooperative systems] \label{def:self-cooperative}
A conditional system $(p,d)$ is said to be \emph{self-cooperative}
if its cooperation index in the absence of execution errors, denoted by $\left(\kappa_{p,d}\right)_{\epsilon=0}$, approaches one in the limit of vanishing assignment errors,
that is, if
\begin{equation*}
 \lim_{\chi \to 0} \left(\kappa_{p,d}\right)_{\epsilon=0} = 1
\end{equation*}
holds.
If $(p,d)$ is self-cooperative, we also say that $p$ is self-cooperative under $d$.
\end{definition}

We denote the set of self-cooperative systems by $\Kset$.
Thus,
\begin{equation*}
 \Kset = \left\{(p,d) \left| (p,d) \in \Cset \text{ and } \lim_{\chi \to 0} \left(\kappa_{p,d}\right)_{\epsilon=0} = 1 \right. \right\} .
\end{equation*}

Combining self-approval and self-cooperation yields the notion of a righteous system.

\begin{definition}[Righteous systems] \label{def:righteous}
A conditional system $(p,d)$ is said to be \emph{righteous} if it is both self-approving and self-cooperative.
If $(p,d)$ is righteous, we also say that $p$ is righteous under $d$.
\end{definition}

We denote the set of righteous systems by $\Rset$.
Thus,
\begin{equation} \label{righteous-set}
 \Rset = \Aset \cap \Kset .
\end{equation}

A population implementing a righteous system $(p,d)$ is sensitive to reputation information (since $p$ is conditional),
mostly regards itself as good (since $(p,d)$ is self-approving), and mostly cooperates (since $(p,d)$ is self-cooperative).

The definitions above characterize a moral system's equilibrium behavior in a resident population. We now turn to its robustness against invasion by alternative strategies.

\begin{definition}[Universally stable systems] \label{def:universally-stable}
A moral system $(p,d)$ is said to be \emph{universally stable} if its stability indices in the absence of execution errors, denoted by $\left(\sigmaPD_{p,d}\right)_{\epsilon=0}$, $\left(\sigmaSG_{p,d}\right)_{\epsilon=0}$, and $\left(\sigmaSH_{p,d}\right)_{\epsilon=0}$, all approach one in the limit of vanishing assignment errors,
that is, if
\begin{equation*}
 \lim_{\chi \to 0} \left(\sigmaPD_{p,d}\right)_{\epsilon=0} = \lim_{\chi \to 0} \left(\sigmaSG_{p,d}\right)_{\epsilon=0} = \lim_{\chi \to 0} \left(\sigmaSH_{p,d}\right)_{\epsilon=0} = 1
\end{equation*}
holds.
If a moral system $(p,d)$ is universally stable, we also say that $p$ is universally stable under $d$.
\end{definition}

By Definition \ref{def:universally-stable}, if $(p,d)$ is universally stable, then $p$ is stable under $d$ across all cooperative dilemmas, in the sense of resisting invasion by any alternative strategy $\tilde{p} \neq p$ for any payoff matrix we consider, in the absence of execution errors ($\epsilon=0$) and in the limit of vanishing assignment errors ($\chi \to 0$).
If a system $(p,d)$ is universally stable, then its weighted average stability index satisfies
\[
 \lim_{\chi \to 0} \sigmabarw_{p,d} = 1
\]
for any vector of weights $(\mathsf{w}_\mathrm{PD},\mathsf{w}_\mathrm{SG},\mathsf{w}_\mathrm{SH})$.

\subsection{Consistent systems and consistent discriminating systems}\label{subsec:def-consistent}

\subsubsection{Consistent systems}

Consistency captures the idea that a moral system's action rule and assessment rule do not contradict each other:
following the prescribed behavior is rewarded with a good reputation, and deviating from it is punished with a bad one.

\begin{definition}[Consistent systems] \label{def:consistent}
A conditional system $(p,d)$ is said to be \emph{consistent}
if both
\begin{subequations} \label{si-consistent-cond}
\begin{align}
\dBB & = \dBG = \dGB = \dGG = \G, \text{ and } \label{si-consistent-cond-G} \\
\notdBB & = \notdBG = \notdGB = \notdGG = \B \label{si-consistent-cond-B}
\end{align}
\end{subequations}
hold.
If $(p,d)$ is consistent, we also say that $p$ is consistent with $d$, or that $d$ is consistent with $p$.
\end{definition}

We denote the set of consistent systems by $\Qset$.
Thus,
\begin{equation*}
 \Qset = \left\{(p,d) \left| (p,d) \in \Cset, \drfrc = \G, \text{and } \notdrfrc = \B \text{ for all } \rf, \rc \in \left\{\B,\G\right\} \right. \right\} .
\end{equation*}

Recall our notational conventions for $\drfrc$ \eqref{dr1r2-shorthand} and $\notdrfrc$ \eqref{notdr1r2-shorthand}.
Conditions \eqref{si-consistent-cond-G} and \eqref{si-consistent-cond-B} together say that,
whenever the co-player plays according to action rule $p$,
assessment rule $d$ regards the focal player as good if they follow $p$'s prescription and as bad if they deviate from it.
The moral system $(p,d)$ is then consistent in the sense that there is no contradiction between the behavior prescribed by the action rule and the moral judgment expressed by the assessment rule.

For each strategy $p$,
the two conditions in Eq. \eqref{si-consistent-cond} specify the values of eight out of the sixteen bits of an assessment rule that is consistent with $p$.
Once these values are set, the remaining eight bits can be set to any binary value, and $d$ will still be consistent with $p$, according to Definition \ref{def:consistent}.
It follows that each strategy $p$ gives rise to $2^8=256$ consistent systems $(p,d)$.

\begin{remark}[Equivalence of consistent systems at $\epsilon=0$] \label{rem:consistent-equivalence}
The eight unconstrained bits of a consistent assessment rule are precisely those of the form $d(r_1, r_2, a_1, \neg p(r_2,r_1))$ for $r_1, r_2 \in \{\B,\G\}$ and $a_1 \in \{\D,\C\}$ -- that is, the bits specifying how to assess a focal player when the co-player's action deviates from the resident prescription.
In the reputation transition probability~\eqref{deltaij}, these bits appear only in the terms proportional to $\epsilon(1-\epsilon)$ and~$\epsilon^2$, which vanish when $\epsilon = 0$.
The same holds for the mutant reputation transition~\eqref{tildedeltaij}: at $\epsilon=0$, the surviving term $d(r_1, r_2, \tilde{p}(r_1,r_2), p(r_2,r_1))$ involves a constrained bit regardless of whether the mutant's action agrees with the resident's, since conditions~\eqref{si-consistent-cond-G} and~\eqref{si-consistent-cond-B} together fix $d(r_1, r_2, a_1, p(r_2,r_1))$ for both values of~$a_1$.
Consequently, all $256$ consistent systems sharing the same action rule~$p$ have identical resident equilibria~$g^*_{p,d}$, mutant equilibria~$h^*_{\tilde{p},p,d}$, payoffs, and hence stability and cooperation indices, whenever $\epsilon = 0$.
\end{remark}

\subsubsection{Consistent discriminating systems}

Of special interest are consistent systems whose action rule is also discriminating.

\begin{definition}[Consistent discriminating systems]
\label{def:consistent-disc}
A conditional system $(p,d)$ is said to be \emph{consistent discriminating} if it is both discriminating and consistent.
\end{definition}

Thus,
\begin{equation*}
\QDset = \Qset \cap \Dset .
\end{equation*}
Depending on the action rule, $\QDset$ can be partitioned into four subsets:
\begin{enumerate}
 \item `Standard consistent discriminating systems' ($p = \DCDC$, subset $\QDisc$):
 \begin{align} \label{QDisc-set}
 \QDisc & = \left\{(p,d) \left| p = \DCDC, \dBBCD = \dBGDD = \dGBCC = \dGGDC = \B, \right. \right. \nonumber \\
 & \left. \text{ and } \dBBDD = \dBGCD = \dGBDC = \dGGCC = \G  \right\} .
 \end{align}
 \item `Contrite consistent discriminating systems' ($p = \CCDC$, subset $\QCDisc$):
 \begin{align} \label{QCDisc-set}
 \QCDisc & = \left\{(p,d) \left| p = \CCDC, \dBBDC = \dBGDD = \dGBCC = \dGGDC = \B, \right. \right. \nonumber \\
 & \left. \text{ and } \dBBCC = \dBGCD = \dGBDC = \dGGCC = \G \right\} .
 \end{align}
 \item `Obfuscated consistent discriminating systems' ($p = \CDDC$, subset $\QODisc$):
 \begin{align} \label{QODisc-set}
 \QODisc & = \left\{(p,d) \left| p = \CDDC, \dBBDC = \dBGCD = \dGBCD = \dGGDC = \B, \right. \right. \nonumber \\
 & \left. \text{ and } \dBBCC = \dBGDD = \dGBDD = \dGGCC = \G \right\} .
 \end{align}
 \item `Petulant consistent discriminating systems' ($p = \DDDC$, subset $\QPDisc$):
 \begin{align} \label{QPDisc-set}
 \QPDisc & = \left\{(p,d) \left| p = \DDDC, \dBBCD = \dBGCD = \dGBCD = \dGGDC = \B, \right. \right. \nonumber \\
 & \left. \text{ and } \dBBDD = \dBGDD = \dGBDD = \dGGCC = \G \right\} .
 \end{align}
\end{enumerate}

The action and assessment rules of each of these subsets of consistent discriminating systems are shown in Fig. \ref{fig:consistent-discriminating}.

\begin{figure}[tp]
    \centering
        \includegraphics{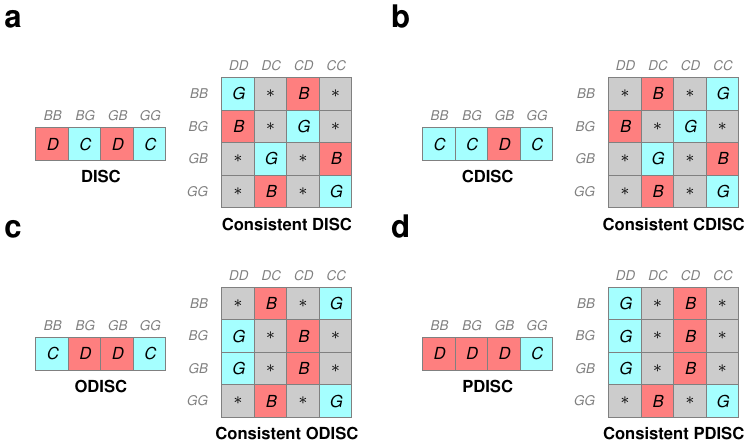}
    \caption{Action and assessment rules of the four subsets of consistent discriminating systems.}
    \label{fig:consistent-discriminating}
\end{figure}

\subsection{The leading eight}
\label{subsec:leading8}

The leading eight, first identified by Ohtsuki and Iwasa\citep{si-ohtsuki:JTB:2004a}, are a set of eight moral systems.
Each uses either $\DISC$ or $\CDISC$ as its action rule, together with a third-order assessment rule of type $\af \rf \rc$ -- one that ignores the co-player's action $\ac$ -- satisfying
\begin{subequations} \label{leading-eight-assessment}
\begin{align}
 \dBGD & = \dGGD = \B, \label{dBGD=dGGD=B} \\
 \dBGC & = \dGGC = \G, \label{dBGC=dGGC=G} \\
 \dGBD & = \G . \label{dGBD=G}
\end{align}
\end{subequations}
Conditions \eqref{dBGD=dGGD=B} and \eqref{dBGC=dGGC=G} state that,
irrespective of the focal player's reputation,
cooperation with good co-players is judged as good and defection against them as bad.
Condition \eqref{dGBD=G} states that defection against bad co-players by good individuals is judged as good,
encapsulating the concept of `justified defection'\citep{si-ohtsuki:JTB:2004a}.
The leading-eight moral systems are named and described in Fig. \ref{fig:leading-eight}.

\begin{figure}[tp]
    \centering
            \includegraphics{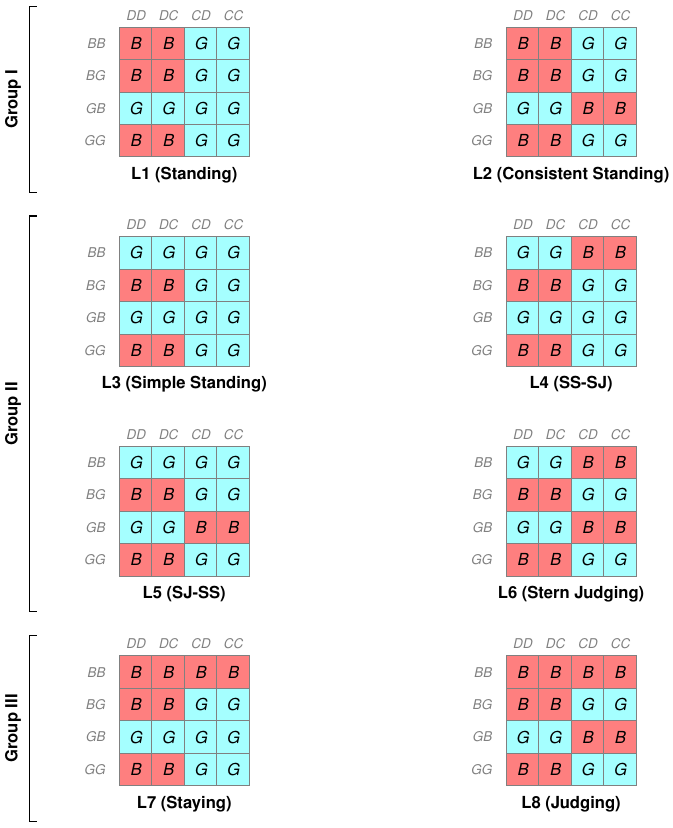}
    \caption{Social norms of the leading-eight moral systems, arranged by group. Group I systems (L1 and L2) use $\CDISC$ as their action rule; Group II (L3--L6) and Group III (L7 and L8) systems use $\DISC$.}
    \label{fig:leading-eight}
\end{figure}

The leading-eight moral systems differ in two respects:
(i) whether they use $\DISC$ or $\CDISC$ as their action rule, and
(ii) their assignment of the bits $\dBBD$, $\dBBC$, and $\dGBC$,
which govern how interactions involving bad co-players are judged.
In line with the original analysis by Ohtsuki and Iwasa\citep{si-ohtsuki:JTB:2004a}, we distinguish the following three groups.

\paragraph{Group I.}

Group I systems use $p=\CCDC$ (i.e.\ $\CDISC$), with $\dBBD=\B$ and $\dBBC=\G$.
We denote this set by $\LI$:
\begin{align} \label{LI-set}
 \LI & = \left\{(p,d) \left| p = \CCDC, \dBBD = \dBGD = \dGGD = \B, \right. \right. \nonumber \\
 & \left. \dBBC = \dBGC = \dGBD = \dGGC = \G, \text{ and } \dGBCD = \dGBCC \right\} .
\end{align}
The two members are L1 ($\CDISC$, Standing), with $\dGBC = \G$, and L2 ($\CDISC$, Consistent Standing), with $\dGBC = \B$
(see Fig.~\ref{fig:leading-eight}):
\begin{align}
 \text{L1} & = \left\{(p,d) \left| p = \CCDC, \dBBD = \dBGD = \dGGD = \B, \right. \right. \nonumber \\
 & \left. \text{ and } \dBBC = \dBGC = \dGBD = \dGBC = \dGGC = \G \right\}, \\
 \text{L2} & = \left\{(p,d) \left| p = \CCDC, \dBBD = \dBGD = \dGBC = \dGGD = \B, \right. \right. \nonumber \\
 & \left. \text{ and } \dBBC = \dBGC = \dGBD = \dGGC = \G \right\} . \label{L2}
\end{align}
Note that L2 belongs to $\QCDisc$ (cf.\ Eqs.~\eqref{L2} and~\eqref{QCDisc-set}).

\paragraph{Group II.}

Group II systems use $p=\DCDC$ (i.e.\ $\DISC$), with $\dBBD=\G$.
We denote this set by $\LII$:
\begin{align} \label{LII-set}
 \LII & = \left\{(p,d) \left| p = \DCDC, \dBGD = \dGGD = \B, \right. \right. \nonumber \\
 & \dBBD = \dBGC = \dGBD = \dGGC = \G, \nonumber \\
 & \left. \dBBCD = \dBBCC, \text{ and } \dGBCD = \dGBCC \right\} .
\end{align}
The four members are
L3 ($\DISC$, Simple Standing), with $\dBBC=\dGBC=\G$;
L4 ($\DISC$, SS-SJ), with $\dBBC=\B$ and $\dGBC=\G$;
L5 ($\DISC$, SJ-SS), with $\dBBC=\G$ and $\dGBC=\B$; and
L6 ($\DISC$, Stern Judging), with $\dBBC=\dGBC=\B$
(see Fig.~\ref{fig:leading-eight}):
\begin{align}
 \text{L3} & = \left\{(p,d) \left| p = \DCDC, \dBGD = \dGGD = \B, \right. \right. \nonumber \\
 & \left. \text{ and } \dBBD = \dBBC = \dBGC = \dGBD = \dGBC = \dGGC = \G \right\}, \\
 \text{L4} & = \left\{(p,d) \left| p = \DCDC, \dBBC = \dBGD = \dGGD = \B, \right. \right. \nonumber \\
 & \left. \text{ and } \dBBD = \dBGC = \dGBD = \dGBC = \dGGC = \G \right\}, \\
 \text{L5} & = \left\{(p,d) \left| p = \DCDC, \dBGD = \dGBC = \dGGD = \B, \right. \right. \nonumber \\
 & \left. \text{ and } \dBBC = \dBBD = \dBGC = \dGBD = \dGGC = \G \right\}, \\
 \text{L6} & = \left\{(p,d) \left| p = \DCDC, \dBBC = \dBGD = \dGBC = \dGGD = \B, \right. \right. \nonumber \\
 & \left. \text{ and } \dBBD = \dBGC = \dGBD = \dGGC = \G \right\} \label{L6} .
\end{align}
Note that L6 belongs to $\QDisc$ (cf.\ Eqs.~\eqref{L6} and~\eqref{QDisc-set}).

\paragraph{Group III.}

Group III systems use $p=\DCDC$ (i.e.\ $\DISC$), with $\dBBD=\dBBC=\B$.
We denote this set by $\LIII$:
\begin{align} \label{LIII-set}
 \LIII & = \left\{(p,d) \left| p = \DCDC, \dBBD = \dBBC = \dBGD = \dGGD = \B, \right. \right. \nonumber \\
 & \left. \dBGC = \dGBD = \dGGC = \G, \text{ and } \dGBCD = \dGBCC \right\} .
\end{align}
The two members are L7 ($\DISC$, Staying), with $\dGBC = \G$, and L8 ($\DISC$, Judging), with $\dGBC = \B$
(see Fig.~\ref{fig:leading-eight}):
\begin{align}
 \text{L7} & = \left\{(p,d) \left| p = \DCDC, \dBBD = \dBBC = \dBGD = \dGGD = \B, \right. \right. \nonumber \\
 & \left. \text{ and } \dBGC = \dGBD = \dGBC = \dGGC = \G \right\}, \\
 \text{L8} & = \left\{(p,d) \left| p = \DCDC, \dBBD = \dBBC = \dBGD = \dGBC = \dGGD = \B, \right. \right. \nonumber \\
 & \left. \text{ and } \dBGC = \dGBD = \dGGC = \G \right\} .
\end{align}

We denote by $\Lset$ the set of all leading-eight moral systems, so that
\begin{equation} \label{leading-eight}
 \Lset = \LI \cup \LII \cup \LIII .
\end{equation}

\subsection{Mirror systems}
\label{subsec:mirror-systems}

The reputation labels $\G$ (`good') and $\B$ (`bad') are arbitrary bookkeeping devices: unlike actions, which determine payoffs through the payoff function $v$, reputations have no intrinsic consequences.
Swapping the two labels throughout a moral system therefore yields a system that is behaviorally indistinguishable from the original -- an observation made by Ohtsuki and Iwasa\citep{si-ohtsuki:JTB:2004a} for the third-order assessment rules they studied, and which they termed `mirror symmetry' (their Appendix~A).
This subsection formalizes this relabeling operation and verifies that the quantities of interest (reputation equilibria, cooperation index, stability indices) are invariant under it; Section~\ref{subsec:mirror-classes} builds on these results to count moral systems up to relabeling and to specify the representatives evaluated in our numerical explorations.

\begin{definition}[Mirror systems] \label{def:mirror-system}
The \emph{mirror} of a moral system $(p,d)$ is the moral system $(\hat{p},\hat{d})$ given by
\begin{subequations} \label{mirror-def}
\begin{align}
 \hat{p}(\rf,\rc) & = p(\neg \rf, \neg \rc), \label{mirror-def-p}\\
 \hat{d}(\rf,\rc,\af,\ac) & = \neg\, d(\neg \rf, \neg \rc, \af, \ac), \label{mirror-def-d}
\end{align}
\end{subequations}
for all $\rf,\rc \in \left\{\B,\G\right\}$ and $\af,\ac \in \left\{\D,\C\right\}$, where $\neg$ denotes the swap $\B \leftrightarrow \G$ (equivalently, with our labeling convention \eqref{C=G=1,D=B=0}, $\neg x = 1 - x$).
A moral system that is equal to its mirror is said to be a \emph{self-mirror}.
\end{definition}

The mirror relabels reputations wherever they occur -- in the reputation arguments of $p$ and $d$, and in the reputation prescribed by $d$ -- while leaving actions untouched.
As a bitstring (Table~\ref{tab:strategies}), $\hat{p}$ is simply $p$ read from right to left, so that, for instance, the mirror of $\DISC$ ($p=\DCDC$) is the strategy $\CDCD$, which cooperates with co-players in bad standing.
Applying the mirror twice recovers the original system, so the mirror map is an involution on the space of moral systems: it partitions this space into `mirror classes' of size one (self-mirrors) or two (`mirror pairs').

The next lemma shows that the mirror system reproduces the reputation dynamics of the original system with the roles of $\G$ and $\B$ interchanged.

\begin{lemma}[Reputation dynamics of mirror systems] \label{lem:mirror-reputation}
Let $(\hat{p},\hat{d})$ be the mirror of $(p,d)$, and let $\hat{\tilde{p}}$ be the mirror of a mutant strategy $\tilde{p}$.
Denote by $\hat{\delta}_{\rf\rc}$ the probability~\eqref{deltaij} that a resident of a population implementing $(\hat{p},\hat{d})$ is assigned a good reputation, and by $\hat{\tilde{\delta}}_{\rf\rc}$ the corresponding probability~\eqref{tildedeltaij} for a mutant adopting $\hat{\tilde{p}}$ in this population.
Then, for all error rates $\chi \in (0,1/2)$ and $\epsilon \in [0,1/2)$, and for all $\rf,\rc \in \left\{\B,\G\right\}$,
\begin{equation} \label{mirror-delta}
 \hat{\delta}_{\rf\rc} = 1 - \delta_{\neg\rf\,\neg\rc}
 \qquad \text{and} \qquad
 \hat{\tilde{\delta}}_{\rf\rc} = 1 - \tilde{\delta}_{\neg\rf\,\neg\rc} .
\end{equation}
Consequently, the resident and mutant reputation equilibria satisfy
\begin{equation} \label{mirror-g-h}
 g^*_{\hat{p},\hat{d}} = 1 - g^*_{p,d}
 \qquad \text{and} \qquad
 h^*_{\hat{\tilde{p}},\hat{p},\hat{d}} = 1 - h^*_{\tilde{p},p,d} .
\end{equation}
\end{lemma}

\begin{proof}
Substituting \eqref{mirror-def} into \eqref{deltaij}, each assessment value appearing in $\hat{\delta}_{\rf\rc}$ equals one minus the corresponding value appearing in $\delta_{\neg\rf\,\neg\rc}$.
For instance, in the absence of execution errors, $\hat{d}(\rf,\rc,\hat{p}(\rf,\rc),\hat{p}(\rc,\rf)) = 1 - d(\neg\rf,\neg\rc,p(\neg\rf,\neg\rc),p(\neg\rc,\neg\rf))$; the same identity holds for the three remaining assessment values, in which execution errors replace one or both actions with $\D$, since the mirror leaves actions -- and hence these erroneous defections -- unchanged.
As the weights $(1-\epsilon)^2$, $\epsilon(1-\epsilon)$, $\epsilon(1-\epsilon)$, and $\epsilon^2$ sum to one, the term in braces in \eqref{deltaij} for the mirror system equals one minus its counterpart for the original system evaluated at $(\neg\rf,\neg\rc)$, whence
$\hat{\delta}_{\rf\rc} = \chi + (1-2\chi)\left(1 - \frac{\delta_{\neg\rf\,\neg\rc}-\chi}{1-2\chi}\right) = 1 - \delta_{\neg\rf\,\neg\rc}$.
The same substitution in \eqref{tildedeltaij} yields the mutant counterpart.

For the resident equilibrium, write the right-hand side of the reputation dynamics \eqref{gdot} as $F_{p,d}(g)$; expanding, $F_{p,d}(g) = a g^2 + b g + c$, with coefficients given in Eq.~\eqref{abc}.
Using \eqref{mirror-delta},
\begin{align*}
 F_{\hat{p},\hat{d}}(1-g) & = \left(1-\delta_{\B\B}\right)(1-g)^2 + \left(2 - \delta_{\G\B} - \delta_{\B\G}\right) g (1-g) + \left(1-\delta_{\G\G}\right) g^2 - (1-g) \\
 & = 1 - \left[ \delta_{\B\B} (1-g)^2 + \left(\delta_{\B\G}+\delta_{\G\B}\right) g(1-g) + \delta_{\G\G} g^2 \right] - 1 + g
 = - F_{p,d}(g) .
\end{align*}
Hence, if $g(t)$ solves the reputation dynamics $\dot{g} = F_{p,d}(g)$ of $(p,d)$, then $\hat{g}(t) := 1 - g(t)$ satisfies $\dot{\hat{g}} = -\dot{g} = -F_{p,d}(g) = F_{\hat{p},\hat{d}}(\hat{g})$, that is, it solves the reputation dynamics of $(\hat{p},\hat{d})$; the converse follows by applying the same argument to the mirror system, since the mirror map is an involution.
The unique asymptotically stable equilibria \eqref{g*} of the two systems thus satisfy $g^*_{\hat{p},\hat{d}} = 1 - g^*_{p,d}$.
An identical computation applied to the mutant dynamics \eqref{hdot}, with coefficients $\hat{\tilde{\delta}}_{\rf\rc}$ and resident equilibrium $1-g^*_{p,d}$, shows that $h(t)$ solves the mutant dynamics of $\tilde{p}$ in $(p,d)$ if and only if $1-h(t)$ solves those of $\hat{\tilde{p}}$ in $(\hat{p},\hat{d})$, whence the second identity in \eqref{mirror-g-h}.
\end{proof}

With Lemma~\ref{lem:mirror-reputation} in place, the behavioral equivalence of a moral system and its mirror follows readily.

\begin{proposition}[Mirror invariance of payoffs, cooperation, and stability] \label{prop:mirror-invariance}
Let $(\hat{p},\hat{d})$ be the mirror of $(p,d)$.
Then, for all error rates $\chi \in (0,1/2)$ and $\epsilon \in [0,1/2)$ and all payoff configurations:
\begin{enumerate}
 \item[(i)] $w_{\hat{p},\hat{d}} = w_{p,d}$, and $\tilde{w}_{\hat{\tilde{p}},\hat{p},\hat{d}} = \tilde{w}_{\tilde{p},p,d}$ for every mutant strategy $\tilde{p}$;
 \item[(ii)] $\kappa_{\hat{p},\hat{d}} = \kappa_{p,d}$;
 \item[(iii)] $(\hat{p},\hat{d})$ is stable if and only if $(p,d)$ is stable; consequently, $\sigma^\Gamma_{\hat{p},\hat{d}} = \sigma^\Gamma_{p,d}$ for each game class $\Gamma \in \left\{\mathrm{PD},\mathrm{SG},\mathrm{SH}\right\}$, and the average stability indices likewise coincide.
\end{enumerate}
\end{proposition}

\begin{proof}
The expected payoffs \eqref{omegaij} and \eqref{tildeomega} involve reputations only through the actions prescribed by the strategies, so \eqref{mirror-def-p} gives $\hat{\omega}_{\rf\rc} = \omega_{\neg\rf\,\neg\rc}$ and $\hat{\tilde{\omega}}_{\rf\rc} = \tilde{\omega}_{\neg\rf\,\neg\rc}$.
Substituting these identities and \eqref{mirror-g-h} into \eqref{w}, \eqref{tildew}, and \eqref{coop-index}, each term with reputation profile $(\rf,\rc)$ for the mirror system equals the term with profile $(\neg\rf,\neg\rc)$ for the original system (for instance, in \eqref{w}, $(1-g)^2 \omega_{\B\B}$ evaluated at $g = 1-g^*_{p,d}$ for the mirror system equals $\left(g^*_{p,d}\right)^2 \omega_{\G\G}$ for the original one), which proves (i) and (ii).
For (iii), the mirror map is a bijection of the strategy space onto itself with $\tilde{p} \neq p$ if and only if $\hat{\tilde{p}} \neq \hat{p}$, so by (i) the stability condition \eqref{evol-stab-cond} holds for $(p,d)$ if and only if it holds for $(\hat{p},\hat{d})$.
Since this equivalence holds at every payoff configuration, the stability regions $\mathcal{S}^\Gamma_{p,d}$ and $\mathcal{S}^\Gamma_{\hat{p},\hat{d}}$ coincide, and so do the stability indices $\sigma^\Gamma$ and their (weighted) averages.
\end{proof}

Whereas quantities defined from behavior alone are mirror-invariant, properties whose definitions refer to the reputation labels are not.
The following corollary records the consequences for the classes of moral systems defined above.

\begin{corollary}[Mirrors and classes of moral systems] \label{cor:mirror-classes}
Let $(\hat{p},\hat{d})$ be the mirror of a conditional system $(p,d)$.
Then:
\begin{enumerate}
 \item[(i)] $(p,d)$ is self-approving if and only if $(\hat{p},\hat{d})$ is self-disapproving; in particular, the mirror map restricts to a bijection $\Aset \to \Zset$, so $|\Aset| = |\Zset|$.
 \item[(ii)] $(p,d)$ is self-cooperative (resp.\ universally stable) if and only if $(\hat{p},\hat{d})$ is self-cooperative (resp.\ universally stable).
 \item[(iii)] Every self-mirror satisfies $g^*_{p,d} = 1/2$ for all error rates $\chi \in (0,1/2)$ and $\epsilon \in [0,1/2)$; in particular, no self-mirror is self-approving or self-disapproving, and no righteous system is a self-mirror.
\end{enumerate}
\end{corollary}

\begin{proof}
Claims (i) and (ii) follow from Lemma~\ref{lem:mirror-reputation} and Proposition~\ref{prop:mirror-invariance} by taking the limits in Definitions \ref{def:self-approving}, \ref{def:self-disapproving}, \ref{def:self-cooperative}, and \ref{def:universally-stable} (the mirror of a conditional system is conditional, since $\hat{p}$ is unconditional if and only if $p$ is).
For (iii), a self-mirror satisfies $g^*_{p,d} = g^*_{\hat{p},\hat{d}} = 1 - g^*_{p,d}$ by \eqref{mirror-g-h}, hence $g^*_{p,d} = 1/2$; a righteous system is self-approving by Definition~\ref{def:righteous}, so it cannot be a self-mirror.
\end{proof}

\section{Analytical results}

This section derives the analytical results that underpin the main text's central claims.
We proceed in two stages.
The first stage characterizes the structural properties of moral systems.
Sections~\ref{subsec:self-approving}--\ref{subsec:consistent-systems} identify necessary and sufficient conditions for a system to be self-approving, self-cooperative, righteous, consistent, or consistent discriminating, and establish the inclusion relations among these sets.
Section~\ref{subsec:leading-eight-righteous} then verifies that the leading-eight moral systems fall into the appropriate sets in this hierarchy.
The second stage turns to evolutionary stability.
Section~\ref{subsec:which-righteous-are-stable} identifies which righteous systems are universally stable, establishing that universal stability is equivalent to being a contrite consistent discriminating system.
Sections~\ref{subsec:stab-anal-disc} and~\ref{subsec:stab-anal-leading-eight} then provide detailed stability analyses for all consistent discriminating systems (grouped by action rule) and for each of the leading-eight moral systems individually.
Throughout, we focus on self-approving systems; self-disapproving systems are behaviorally equivalent under the mirror operation of Section~\ref{subsec:mirror-systems} and require no separate treatment.

\subsection{Characterization of self-approving systems}\label{subsec:self-approving}

Here, we characterize the sets of self-approving and self-disapproving moral systems.
We begin with the following result, which provides simple necessary and sufficient conditions for a moral system to be self-approving (or self-disapproving).

\begin{proposition}[Characterization of self-approving and self-disapproving systems] \label{prop:self-approving}
A conditional system $(p,d)$ is \emph{self-approving} if and only if
\begin{enumerate}
    \item[(i)] $\dGG = \G$, and
    \item[(ii)] %
    $\dBB+\dBG+\dGB \geq 2$, i.e., at least two of $\left\{\dBB,\dBG,\dGB\right\}$ are $\G$.
\end{enumerate}
Similarly, a conditional system $(p,d)$ is \emph{self-disapproving} if and only if
\begin{enumerate}
    \item[(i)] $\dBB = \B$, and
    \item[(ii)] %
    $\dBG+\dGB+\dGG \leq 1$, i.e., at least two of $\left\{\dBG,\dGB,\dGG\right\}$ are $\B$.
\end{enumerate}
\end{proposition}

\begin{proof}

To identify the set of self-approving systems for which condition \eqref{self-approving-cond} in Definition \ref{def:self-approving} holds,
and the set of self-disapproving systems for which condition \eqref{self-disapproving-cond} in Definition \ref{def:self-disapproving} holds,
we set $\epsilon=0$ in Eq. \eqref{deltaij}, so that Eq. \eqref{abc} can be written as
 \begin{align*}
  a(\chi) & = (1-2 \chi) \left[ \dBB + \dGG - \dBG - \dGB \right], \\
  b(\chi) & = (1-2\chi) \left[\dBG + \dGB - 2 \dBB \right] - 1, \\
  c(\chi) & = \chi + (1-2\chi)\dBB,
 \end{align*}
where we have expressed $a$, $b$, and $c$ as functions of $\chi$, and where we have used the shorthand notation \eqref{dr1r2-shorthand}.
There are $2^4=16$ different possible values for the binary string $(\dBB,\dBG,\dGB,\dGG)$, and so 16 different cases to analyze.
We begin by noting that, if
\begin{equation} \label{dBB+dGG=dBG+dGB}
 \dBB + \dGG = \dBG + \dGB
\end{equation}
holds, then $a(\chi)=0$ holds for all $\chi$.
In this case, %
\begin{equation} \label{limg_linear}
 \lim_{\chi \to 0} \left(g^*_{p,d}\right)_{\epsilon=0} = - \lim_{\chi \to 0}\frac{c(\chi)}{b(\chi)} = -\frac{c(0)}{b(0)} = \frac{\dBB}{1+\dBB-\dGG}.
\end{equation}
Depending on the values of $\dBB$ and $\dGG$, we have four different cases.
\begin{enumerate}
 \item If $(\dBB,\dGG)=(\B,\B)$, then \eqref{limg_linear} simplifies to $0$.
 \item If $(\dBB,\dGG)=(\B,\G)$, then
\eqref{limg_linear} results in the indeterminate form $0/0$.
Applying L'H\^opital's rule, we evaluate the limit as
\begin{align*}
 \lim_{\chi \to 0} \left(g^*_{p,d}\right)_{\epsilon=0} & = - \lim_{\chi \to 0}\frac{c'(\chi)}{b'(\chi)} \nonumber \\
 & = -\frac{1-2 \dBB}{-2 \left(\dBG+\dGB-2 \dBB\right)} = \frac{1-2 \dBB}{2 \left(\dGG-\dBB\right)} = \frac{1}{2} .
\end{align*}
\item If $(\dBB,\dGG)=(\G,\B)$, then \eqref{limg_linear} simplifies to $1/2$.
\item If $(\dBB,\dGG)=(\G,\G)$, then \eqref{limg_linear} simplifies to $1$.
\end{enumerate}
Suppose now that \eqref{dBB+dGG=dBG+dGB} does not hold, that is, that
\begin{equation} \label{dBB+dGGneqdBG+dGB}
 \dBB + \dGG \neq \dBG + \dGB
\end{equation}
holds.
In this case, $a(\chi) \neq 0$ holds for all $\chi \in (0,1/2)$, and
\begin{equation} \label{limg_quadratic}
 \lim_{\chi \to 0} \left(g^*_{p,d}\right)_{\epsilon=0} = \lim_{\chi \to 0} \frac{-b(\chi)-\sqrt{b(\chi)^2-4 a(\chi) c(\chi)}}{2 a(\chi)} .
\end{equation}
Here, we have two main cases.
\begin{enumerate}
 \item If $\dBB = \B$, then $c(0)=0$, and Eq. \eqref{limg_quadratic} reduces to
\begin{equation} \label{limg_quadratic_1}
 \lim_{\chi \to 0} \left(g^*_{p,d}\right)_{\epsilon=0} = \lim_{\chi \to 0} \frac{-b(\chi)-\sqrt{b(\chi)^2}}{2 a(\chi)} .
\end{equation}
There are five different subcases for which both $\dBB=\B$ and \eqref{dBB+dGGneqdBG+dGB} hold.
\begin{enumerate}
 \item If $(\dBB,\dBG,\dGB,\dGG)=(\B,\B,\B,\G)$, then $a(0)=1$, $b(0)=-1$, and \eqref{limg_quadratic_1} simplifies to $0$.
 \item If $(\dBB,\dBG,\dGB,\dGG)=(\B,\B,\G,\B)$, then $a(0)=-1$, $b(0)=0$, and \eqref{limg_quadratic_1} simplifies to $0$.
 \item If $(\dBB,\dBG,\dGB,\dGG)=(\B,\G,\B,\B)$, then $a(0)=-1$, $b(0)=0$, and \eqref{limg_quadratic_1} simplifies to $0$.
 \item If $(\dBB,\dBG,\dGB,\dGG)=(\B,\G,\G,\B)$, then $a(0)=-2$, $b(0)=1$, and \eqref{limg_quadratic_1} simplifies to $1/2$.
 \item If $(\dBB,\dBG,\dGB,\dGG)=(\B,\G,\G,\G)$, then $a(0)=-1$, $b(0)=1$, and \eqref{limg_quadratic_1} simplifies to $1$.
\end{enumerate}
\item If $\dBB=\G$, then $c(0)=1$, and \eqref{limg_quadratic} reduces to
\begin{equation} \label{limg_quadratic_2}
 \lim_{\chi \to 0} \left(g^*_{p,d}\right)_{\epsilon=0} = \lim_{\chi \to 0} \frac{-b(\chi)-\sqrt{b(\chi)^2-4a(\chi)}}{2 a(\chi)} .
\end{equation}
There are five different subcases for which both $\dBB=\G$ and \eqref{dBB+dGGneqdBG+dGB} hold.
\begin{enumerate}
 \item If $(\dBB,\dBG,\dGB,\dGG)=(\G,\B,\B,\B)$, then $a(0)=1$, $b(0)=-3$, and \eqref{limg_quadratic_2} simplifies to $\frac{3-\sqrt{5}}{2} \approx 0.382$.
 \item If $(\dBB,\dBG,\dGB,\dGG)=(\G,\B,\B,\G)$, then $a(0)=2$, $b(0)=-3$, and \eqref{limg_quadratic_2} simplifies to $\frac{1}{2}$.
 \item If $(\dBB,\dBG,\dGB,\dGG)=(\G,\B,\G,\G)$, then $a(0)=1$, $b(0)=-2$, and \eqref{limg_quadratic_2} simplifies to $1$.
 \item If $(\dBB,\dBG,\dGB,\dGG)=(\G,\G,\B,\G)$, then $a(0)=1$, $b(0)=-2$, and \eqref{limg_quadratic_2} simplifies to $1$.
 \item If $(\dBB,\dBG,\dGB,\dGG)=(\G,\G,\G,\B)$, then $a(0)=b(0)=-1$, and \eqref{limg_quadratic_2} simplifies to $\frac{\sqrt{5}-1}{2} \approx 0.618$.
\end{enumerate}
\end{enumerate}

Table \ref{tab:limg} shows a summary of the case distinctions and the results.
By Definitions~\ref{def:self-approving} and~\ref{def:self-disapproving}, the self-approving systems are those for which this limit equals one and the self-disapproving systems those for which it equals zero.
An inspection of the table thus shows that the self-approving cases are precisely those with $\dGG=\G$ and at least two of $\left\{\dBB,\dBG,\dGB\right\}$ equal to $\G$, whereas the self-disapproving cases are exactly those with $\dBB=\B$ and at least two of $\left\{\dBG,\dGB,\dGG\right\}$ equal to $\B$.
\end{proof}

\begin{table}
\caption{Limiting resident reputation equilibrium $\lim_{\chi \to 0} \left(g^*_{p,d}\right)_{\epsilon=0}$ for each of the 16 possible assessment-bit combinations $(\dBB,\dBG,\dGB,\dGG)$, in the absence of execution errors ($\epsilon=0$) and in the limit of vanishing assignment error ($\chi\to0$). Self-approving systems are those for which this limit equals one; self-disapproving systems those for which it equals zero.}
\begin{center}
\begin{tabular}{c c c c | c}
  $\dBB$ & $\dBG$ & $\dGB$ & $\dGG$ & $\lim_{\chi \to 0} \left(g^*_{p,d}\right)_{\epsilon=0}$ \\
  \hline
  $\B$ & $\B$ & $\B$ & $\B$ & $0$\\
  $\B$ & $\B$ & $\B$ & $\G$ & $0$ \\
  $\B$ & $\B$ & $\G$ & $\B$ & $0$  \\
  $\B$ & $\B$ & $\G$ & $\G$ & $1/2$ \\
  $\B$ & $\G$ & $\B$ & $\B$ & $0$ \\
  $\B$ & $\G$ & $\B$ & $\G$ & $1/2$ \\
  $\B$ & $\G$ & $\G$ & $\B$ & $1/2$ \\
  $\B$ & $\G$ & $\G$ & $\G$ & $1$ \\
  $\G$ & $\B$ & $\B$ & $\B$ & $\frac{3-\sqrt{5}}{2}$\\
  $\G$ & $\B$ & $\B$ & $\G$ & $1/2$ \\
  $\G$ & $\B$ & $\G$ & $\B$ & $1/2$ \\
  $\G$ & $\B$ & $\G$ & $\G$ & $1$ \\
  $\G$ & $\G$ & $\B$ & $\B$ & $1/2$ \\
  $\G$ & $\G$ & $\B$ & $\G$ & $1$ \\
  $\G$ & $\G$ & $\G$ & $\B$ & $\frac{\sqrt{5}-1}{2}$ \\
  $\G$ & $\G$ & $\G$ & $\G$ & $1$ \\
\end{tabular}
\end{center}
\label{tab:limg}
\end{table}

Proposition~\ref{prop:self-approving} lets us distinguish three subsets of self-approving systems, which differ in their resident reputation equilibrium.
The following lemma makes this partition explicit.

\begin{lemma}[Partition of self-approving systems] \label{lem:partition-SA}
The set of self-approving systems $\Aset$ can be partitioned into the following three subsets, each characterized by a distinct resident reputation equilibrium in the absence of execution errors, $\left(g^*_{p,d}\right)_{\epsilon=0}$, as a function of the assignment error $\chi$:
\begin{enumerate}
 \item The set $\AsetGXXG = \AsetGBGG \cup \AsetGGBG$ of self-approving systems for which
 \begin{equation} \label{cond-SA-1**1}
 \dBB = \dGG = \G \ \text{ and } \ \dBG \neq \dGB
 \end{equation}
 holds, characterized by
\begin{equation} \label{g-SA-1**1}
 \left(g^*_{p,d}\right)_{\epsilon=0} = \frac{1-\chi-\sqrt{\chi-\chi^2}}{1-2\chi} . %
\end{equation}
Here, $\AsetGBGG$ denotes the subset with $(\dBG,\dGB) = (\B,\G)$
and $\AsetGGBG$ the subset with $(\dBG,\dGB) = (\G,\B)$.
\item The set $\AsetBGGG$ of self-approving systems for which
\begin{equation} \label{cond-SA-0111}
\dBB = \B \ \text{ and } \ \dBG = \dGB = \dGG = \G
\end{equation}
holds, characterized by
\begin{equation} \label{g-SA-0111}
 \left(g^*_{p,d}\right)_{\epsilon=0} = \frac{1-4\chi+\sqrt{1-4\chi+8\chi^2}}{2-4\chi} . %
\end{equation}
\item The set $\AsetGGGG$ of self-approving systems for which
\begin{equation} \label{cond-SA-1111}
\dBB = \dBG = \dGB = \dGG = \G
\end{equation}
holds,
characterized by
\begin{equation} \label{g-SA-1111}
 \left(g^*_{p,d}\right)_{\epsilon=0} = 1-\chi .
\end{equation}
\end{enumerate}
\end{lemma}

\begin{proof}
The partition follows from Proposition~\ref{prop:self-approving}: the systems satisfying its conditions are precisely those with $(\dBB,\dBG,\dGB,\dGG)$ equal to $(\G,\B,\G,\G)$, $(\G,\G,\B,\G)$, $(\B,\G,\G,\G)$, or $(\G,\G,\G,\G)$, which group into the three subsets above.
The explicit equilibria are obtained by substituting the corresponding assessment bits into Eqs.~\eqref{g*} and~\eqref{abc} at $\epsilon=0$ (equivalently, into the $a(\chi)$, $b(\chi)$, and $c(\chi)$ of the proof of Proposition~\ref{prop:self-approving}) and simplifying, without taking the $\chi\to0$ limit.
\end{proof}

\subsection{Characterization of self-cooperative and righteous systems}\label{subsec:self-cooperative-righteous}

The cooperation index in the absence of execution errors, $(\kappa_{p,d})_{\epsilon=0}$, and its limit when the assignment error tends to zero, $\lim_{\chi \to 0}(\kappa_{p,d})_{\epsilon=0}$, are important quantities in our analysis.
The following lemma expresses these two quantities in terms of a simple quadratic function.

\begin{lemma}[Cooperation index as a function of the reputation equilibrium] \label{lem:kappa-q}
Let $q_p: [0,1] \to [0,1]$ be defined by
\begin{align} \label{q}
 q_p(x) = \pBB (1-x)^2 + \left(\pBG + \pGB\right) x(1-x) + \pGG x^2.
\end{align}
Then the cooperation index $(\kappa_{p,d})_{\epsilon=0}$ of a moral system $(p,d)$ and its limit $\lim_{\chi \to 0}(\kappa_{p,d})_{\epsilon=0}$ can be written, respectively, as
\begin{equation} \label{kappa-q}
\left(\kappa_{p,d}\right)_{\epsilon=0} = q_p\left((g^*_{p,d})_{\epsilon=0}\right),
\end{equation}
and
\begin{equation} \label{lim-kappa-q}
\lim_{\chi \to 0} \left(\kappa_{p,d}\right)_{\epsilon=0} = q_p\left(\lim_{\chi \to 0} \left(g^*_{p,d}\right)_{\epsilon=0}\right).
\end{equation}
\end{lemma}

\begin{proof}
Evaluating the cooperation index~\eqref{coop-index} at $\epsilon=0$ yields~\eqref{kappa-q}.
Taking the limit of~\eqref{kappa-q} as $\chi$ tends to zero and applying the algebraic limit theorem yields~\eqref{lim-kappa-q}.
\end{proof}

Table~\ref{tab:q} lists the cooperation-index function $q_p(x)$~\eqref{q} for each action rule $p$.

\begin{table}
\caption{The cooperation-index function $q_p(x)$~\eqref{q} for each of the 16 action rules $p$, where $x\in[0,1]$ is the fraction of individuals with a good reputation. Actions and reputations follow the encoding $\C=\G=1$, $\D=\B=0$~\eqref{C=G=1,D=B=0}.}
\begin{center}
\begin{tabular}{c c c c | c}
  $\pBB$ & $\pBG$ & $\pGB$ & $\pGG$ & $q_p(x)$ \\
  \hline
  $\D$ & $\D$ & $\D$ & $\D$ & $0$ \\
  $\C$ & $\D$ & $\D$ & $\D$ & $1-2x+x^2$ \\
  $\D$ & $\C$ & $\D$ & $\D$ & $x-x^2$ \\
  $\C$ & $\C$ & $\D$ & $\D$ & $1-x$ \\
  $\D$ & $\D$ & $\C$ & $\D$ & $x-x^2$ \\
  $\C$ & $\D$ & $\C$ & $\D$ & $1-x$ \\
  $\D$ & $\C$ & $\C$ & $\D$ & $2x-2x^2$\\
  $\C$ & $\C$ & $\C$ & $\D$ & $1-x^2$ \\
  $\D$ & $\D$ & $\D$ & $\C$ & $x^2$ \\
  $\C$ & $\D$ & $\D$ & $\C$ & $1-2x+2x^2$ \\
  $\D$ & $\C$ & $\D$ & $\C$ & $x$ \\
  $\C$ & $\C$ & $\D$ & $\C$ & $1-x+x^2$ \\
  $\D$ & $\D$ & $\C$ & $\C$ & $x$ \\
  $\C$ & $\D$ & $\C$ & $\C$ & $1-x+x^2$ \\
  $\D$ & $\C$ & $\C$ & $\C$ & $2x-x^2$\\
  $\C$ & $\C$ & $\C$ & $\C$ & $1$ \\
\end{tabular}
\end{center}
\label{tab:q}
\end{table}

The following proposition provides simple necessary and sufficient conditions for a moral system to be self-cooperative.
Its proof uses Lemma \ref{lem:kappa-q}.

\begin{proposition}[Necessary and sufficient conditions for a conditional system to be self-cooperative]
\label{prop:self-cooperative}
A conditional system is self-cooperative if and only if either
\begin{enumerate}
 \item[(i)] it is self-approving and $\pGG = \C$ holds, or
 \item[(ii)] it is self-disapproving and $\pBB = \C$ holds.
\end{enumerate}
\end{proposition}

\begin{proof}
By Lemma \ref{lem:kappa-q}, we can write $\lim_{\chi \to 0}\left(\kappa_{p,d}\right)_{\epsilon=0} = q_p(\gamma)$, where
\begin{equation} \label{gamma}
\gamma = \lim_{\chi \to 0} \left(g^*_{p,d}\right)_{\epsilon=0}
\end{equation}
and where $q_p$ is given by \eqref{q}.
Substituting all possible choices of $(\pBB,\pBG,\pGB,\pGG)$ into $q_p(\gamma)$ yields the 16 cases listed in Table~\ref{tab:q}.

Since $p$ is conditional, there are three subcases to consider.
\begin{enumerate}
\item Suppose that $\pBB = \pGG = \D$ holds.
Then, $q_p(\gamma) = \left(\pBG+\pGB\right) \gamma(1-\gamma)$ with $1 \leq \pBG+\pGB \leq 2$ (since $p$ is conditional, $\pBG$ and $\pGB$ cannot both be $\D$).
Since $\gamma(1-\gamma) \leq 1/4$ for $\gamma \in [0,1]$, it follows that $\max_{0 \leq \gamma \leq 1} q_p(\gamma) \leq 1/2$.
Hence, $p$ can never be self-cooperative under $d$.
\item Suppose that $\pGG = \C$ holds.
Then, $q_p(1) = \pGG = 1$.
Therefore, if $(p,d)$ is self-approving (so that $\gamma=1$ by Definition~\ref{def:self-approving} and Eq.~\eqref{gamma}), then it is self-cooperative.
This proves the sufficient condition of part (i).
\item Suppose that $\pBB = \C$ holds.
Then, $q_p(0) = \pBB = 1$.
Hence, if $(p,d)$ is self-disapproving (so that $\gamma=0$ by Definition~\ref{def:self-disapproving} and Eq.~\eqref{gamma}), then it is self-cooperative.
This proves the sufficient condition of part (ii).
\end{enumerate}
Finally, we prove the necessity of conditions (i)--(ii).
Suppose $(p,d)$ is self-cooperative, so that $q_p(\gamma)=1$.
If $0<\gamma<1$, all four weights $(1-\gamma)^2$, $\gamma(1-\gamma)$, $\gamma(1-\gamma)$, and $\gamma^2$ in $q_p(\gamma)$ are strictly positive, so $q_p(\gamma)=1$ would require $\pBB=\pGG=\C$ and $\pBG=\pGB=\C$, that is, $p=\ALLC$; since $p$ is conditional, this is impossible, and hence $q_p(\gamma)<1$, a contradiction.
Therefore $\gamma\in\{0,1\}$, so $(p,d)$ is self-approving or self-disapproving.
If $\gamma=1$, then $\pGG=q_p(1)=1$, so $\pGG=\C$ and (i) holds; if $\gamma=0$, then $\pBB=q_p(0)=1$, so $\pBB=\C$ and (ii) holds.
This concludes the proof.
\end{proof}

\begin{corollary}[Characterization of righteous systems] \label{cor:righteous}
Let $(p,d)$ be self-approving.
Then $(p,d)$ is self-cooperative (and hence also righteous) if and only if $\pGG = \C$.
\end{corollary}

\begin{proof}
Being self-approving, $(p,d)$ is not self-disapproving, so condition~(ii) of Proposition~\ref{prop:self-cooperative} cannot hold.
Hence $(p,d)$ is self-cooperative if and only if condition~(i) holds, that is, if and only if $\pGG=\C$.
Righteousness then follows from Definition~\ref{def:righteous}, since $(p,d)$ is self-approving by assumption.
\end{proof}

Proposition~\ref{prop:self-cooperative} shows that the self-cooperative systems form two disjoint classes.
The first is the set of self-approving systems satisfying $\pGG = \C$; by Corollary~\ref{cor:righteous}, these coincide with the righteous systems of Definition~\ref{def:righteous}.
The second is the set of self-disapproving systems satisfying $\pBB = \C$.

Together with Proposition~\ref{prop:self-approving}, Corollary~\ref{cor:righteous} lets us simplify the characterization of the set $\Rset$ of righteous systems given in Eq.~\eqref{righteous-set} to
\begin{equation} \label{righteous-set-simple}
 \Rset = \left\{ (p,d) \left| (p,d) \in \Cset,\ \pGG = \C,\ \dGG = \G, \text{ and } \dBB+\dBG+\dGB \geq 2 \right. \right\} .
\end{equation}

This characterization also lets us count the righteous systems, a figure reported in Fig.~\ref{fig:venn_analytic_result}.

\begin{proposition}[Cardinality of $\Rset$] \label{prop:cardinality-R}
The set of righteous systems contains $|\Rset| = \numprint{114688}$ moral systems.
\end{proposition}
\begin{proof}
An assessment rule $d$ is a function of $2^4=16$ argument combinations, hence a string of $16$ bits; a condition that leaves $k$ of these bits unconstrained is satisfied by $2^k$ assessment rules. We count the admissible systems by multiplying, over action rules and constrained bits, the number of choices each leaves open.
By Eq.~\eqref{righteous-set-simple} a righteous system is conditional and has $\pGG=\C$. Fixing $\pGG=\C$ leaves $2^3=8$ strategies; one of them, the unconditional cooperator $\ALLC$ ($p=\CCCC$), is excluded because a righteous system is conditional (Definition~\ref{def:righteous}), leaving $7$ admissible action rules. Each moreover satisfies $\dGG=\G$ together with $\dBB+\dBG+\dGB\ge2$, which constrain four distinct bits of $d$: $\dGG$ is fixed to $\G$, while of the eight assignments of $(\dBB,\dBG,\dGB)$ only the four with at least two $\G$'s are admissible. The remaining $12$ bits are free, so
\begin{equation*}
 |\Rset| = 7\cdot 4\cdot 2^{12} = \numprint{114688}.
\end{equation*}
\end{proof}

\subsection{Righteousness and reputational dominance of consistent systems}\label{subsec:consistent-systems}

Every consistent system is self-approving, as the following proposition records.

\begin{proposition}[Consistent systems are self-approving]
\label{prop:cons->SA}
If $(p,d)$ is consistent, then it lies in $\AsetGGGG$ and is, in particular, self-approving. That is, $\Qset \subset \AsetGGGG \subset \Aset$.
\end{proposition}
\begin{proof}
By condition~\eqref{si-consistent-cond-G} of Definition~\ref{def:consistent}, a consistent system satisfies $\dGG=\G$ and $\dBB=\dBG=\dGB=\G$, so that $\dBB+\dBG+\dGB=3\geq2$.
By Proposition~\ref{prop:self-approving}, it is therefore self-approving.
Since, in addition, all four assessment bits equal $\G$, it lies in the subset $\AsetGGGG$ (item~3 of Lemma~\ref{lem:partition-SA}); hence $\Qset\subset\AsetGGGG\subset\Aset$.
\end{proof}

By Definition~\ref{def:consistent-disc}, a consistent discriminating system $(p,d)$ is one whose action rule $p$ is discriminating and consistent with $d$.

\begin{proposition}[Consistent discriminating systems are righteous]
\label{prop:consistentdisc->righteous}
If $(p,d)$ is consistent discriminating, then it is righteous. That is, $\QDset \subset \Rset$.
\end{proposition}
\begin{proof}
Since $p$ is discriminating, it is conditional and satisfies $\pGG=\C$ (Definitions~\ref{def:cond} and~\ref{def:disc}).
Being consistent, $(p,d)$ is self-approving (Proposition~\ref{prop:cons->SA}); together with $\pGG=\C$, Proposition~\ref{prop:self-cooperative}(i) then makes it self-cooperative.
A system that is both self-approving and self-cooperative is righteous (Definition~\ref{def:righteous}), so $\QDset\subset\Rset$.
\end{proof}

The remainder of this subsection characterizes consistency in reputational terms.
The intuition is that a consistent assessment rule labels as good exactly those who follow the resident prescription, so no alternative action rule should be able to match the resident's reputation under~$d$.
The next two propositions make this claim precise at $\epsilon=0$: together, they show that a system is consistent if and only if its resident action rule attains a strictly higher reputation than every mutant action rule.

\begin{proposition}[Reputational dominance of consistent systems] \label{prop:consistent-rep-dominance}
Let $\epsilon=0$ and $\chi \in (0,1/2)$.
If $(p,d) \in \Qset$, then for every mutant action rule $\tilde{p} \neq p$,
\[
 h^*_{\tilde{p},p,d} < g^*_{p,d} .
\]
\end{proposition}

\begin{proof}
By Definition~\ref{def:consistent}, $\dBB=\dBG=\dGB=\dGG=\G$.
For $\epsilon=0$, Eq.~\eqref{deltaij} reduces to $\delta_{\rf\rc} = \chi + (1-2\chi)\,\drfrc = 1-\chi$ for all $(\rf,\rc) \in \{\B,\G\}^2$.
Substituting into Eq.~\eqref{g*} yields $g^*_{p,d} = 1-\chi$, in agreement with Lemma~\ref{lem:partition-SA}, item~3.

Fix a mutant $\tilde{p} \neq p$ and, for each $(\rf,\rc) \in \{\B,\G\}^2$, let $b_{\rf\rc}=1$ if $\tilde{p}(\rf,\rc)=p(\rf,\rc)$ and $b_{\rf\rc}=0$ otherwise.
Because $(p,d)$ is consistent, whenever $\tilde{p}(\rf,\rc) \neq p(\rf,\rc)$ we have $\drfrcm = \notdrfrc = \B$ (by condition~\eqref{si-consistent-cond-B}); and whenever $\tilde{p}(\rf,\rc) = p(\rf,\rc)$ we have $\drfrcm = \drfrc = \G$ (by condition~\eqref{si-consistent-cond-G}).
Hence $\drfrcm = \G$ if and only if $b_{\rf\rc}=1$, and, for $\epsilon=0$, Eq.~\eqref{tildedeltaij} reduces to
\begin{equation} \label{tilde-delta-consistent}
 \tilde\delta_{\rf\rc} = \chi + (1-2\chi)\, b_{\rf\rc}.
\end{equation}

The mutant reputation equilibrium $h^*_{\tilde{p},p,d}$, obtained by setting $\dot{h}_{\tilde{p},p,d}=0$ in Eq.~\eqref{hdot}, satisfies the fixed-point equation
\begin{equation} \label{h-fixed-point}
 h^*_{\tilde{p},p,d} = \alpha + h^*_{\tilde{p},p,d}(\beta-\alpha),
\end{equation}
with
\begin{subequations} \label{alpha-beta}
\begin{align}
 \alpha & := g^*_{p,d}\,\tilde\delta_{\B\G} + (1-g^*_{p,d})\,\tilde\delta_{\B\B},\label{alpha-def} \\
 \beta  & := g^*_{p,d}\,\tilde\delta_{\G\G} + (1-g^*_{p,d})\,\tilde\delta_{\G\B}.\label{beta-def}
\end{align}
\end{subequations}
Since $\tilde\delta_{\rf\rc} \in \{\chi,1-\chi\}$ for every $(\rf,\rc)\in\{\B,\G\}^2$, and since $\alpha$ and $\beta$ are convex combinations of these values (the weights $g^*_{p,d}$ and $1-g^*_{p,d}$ being nonnegative and summing to one), both $\alpha$ and $1-\beta$ lie in $[\chi,1-\chi]$, so $1-\beta+\alpha \geq 2\chi > 0$.
Solving Eq.~\eqref{h-fixed-point} for $h^*_{\tilde{p},p,d}$ (legitimate since $1-\beta+\alpha>0$) gives
\[
 h^*_{\tilde{p},p,d} = \frac{\alpha}{1-\beta+\alpha},
\]
in agreement with Eq.~\eqref{h*},
and a direct computation gives
\[
 g^*_{p,d} - h^*_{\tilde{p},p,d} = \frac{(1-\chi)(1-\beta) - \chi\,\alpha}{1-\beta+\alpha} .
\]

To expand the numerator, let $\rho_{\rf\rc}$ denote the equilibrium probability that a focal player has reputation $\rf$ and the co-player has reputation $\rc$; these are the reputation-pair weights $(g^*_{p,d})^2$, $g^*_{p,d}(1-g^*_{p,d})$, and $(1-g^*_{p,d})^2$ that already appear in Eqs.~\eqref{gdot} and~\eqref{w}.
With $g^*_{p,d}=1-\chi$,
\[
 \rho_{\G\G}=(1-\chi)^2,\quad \rho_{\G\B}=\rho_{\B\G}=\chi(1-\chi),\quad \rho_{\B\B}=\chi^2,
\]
which satisfy
\begin{equation} \label{rho-identity}
 \sum_{\rf,\rc\in\{\B,\G\}} \rho_{\rf\rc}=\big[(1-\chi)+\chi\big]^2=1.
\end{equation}
Within $\alpha$ and $\beta$ the $\tilde\delta_{\rf\rc}$ are weighted by $g^*_{p,d}$ and $1-g^*_{p,d}$, and the outer coefficients $1-\chi=g^*_{p,d}$ and $\chi=1-g^*_{p,d}$ match these; hence in $(1-\chi)\beta+\chi\,\alpha$ the coefficient multiplying each $\tilde\delta_{\rf\rc}$ is exactly the product $\rho_{\rf\rc}$.
Collecting the four terms by reputation pair gives
\begin{equation} \label{grouping-identity}
 (1-\chi)\beta+\chi\,\alpha=\sum_{\rf,\rc\in\{\B,\G\}} \rho_{\rf\rc}\,\tilde\delta_{\rf\rc}.
\end{equation}
Substituting $\tilde\delta_{\rf\rc}=\chi+(1-2\chi)b_{\rf\rc}$ from Eq.~\eqref{tilde-delta-consistent}, the numerator becomes
\begin{align} \label{D-expansion}
 (1-\chi)(1-\beta) - \chi\,\alpha
 &= (1-\chi) - \big[(1-\chi)\beta+\chi\,\alpha\big] \nonumber \\
 &= (1-\chi) - \sum_{\rf,\rc\in\{\B,\G\}} \rho_{\rf\rc}\,\big[\chi+(1-2\chi)b_{\rf\rc}\big] \nonumber \\
 &= (1-\chi) - \chi\sum_{\rf,\rc\in\{\B,\G\}}\rho_{\rf\rc} - (1-2\chi)\sum_{\rf,\rc\in\{\B,\G\}}\rho_{\rf\rc}\,b_{\rf\rc} \nonumber \\
 &= (1-\chi) - \chi - (1-2\chi)\sum_{\rf,\rc\in\{\B,\G\}}\rho_{\rf\rc}\,b_{\rf\rc} \nonumber \\
 &= (1-2\chi)\bigg(\sum_{\rf,\rc\in\{\B,\G\}}\rho_{\rf\rc} - \sum_{\rf,\rc\in\{\B,\G\}}\rho_{\rf\rc}\,b_{\rf\rc}\bigg) \nonumber \\
 &= (1-2\chi)\sum_{\rf,\rc\in\{\B,\G\}}\rho_{\rf\rc}\,(1-b_{\rf\rc}) \nonumber \\
 &= (1-2\chi)\Big\{ (1-\chi)^2 (1-b_{\G\G}) + \chi(1-\chi)(1-b_{\G\B}) \nonumber \\
 & \quad + \chi(1-\chi)(1-b_{\B\G}) + \chi^2 (1-b_{\B\B}) \Big\} ,
\end{align}
where the second equality uses Eq.~\eqref{grouping-identity}; the fourth and fifth equalities both invoke the identity~\eqref{rho-identity} --- the fourth to replace $\chi\sum_{\rf,\rc}\rho_{\rf\rc}$ by $\chi$, the fifth to reinstate $1=\sum_{\rf,\rc}\rho_{\rf\rc}$ (after combining $1-\chi-\chi=1-2\chi$).
Each $(1-b_{\rf\rc}) \in \{0,1\}$, the weights $\rho_{\rf\rc}$ are strictly positive for $\chi \in (0,1/2)$, and at least one $(1-b_{\rf\rc})=1$ since $\tilde{p}\neq p$.
Hence the right-hand side of Eq.~\eqref{D-expansion} is strictly positive, and $g^*_{p,d} - h^*_{\tilde{p},p,d} > 0$.
\end{proof}

\begin{proposition}[Reputational non-dominance of inconsistent systems] \label{prop:inconsistent-rep-nondominance}
Let $\epsilon=0$ and $\chi \in (0,1/2)$.
If $(p,d) \notin \Qset$, then there exists $\tilde{p} \neq p$ such that
\[
 h^*_{\tilde{p},p,d} \geq g^*_{p,d} .
\]
\end{proposition}

\begin{proof}
Because $(p,d) \notin \Qset$, Definition~\ref{def:consistent} guarantees that at least one pair $(\drfrc,\notdrfrc)$ differs from $(\G,\B)$.
Throughout, $(\rf,\rc)$ ranges over the four reputation pairs in $\{\B,\G\}^2$.
We distinguish two mutually exclusive and exhaustive cases.

\begin{enumerate}
 \item Suppose some pair $(\rf^\ast,\rc^\ast)$ has $\mathsf{d}_{\rf^\ast\rc^\ast} = \bar{\mathsf{d}}_{\rf^\ast\rc^\ast}$.
 Let $\tilde{p}$ be the single-bit flip of $p$ at $(\rf^\ast,\rc^\ast)$: $\tilde{p}(\rf^\ast,\rc^\ast) = \neg p(\rf^\ast,\rc^\ast)$ and $\tilde{p} = p$ elsewhere, so that $\tilde{p} \neq p$.
 We claim that $\drfrcm = \drfrc$ at every pair.
 At $(\rf^\ast,\rc^\ast)$, the flip gives $\tilde{\mathsf{d}}_{\rf^\ast\rc^\ast} = \bar{\mathsf{d}}_{\rf^\ast\rc^\ast}$ (by the shorthands~\eqref{tildedr1r2-shorthand} and~\eqref{notdr1r2-shorthand}), which equals $\mathsf{d}_{\rf^\ast\rc^\ast}$ by the case hypothesis; at every other pair, $\tilde{p} = p$, so $\drfrcm = \drfrc$ trivially.
 At $\epsilon=0$, Eqs.~\eqref{deltaij} and~\eqref{tildedeltaij} reduce to $\delta_{\rf\rc} = \chi + (1-2\chi)\,\drfrc$ and $\tilde\delta_{\rf\rc} = \chi + (1-2\chi)\,\drfrcm$, so the equalities $\drfrcm = \drfrc$ yield $\tilde\delta_{\rf\rc} = \delta_{\rf\rc}$ at every pair.
 Substituting these equalities into the mutant reputation dynamics~\eqref{hdot} and evaluating its right-hand side at $h_{\tilde{p},p,d} = g^*_{p,d}$ gives
 \[
  \delta_{\G\G}\,(g^*_{p,d})^2 + (\delta_{\B\G}+\delta_{\G\B})\,g^*_{p,d}(1-g^*_{p,d}) + \delta_{\B\B}\,(1-g^*_{p,d})^2 - g^*_{p,d}.
 \]
 This is exactly the right-hand side of the resident dynamics~\eqref{gdot} at $g_{p,d}=g^*_{p,d}$, which vanishes by the definition of $g^*_{p,d}$.
 Thus $h_{\tilde{p},p,d} = g^*_{p,d}$ is an equilibrium of~\eqref{hdot}, and by the uniqueness of the reputation equilibria (Eqs.~\eqref{g*} and~\eqref{h*}), $h^*_{\tilde{p},p,d} = g^*_{p,d}$.

 \item Suppose instead that $\drfrc \neq \notdrfrc$ at every pair, so that each $(\drfrc,\notdrfrc)$ is either $(\G,\B)$ or $(\B,\G)$.
 Were all four pairs $(\G,\B)$, Definition~\ref{def:consistent} would make $(p,d)$ consistent, contradicting the hypothesis; hence the set
 \[
  S := \{(\rf,\rc) : (\drfrc,\notdrfrc) = (\B,\G) \}
 \]
 is non-empty.
 Let $\tilde{p}$ flip $p$ on $S$ and agree with $p$ elsewhere, so that $\tilde{p} \neq p$.
 Then $\drfrcm = \G$ at every pair: for $(\rf,\rc) \in S$, $\drfrcm = \notdrfrc = \G$; for $(\rf,\rc) \notin S$, $\drfrcm = \drfrc = \G$, since $(\drfrc,\notdrfrc) = (\G,\B)$.
 At $\epsilon=0$, Eq.~\eqref{tildedeltaij} therefore reduces to $\tilde\delta_{\rf\rc} = \chi + (1-2\chi)\,\drfrcm = 1-\chi$ at every pair.
 Substituting into Eq.~\eqref{h*}, we find that the coefficients of $g^*_{p,d}$ in the numerator, $\tilde\delta_{\B\G}-\tilde\delta_{\B\B}$, and in the denominator, $\tilde\delta_{\B\G}+\tilde\delta_{\G\B}-\tilde\delta_{\B\B}-\tilde\delta_{\G\G}$, both vanish; the dependence on $g^*_{p,d}$ drops out, leaving $h^*_{\tilde{p},p,d} = 1-\chi$.
 It remains to compare this with $g^*_{p,d}$.
 Setting $\dot{g}_{p,d}=0$ in~\eqref{gdot} expresses $g^*_{p,d}$ as a convex combination of the $\delta_{\rf\rc}$.
 At $\epsilon=0$, each $\delta_{\rf\rc} = \chi + (1-2\chi)\,\drfrc$ equals $\chi$ (if $\drfrc = \B$) or $1-\chi$ (if $\drfrc = \G$), and $\chi < 1-\chi$ for $\chi \in (0,1/2)$, so every $\delta_{\rf\rc} \leq 1-\chi$; hence $g^*_{p,d} \leq 1-\chi$.
 Therefore $h^*_{\tilde{p},p,d} = 1-\chi \geq g^*_{p,d}$.
\end{enumerate}
\end{proof}

Combining Propositions~\ref{prop:consistent-rep-dominance} and~\ref{prop:inconsistent-rep-nondominance} yields the reputational characterization of consistency anticipated at the start of this discussion.

\begin{corollary}[Reputational characterization of consistency] \label{cor:rep-consistency}
Let $\epsilon=0$ and $\chi \in (0,1/2)$.
A system $(p,d)$ is consistent if and only if $h^*_{\tilde{p},p,d} < g^*_{p,d}$ for every mutant action rule $\tilde{p} \neq p$.
\end{corollary}

\begin{proof}
Proposition~\ref{prop:consistent-rep-dominance} gives the forward direction.
The contrapositive of the reverse direction is Proposition~\ref{prop:inconsistent-rep-nondominance}: if $(p,d) \notin \Qset$, then some $\tilde{p} \neq p$ satisfies $h^*_{\tilde{p},p,d} \geq g^*_{p,d}$.
\end{proof}

\subsection{Righteousness of the leading eight}\label{subsec:leading-eight-righteous}

The leading eight~\eqref{leading-eight} have action rule $p=\wildcard\C\D\C$.
Together with the shorthand~\eqref{dr1r2-shorthand} and the leading-eight assessment rule~\eqref{leading-eight-assessment}, this implies
\begin{equation} \label{dBG=dGB=dGG=G}
\dBG=\dGB=\dGG=\G .
\end{equation}
These values already place the leading eight among the righteous systems.

\begin{proposition}[Leading-eight systems are righteous]
\label{prop:leading-eight-righteous}
Every leading-eight moral system is righteous; that is, $\Lset \subset \Rset$.
Moreover, those in groups~I and~II belong to $\AsetGGGG$, while those in group~III belong to $\AsetBGGG$.
\end{proposition}
\begin{proof}
By Eq.~\eqref{dBG=dGB=dGG=G}, $\dGG=\G$ and at least two of $\{\dBB,\dBG,\dGB\}$ equal $\G$, so the leading eight are self-approving (Proposition~\ref{prop:self-approving}).
Since $\pGG=\C$, they are also self-cooperative (Corollary~\ref{cor:righteous}), and hence righteous (Definition~\ref{def:righteous}); thus $\Lset\subset\Rset$.
For the finer placement, Eqs.~\eqref{LI-set}--\eqref{LIII-set} give $\dBB=\G$ for groups~I and~II and $\dBB=\B$ for group~III.
Combined with Eq.~\eqref{dBG=dGB=dGG=G}, the former therefore have all four assessment bits equal to $\G$ and lie in $\AsetGGGG$, while the latter lie in $\AsetBGGG$.
\end{proof}

\subsection{Universal stability within the set of righteous systems}
\label{subsec:which-righteous-are-stable}

Many moral systems of interest are righteous, notably the consistent discriminating systems
(Proposition~\ref{prop:consistentdisc->righteous}) and the leading-eight systems
(Proposition~\ref{prop:leading-eight-righteous}). It is therefore natural to ask which righteous
systems are universally stable (Definition~\ref{def:universally-stable}).
We show that a righteous system is universally stable if and only if it is a contrite consistent
discriminating system; that is, the universally stable systems in the righteous set $\Rset$
(Eq.~\eqref{righteous-set-simple}) are exactly $\QCDisc$ (Eq.~\eqref{QCDisc-set}).

We establish necessity by eliminating systems one class of invader at a time, producing a strictly
decreasing chain $\Rset \supset \Rset' \supset \Rset'' \supset \Rset''' = \QCDisc$: we discard in turn
the systems invaded by $\ALLD$ (leaving $\Rset'$; Section~\ref{subsubsec:invasion-ALLD}), by $\ALLC$
(leaving $\Rset''$; Section~\ref{subsubsec:invasion-ALLC}), and by some discriminating strategy
(leaving $\Rset''' = \QCDisc$; Section~\ref{subsubsec:invasion-disc}). Sufficiency then
follows by verifying that every system in $\QCDisc$ resists invasion by all alternative strategies
(Section~\ref{subsubsec:invasion-other-mutants}).
Fig.~\ref{fig:venn_analytic_result} maps this chain and locates the consistent
discriminating systems and the leading eight within it.

\begin{figure}[tp]
    \centering
    \includegraphics[width=\textwidth]{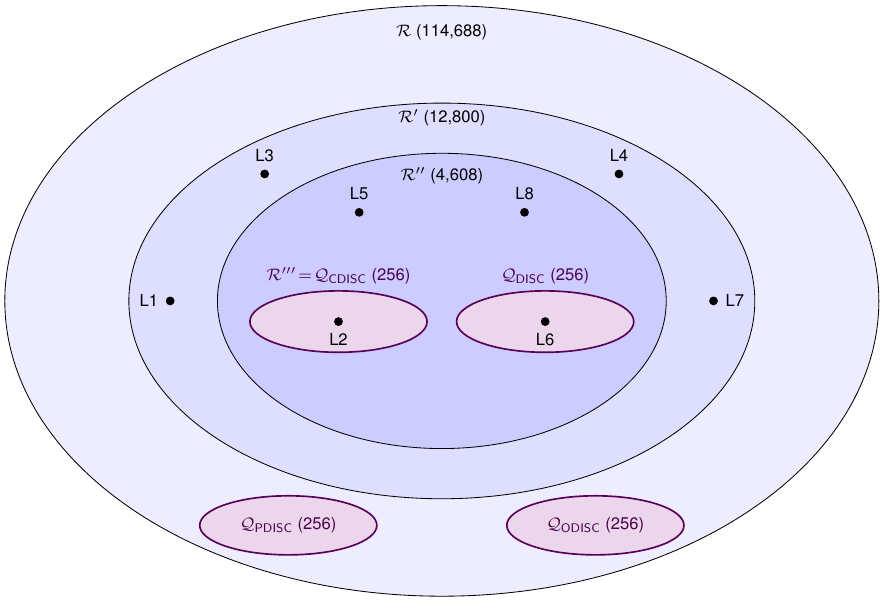}
    \caption{The nested righteous sets produced by successively removing the systems invaded by
    $\ALLD$, by $\ALLC$, and by a discriminating strategy: $\Rset \supset \Rset' \supset \Rset''
    \supset \Rset''' = \QCDisc$. This leaves $\QCDisc$ as the only candidate for universal stability; a final check that
    it resists every remaining invader shows that the contrite consistent discriminating systems
    $\QCDisc$ are the only universally stable righteous systems
    (Section~\ref{subsec:which-righteous-are-stable}).
    The diagram is schematic (not to scale), but sets of equal cardinality are drawn at equal size,
    and the parenthetical figures give the number of moral systems in each set. The four consistent
    discriminating families appear as equal purple ellipses; the leading eight, whose assessment rules
    are third-order, are single systems shown as dots, with $\mathrm{L2}$ ($\CDISC$, Consistent
    Standing) in $\QCDisc$, and $\mathrm{L6}$ ($\DISC$, Stern Judging) in $\QDisc$.}
    \label{fig:venn_analytic_result}
\end{figure}

\subsubsection{Preliminaries}
\label{subsubsec:preliminaries-stability}

\paragraph{Reputation equilibria.}

Given a resident system $(p,d) \in \Rset$ and a mutant strategy $\tilde{p} \neq p$, we write
\begin{equation} \label{gchi-def}
    g(\chi) := (g^*_{p,d})_{\epsilon=0}
\end{equation}
and
\begin{equation} \label{hchi-def}
    h(\chi) := \left(h^*_{\tilde{p},p,d}\right)_{\epsilon=0}
\end{equation}
for the resident and mutant reputation equilibria at $\epsilon=0$, as functions of $\chi$.

Since a righteous resident $(p,d)\in\Rset$ is self-approving (Definition~\ref{def:righteous}), it belongs to exactly one of the three classes $\AsetGXXG$, $\AsetBGGG$, and $\AsetGGGG$ that partition the self-approving systems (Lemma~\ref{lem:partition-SA}); accordingly, $g(\chi)$ is given (Eqs.~\eqref{g-SA-1**1}--\eqref{g-SA-1111}) by
\begin{equation} \label{gchi-explicit}
 g(\chi) =
    \begin{cases}
      \frac{1-\chi-\sqrt{\chi-\chi^2}}{1-2\chi} & \text{if } (p,d) \in \AsetGXXG \\
      \frac{1-4\chi+\sqrt{1-4\chi+8\chi^2}}{2-4\chi} & \text{if } (p,d) \in \AsetBGGG \\
      1-\chi & \text{if } (p,d) \in \AsetGGGG
   \end{cases} .
\end{equation}

Substituting $\epsilon=0$ into the mutant equilibrium~\eqref{h*} --- and thus into the assessment probabilities $\tilde{\delta}_{\rf\rc}$~\eqref{tildedeltaij} on which it depends --- and simplifying gives
\begin{equation*}
    h(\chi) = \frac{\chi + (1-2 \chi) \dBBm + (1-2 \chi)(\dBGm - \dBBm) g(\chi)}{1+(1-2\chi)(\dBBm-\dGBm) + (1-2\chi)(\dBGm+\dGBm-\dBBm-\dGGm)g(\chi)},
\end{equation*}
which depends on the triple $(p,d,\tilde{p})$ only through $g(\chi)$~\eqref{gchi-explicit} and the four mutant assessment bits $\dBBm$, $\dBGm$, $\dGBm$, and $\dGGm$~\eqref{tildedr1r2-shorthand}.

\paragraph{Expected payoffs.}

We similarly write
\begin{equation} \label{wchi-def}
 w(\chi) := \left(w_{p,d}\right)_{\epsilon=0}
\end{equation}
and
\begin{equation} \label{tildewchi-def}
  \tilde{w}(\chi) := \left(\tilde{w}_{\tilde{p},p,d}\right)_{\epsilon=0}
\end{equation}
for the expected payoffs to a resident and a mutant, evaluated at the reputation equilibrium and at $\epsilon=0$, as functions of $\chi$.
Substituting $\epsilon=0$ into Eq.~\eqref{w}, we obtain
\begin{align}
 w(\chi) & = (1-g(\chi))^2 \wBB \nonumber \\
 & + g(\chi)(1-g(\chi))\left(\wBG + \wGB \right) \nonumber \\
 & + g(\chi)^2 \wGG, \label{wchi-explicit}
\end{align}
where, by Eq.~\eqref{omegaij} at $\epsilon=0$,
\begin{align*}
 \left(\omega_{\rf\rc}\right)_{\epsilon=0} & = v(\mathsf{p}_{\rf \rc},\mathsf{p}_{\rc\rf}).
\end{align*}
Likewise, substituting $\epsilon=0$ into Eq.~\eqref{tildew}, we obtain
\begin{align}
 \tilde{w}(\chi) & = (1-h(\chi))(1-g(\chi)) \wBBm + (1-h(\chi))g(\chi) \wBGm \nonumber \\
 & + h(\chi)(1-g(\chi)) \wGBm + h(\chi) g(\chi) \wGGm, \label{tildewchi-explicit}
\end{align}
where, by Eq.~\eqref{tildeomega} at $\epsilon=0$,
\begin{align*}
 \left(\tilde{\omega}_{\rf\rc}\right)_{\epsilon=0} & = v(\tilde{\mathsf{p}}_{\rf\rc},\mathsf{p}_{\rc\rf}) .
\end{align*}
Thus, $w(\chi)$ is determined by $g(\chi)$ and the resident payoffs $\wBB$, $\wBG$, $\wGB$, $\wGG$, and $\tilde{w}(\chi)$ additionally by $h(\chi)$ and the mutant payoffs $\wBBm$, $\wBGm$, $\wGBm$, $\wGGm$; each of these payoffs equals one of $R$, $S$, $T$, $P$, according to the actions the focal player and co-player take.

\paragraph{Stability criterion.}

We assess whether $p$ resists invasion by $\tilde{p}$ under $d$
by expanding $w(\chi)$ and $\tilde{w}(\chi)$ as $\chi \to 0$ and checking whether
the leading term of
\begin{equation} \label{Deltaw-def}
 \Delta w(\chi) := w(\chi) - \tilde{w}(\chi)
\end{equation}
is positive, so that condition~\eqref{evol-stab-cond} holds at $\epsilon=0$ for sufficiently small $\chi$.
Universal stability of $p$ under $d$ (Definition~\ref{def:universally-stable}) then requires this to hold for all $\tilde{p} \neq p$ across all cooperative dilemmas we study.
Since $g(\chi)$ takes a different form on each of $\AsetGXXG$, $\AsetBGGG$, and $\AsetGGGG$ (Eq.~\eqref{gchi-explicit}), we treat these classes separately.

\paragraph{Resident system.}

By Eq.~\eqref{righteous-set-simple}, a resident system $(p,d) \in \Rset$ satisfies
\begin{equation} \label{pGG=C}
 \pGG=\C
\end{equation}
and
\begin{subequations}
 \begin{align}
  \dGG &= \G, \label{dGG=G} \\
  \dBB+\dBG+\dGB &\geq 2 .
 \end{align}
\end{subequations}
Recalling the shorthand~\eqref{dr1r2-shorthand}, $\dGG = d(\G,\G,\pGG,\pGG)$, so Eqs.~\eqref{pGG=C} and~\eqref{dGG=G} together imply $\dGGCC = \G$.

\subsubsection{Invasion attempts by \texorpdfstring{$\ALLD$}{ALLD}}
\label{subsubsec:invasion-ALLD}

For the mutant $\tilde{p}=\ALLD$,
\begin{equation} \label{pm-ALLD}
  \pBBm = \pBGm = \pGBm = \pGGm = \D.
\end{equation}
We seek the subset $\Rset' \subseteq \Rset$ of righteous systems that resist such invasion across cooperative dilemmas.
This subset is characterized by the following two propositions, according to whether $\dGB=\G$ or $\dGB=\B$; both require $\pGB=\D$ and differ only in condition~(ii). We prove them below by a single argument.

\begin{proposition}[Stability against $\ALLD$ for righteous systems with $\dGB=\G$]
\label{prop:stability-against-ALLD-dGB=G}
Let $\epsilon=0$ and $\chi \to 0$.
A righteous system $(p,d)$ with $\dGB=\G$ resists invasion attempts by $\ALLD$ across cooperative dilemmas if and only if
\begin{enumerate}
 \item[(i)] $\pGB = \D$, and
 \item[(ii)] $\dBGDD = \dGGDC = \B$.
\end{enumerate}
\end{proposition}

\begin{proposition}[Stability against $\ALLD$ for righteous systems with $\dGB=\B$]
\label{prop:stability-against-ALLD-dGB=B}
Let $\epsilon=0$ and $\chi \to 0$.
A righteous system $(p,d)$ with $\dGB=\B$ resists invasion attempts by $\ALLD$ across cooperative dilemmas if and only if
\begin{enumerate}
 \item[(i)] $\pGB = \D$, and
 \item[(ii)] $\dBGDD = \B$ and either $\dGGDC = \B$ or $d(\B,\B,\D,\pBB) = \B$.
\end{enumerate}
\end{proposition}

\paragraph{Proof of Propositions~\ref{prop:stability-against-ALLD-dGB=G} and~\ref{prop:stability-against-ALLD-dGB=B}.}
Since any righteous system is self-approving (Definitions~\ref{def:self-approving} and~\ref{def:righteous}), $\lim_{\chi \to 0} g(\chi)=1$, and hence
\begin{equation} \label{limchi0w}
\lim_{\chi \to 0} w(\chi) = \wGG = v(\pGG,\pGG) = \vCC = R .
\end{equation}
Since $w(\chi) \to R$ as $\chi \to 0$, the leading term of $\Delta w(\chi)$ has the sign of $R - \lim_{\chi\to0}\tilde{w}(\chi)$: the resident resists $\ALLD$ when $\lim_{\chi\to0}\tilde{w}(\chi) < R$ and is invaded when $\lim_{\chi\to0}\tilde{w}(\chi) > R$.

As $\chi\to0$, $g(\chi)\to1$, so the terms of $\tilde{w}(\chi)$~\eqref{tildewchi-explicit} carrying a factor $1-g(\chi)$ vanish, leaving
\begin{equation*}
 \lim_{\chi\to0}\tilde{w}(\chi) = (1-h_0)\,\wBGm + h_0\,\wGGm ,
\end{equation*}
where
\begin{equation*}
 h_0 := \lim_{\chi\to0} h(\chi) .
\end{equation*}
Each limit is thus a convex combination of $\wBGm$ and $\wGGm$, with a weight $h_0$ that depends on the mutant bits $(\dBBm,\dBGm,\dGBm,\dGGm)$ and, through $g(\chi)$, on the resident class.
Since $\tilde{p}=\ALLD$ and $\pGG=\C$ by Eq.~\eqref{pGG=C}, the two payoffs are
\begin{subequations} \label{tildewij-ALLD}
\begin{align}
 \wBGm & = v(\pBGm,\pGB) = v(\D,\pGB) = \begin{cases}
      \vDD = P & \text{if } \pGB = \D \\
      \vDC = T & \text{if } \pGB = \C
   \end{cases},\\
 \wGGm & = v(\pGGm,\pGG) = \vDC = T .
\end{align}
\end{subequations}
Table~\ref{tab:limtildew-righteous} reports $\lim_{\chi\to0}\tilde{w}(\chi)$ for every combination of the mutant assessment bits, separately for the two mutually exclusive resident classes $(p,d)\in\AsetGXXG$ and $(p,d)\in\AsetBGGG\cup\AsetGGGG$.

\begin{table}
\caption{$\lim_{\chi\to 0}\tilde{w}(\chi)$ for invasion attempts by $\ALLD$ against righteous resident systems $(p,d)\in\Rset$. Rows enumerate all combinations of the mutant assessment bits $(\dBBm,\dBGm,\dGBm,\dGGm)$. The last two columns give $\lim_{\chi\to 0}\tilde{w}(\chi)$ separately for the two mutually exclusive resident classes $(p,d)\in\AsetGXXG$ and $(p,d)\in\AsetBGGG \cup \AsetGGGG$, in terms of $\wBGm$ and $\wGGm$ defined in Eqs.~\eqref{tildewij-ALLD}.}
\begin{center}
\begin{tabular}{c c c c | c | c}
 $\dBBm$ & $\dBGm$ & $\dGBm$ & $\dGGm$ & $(p,d) \in \AsetGXXG$ & $(p,d) \in \AsetBGGG \cup \AsetGGGG$ \\
  \hline
 $\B$ & $\B$ & $\B$ & $\B$ & $\wBGm$ & $\wBGm$ \\
 $\B$ & $\B$ & $\B$ & $\G$ & $\wBGm$ & $\frac{2}{3}\wBGm + \frac{1}{3}\wGGm$ \\
 $\B$ & $\B$ & $\G$ & $\B$ & $\wBGm$ & $\wBGm$ \\
 $\B$ & $\B$ & $\G$ & $\G$ & $\frac{1}{2} \wBGm + \frac{1}{2} \wGGm$ & $\frac{1}{2} \wBGm + \frac{1}{2} \wGGm$ \\
 $\B$ & $\G$ & $\B$ & $\B$ & $\frac{1}{2} \wBGm + \frac{1}{2} \wGGm$ & $\frac{1}{2} \wBGm + \frac{1}{2} \wGGm$ \\
 $\B$ & $\G$ & $\B$ & $\G$ & $\wGGm$ & $\wGGm$ \\
 $\B$ & $\G$ & $\G$ & $\B$ & $\frac{1}{2} \wBGm + \frac{1}{2} \wGGm$ & $\frac{1}{2} \wBGm + \frac{1}{2} \wGGm$ \\
 $\B$ & $\G$ & $\G$ & $\G$ & $\wGGm$ & $\wGGm$ \\
 $\G$ & $\B$ & $\B$ & $\B$ & $\wBGm$ & $\wBGm$ \\
 $\G$ & $\B$ & $\B$ & $\G$ & $\frac{1}{2} \wBGm + \frac{1}{2} \wGGm$ & $\frac{1}{2} \wBGm + \frac{1}{2} \wGGm$ \\
 $\G$ & $\B$ & $\G$ & $\B$ & $\wBGm$ & $\wBGm$ \\
 $\G$ & $\B$ & $\G$ & $\G$ & $\wGGm$ & $\frac{1}{3}\wBGm + \frac{2}{3} \wGGm$ \\
 $\G$ & $\G$ & $\B$ & $\B$ & $\frac{1}{2} \wBGm + \frac{1}{2} \wGGm$ & $\frac{1}{2} \wBGm + \frac{1}{2} \wGGm$ \\
 $\G$ & $\G$ & $\B$ & $\G$ & $\wGGm$ & $\wGGm$ \\
 $\G$ & $\G$ & $\G$ & $\B$ & $\frac{1}{2} \wBGm + \frac{1}{2} \wGGm$ & $\frac{1}{2} \wBGm + \frac{1}{2} \wGGm$ \\
 $\G$ & $\G$ & $\G$ & $\G$ & $\wGGm$ & $\wGGm$ \\
\end{tabular}
\end{center}
\label{tab:limtildew-righteous}
\end{table}

If $\pGB=\C$, then $\wBGm=\wGGm=T$ by Eq.~\eqref{tildewij-ALLD}; since $\lim_{\chi\to0}\tilde{w}(\chi)$ is a convex combination of these two payoffs, it equals $T$ regardless of the mutant bits and the resident class (Table~\ref{tab:limtildew-righteous}).
As $T>R$ for prisoner's dilemmas and snowdrift games, $\ALLD$ invades there, and a universally stable righteous system must therefore satisfy
\begin{equation} \label{pGB=D}
 \pGB=\D .
\end{equation}

Under condition~\eqref{pGB=D}, Eq.~\eqref{tildewij-ALLD} gives $\wBGm=P$ and $\wGGm=T$; the limit $\lim_{\chi\to0}\tilde{w}(\chi)=(1-h_0)P+h_0 T$ is then a convex combination of $P$ and $T$.
Since $R>P$~\eqref{R>P} always holds and $T>R$ holds for prisoner's dilemmas and snowdrift games (so $P<R<T$ there), Table~\ref{tab:limtildew-righteous} leaves three possibilities:
\begin{enumerate}
 \item $\lim_{\chi\to0}\tilde{w}(\chi)=T$: since $T>R$, $\ALLD$ invades for prisoner's dilemmas and snowdrift games.
 \item $\lim_{\chi\to0}\tilde{w}(\chi)$ lies strictly between $P$ and $T$: it then exceeds $R$ for some prisoner's dilemmas and snowdrift games, so $\ALLD$ can invade.
 \item $\lim_{\chi\to0}\tilde{w}(\chi)=P$: since $R>P$, the resident resists invasion by $\ALLD$ across all cooperative dilemmas.
\end{enumerate}

We conclude that a necessary and sufficient condition for $(p,d)\in\Rset$ to resist invasion attempts by $\ALLD$ across cooperative dilemmas is
\begin{equation} \label{limchi0tildew=wBGm}
 \lim_{\chi\to0}\tilde{w}(\chi)=\wBGm=P.
\end{equation}
Equivalently, $h_0=0$: the invading $\ALLD$ mutants acquire a bad reputation in the limit, so residents defect against them and they earn only $\wBGm=P<R$.
We now translate this condition into constraints on the bits of the resident action rule $p$ and assessment rule $d$.

We distinguish two mutually exclusive cases: (i) $(p,d)\in\AsetBGGG \cup \AsetGGGG$, and (ii) $(p,d)\in\AsetGXXG$.
In the latter, we also use the identity
\begin{equation} \label{dGBm=dGB}
 \dGBm = d(\G,\B,\pGBm,\pBG)=d(\G,\B,\pGB,\pBG)= \dGB,
\end{equation}
which follows from Eqs.~\eqref{pGB=D} and~\eqref{pm-ALLD}.

\begin{enumerate}
 \item Suppose that $(p,d)\in\AsetBGGG \cup \AsetGGGG$.
 Inspection of Table~\ref{tab:limtildew-righteous} shows that condition~\eqref{limchi0tildew=wBGm} holds if and only if
 \begin{equation} \label{dBGm=B}
    \dBGm = d(\B,\G,\pBGm,\pGB) = \dBGDD = \B
 \end{equation}
 and
 \begin{equation} \label{dGGm=B}
    \dGGm = d(\G,\G,\pGGm,\pGG) = \dGGDC = \B
 \end{equation}
 both hold.

 \item Suppose that $(p,d)\in\AsetGXXG$, so that $\dBB=\dGG=\G$ and $\dBG\neq\dGB$ (Eq.~\eqref{cond-SA-1**1}). By the identity~\eqref{dGBm=dGB}, $\dGBm=\dGB$, giving two subcases.

 \begin{enumerate}
  \item If $(\dBG,\dGB)=(\B,\G)$, then $\dGBm=\dGB=\G$, and Table~\ref{tab:limtildew-righteous} shows that condition~\eqref{limchi0tildew=wBGm} holds if and only if both~\eqref{dBGm=B} and~\eqref{dGGm=B} hold --- the same condition as in case~(i).

  \item If $(\dBG,\dGB)=(\G,\B)$, then $\dGBm=\dGB=\B$, and Table~\ref{tab:limtildew-righteous} shows that condition~\eqref{limchi0tildew=wBGm} holds if and only if~\eqref{dBGm=B} holds, together with either~\eqref{dGGm=B} or
  \begin{equation*}
   \dBBm = d(\B,\B,\pBBm,\pBB)=d(\B,\B,\D,\pBB)=\B.
  \end{equation*}
 \end{enumerate}
\end{enumerate}
Together with the necessary condition $\pGB=\D$~\eqref{pGB=D}, these cases give condition~(ii) of Proposition~\ref{prop:stability-against-ALLD-dGB=G} when $\dGB=\G$ (case~(i) and subcase~(ii-a)) and of Proposition~\ref{prop:stability-against-ALLD-dGB=B} when $\dGB=\B$ (subcase~(ii-b)).
\qed

These propositions have a further consequence for righteous systems with $\dBG=\G$, which we use repeatedly in the analyses that follow.

\begin{corollary}[Stability against $\ALLD$ forces $\pBG=\C$ when $\dBG=\G$]\label{cor:pBG=C}
Let $\epsilon=0$ and $\chi \to 0$.
If a righteous system $(p,d)$ with $\dBG=\G$ resists invasion attempts by $\ALLD$ across cooperative dilemmas, then $\pBG=\C$.
\end{corollary}
\begin{proof}
Suppose instead that $\pBG=\D$. Resistance to $\ALLD$ forces $\pGB=\D$~\eqref{pGB=D}, so that $\dBG=d(\B,\G,\pBG,\pGB)=d(\B,\G,\D,\D)=\dBGDD$; with $\dBG=\G$ this gives $\dBGDD=\G$, contradicting the condition $\dBGDD=\B$ that resistance requires (Propositions~\ref{prop:stability-against-ALLD-dGB=G} and~\ref{prop:stability-against-ALLD-dGB=B}).
\end{proof}

We denote by $\Rset'$ the subset of righteous systems that resist invasion by $\ALLD$ across cooperative dilemmas, as characterized by Propositions~\ref{prop:stability-against-ALLD-dGB=G} and~\ref{prop:stability-against-ALLD-dGB=B}.
Combining these two propositions with Eq.~\eqref{righteous-set-simple}:
\begin{align} \label{R'}
 \Rset' &=
  \left\{ (p,d) \in \Rset \,\middle|\, \pGB=\D,\ \dGB=\G,\ \dBGDD=\B,\ \dGGDC=\B \right\} \nonumber \\
  &\quad\cup
  \left\{ (p,d) \in \Rset \,\middle|\, \pGB=\D,\ \dGB=\B,\ \dBGDD=\B, \right. \nonumber \\
  &\hspace{6.5em} \left. \bigl[\dGGDC=\B \text{ or } d(\B,\B,\D,\pBB)=\B\bigr] \right\} .
\end{align}

Since systems in $\Rset'$ satisfy both $\pGG=\C$ and $\pGB=\D$, every $(p,d)\in\Rset'$ has a discriminating action rule.
This characterization lets us immediately locate two families of systems of special interest within the subset $\Rset'$.
First, we can verify that all leading-eight systems satisfy the conditions defining $\Rset'$.
Second, among the four subsets of consistent discriminating systems, two belong to $\Rset'$ whereas the other two are excluded already at this stage.

\begin{proposition}[Leading-eight systems belong to $\Rset'$]
\label{prop:L-subset-R'}
$\Lset \subset \Rset'$.
\end{proposition}
\begin{proof}
All leading-eight action rules are discriminating, so $\pGB=\D$, and their assessment rules are third-order with $\dBGD=\B$, $\dGGD=\B$, and $\dGBD=\G$ (Eqs.~\eqref{LI-set}--\eqref{LIII-set}).
Since the norm ignores the co-player's action, $\dBGDD=\B$, $\dGGDC=\B$, and $\dGB=d(\G,\B,\D,\pBG)=\dGBD=\G$.
Together with $\Lset\subset\Rset$ (Proposition~\ref{prop:leading-eight-righteous}), these satisfy the $\dGB=\G$ branch of Eq.~\eqref{R'}.
\end{proof}

\begin{proposition}[Consistent discriminating systems and $\Rset'$]
\label{prop:QD-R'-membership}
$\QCDisc \cup \QDisc \subset \Rset'$, while $\QPDisc \cap \Rset' = \QODisc \cap \Rset' = \emptyset$.
\end{proposition}
\begin{proof}
For $\QCDisc$ and $\QDisc$: both are discriminating with $\pBG=\C$, so that $\dGB=d(\G,\B,\D,\C)=\dGBDC$; Eqs.~\eqref{QCDisc-set} and~\eqref{QDisc-set} give $\dBGDD=\B$, $\dGGDC=\B$, and $\dGBDC=\G$, whence $\dGB=\G$ and the $\dGB=\G$ branch of Eq.~\eqref{R'} is satisfied.
For $\QPDisc$ and $\QODisc$: Eqs.~\eqref{QPDisc-set} and~\eqref{QODisc-set} include $\dBGDD=\G$, which violates the necessary condition $\dBGDD=\B$ for $\Rset'$ membership.
\end{proof}

The characterization~\eqref{R'} also determines the cardinality of $\Rset'$.

\begin{proposition}[Cardinality of $\Rset'$] \label{prop:cardinality-R'}
The set $\Rset'$ contains $|\Rset'| = \numprint{12800}$ moral systems.
\end{proposition}
\begin{proof}
We count the admissible systems as in the proof of Proposition~\ref{prop:cardinality-R}, recalling that a condition leaving $k$ bits of $d$ unconstrained is satisfied by $2^k$ assessment rules.
By Eq.~\eqref{R'} every system in $\Rset'$ has $\pGG=\C$ and $\pGB=\D$, leaving four action rules indexed by $(\pBB,\pBG)\in\{\C,\D\}^2$; Eq.~\eqref{R'} is moreover the disjoint union of a branch with $\dGB=\G$ and one with $\dGB=\B$. For each action rule and branch, the defining conditions fix or restrict a few bits of $d$ and leave the rest free. For instance, in the branch $\dGB=\G$ with $\pBG=\C$, the bits $\dGB=d(\G,\B,\D,\C)$, $\dBG=d(\B,\G,\C,\D)$ and $\dBGDD=d(\B,\G,\D,\D)$ are distinct; the conditions $\dGG=\G$, $\dGGDC=\B$, $\dGB=\G$, $\dBGDD=\B$ fix four bits, and $\dBB+\dBG+\dGB\ge2$ (with $\dGB=\G$) admits three of the four assignments of $(\dBB,\dBG)$, giving $3\cdot 2^{10}=\numprint{3072}$. Coincidences among the constrained bits occur when $\pBG=\D$ (then $\dBG=\dBGDD$) or $\pBB=\D$ (then $\dBB=d(\B,\B,\D,\pBB)$) and lower the counts; in particular, the $\dGB=\B$ branch is empty when $\pBG=\D$, since then $\dBG=\dBGDD=\B$ and $\dGB=\B$ would violate $\dBB+\dBG+\dGB\ge2$. Table~\ref{tab:cardinality-R'} collects all cases; they sum to $|\Rset'|=\numprint{12800}$.
\end{proof}

\begin{table}
\caption{Number of assessment rules $d$ yielding a system in $\Rset'$, for each of the four action rules $(\pBB,\pBG)$ compatible with $\pGG=\C$, $\pGB=\D$, split by the two branches ($\dGB=\G$ and $\dGB=\B$) of Eq.~\eqref{R'}.}
\begin{center}
\begin{tabular}{c c | c c | c}
 $\pBB$ & $\pBG$ & $\dGB=\G$ & $\dGB=\B$ & total \\
 \hline
 $\C$ & $\C$ & \numprint{3072} & \numprint{1536} & \numprint{4608} \\
 $\C$ & $\D$ & \numprint{2048} & $0$ & \numprint{2048} \\
 $\D$ & $\C$ & \numprint{3072} & \numprint{1024} & \numprint{4096} \\
 $\D$ & $\D$ & \numprint{2048} & $0$ & \numprint{2048} \\
 \hline
 \multicolumn{4}{r|}{$|\Rset'|=$} & \numprint{12800} \\
\end{tabular}
\end{center}
\label{tab:cardinality-R'}
\end{table}

\subsubsection{Invasion attempts by \texorpdfstring{$\ALLC$}{ALLC}}
\label{subsubsec:invasion-ALLC}

For the mutant $\tilde{p}=\ALLC$,
\begin{equation}
  \pBBm = \pBGm = \pGBm = \pGGm = \C. \label{pm-ALLC}
\end{equation}
We seek which systems in $\Rset'$ also resist such invasion across cooperative dilemmas.
These are characterized by the following three propositions: the first shows that resistance requires $\dGB=\G$; the other two then characterize the systems with $\dGB=\G$, according to whether $\dBG=\B$ or $\dBG=\G$. We prove all three together below.

\begin{proposition}[Stability against $\ALLC$ requires $\dGB=\G$] \label{prop:stability-against-ALLC-dGB=G}
Let $\epsilon=0$ and $\chi \to 0$.
If a righteous system $(p,d)$ that resists invasion attempts by $\ALLD$ also resists invasion attempts by $\ALLC$ across cooperative dilemmas, then $\dGB = \G$.
\end{proposition}

Proposition \ref{prop:stability-against-ALLC-dGB=G} implies that $(p,d) \in \Rset \cap \AsetGGBG$ cannot be universally stable.

\begin{proposition}[Stability against $\ALLC$ for righteous systems with $\dGB=\G$ and $\dBG=\B$] \label{prop:stability-against-ALLC-dBG=B}
Let $\epsilon=0$ and $\chi \to 0$.
A righteous system $(p,d)$ with $\dGB=\G$ and $\dBG=\B$ that resists invasion attempts by $\ALLD$ also resists invasion attempts by $\ALLC$ across cooperative dilemmas if and only if
\begin{enumerate}
 \item[(i)] $\dBGCD = \B$, and
 \item[(ii)] either $d(\G,\B,\C,\pBG) = \B$ or $d(\B,\B,\C,\pBB) = \B$.
\end{enumerate}
\end{proposition}

\begin{proposition}[Stability against $\ALLC$ for righteous systems with $\dGB=\G$ and $\dBG=\G$] \label{prop:stability-against-ALLC-dBG=G}
Let $\epsilon=0$ and $\chi \to 0$.
A righteous system $(p,d)$ with $\dGB=\G$ and $\dBG=\G$ that resists invasion attempts by $\ALLD$ also resists invasion attempts by $\ALLC$ across cooperative dilemmas if and only if $\dGBCC = \B$.
\end{proposition}

\paragraph{Proof of Propositions~\ref{prop:stability-against-ALLC-dGB=G}, \ref{prop:stability-against-ALLC-dBG=B}, and~\ref{prop:stability-against-ALLC-dBG=G}.}
Since the resident is discriminating, the resident payoffs are
\begin{subequations}
 \begin{align}
  \wBG & = v(\pBG,\D) = \begin{cases}
      \vDD = P & \text{if } \pBG = \D \\
      \vDC = S & \text{if } \pBG = \C
   \end{cases}, \label{wBG-disc} \\
  \wGB & = v(\D,\pBG) = \begin{cases}
      \vDD = P & \text{if } \pBG = \D \\
      \vDC = T & \text{if } \pBG = \C
   \end{cases}, \\
    \wGG & = v(\pGG,\pGG) = \vCC = R,
 \end{align}
\end{subequations}
while the $\ALLC$ mutant always cooperates, so the mutant payoffs are
\begin{align*}
 \wBGm & = v(\pBGm,\pGB) = \vCD = S,\\
 \wGBm & = v(\pGBm,\pBG) = v(\C,\pBG) = \begin{cases}
      \vCD = S & \text{if } \pBG = \D \\
      \vCC = R & \text{if } \pBG = \C
   \end{cases},\\
 \wGGm & = v(\pGGm,\pGG) = \vCC = R.
\end{align*}

In both cases $g(\chi)\to1$ as $\chi\to0$, so $w(\chi)\to\wGG=R$; the resident resists $\ALLC$ exactly when the leading term of $\Delta w(\chi)=w(\chi)-\tilde{w}(\chi)$ is positive, which we obtain by expanding $w(\chi)$ and $\tilde{w}(\chi)$ as $\chi\to0$.
Because $g(\chi)$ takes different forms on $\AsetGXXG$ and on $\AsetBGGG\cup\AsetGGGG$ (Eq.~\eqref{gchi-explicit}), we treat cases (i) and (ii) separately.
Since $\pGGm=\pGG$ (Eqs.~\eqref{pGG=C} and~\eqref{pm-ALLC}) and $\dGG=\G$~\eqref{dGG=G}, we have
\begin{equation} \label{dGGm=G-ALLC}
 \dGGm = d(\G,\G,\pGGm,\pGG) = d(\G,\G,\pGG,\pGG) = \dGG = \G
\end{equation}
in both cases.

\begin{enumerate}
 \item Suppose that $(p,d) \in \AsetGXXG$.
The first-order Puiseux series expansion of $w(\chi)$~\eqref{wchi-explicit} is then
\begin{equation} \label{wchi-expansion-ALLC-AGXXG}
 w(\chi) = R + \left[ \wBG + \wGB - 2 R \right] \sqrt{\chi} + \BigO(\chi).
\end{equation}
Table~\ref{tab:expansion-tildew-ALLC-AGXXG} reports the corresponding expansion of $\tilde{w}(\chi)$ for each combination of the bits $\dBBm$, $\dBGm$, and $\dGBm$, using $\dGGm=\G$~\eqref{dGGm=G-ALLC}, $\wGGm=\vCC=R$, and $\wBGm=\vCD=S$.
The entries take three forms: $\tilde{w}(\chi)$ may have a leading term below $R$ (the rows equal to $S$ or $(S+R)/2$), so that the resident resists; or it may match the leading term $R$ of $w(\chi)$, so that the $\sqrt{\chi}$-coefficients decide (the two remaining forms).

\begin{table}
\caption{Puiseux-series expansion of $\tilde{w}(\chi)$ for invasion attempts by $\ALLC$ against resident systems in $(p,d)\in\AsetGXXG$. Rows enumerate all possible combinations of the mutant assessment bits $(\dBBm,\dBGm,\dGBm,\dGGm)$ consistent with $\dGGm=\G$ in Eq.~\eqref{dGGm=G-ALLC}. Entries are written using $\wGGm=\vCC=R$ and $\wBGm=\vCD=S$.}
\begin{center}
\begin{tabular}{c c c c | c}
 $\dBBm$ & $\dBGm$ & $\dGBm$ & $\dGGm$ & $\tilde{w}(\chi)$ \\
  \hline
 $\B$ & $\B$ & $\B$ & $\G$ & $S + \BigO(\sqrt{\chi})$ \\
 $\B$ & $\B$ & $\G$ & $\G$ & $(S+R)/2 + \BigO(\sqrt{\chi})$ \\
 $\B$ & $\G$ & $\B$ & $\G$ & $R + \left[S + \wGBm - 2 R \right]\sqrt{\chi} + \BigO(\chi)$ \\
 $\B$ & $\G$ & $\G$ & $\G$ & $R + \left[\wGBm - R \right] \sqrt{\chi} + \BigO(\chi)$ \\
 $\G$ & $\B$ & $\B$ & $\G$ & $(S+R)/2 + \BigO(\sqrt{\chi})$ \\
 $\G$ & $\B$ & $\G$ & $\G$ & $R + \left[S + \wGBm - 2 R \right]\sqrt{\chi} + \BigO(\chi)$ \\
 $\G$ & $\G$ & $\B$ & $\G$ & $R + \left[S + \wGBm - 2 R \right]\sqrt{\chi} + \BigO(\chi)$ \\
 $\G$ & $\G$ & $\G$ & $\G$ & $R + \left[\wGBm - R \right] \sqrt{\chi} + \BigO(\chi)$ \\
\end{tabular}
\end{center}
\label{tab:expansion-tildew-ALLC-AGXXG}
\end{table}

\begin{enumerate}
 \item The expansion of $\tilde{w}(\chi)$ is either $S + \BigO(\sqrt{\chi})$ or $(S+R)/2 + \BigO(\sqrt{\chi})$.
 In both cases, the leading term of $\tilde{w}(\chi)$ is strictly smaller than $R$, the leading term of $w(\chi)$ in Eq.~\eqref{wchi-expansion-ALLC-AGXXG}, because $R>S$ holds across cooperative dilemmas~\eqref{R>S}.
 Thus, the resident system $(p,d)$ resists invasion attempts by $\ALLC$.
 \item The expansion of $\tilde{w}(\chi)$ is given by
 $R + \left[S + \wGBm - 2 R \right]\sqrt{\chi} + \BigO(\chi)$.
 The leading term of $\Delta w(\chi)$ is
 \begin{equation} \label{wBG+wGB-wBGm-wGBm}
  \wBG + \wGB - S - \wGBm .
 \end{equation}
 If $\pBG=\D$, then $\wBG=\wGB=P$ and $\wGBm=S$, so Eq.~\eqref{wBG+wGB-wBGm-wGBm} reduces to $2(P-S)$, which is negative for snowdrift games.
 If $\pBG=\C$, then $\wBG=S$, $\wGB=T$, and $\wGBm=R$, so Eq.~\eqref{wBG+wGB-wBGm-wGBm} reduces to $T-R$, which is negative for stag hunts.
 In either case, the resident system is not universally stable.
 \item The expansion of $\tilde{w}(\chi)$ is given by $R + \left[\wGBm - R \right] \sqrt{\chi} + \BigO(\chi)$.
 The leading term of $\Delta w(\chi)$ is
 \begin{equation} \label{wBG+wGB-R-wGBm}
  \wBG + \wGB - R - \wGBm .
 \end{equation}
 If $\pBG=\D$, then $\wBG=\wGB=P$ and $\wGBm=S$, so Eq.~\eqref{wBG+wGB-R-wGBm} reduces to $2P-R-S$, which fails to be positive unless $P>(S+R)/2$; this condition is violated for all snowdrift games and for some prisoner's dilemmas and stag hunts.
 If $\pBG=\C$, then $\wBG=S$, $\wGB=T$, and $\wGBm=R$, so Eq.~\eqref{wBG+wGB-R-wGBm} reduces to $S+T-2R<0$ by condition~\eqref{2R>T+S}, so $\ALLC$ always invades.
 In either case, the resident system is not universally stable.
\end{enumerate}

Within the class $(p,d) \in \AsetGXXG$, the only case in which the resident system resists invasion is case~1 above.
Inspecting Table~\ref{tab:expansion-tildew-ALLC-AGXXG}, this holds if and only if
\begin{equation} \label{dBGm=B-ALLC}
 \dBGm = d(\B,\G,\pBGm,\pGB) = \dBGCD = \B
\end{equation}
and either $\dGBm=d(\G,\B,\C,\pBG)=\B$ or $\dBBm=d(\B,\B,\C,\pBB)=\B$ also holds.
Furthermore, condition~\eqref{dBGm=B-ALLC} forces $\dBG=\B$: if instead $\dBG=\G$, then either $\pBG=\C$ (giving $\dBGm=\dBG=\G$, contradicting condition~\eqref{dBGm=B-ALLC}) or $\pBG=\D$ (giving $\dBGDD=\G$, contradicting the $\ALLD$-stability condition $\dBGDD=\B$ from Propositions~\ref{prop:stability-against-ALLD-dGB=G} and~\ref{prop:stability-against-ALLD-dGB=B}).
Hence, within $\AsetGXXG$, resistance against $\ALLC$ requires $(p,d) \in \AsetGBGG$, i.e., $\dGB=\G$.

\item Suppose that $(p,d) \in \AsetBGGG \cup \AsetGGGG$.
Then $\dBG=\G$ holds.
By Corollary~\ref{cor:pBG=C}, this forces $\pBG=\C$, and therefore $\wBG=\wBGm=\vCD=S$, $\wGB=\vDC=T$, and $\wGBm=\vCC=R$.
Using also Eq.~\eqref{limchi0w}, the first-order Taylor series expansion of $w(\chi)$ in Eq.~\eqref{wchi-explicit} reduces to
\begin{equation*}
 w(\chi) = R + \left( S + T - 2 R \right) \chi + \BigO(\chi^2).
\end{equation*}

Table~\ref{tab:expansion-tildew-ALLC-A*GGG} reports the corresponding expansion of $\tilde{w}(\chi)$ for each combination of $\dBBm$ and $\dGBm$, using $\dGGm=\G$~\eqref{dGGm=G-ALLC} and $\dBGm=\dBG=\G$ (since $\pBG=\C$).

\begin{table}
\caption{Taylor-series expansion of $\tilde{w}(\chi)$ for invasion attempts by $\ALLC$ against resident systems in $(p,d)\in\AsetBGGG \cup \AsetGGGG$. Rows enumerate all possible combinations of the mutant assessment bits $(\dBBm,\dBGm,\dGBm,\dGGm)$ consistent with $\dBGm=\dGGm=\G$.}
\begin{center}
\begin{tabular}{c c c c | c}
 $\dBBm$ & $\dBGm$ & $\dGBm$ & $\dGGm$ & $\tilde{w}(\chi)$ \\
  \hline
 $\B$ & $\G$ & $\B$ & $\G$ & $R + 2 \left(S - R \right)\chi + \BigO(\chi^2)$ \\
 $\B$ & $\G$ & $\G$ & $\G$ & $R + (S-R) \chi + \BigO(\chi^2)$ \\
 $\G$ & $\G$ & $\B$ & $\G$ & $R + 2 \left(S - R \right)\chi + \BigO(\chi^2)$ \\
 $\G$ & $\G$ & $\G$ & $\G$ & $R + (S-R) \chi + \BigO(\chi^2)$ \\
\end{tabular}
\end{center}
\label{tab:expansion-tildew-ALLC-A*GGG}
\end{table}

Both entries of Table~\ref{tab:expansion-tildew-ALLC-A*GGG} match the leading term $R$ of $w(\chi)$, so the $\chi$-coefficients decide; there are two cases.

\begin{enumerate}
 \item The expansion of $\tilde{w}(\chi)$ is given by $R + (S-R) \chi + \BigO(\chi^2)$. In this case, the leading term of $\Delta w(\chi)$ is $T-R$, which is negative for stag hunts. Thus, the resident system can be invaded by $\ALLC$ for that class of cooperative dilemmas, and hence is not universally stable.
 \item The expansion of $\tilde{w}(\chi)$ is given by
 $R + 2(S-R)\chi + \BigO(\chi^2)$.
 In this case, the leading term of $\Delta w(\chi)$ is given by $T-S$, which is positive across cooperative dilemmas since $T>S$ always holds.
 Hence, the resident system $(p,d)$ resists invasion attempts by $\ALLC$.
\end{enumerate}

Within the class $(p,d) \in \AsetBGGG \cup \AsetGGGG$, the resident system resists invasion only in case~2 above.
By inspection of Table~\ref{tab:expansion-tildew-ALLC-A*GGG}, this holds if and only if
\begin{equation*}
 \dGBm = d(\G,\B,\pGBm,\pBG) = \dGBCC = \B.
\end{equation*}
\end{enumerate}

Together, the two cases establish the three propositions: case~(i) rules out $\AsetGGBG$ --- so resistance requires $\dGB=\G$ (Proposition~\ref{prop:stability-against-ALLC-dGB=G}) --- and characterizes the $\dBG=\B$ systems of $\AsetGBGG$ (Proposition~\ref{prop:stability-against-ALLC-dBG=B}), while case~(ii) characterizes the systems with $\dGB=\dBG=\G$ via $\dGBCC=\B$ (Proposition~\ref{prop:stability-against-ALLC-dBG=G}).
\qed

We denote by $\Rset''$ the subset of $\Rset'$ that resists invasion by $\ALLC$ across cooperative dilemmas, as characterized by Propositions~\ref{prop:stability-against-ALLC-dGB=G},~\ref{prop:stability-against-ALLC-dBG=B} and~\ref{prop:stability-against-ALLC-dBG=G}.
Proposition~\ref{prop:stability-against-ALLC-dGB=G} forces $\dGB=\G$ for all systems in $\Rset''$, and the remaining conditions split on $\dBG$:
\begin{align} \label{R''}
 \Rset'' &=
  \left\{ (p,d) \in \Rset' \,\middle|\, \dGB=\G,\ \dBG=\B,\ \dBGCD=\B, \right. \nonumber \\
  &\hspace{6.5em} \left. \bigl[d(\G,\B,\C,\pBG)=\B \text{ or } d(\B,\B,\C,\pBB)=\B\bigr] \right\} \nonumber \\
  &\quad\cup
  \left\{ (p,d) \in \Rset' \,\middle|\, \dGB=\G,\ \dBG=\G,\ \dGBCC=\B \right\} .
\end{align}

As for $\Rset'$, we locate the leading-eight and consistent discriminating systems within $\Rset''$.
Resistance to $\ALLC$ now removes half of the leading eight, whereas the consistent discriminating families split exactly as they did for $\Rset'$.

\begin{proposition}[Leading-eight systems and $\Rset''$]
\label{prop:L-R''-membership}
A leading-eight system belongs to $\Rset''$ if and only if $\dGBC=\B$.
The systems satisfying this condition are $\{\text{L2}, \text{L5}, \text{L6}, \text{L8}\}$; the remaining four systems $\{\text{L1}, \text{L3}, \text{L4}, \text{L7}\}$ belong to $\Rset'\setminus\Rset''$.
\end{proposition}
\begin{proof}
By Proposition~\ref{prop:L-subset-R'}, all leading-eight systems belong to $\Rset'$ with $\dGB=\G$.
Since all have $\dBGC=\G$ (Eqs.~\eqref{LI-set}--\eqref{LIII-set}), it follows that $\dBG=\dBGCD=\G$, so every leading-eight system falls into the second branch of Eq.~\eqref{R''}, which requires $\dGBCC=\B$.
Since the norm ignores the co-player's action, $\dGBCC=\dGBC$, so the condition $\dGBCC=\B$ is just $\dGBC=\B$.
Inspecting the definitions: L2, L5, L6, and L8 have $\dGBC=\B$; L1, L3, L4, and L7 have $\dGBC=\G$.
\end{proof}

\begin{proposition}[Consistent discriminating systems and $\Rset''$]
\label{prop:QD-R''-membership}
$\QCDisc \cup \QDisc \subset \Rset''$, while $\QPDisc \cap \Rset'' = \QODisc \cap \Rset'' = \emptyset$.
\end{proposition}
\begin{proof}
Since $\QPDisc \cap \Rset' = \QODisc \cap \Rset' = \emptyset$ (Proposition~\ref{prop:QD-R'-membership}) and $\Rset''\subset\Rset'$, the same holds for $\Rset''$.
For $\QCDisc$ and $\QDisc$: both have $\pBG=\C$, so $\dBG=d(\B,\G,\C,\D)=\dBGCD$ and $\dGB=d(\G,\B,\D,\C)=\dGBDC$; Eqs.~\eqref{QCDisc-set} and~\eqref{QDisc-set} give $\dBGCD=\G$, $\dGBDC=\G$, and $\dGBCC=\B$, whence $\dBG=\dGB=\G$ and the second branch of Eq.~\eqref{R''} is satisfied.
\end{proof}

The characterization~\eqref{R''} likewise determines the cardinality of $\Rset''$.

\begin{proposition}[Cardinality of $\Rset''$] \label{prop:cardinality-R''}
The set $\Rset''$ contains $|\Rset''| = \numprint{4608}$ moral systems.
\end{proposition}
\begin{proof}
We count the admissible systems as before, recalling that a condition leaving $k$ bits of $d$ unconstrained is satisfied by $2^k$ assessment rules.
Both branches of Eq.~\eqref{R''} require $\dGB=\G$, so $\Rset''$ lies within the $\dGB=\G$ branch of $\Rset'$; on top of the conditions imposed there it adds constraints that split on $\dBG$. For instance, with $\pBG=\C$ and $\dBG=\G$: since $\pBG=\C$ we have $\dBG=d(\B,\G,\C,\D)$, so $\dBG=\G$ and $\dGBCC=\B$ fix two further bits, the majority condition $\dBB+\dBG+\dGB\ge2$ is then automatic and $\dBB$ is free; six bits are fixed and ten free, giving $2^{10}=\numprint{1024}$. The $\dBG=\G$ branch is moreover empty when $\pBG=\D$, since then $\dBG=\dBGDD$, which $\Rset'$ already forces to $\B$. Table~\ref{tab:cardinality-R''} collects all cases; they sum to $|\Rset''|=\numprint{4608}$.
\end{proof}

A direct enumeration of all $14 \times 2^{16}$ conditional systems confirms the cardinalities $|\Rset|$, $|\Rset'|$, and $|\Rset''|$ reported in Propositions~\ref{prop:cardinality-R}--\ref{prop:cardinality-R''} and in Fig.~\ref{fig:venn_analytic_result}.

\begin{table}
\caption{Number of assessment rules $d$ yielding a system in $\Rset''$, for each of the four action rules $(\pBB,\pBG)$, split by the two branches ($\dBG=\G$ and $\dBG=\B$) of Eq.~\eqref{R''}.}
\begin{center}
\begin{tabular}{c c | c c | c}
 $\pBB$ & $\pBG$ & $\dBG=\G$ & $\dBG=\B$ & total \\
 \hline
 $\C$ & $\C$ & \numprint{1024} & $512$ & \numprint{1536} \\
 $\C$ & $\D$ & $0$ & $512$ & $512$ \\
 $\D$ & $\C$ & \numprint{1024} & $768$ & \numprint{1792} \\
 $\D$ & $\D$ & $0$ & $768$ & $768$ \\
 \hline
 \multicolumn{4}{r|}{$|\Rset''|=$} & \numprint{4608} \\
\end{tabular}
\end{center}
\label{tab:cardinality-R''}
\end{table}

\subsubsection{Invasion attempts by a discriminating mutant}
\label{subsubsec:invasion-disc}

We now determine which systems in $\Rset''$ also resist invasion by a discriminating mutant $\tilde{p} \neq p$ across cooperative dilemmas. Such a mutant satisfies
\begin{equation} \label{pm-disc}
  \pGGm = \C \quad \text{and} \quad \pGBm = \D.
\end{equation}
The main results are the following two propositions, which we prove in turn below.

\begin{proposition}[Invasion by discriminating mutants for righteous systems with $\dBG=\B$ or $\dBB=\B$] \label{prop:dBG=B-dBB=B-invaded}
Let $\epsilon=0$ and $\chi \to 0$.
A righteous system $(p,d)$ with $\dBG=\B$ or $\dBB=\B$ that resists invasion attempts by $\ALLD$ and $\ALLC$ cannot resist invasion attempts by at least one alternative discriminating strategy.
\end{proposition}

Propositions~\ref{prop:stability-against-ALLC-dGB=G} and~\ref{prop:dBG=B-dBB=B-invaded} together imply that $(p,d) \in \Rset \cap \left(\AsetGGBG \cup \AsetGBGG \cup \AsetBGGG \right)$ cannot be universally stable; among the righteous classes, this leaves only $\AsetGGGG$, which the next proposition characterizes.

\begin{proposition}[Stability against discriminating mutants for righteous systems in $\AsetGGGG$] \label{prop:stability-against-alternative-disc}
Let $\epsilon=0$ and $\chi \to 0$.
A righteous system $(p,d)$ with $\dBB=\dBG=\dGB=\dGG=\G$ that resists invasion attempts by $\ALLD$ and $\ALLC$ also resists invasion attempts by $\tilde{p} \in \left\{\DISC,\CDISC\right\}$ with $\tilde{p} \neq p$ across cooperative dilemmas if and only if $\pBB = \C$ and $\dBBDC = \B$.
\end{proposition}

\paragraph{Preliminary observations.}
Since both $p$ and $\tilde{p}$ are discriminating, the following equalities hold:
\begin{subequations} \label{wGGm-wGBm-discs}
\begin{align}
 \wGGm & = \wGG = R,\\
 \wGBm & = \wGB = v(\pGB,\pBG) = v(\D,\pBG).
\end{align}
\end{subequations}
Moreover, $\pGGm=\pGG$ and $\pGBm=\pGB$, so from $\dGG=\G$~\eqref{dGG=G} and $\dGB=\G$ (Proposition~\ref{prop:stability-against-ALLC-dGB=G}) we obtain
\begin{equation} \label{dGGm=G-discs}
 \dGGm = d(\G,\G,\pGGm,\pGG) = d(\G,\G,\pGG,\pGG) = \dGG = \G
\end{equation}
and
\begin{equation} \label{dGBm=G-discs}
 \dGBm = d(\G,\B,\pGBm,\pBG) = d(\G,\B,\pGB,\pBG) = \dGB = \G.
\end{equation}
These two equalities, $\dGGm=\dGBm=\G$, hold throughout the analysis below. The mutant assessment bits $\dBGm$ and $\dBBm$ depend on which class of systems is considered, and are determined case by case.

\begin{proof}[Proof of Proposition~\ref{prop:dBG=B-dBB=B-invaded}]
We show that for each of the classes $\AsetGBGG$ and $\AsetBGGG$ there exists a discriminating mutant that invades for some cooperative dilemma.

\medskip

\noindent\textbf{Case $\boldsymbol{\AsetGBGG}$ (i.e., $\boldsymbol{\dBG=\B}$).}
The Puiseux-series expansion of $w(\chi)$~\eqref{wchi-explicit} is
\begin{equation} \label{wchi-expansion-GBGG}
 w(\chi) = \wGG + \left[ \wBG+\wGB-2\wGG\right] \sqrt{\chi} + \BigO(\chi) .
\end{equation}
By Propositions~\ref{prop:stability-against-ALLD-dGB=G} and~\ref{prop:stability-against-ALLC-dBG=B}, $\dBGDD = \dBGCD = \B$ holds. Hence, regardless of $\pBGm$,
\begin{equation} \label{dBGm=B-discs-AGBGG}
 \dBGm = d(\B,\G,\pBGm,\pGB) = \B.
\end{equation}
Consider the discriminating mutant $\tilde{p}$ with $\pBGm \neq \pBG$ and $\pBBm = \pBB$.
Then $\dBBm = \dBB = \G$, and the mutant assessment bits are $(\dBBm,\dBGm,\dGBm,\dGGm)=(\G,\B,\G,\G)$, so $\tilde{w}(\chi)$ has the same expansion form as $w(\chi)$.
The leading term of $\Delta w(\chi)$ reduces to
\begin{equation} \label{wBG-wBGm-discs-AGBGG}
 \wBG - \wBGm,
\end{equation}
which equals $P-S$ if $\pBG=\D$ (negative for snowdrift games) or $S-P$ if $\pBG=\C$ (negative for prisoner's dilemmas and stag hunts), by Eq.~\eqref{wBG-disc}.
Hence $(p,d) \in \AsetGBGG$ cannot be universally stable.

\medskip

\noindent\textbf{Case $\boldsymbol{\AsetBGGG}$ (i.e., $\boldsymbol{\dBG=\G}$, $\boldsymbol{\dBB=\B}$).}
Since $\dBG=\G$, Corollary~\ref{cor:pBG=C} gives $\pBG=\C$.
Consider a discriminating mutant $\tilde{p}$ with $\pBGm = \pBG$ but $\pBBm \neq \pBB$.
Then
\begin{equation} \label{dBGm=G-discs-A*GGG}
 \dBGm = d(\B,\G,\pBG,\pGB) = \dBG = \G
\end{equation}
and
\begin{equation} \label{wBGm-discs-A*GGG}
 \wBGm = \wBG = \vCD = S.
\end{equation}
The second-order expansion of $w(\chi)$ is
\begin{align}
 w(\chi) & = \wGG + \left[ \wBG+\wGB-2\wGG\right] \chi \nonumber \\
 & + \left[\wBB - \wGG \right] \chi^2 + \BigO(\chi^3) . \label{wchi-expansion-BGGG}
\end{align}
Since $\dGGm=\dGBm=\dBGm=\G$, the mutant assessment bits are $(\dBBm,\G,\G,\G)$ with $\dBBm \in \{\G,\B\}$.
Computing $\tilde{w}(\chi)$ from \eqref{tildewchi-explicit} and subtracting, the leading term of $\Delta w(\chi)$ at order $\chi^2$ is reported in Table~\ref{tab:Deltaw-disc-ABGGG} for each subcase.

\begin{table}
\caption{Leading term of $\Delta w(\chi)$ for a discriminating mutant with $\pBBm \neq \pBB$ invading a resident system in $\AsetBGGG$. In each row, the expression is evaluated for $\pBB = \C$ (so $\pBBm = \D$) and $\pBB = \D$ (so $\pBBm = \C$).}
\begin{center}
\begin{tabular}{c | c | c c}
 $\dBBm$ & Leading term & $\pBB = \C$ & $\pBB = \D$ \\
 \hline
 $\G$ & $\wBB - R - \wBBm + S$ & $S-T$ & $P-R$ \\
 $\B$ & $\wBB - \wBBm$ & $R-T$ & $P-S$
\end{tabular}
\end{center}
\label{tab:Deltaw-disc-ABGGG}
\end{table}

When $\dBBm=\G$, the leading term evaluates to $S-T < 0$ (if $\pBB=\C$) or $P-R < 0$ (if $\pBB=\D$), both always negative. When $\dBBm=\B$, it evaluates to $R-T$ (if $\pBB=\C$), which is negative for prisoner's dilemmas and snowdrift games, or to $P-S$ (if $\pBB=\D$), which is negative for snowdrift games. In every case, the mutant invades for at least one class of cooperative dilemmas.
Hence $(p,d) \in \AsetBGGG$ cannot be universally stable.
\end{proof}

\begin{proof}[Proof of Proposition~\ref{prop:stability-against-alternative-disc}]
Suppose $(p,d) \in \AsetGGGG$, i.e., $\dBB=\dBG=\dGB=\dGG=\G$.
As in the $\AsetBGGG$ case, Corollary~\ref{cor:pBG=C} gives $\pBG = \C$, and a discriminating mutant $\tilde{p}$ with $\pBGm = \pBG$ and $\pBBm \neq \pBB$ satisfies Eqs.~\eqref{dBGm=G-discs-A*GGG} and~\eqref{wBGm-discs-A*GGG}. The second-order expansion of $w(\chi)$ is
\begin{align}
 w(\chi) & = \wGG + \left[ \wBG+\wGB-2\wGG\right] \chi \nonumber \\
 & + \left[\wBB - \wBG -\wGB + \wGG \right] \chi^2 + \BigO(\chi^3) . \label{wchi-expansion-GGGG}
\end{align}
The mutant assessment bits are again $(\dBBm,\G,\G,\G)$ with $\dBBm \in \{\G,\B\}$.
Computing $\tilde{w}(\chi)$ and subtracting, we obtain the leading terms of $\Delta w(\chi)$ reported in Table~\ref{tab:Deltaw-disc-AGGGG}.

\begin{table}
\caption{Leading term of $\Delta w(\chi)$ for a discriminating mutant with $\pBBm \neq \pBB$ invading a resident system in $\AsetGGGG$. In each row, the expression is evaluated for $\pBB = \C$ (so $\pBBm = \D$) and $\pBB = \D$ (so $\pBBm = \C$).}
\begin{center}
\begin{tabular}{c | c | c c}
 $\dBBm$ & Leading term & $\pBB = \C$ & $\pBB = \D$ \\
 \hline
 $\G$ & $\wBB - \wBBm$ & $R - T$ & $P - S$ \\
 $\B$ & $\wBB - \wBG + \wGG - \wBBm$ & $2R-T-S$ & $R+P-2S$
\end{tabular}
\end{center}
\label{tab:Deltaw-disc-AGGGG}
\end{table}

When $\dBBm=\G$, the leading term $\wBB - \wBBm$ evaluates to $R-T$ for $\pBB=\C$ (so $\pBBm=\D$), negative for prisoner's dilemmas and snowdrift games, or to $P-S$ for $\pBB=\D$ (so $\pBBm=\C$), negative for snowdrift games. In either case the mutant invades for at least one class, so $\dBBm = \B$ is necessary for universal stability.

When $\dBBm = \B$, the leading term equals $\wBB - \wBG + \wGG - \wBBm$.
\begin{itemize}
 \item If $\pBB = \D$ (i.e., $p=\DISC$, $\tilde{p}=\CDISC$), this reduces to $R+P-2S$.
 While positive for prisoner's dilemmas and stag hunts, this quantity is negative for some snowdrift games, so the resident can be invaded.
 \item If $\pBB = \C$ (i.e., $p=\CDISC$, $\tilde{p}=\DISC$), this reduces to $2R-T-S$, which is positive across all cooperative dilemmas by condition~\eqref{2R>T+S}. Hence the resident resists invasion.
\end{itemize}

We conclude that universal stability requires $(p,d) \in \AsetGGGG$, $\pBB=\C$, and also $\dBBm = d(\B,\B,\pBBm,\pBB) = \dBBDC = \B$.
\end{proof}

We denote by $\Rset'''$ the subset of $\Rset''$ that resists invasion by any discriminating strategy across cooperative dilemmas, as characterized by Propositions~\ref{prop:dBG=B-dBB=B-invaded} and~\ref{prop:stability-against-alternative-disc}.
These, together with Propositions~\ref{prop:stability-against-ALLD-dGB=G}, \ref{prop:stability-against-ALLC-dGB=G}, \ref{prop:stability-against-ALLC-dBG=G} and Corollary~\ref{cor:pBG=C}, imply that the systems in $\Rset'''$ satisfy $p = \C\C\D\C$ together with
$\dBBCC = \dBGCD = \dGBDC = \dGGCC = \G$
and
$\dBBDC = \dBGDD = \dGBCC = \dGGDC = \B$;
by Eq.~\eqref{QCDisc-set}, this set of righteous systems coincides with the contrite consistent discriminating systems, i.e.,
\begin{equation} \label{Rset'''=QCDisc}
  \Rset''' = \QCDisc.
\end{equation}
Since every discriminating action rule admits $2^8=256$ consistent assessment rules (Section~\ref{subsec:def-consistent}), it follows in particular that $|\Rset'''|=|\QCDisc|=256$, the last of the cardinalities reported in Fig.~\ref{fig:venn_analytic_result}.

As for $\Rset'$ and $\Rset''$, we locate the leading-eight and consistent discriminating systems within $\Rset'''$.
Only one leading-eight system survives, and among the consistent discriminating families only $\QCDisc$ meets $\Rset'''$ --- in fact equal to it.

\begin{proposition}[Leading-eight systems and $\Rset'''$]
\label{prop:L-R'''-membership}
Among the leading eight, only \emph{L2} belongs to $\Rset'''$; the other seven ($\text{L1}, \text{L3}, \text{L4}, \text{L5}, \text{L6}, \text{L7}, \text{L8}$) do not.
\end{proposition}
\begin{proof}
Since $\Rset'''=\QCDisc$~\eqref{Rset'''=QCDisc}, membership in $\Rset'''$ requires $p=\CCDC$.
Among the leading eight, only L1 and L2 use $p=\CCDC$; all others use $p=\DCDC$ and are thus excluded.
L1 fails because $\dGBC=\G$ (i.e., $\dGBCC=\G\neq\B$), violating the condition $\dGBCC=\B$ in Eq.~\eqref{QCDisc-set}.
L2 satisfies all conditions of $\QCDisc$ (cf.\ Eqs.~\eqref{L2} and~\eqref{QCDisc-set}).
\end{proof}

\begin{proposition}[Consistent discriminating systems and $\Rset'''$]
\label{prop:QD-R'''-membership}
$\QCDisc = \Rset'''$, while $\QDisc \cap \Rset''' = \QPDisc \cap \Rset''' = \QODisc \cap \Rset''' = \emptyset$.
\end{proposition}
\begin{proof}
The equality $\QCDisc=\Rset'''$ is Eq.~\eqref{Rset'''=QCDisc}.
Since $\QDisc$, $\QPDisc$, and $\QODisc$ use the action rules $\DCDC$, $\DDDC$, and $\CDDC$ respectively --- all different from $\CCDC$ --- each is disjoint from $\QCDisc=\Rset'''$.
\end{proof}

\subsubsection{The universal stability theorem}
\label{subsubsec:invasion-other-mutants}

We are now in a position to prove the main result of this section.

\begin{theorem}[Universal stability within righteous systems] \label{thm:universal-stability}
Let $\epsilon=0$ and $\chi \to 0$.
A righteous system is universally stable if and only if it is a contrite consistent discriminating system.
\end{theorem}
\begin{proof}
\emph{Necessity.} A universally stable righteous system resists $\ALLD$, $\ALLC$, and every discriminating strategy; by the successive eliminations of Sections~\ref{subsubsec:invasion-ALLD}, \ref{subsubsec:invasion-ALLC}, and~\ref{subsubsec:invasion-disc}, it therefore lies in $\Rset'''=\QCDisc$~\eqref{Rset'''=QCDisc}, i.e., it is contrite consistent discriminating.

\emph{Sufficiency.} Conversely, let $(p,d) \in \QCDisc$, so that $p=\CDISC$. We compute $\Delta w(\chi)$ for each of the 15 mutant strategies $\tilde{p} \neq p$; since all systems in $\QCDisc$ share the same reputation equilibria and payoffs at $\epsilon=0$ (Remark~\ref{rem:consistent-equivalence}), $\Delta w(\chi)$ does not depend on which consistent assessment rule accompanies $\CDISC$, and the results, collected in Table~\ref{tab:stability-C-CDISC}, hold for all of $\QCDisc$ at once.

\begin{table}
\caption{Taylor series expansion of the difference between resident and mutant expected payoffs for $p=\C\C\D\C$ (i.e., $\CDISC$) under a consistent assessment rule $d$ for different mutant strategies $\tilde{p} \neq p$. Here, no execution errors are considered ($\epsilon=0$).}
\begin{center}
\begin{tabular}{c | c | c}
 $\tilde{p}$ (bitstring) & $\tilde{p}$ (name) & $\Delta w(\chi)$ \\
 \hline
 $\DDDD$ & $\ALLD$ & $R-P + \BigO(\chi)$ \\
 $\CDDD$ & & $R-P + \BigO(\chi)$ \\
 $\DCDD$ & & $\frac{1}{2}(2R - T - S) + \BigO(\chi)$ \\
 $\CCDD$ & & $\frac{1}{2}(2R - T - S) + \BigO(\chi)$ \\
 $\DDCD$ & & $R-P + \BigO(\chi)$ \\
 $\CDCD$ & & $R-P + \BigO(\chi)$ \\
 $\DCCD$ & & $\frac{1}{2}(2R - T - S) + \BigO(\chi)$ \\
 $\CCCD$ & & $\frac{1}{2}(2R - T - S) + \BigO(\chi)$ \\
 $\DDDC$ & $\PDISC$ & $\frac{1}{2} (R-P) + \BigO(\chi)$ \\
 $\CDDC$ & $\ODISC$ & $\frac{1}{3} (R-P) + \BigO(\chi)$ \\
 $\DCDC$ & $\DISC$ & $(2 R - T - S) \chi^2 + \BigO(\chi^3)$ \\
 $\CCDC$ & $\CDISC$ & --- \\
 $\DDCC$ & & $\frac{2}{3}(R-P) + \BigO(\chi)$ \\
 $\CDCC$ & & $\frac{1}{2}(R-P) + \BigO(\chi)$ \\
 $\DCCC$ & & $(T-S) \chi + \BigO(\chi^2)$ \\
 $\CCCC$ & $\ALLC$ & $(T-S) \chi + \BigO(\chi^2)$ \\
\end{tabular}
\end{center}
\label{tab:stability-C-CDISC}
\end{table}

A direct inspection of Table~\ref{tab:stability-C-CDISC} reveals three conditions for $(p,d) \in \QCDisc$ to resist invasion by a mutant strategy $\tilde{p} \neq p$:
\begin{enumerate}
 \item If $\tilde{p}=\wildcard\D\wildcard\wildcard$, then $\CDISC$ resists invasion attempts if and only if $R>P$, which always holds across cooperative dilemmas \eqref{R>P}.
 \item If $\tilde{p}=\wildcard\C\C\C$, then $\CDISC$ resists invasion attempts if and only if $T>S$, which always holds across cooperative dilemmas \eqref{T>S}.
 \item Otherwise, $\CDISC$ resists invasion attempts if and only if $2R>T+S$, which always holds across cooperative dilemmas \eqref{2R>T+S}.
\end{enumerate}
As all three conditions hold across cooperative dilemmas, $(p,d)$ resists every mutant strategy and is therefore universally stable.
\end{proof}

\subsection{Stability analysis of consistent discriminating systems}
\label{subsec:stab-anal-disc}

Consistent discriminating systems occupy a central role in our analysis: they are righteous (Proposition~\ref{prop:consistentdisc->righteous}), and they form the dense cluster in the upper-right corner of the stability--cooperation scatter plot shown in Fig.~\ref{MT-fig:scatter}a of the main text.
Section~\ref{subsec:which-righteous-are-stable} establishes that, among righteous systems, universal stability --- resistance to invasion across the full space of cooperative dilemmas --- is exclusive to consistent discriminating systems with $\CDISC$ as action rule (Theorem~\ref{thm:universal-stability}).

We now ask how consistent discriminating systems with the three other action rules ($\DISC$, $\ODISC$, and $\PDISC$) fare.
Universal stability being out of reach for these systems, we determine instead the limiting values of their per-class stability indices $\sigmaPD$, $\sigmaSG$, and $\sigmaSH$, thereby identifying the classes of cooperative dilemma (prisoner's dilemma, snowdrift game, stag hunt) in which each system remains stable.
We proceed as in the previous section, expanding the payoff difference $\Delta w(\chi)$ around $\chi=0$ for every alternative mutant strategy; since all $256$ consistent systems sharing an action rule have identical reputation equilibria and payoffs at $\epsilon=0$ (Remark~\ref{rem:consistent-equivalence}), these expansions depend only on the action rule, and each result below holds for every consistent assessment rule at once.
For ease of comparison, we first restate the result for $\CDISC$ (Corollary~\ref{prop:stability-QCDisc}) and then treat each remaining case in turn; Table~\ref{tab:consistent-disc-summary} collects the results.

\subsubsection{Consistent discriminating systems with CDISC as action rule}
\label{subsubsec:CDISC-asymptotics-stability}

Since systems in $\QCDisc$ are righteous (Proposition~\ref{prop:consistentdisc->righteous}) and are precisely the contrite consistent discriminating systems (Definition~\ref{def:contrite-system}), Theorem~\ref{thm:universal-stability} implies that they are universally stable, and Definition~\ref{def:universally-stable} immediately gives:

\begin{corollary}[Stability of consistent discriminating systems with $\CDISC$ as action rule]
\label{prop:stability-QCDisc}
Systems in $\QCDisc$ are universally stable, that is,
\begin{equation*}
 \lim_{\chi \to 0} \left( \sigmaPD_{p,d} \right)_{\epsilon=0} = \lim_{\chi \to 0} \left( \sigmaSG_{p,d} \right)_{\epsilon=0} = \lim_{\chi \to 0} \left( \sigmaSH_{p,d} \right)_{\epsilon=0} = 1 \quad \text{for all } (p,d) \in \QCDisc .
\end{equation*}
\end{corollary}

\subsubsection{Consistent discriminating systems with DISC as action rule}
\label{subsubsec:DISC-asymptotics-stability}

It is clear from the characterizations of the sets $\QDisc$ \eqref{QDisc-set} and $\Rset''$ \eqref{R''} that $\QDisc \subset \Rset''$ holds, so that a moral system $(p,d) \in \QDisc$ resists invasion attempts by $\ALLD$ and $\ALLC$ across cooperative dilemmas.
It is equally clear, from the analysis in Section~\ref{subsubsec:invasion-disc}, that such a system can be invaded by $\CDISC$ for some games in our payoff space.
The following proposition completes the analysis by considering all alternative mutant strategies and quantifies the resulting stability indices.

\begin{proposition}[Stability of consistent discriminating systems with $\DISC$ as action rule]\label{prop:stability-QDisc}
For $\epsilon=0$ and in the limit $\chi \to 0$, a consistent discriminating system with $\DISC$ as action rule is stable across prisoner's dilemmas and stag hunts, but it can be invaded by $\CDISC$ if the cooperative dilemma is a snowdrift game for which $R+P<2S$ holds.
Thus:
\begin{align*}
 \lim_{\chi \to 0} \left( \sigmaPD_{p,d} \right)_{\epsilon=0} &= 1, &
 \lim_{\chi \to 0} \left( \sigmaSG_{p,d} \right)_{\epsilon=0} &= \tfrac{2}{3}, &
 \lim_{\chi \to 0} \left( \sigmaSH_{p,d} \right)_{\epsilon=0} &= 1 \quad \text{for all } (p,d) \in \QDisc.
\end{align*}
\end{proposition}

\begin{proof}
We compute $\Delta w(\chi)$ for each of the 15 mutant strategies $\tilde{p} \neq p$; the results are collected in Table~\ref{tab:stability-C-DISC}.

\begin{table}
\caption{Taylor series expansion of the difference between resident and mutant expected payoffs for $p=\DCDC$ (i.e., $\DISC$) under a consistent assessment rule $d$ for different mutant strategies $\tilde{p} \neq p$. Here, no execution errors are considered ($\epsilon=0$).}
\begin{center}
\begin{tabular}{c | c | c}
 $\tilde{p}$ (bitstring) & $\tilde{p}$ (name) & $\Delta w(\chi)$ \\
 \hline
 $\DDDD$ & $\ALLD$ & $R-P + \BigO(\chi)$ \\
 $\CDDD$ & & $R-P + \BigO(\chi)$ \\
 $\DCDD$ & & $\frac{1}{2}(2R-T-S) + \BigO(\chi)$ \\
 $\CCDD$ & & $\frac{1}{2}(2R-T-S) + \BigO(\chi)$ \\
 $\DDCD$ & & $R-P + \BigO(\chi)$ \\
 $\CDCD$ & & $R-P + \BigO(\chi)$ \\
 $\DCCD$ & & $\frac{1}{2}(2R-T-S) + \BigO(\chi)$ \\
 $\CCCD$ & & $\frac{1}{2}(2R-T-S) + \BigO(\chi)$ \\
 $\DDDC$ & $\PDISC$ & $\frac{1}{3} (R-P) + \BigO(\chi)$ \\
 $\CDDC$ & $\ODISC$ & $\frac{1}{2} (R-P) + \BigO(\chi)$ \\
 $\DCDC$ & $\DISC$ & --- \\
 $\CCDC$ & $\CDISC$ & $(R + P - 2 S) \chi^2 + \BigO(\chi^3)$ \\
 $\DDCC$ & & $\frac{1}{2}(R-P) + \BigO(\chi)$ \\
 $\CDCC$ & & $\frac{2}{3}(R-P) + \BigO(\chi)$ \\
 $\DCCC$ & & $(T-S) \chi + \BigO(\chi^2)$ \\
 $\CCCC$ & $\ALLC$ & $(T-S) \chi + \BigO(\chi^2)$ \\
\end{tabular}
\end{center}
\label{tab:stability-C-DISC}
\end{table}

A direct inspection of Table~\ref{tab:stability-C-DISC} reveals four conditions for $(p,d) \in \QDisc$ to resist invasion by a mutant strategy $\tilde{p} \neq p$:
\begin{enumerate}
 \item If $\tilde{p}=\wildcard\D\wildcard\wildcard$, then $\DISC$ resists invasion attempts if and only if $R>P$, which always holds across cooperative dilemmas \eqref{R>P}.
 \item If $\tilde{p}=\wildcard\C\C\C$, then $\DISC$ resists invasion attempts if and only if $T>S$, which always holds across cooperative dilemmas \eqref{T>S}.
 \item If $\tilde{p} = \CCDC$ (i.e., $\CDISC$), then $\DISC$ resists invasion attempts if and only if $R+P>2S$, which holds for all prisoner's dilemmas~\eqref{PD} and stag hunts~\eqref{SH}, but not for all snowdrift games~\eqref{SG}.
 \item For all remaining mutant strategies, $\DISC$ resists invasion attempts if and only if $2R>T+S$, which always holds across cooperative dilemmas~\eqref{2R>T+S}.
\end{enumerate}
The conditions in items 1, 2, and 4 hold across all cooperative dilemmas, and the condition in item 3 holds across all prisoner's dilemmas and stag hunts; the stability indices for these two classes thus converge to one.
Within snowdrift games, the only binding condition is $R+P>2S$: it fails precisely for the games with $R+P<2S$, in which $\CDISC$ invades, and by Lemma~\ref{lem:sigma-SG-R+P>2S} the snowdrift games satisfying $R+P>2S$ make up a fraction $2/3$ of the class.
\end{proof}

\subsubsection{Consistent discriminating systems with ODISC as action rule}
\label{subsubsec:ODISC-asymptotics-stability}

It is clear from the characterizations of the sets $\QODisc$ \eqref{QODisc-set} and $\Rset'$ \eqref{R'} that $\QODisc \not\subset \Rset'$ holds, so that a moral system $(p,d) \in \QODisc$ can be invaded by $\ALLD$ for some cooperative dilemmas.
The following proposition completes the analysis by considering all alternative mutant strategies and quantifies the resulting stability indices.

\begin{proposition}[Stability of consistent discriminating systems with $\ODISC$ as action rule]\label{prop:stability-QODisc}
For $\epsilon=0$ and in the limit $\chi \to 0$, a consistent discriminating system with $\ODISC$ as action rule is stable across snowdrift games and stag hunts, but it can be invaded by mutants with $\pBGm=\pGGm=\D$, or by $\PDISC$, if the cooperative dilemma is a prisoner's dilemma for which $2R<T+P$ holds.
Thus:
\begin{align*}
 \lim_{\chi \to 0} \left( \sigmaPD_{p,d} \right)_{\epsilon=0} &= \tfrac{2}{3}, &
 \lim_{\chi \to 0} \left( \sigmaSG_{p,d} \right)_{\epsilon=0} &= 1, &
 \lim_{\chi \to 0} \left( \sigmaSH_{p,d} \right)_{\epsilon=0} &= 1 \quad \text{for all } (p,d) \in \QODisc.
\end{align*}
\end{proposition}

\begin{proof}
We compute $\Delta w(\chi)$ for each of the 15 mutant strategies $\tilde{p} \neq p$; the results are collected in Table~\ref{tab:stability-C-ODISC}.

\begin{table}
\caption{Taylor series expansion of the difference between resident and mutant expected payoffs for $p=\CDDC$ (i.e., $\ODISC$) under a consistent assessment rule $d$ for different mutant strategies $\tilde{p} \neq p$. Here, no execution errors are considered ($\epsilon=0$).}
\begin{center}
\begin{tabular}{c | c | c}
 $\tilde{p}$ (bitstring) & $\tilde{p}$ (name) & $\Delta w(\chi)$ \\
 \hline
 $\DDDD$ & $\ALLD$ & $\frac{1}{2}\left(2 R - T - P \right) + \BigO(\chi)$ \\
 $\CDDD$ & & $\frac{1}{2}\left(2 R - T - P \right) + \BigO(\chi)$ \\
 $\DCDD$ & & $R-S + \BigO(\chi)$ \\
 $\CCDD$ & & $R-S + \BigO(\chi)$ \\
 $\DDCD$ & & $\frac{1}{2}\left(2 R - T - P \right) + \BigO(\chi)$ \\
 $\CDCD$ & & $\frac{1}{2}\left(2 R - T - P \right) + \BigO(\chi)$ \\
 $\DCCD$ & & $R-S + \BigO(\chi)$ \\
 $\CCCD$ & & $R-S + \BigO(\chi)$ \\
 $\DDDC$ & $\PDISC$ & $\left(2 R - T - P \right)\chi^2 + \BigO(\chi^3)$ \\
 $\CDDC$ & $\ODISC$ & --- \\
 $\DCDC$ & $\DISC$ & $\frac{1}{2}\left(R-S\right) + \BigO(\chi)$ \\
 $\CCDC$ & $\CDISC$ & $\frac{1}{3}\left(R-S\right) + \BigO(\chi)$ \\
 $\DDCC$ & & $\left(R-S\right) \chi + \BigO(\chi^2)$ \\
 $\CDCC$ & & $\left(R-S\right) \chi + \BigO(\chi^2)$ \\
 $\DCCC$ & & $\frac{2}{3}\left(R-S\right) + \BigO(\chi)$ \\
 $\CCCC$ & $\ALLC$ & $\frac{1}{2}\left(R-S\right) + \BigO(\chi)$ \\
\end{tabular}
\end{center}
\label{tab:stability-C-ODISC}
\end{table}

A direct inspection of Table~\ref{tab:stability-C-ODISC} reveals two conditions for $(p,d) \in \QODisc$ to resist invasion by a mutant strategy $\tilde{p} \neq p$:
\begin{enumerate}
 \item If $\tilde{p}=\wildcard\D\wildcard\D$ or $\tilde{p}=\DDDC$ (i.e., $\PDISC$), then $\ODISC$ resists invasion attempts if and only if
 \begin{equation} \label{2R>T+P}
  2R>T+P
 \end{equation}
 holds. This condition holds for all snowdrift games~\eqref{SG} and stag hunts~\eqref{SH}, but not for all prisoner's dilemmas~\eqref{PD}.
 \item For all remaining mutant strategies, $\ODISC$ resists invasion attempts if and only if $R>S$, which always holds across cooperative dilemmas~\eqref{R>S}.
\end{enumerate}
The condition in item 2 holds across all cooperative dilemmas, and the condition in item 1 holds across all snowdrift games and stag hunts; the stability indices for these two classes thus converge to one.
Within prisoner's dilemmas, the only binding condition is $2R>T+P$~\eqref{2R>T+P}: it fails precisely for the games with $2R<T+P$, in which the mutants of item 1 invade, and by Lemma~\ref{lem:sigma-PD-2R>T+P} the prisoner's dilemmas satisfying $2R>T+P$ make up a fraction $2/3$ of the class.
\end{proof}

\subsubsection{Consistent discriminating systems with PDISC as action rule}
\label{subsubsec:PDISC-asymptotics-stability}

It is clear from the characterizations of the sets $\QPDisc$ \eqref{QPDisc-set} and $\Rset'$ \eqref{R'} that $\QPDisc \not\subset \Rset'$ holds, so that a moral system $(p,d) \in \QPDisc$ can be invaded by $\ALLD$ for some cooperative dilemmas.
The following proposition completes the analysis by considering all alternative mutant strategies and quantifies the resulting stability indices.

\begin{proposition}[Stability of consistent discriminating systems with $\PDISC$ as action rule]\label{prop:stability-QPDisc}
For $\epsilon=0$ and in the limit $\chi \to 0$, a consistent discriminating system with $\PDISC$ as action rule is stable across snowdrift games and stag hunts, but it can be invaded by mutants with $\pBGm=\pGGm=\D$ if the cooperative dilemma is a prisoner's dilemma for which $2R<T+P$ holds.
Thus:
\begin{align*}
 \lim_{\chi \to 0} \left( \sigmaPD_{p,d} \right)_{\epsilon=0} &= \tfrac{2}{3}, &
 \lim_{\chi \to 0} \left( \sigmaSG_{p,d} \right)_{\epsilon=0} &= 1, &
 \lim_{\chi \to 0} \left( \sigmaSH_{p,d} \right)_{\epsilon=0} &= 1 \quad \text{for all } (p,d) \in \QPDisc.
\end{align*}
\end{proposition}

\begin{proof}
We compute $\Delta w(\chi)$ for each of the 15 mutant strategies $\tilde{p} \neq p$; the results are collected in Table~\ref{tab:stability-C-PDISC}.

\begin{table}
\caption{Taylor series expansion of the difference between resident and mutant expected payoffs for $p=\DDDC$ (i.e., $\PDISC$) under a consistent assessment rule $d$ for different mutant strategies $\tilde{p} \neq p$. Here, no execution errors are considered ($\epsilon=0$).}
\begin{center}
\begin{tabular}{c | c | c}
 $\tilde{p}$ (bitstring) & $\tilde{p}$ (name) & $\Delta w(\chi)$ \\
 \hline
 $\DDDD$ & $\ALLD$ & $\frac{1}{2}\left(2 R - T - P \right) + \BigO(\chi)$ \\
 $\CDDD$ & & $\frac{1}{2}\left(2 R - T - P \right) + \BigO(\chi)$ \\
 $\DCDD$ & & $R-S + \BigO(\chi)$ \\
 $\CCDD$ & & $R-S + \BigO(\chi)$ \\
 $\DDCD$ & & $\frac{1}{2}\left(2 R - T - P \right) + \BigO(\chi)$ \\
 $\CDCD$ & & $\frac{1}{2}\left(2 R - T - P \right) + \BigO(\chi)$ \\
 $\DCCD$ & & $R-S + \BigO(\chi)$ \\
 $\CCCD$ & & $R-S + \BigO(\chi)$ \\
 $\DDDC$ & $\PDISC$ & --- \\
 $\CDDC$ & $\ODISC$ & $\left(R-S\right)\chi^2 + \BigO(\chi^3)$ \\
 $\DCDC$ & $\DISC$ & $\frac{1}{3}\left(R-S\right) + \BigO(\chi)$ \\
 $\CCDC$ & $\CDISC$ & $\frac{1}{2}\left(R-S\right) + \BigO(\chi)$ \\
 $\DDCC$ & & $\left(R-S\right) \chi + \BigO(\chi^2)$ \\
 $\CDCC$ & & $\left(R-S\right) \chi + \BigO(\chi^2)$ \\
 $\DCCC$ & & $\frac{1}{2}\left(R-S\right) + \BigO(\chi)$ \\
 $\CCCC$ & $\ALLC$ & $\frac{2}{3}\left(R-S\right) + \BigO(\chi)$ \\
\end{tabular}
\end{center}
\label{tab:stability-C-PDISC}
\end{table}

A direct inspection of Table~\ref{tab:stability-C-PDISC} reveals two conditions for $(p,d) \in \QPDisc$ to resist invasion by a mutant strategy $\tilde{p} \neq p$:
\begin{enumerate}
 \item If $\tilde{p}=\wildcard\D\wildcard\D$, then $\PDISC$ resists invasion attempts if and only if $2R>T+P$~\eqref{2R>T+P} holds. This condition holds for all snowdrift games~\eqref{SG} and stag hunts~\eqref{SH}, but not for all prisoner's dilemmas~\eqref{PD}.
 \item For all remaining mutant strategies, $\PDISC$ resists invasion attempts if and only if $R>S$, which always holds across cooperative dilemmas~\eqref{R>S}.
\end{enumerate}
The condition in item 2 holds across all cooperative dilemmas, and the condition in item 1 holds across all snowdrift games and stag hunts; the stability indices for these two classes thus converge to one.
Within prisoner's dilemmas, the only binding condition is $2R>T+P$~\eqref{2R>T+P}: it fails precisely for the games with $2R<T+P$, in which the mutants of item 1 invade, and by Lemma~\ref{lem:sigma-PD-2R>T+P} the prisoner's dilemmas satisfying $2R>T+P$ make up a fraction $2/3$ of the class.
\end{proof}

\subsubsection{Summary: consistent discriminating systems}

Table~\ref{tab:consistent-disc-summary} collects the stability results for the four discriminating action rules under a consistent assessment rule (Corollary~\ref{prop:stability-QCDisc} and Propositions~\ref{prop:stability-QDisc}, \ref{prop:stability-QODisc}, and~\ref{prop:stability-QPDisc}).
The results confirm that, among consistent discriminating systems, universal stability is exclusive to those with $\CDISC$ as action rule: every other action rule fails in at least one class of cooperative dilemma.
Systems with $\PDISC$ or $\ODISC$ as action rule can be invaded in the strict subset of prisoner's dilemmas for which $2R<T+P$ holds, while systems with $\DISC$ as action rule can be invaded by $\CDISC$ in the strict subset of snowdrift games for which $R+P<2S$ holds.

\begin{table}[t]
\caption{Stability of consistent discriminating systems across cooperative dilemmas, for $\epsilon = 0$ and in the limit $\chi \to 0$. Each entry gives the limiting value of the stability index for the indicated game class: a value of $1$ indicates stability across the entire class, and a value of $2/3$ indicates stability across a strict subset. The last column gives the average stability index $\sigmabar$~\eqref{avg-sigma}. Action rules are ordered as in Fig.~\ref{MT-fig:stability_indices}a of the main text.}\label{tab:consistent-disc-summary}
\begin{center}
\begin{tabular}{c | c c c | c}
 Action rule & $\sigmaPD$ & $\sigmaSG$ & $\sigmaSH$ & $\sigmabar$ \\
 \hline
 $\PDISC$ & $2/3$ & $1$ & $1$ & $8/9$ \\
 $\ODISC$ & $2/3$ & $1$ & $1$ & $8/9$ \\
 $\DISC$  & $1$ & $2/3$ & $1$ & $8/9$ \\
 $\CDISC$ & $1$ & $1$ & $1$ & $1$ \\
\end{tabular}
\end{center}
\end{table}

\subsection{Stability analysis of the leading eight}
\label{subsec:stab-anal-leading-eight}

We now determine the limiting per-class stability indices of the leading-eight moral systems (Section~\ref{subsec:leading8}).
As we shall see, L2 is the only leading-eight system that achieves universal stability; every other member fails in at least one class of cooperative dilemma.

By Propositions~\ref{prop:L-R''-membership} and~\ref{prop:L-R'''-membership}, the leading eight are distributed across the hierarchy $\Rset \supset \Rset' \supset \Rset'' \supset \Rset'''$ of Section~\ref{subsec:which-righteous-are-stable} as follows: L2 belongs to $\Rset'''$; L5, L6, and L8 belong to $\Rset''\setminus\Rset'''$; and L1, L3, L4, and L7 belong to $\Rset'\setminus\Rset''$.
This placement organizes the analysis into three cases:
\begin{enumerate}
 \item L2 and L6. Since L2 belongs to $\QCDisc$ and L6 to $\QDisc$, their stability indices follow directly from Corollary~\ref{prop:stability-QCDisc} and Proposition~\ref{prop:stability-QDisc}.
 \item L5 and L8. These systems resist invasion by $\ALLD$ and $\ALLC$ across cooperative dilemmas, but can be invaded by some discriminating strategy for some games in our payoff space.
 \item L1, L3, L4, and L7. These systems resist invasion by $\ALLD$ across cooperative dilemmas, but can be invaded by $\ALLC$ for some games in our payoff space.
\end{enumerate}
For the six systems in the last two cases, we proceed as in the previous section, expanding the payoff difference $\Delta w(\chi)$ around $\chi=0$ for every alternative mutant strategy.
We treat the eight systems in turn; Table~\ref{tab:leading-eight-summary} collects the results.

\subsubsection{L1 (CDISC, Standing)}

L1 belongs to $\Rset'\setminus\Rset''$ (Proposition~\ref{prop:L-R''-membership}); its vulnerability stems from the bit $\dGBC=\G$ --- the only bit in which its assessment rule differs from L2's --- which exposes it to invasion by $\wildcard\C\C\C$ mutants (among them $\ALLC$) under stag-hunt payoffs.
The following proposition quantifies the resulting stability indices.

\begin{proposition}[Stability of L1]\label{prop:stability-L1}
For $\epsilon=0$ and in the limit $\chi \to 0$, L1 is stable across prisoner's dilemmas and snowdrift games, but it can always be invaded by $\tilde{p}=\wildcard\C\C\C$ if the cooperative dilemma is a stag hunt.
Thus:
\begin{equation*}
 \lim_{\chi \to 0} \left( \sigmaPD_{\mathrm{L1}} \right)_{\epsilon=0} = 1, \ \lim_{\chi \to 0} \left( \sigmaSG_{\mathrm{L1}} \right)_{\epsilon=0} = 1, \  \lim_{\chi \to 0} \left( \sigmaSH_{\mathrm{L1}} \right)_{\epsilon=0} = 0 .
\end{equation*}
\end{proposition}

\begin{proof}
We compute $\Delta w(\chi)$ for each of the 15 mutant strategies $\tilde{p} \neq p$; the results are collected in Table~\ref{tab:stability-L1}.

\begin{table}
\caption{Taylor series expansion of the difference between resident and mutant expected payoffs for the leading-eight moral system L1 ($\CDISC$, Standing) for different mutant strategies $\tilde{p} \neq \CDISC$. Here, no execution errors are considered ($\epsilon=0$).}\label{tab:stability-L1}
\begin{center}
\begin{tabular}{c | c | c | c}
 $\tilde{p}$ (ID) & $\tilde{p}$ (bitstring) & $\tilde{p}$ (name) & $\Delta w(\chi)$ \\
 \hline
 0 & $\DDDD$ & $\ALLD$ & $R-P + \BigO(\chi)$ \\
 1 & $\CDDD$ & & $R-P + \BigO(\chi)$ \\
 2 & $\DCDD$ & & $\frac{1}{2}(2R - T - S) + \BigO(\chi)$ \\
 3 & $\CCDD$ & & $\frac{1}{2}(2R - T - S) + \BigO(\chi)$ \\
 4 & $\DDCD$ & & $R-P + \BigO(\chi)$ \\
 5 & $\CDCD$ & & $R-P + \BigO(\chi)$ \\
 6 & $\DCCD$ & & $\frac{1}{2}(2R - T - S) + \BigO(\chi)$ \\
 7 & $\CCCD$ & & $\frac{1}{2}(2R - T - S) + \BigO(\chi)$ \\
 8 & $\DDDC$ & $\PDISC$ & $\frac{1}{2} (R-P) + \BigO(\chi)$ \\
 9 & $\CDDC$ & $\ODISC$ & $\frac{1}{3} (R-P) + \BigO(\chi)$ \\
 10 & $\DCDC$ & $\DISC$ & $(2 R - T - S) \chi^2 + \BigO(\chi^3)$ \\
 11 & $\CCDC$ & $\CDISC$ & --- \\
 12 & $\DDCC$ & & $\frac{1}{2}(R-P) + \BigO(\chi)$ \\
 13 & $\CDCC$ & & $\frac{1}{3}(R-P) + \BigO(\chi)$ \\
 14 & $\DCCC$ & & $(T-R) \chi + \BigO(\chi^2)$ \\
 15 & $\CCCC$ & $\ALLC$ & $(T-R) \chi + \BigO(\chi^2)$ \\
\end{tabular}
\end{center}
\end{table}

A direct inspection of Table~\ref{tab:stability-L1} reveals three conditions for L1 to resist invasion by a mutant strategy $\tilde{p} \neq p$:
\begin{enumerate}
 \item If $\tilde{p}=\wildcard\D\wildcard\wildcard$, then L1 resists invasion attempts if and only if $R>P$, which always holds across cooperative dilemmas~\eqref{R>P}.
 \item If $\tilde{p}=\wildcard\C\C\C$, then L1 resists invasion attempts if and only if $T>R$, which holds for all prisoner's dilemmas~\eqref{PD} and snowdrift games~\eqref{SG}, but for no stag hunt~\eqref{SH}.
 \item For all remaining mutant strategies, L1 resists invasion attempts if and only if $2R>T+S$, which always holds across cooperative dilemmas~\eqref{2R>T+S}.
\end{enumerate}
The conditions in items 1 and 3 hold across all cooperative dilemmas, and the condition in item 2 holds across all prisoner's dilemmas and snowdrift games; the stability indices for these two classes thus converge to one.
In contrast, $T>R$ fails for every stag hunt, in which the mutants of item 2 invade; the stability index for this class thus converges to zero.
\end{proof}

\subsubsection{L2 (CDISC, Consistent Standing)}

L2 belongs to $\QCDisc$ (cf.\ Eqs.~\eqref{L2} and~\eqref{QCDisc-set}), the set of contrite consistent discriminating systems.
In contrast with L1, its assessment rule judges cooperation with bad co-players as bad ($\dGBC=\B$), so a $\wildcard\C\C\C$ mutant loses its good standing whenever it cooperates with a bad co-player, and is subsequently exploited by discriminating residents.
As a result, the binding stability condition against these mutants is $T>S$ (Table~\ref{tab:stability-C-CDISC}) rather than L1's $T>R$ (Table~\ref{tab:stability-L1}); since $T>S$ holds across all cooperative dilemmas~\eqref{T>S}, L2 resists these mutants even in stag hunts, where L1 is invaded.

\begin{proposition}[Stability of L2]\label{prop:stability-L2}
L2 is universally stable, that is,
\begin{equation*}
 \lim_{\chi \to 0} \left( \sigmaPD_{\mathrm{L2}} \right)_{\epsilon=0} = 1, \ \lim_{\chi \to 0} \left( \sigmaSG_{\mathrm{L2}} \right)_{\epsilon=0} = 1, \  \lim_{\chi \to 0} \left( \sigmaSH_{\mathrm{L2}} \right)_{\epsilon=0} = 1 .
\end{equation*}
\end{proposition}

\begin{proof}
Immediate from Corollary~\ref{prop:stability-QCDisc}, since L2 belongs to $\QCDisc$.
\end{proof}

\subsubsection{L3 (DISC, Simple Standing)}

L3 belongs to $\Rset'\setminus\Rset''$ (Proposition~\ref{prop:L-R''-membership}); its Simple Standing assessment rule has both $\dGBC=\G$ and $\dBBC=\G$.
The former exposes it, like L1, to invasion by $\wildcard\C\C\C$ mutants (among them $\ALLC$) under stag-hunt payoffs; the latter exposes it to invasion by $\CDISC$ under snowdrift payoffs.
The following proposition quantifies the resulting stability indices.

\begin{proposition}[Stability of L3] \label{prop:stability-L3}
For $\epsilon=0$ and in the limit $\chi \to 0$, L3 is stable across prisoner's dilemmas, but it can always be invaded by $\CDISC$ if the cooperative dilemma is a snowdrift game, and by $\tilde{p}=\wildcard\C\C\C$ if the cooperative dilemma is a stag hunt.
Thus:
\begin{equation*}
 \lim_{\chi \to 0} \left( \sigmaPD_{\mathrm{L3}} \right)_{\epsilon=0} = 1, \  \lim_{\chi \to 0} \left( \sigmaSG_{\mathrm{L3}} \right)_{\epsilon=0} = 0, \  \lim_{\chi \to 0} \left( \sigmaSH_{\mathrm{L3}} \right)_{\epsilon=0} = 0 .
\end{equation*}
\end{proposition}

\begin{proof}
We compute $\Delta w(\chi)$ for each of the 15 mutant strategies $\tilde{p} \neq p$; the results are collected in Table~\ref{tab:stability-L3}.

\begin{table}
\caption{Taylor series expansion of the difference between resident and mutant expected payoffs for the leading-eight moral system L3 ($\DISC$, Simple Standing) for different mutant strategies $\tilde{p} \neq \DISC$. Here, no execution errors are considered ($\epsilon=0$).}\label{tab:stability-L3}
\begin{center}
\begin{tabular}{c | c | c | c}
 $\tilde{p}$ (ID) & $\tilde{p}$ (bitstring) & $\tilde{p}$ (name) & $\Delta w(\chi)$ \\
 \hline
 0 & $\DDDD$ & $\ALLD$ & $R-P + \BigO(\chi)$ \\
 1 & $\CDDD$ & & $R-P + \BigO(\chi)$ \\
 2 & $\DCDD$ & & $\frac{1}{2}(2R - T - S) + \BigO(\chi)$ \\
 3 & $\CCDD$ & & $\frac{1}{2}(2R - T - S) + \BigO(\chi)$ \\
 4 & $\DDCD$ & & $R-P + \BigO(\chi)$ \\
 5 & $\CDCD$ & & $R-P + \BigO(\chi)$ \\
 6 & $\DCCD$ & & $\frac{1}{2}(2R - T - S) + \BigO(\chi)$ \\
 7 & $\CCCD$ & & $\frac{1}{2}(2R - T - S) + \BigO(\chi)$ \\
 8 & $\DDDC$ & $\PDISC$ & $\frac{1}{3} (R-P) + \BigO(\chi)$ \\
 9 & $\CDDC$ & $\ODISC$ & $\frac{1}{3} (R-P) + \BigO(\chi)$ \\
 10 & $\DCDC$ & $\DISC$ & --- \\
 11 & $\CCDC$ & $\CDISC$ & $(P - S) \chi^2 + \BigO(\chi^3)$ \\
 12 & $\DDCC$ & & $\frac{1}{3}(R-P) + \BigO(\chi)$ \\
 13 & $\CDCC$ & & $\frac{1}{3}(R-P) + \BigO(\chi)$ \\
 14 & $\DCCC$ & & $(T-R) \chi + \BigO(\chi^2)$ \\
 15 & $\CCCC$ & $\ALLC$ & $(T-R) \chi + \BigO(\chi^2)$ \\
\end{tabular}
\end{center}
\end{table}

A direct inspection of Table~\ref{tab:stability-L3} reveals four conditions for L3 to resist invasion by a mutant strategy $\tilde{p} \neq p$:
\begin{enumerate}
 \item If $\tilde{p}=\wildcard\D\wildcard\wildcard$, then L3 resists invasion attempts if and only if $R>P$, which always holds across cooperative dilemmas~\eqref{R>P}.
 \item If $\tilde{p}=\wildcard\C\C\C$, then L3 resists invasion attempts if and only if $T>R$, which holds for all prisoner's dilemmas~\eqref{PD} and snowdrift games~\eqref{SG}, but for no stag hunt~\eqref{SH}.
 \item If $\tilde{p}=\CCDC$ (i.e., $\CDISC$), then L3 resists invasion attempts if and only if $P>S$, which holds for all prisoner's dilemmas~\eqref{PD} and stag hunts~\eqref{SH}, but for no snowdrift game~\eqref{SG}.
 \item For all remaining mutant strategies, L3 resists invasion attempts if and only if $2R>T+S$, which always holds across cooperative dilemmas~\eqref{2R>T+S}.
\end{enumerate}
All four conditions hold across prisoner's dilemmas, so the corresponding stability index converges to one.
In contrast, $P>S$ fails for every snowdrift game and $T>R$ fails for every stag hunt; the stability indices for these two classes thus converge to zero.
\end{proof}

\subsubsection{L4 (DISC, SS-SJ)}

L4 belongs to $\Rset'\setminus\Rset''$ (Proposition~\ref{prop:L-R''-membership}); its assessment rule has $\dGBC=\G$ and $\dBBC=\B$.
The former exposes it, like L1 and L3, to invasion by $\wildcard\C\C\C$ mutants (among them $\ALLC$) under stag-hunt payoffs; the latter softens the vulnerability to $\CDISC$ present in L3, restricting it to the snowdrift games for which $R+P<2S$ holds.
The following proposition quantifies the resulting stability indices.

\begin{proposition}[Stability of L4] \label{prop:stability-L4}
For $\epsilon=0$ and in the limit $\chi \to 0$, L4 is stable across prisoner's dilemmas, but it can be invaded by $\CDISC$ if the cooperative dilemma is a snowdrift game for which $R+P<2S$ holds, and it can always be invaded by $\tilde{p}=\wildcard\C\C\C$ if the cooperative dilemma is a stag hunt.
Thus:
\begin{equation*}
 \lim_{\chi \to 0} \left( \sigmaPD_{\mathrm{L4}} \right)_{\epsilon=0} = 1, \  \lim_{\chi \to 0} \left( \sigmaSG_{\mathrm{L4}} \right)_{\epsilon=0} = 2/3, \
 \lim_{\chi \to 0} \left( \sigmaSH_{\mathrm{L4}} \right)_{\epsilon=0} = 0 .
\end{equation*}
\end{proposition}

\begin{proof}
We compute $\Delta w(\chi)$ for each of the 15 mutant strategies $\tilde{p} \neq p$; the results are collected in Table~\ref{tab:stability-L4}.

\begin{table}
\caption{Taylor series expansion of the difference between resident and mutant expected payoffs for the leading-eight moral system L4 ($\DISC$, SS-SJ) for different mutant strategies $\tilde{p} \neq \DISC$. Here, no execution errors are considered ($\epsilon=0$).}\label{tab:stability-L4}
\begin{center}
\begin{tabular}{c | c | c | c}
 $\tilde{p}$ (ID) & $\tilde{p}$ (bitstring) & $\tilde{p}$ (name) & $\Delta w(\chi)$ \\
 \hline
 0 & $\DDDD$ & $\ALLD$ & $R-P + \BigO(\chi)$ \\
 1 & $\CDDD$ & & $R-P + \BigO(\chi)$ \\
 2 & $\DCDD$ & & $\frac{1}{2}(2R - T - S) + \BigO(\chi)$ \\
 3 & $\CCDD$ & & $\frac{1}{2}(2R - T - S) + \BigO(\chi)$ \\
 4 & $\DDCD$ & & $R-P + \BigO(\chi)$ \\
 5 & $\CDCD$ & & $R-P + \BigO(\chi)$ \\
 6 & $\DCCD$ & & $\frac{1}{2}(2R - T - S) + \BigO(\chi)$ \\
 7 & $\CCCD$ & & $\frac{1}{2}(2R - T - S) + \BigO(\chi)$ \\
 8 & $\DDDC$ & $\PDISC$ & $\frac{1}{3} (R-P) + \BigO(\chi)$ \\
 9 & $\CDDC$ & $\ODISC$ & $\frac{1}{2} (R-P) + \BigO(\chi)$ \\
 10 & $\DCDC$ & $\DISC$ & --- \\
 11 & $\CCDC$ & $\CDISC$ & $(R+P - 2S) \chi^2 + \BigO(\chi^3)$ \\
 12 & $\DDCC$ & & $\frac{1}{3}(R-P) + \BigO(\chi)$ \\
 13 & $\CDCC$ & & $\frac{1}{2}(R-P) + \BigO(\chi)$ \\
 14 & $\DCCC$ & & $(T-R) \chi + \BigO(\chi^2)$ \\
 15 & $\CCCC$ & $\ALLC$ & $(T-R) \chi + \BigO(\chi^2)$ \\
\end{tabular}
\end{center}
\end{table}

A direct inspection of Table~\ref{tab:stability-L4} reveals four conditions for L4 to resist invasion by a mutant strategy $\tilde{p} \neq p$:
\begin{enumerate}
 \item If $\tilde{p}=\wildcard\D\wildcard\wildcard$, then L4 resists invasion attempts if and only if $R>P$, which always holds across cooperative dilemmas~\eqref{R>P}.
 \item If $\tilde{p}=\wildcard\C\C\C$, then L4 resists invasion attempts if and only if $T>R$, which holds for all prisoner's dilemmas~\eqref{PD} and snowdrift games~\eqref{SG}, but for no stag hunt~\eqref{SH}.
 \item If $\tilde{p}=\CCDC$ (i.e., $\CDISC$), then L4 resists invasion attempts if and only if $R+P>2S$, which holds for all prisoner's dilemmas~\eqref{PD} and stag hunts~\eqref{SH}, but not for all snowdrift games~\eqref{SG}.
 \item For all remaining mutant strategies, L4 resists invasion attempts if and only if $2R>T+S$, which always holds across cooperative dilemmas~\eqref{2R>T+S}.
\end{enumerate}
All four conditions hold across prisoner's dilemmas, so the corresponding stability index converges to one.
Within snowdrift games, the only binding condition is $R+P>2S$: it fails precisely for the games with $R+P<2S$, in which $\CDISC$ invades, and by Lemma~\ref{lem:sigma-SG-R+P>2S} the snowdrift games satisfying $R+P>2S$ make up a fraction $2/3$ of the class.
Finally, $T>R$ fails for every stag hunt, in which the mutants of item 2 invade; the stability index for this class thus converges to zero.
\end{proof}

\subsubsection{L5 (DISC, SJ-SS)}

L5 belongs to $\Rset''\setminus\Rset'''$ (Proposition~\ref{prop:L-R''-membership}); its assessment rule has $\dGBC=\B$ and $\dBBC=\G$.
The former removes the vulnerability to $\wildcard\C\C\C$ mutants present in L1, L3, and L4 (the binding condition becomes $T>S$ rather than $T>R$); the latter exposes it, like L3, to invasion by $\CDISC$ under snowdrift payoffs.
The following proposition quantifies the resulting stability indices.

\begin{proposition}[Stability of L5] \label{prop:stability-L5}
For $\epsilon=0$ and in the limit $\chi \to 0$, L5 is stable across prisoner's dilemmas and stag hunts, but it can always be invaded by $\CDISC$ if the cooperative dilemma is a snowdrift game.
Thus:
\begin{equation*}
 \lim_{\chi \to 0} \left( \sigmaPD_{\mathrm{L5}} \right)_{\epsilon=0} = 1, \  \lim_{\chi \to 0} \left( \sigmaSG_{\mathrm{L5}} \right)_{\epsilon=0} = 0, \
 \lim_{\chi \to 0} \left( \sigmaSH_{\mathrm{L5}} \right)_{\epsilon=0} = 1 .
\end{equation*}
\end{proposition}

\begin{proof}
We compute $\Delta w(\chi)$ for each of the 15 mutant strategies $\tilde{p} \neq p$; the results are collected in Table~\ref{tab:stability-L5}.

\begin{table}
\caption{Taylor series expansion of the difference between resident and mutant expected payoffs for the leading-eight moral system L5 ($\DISC$, SJ-SS) for different mutant strategies $\tilde{p} \neq \DISC$. Here, no execution errors are considered ($\epsilon=0$).}\label{tab:stability-L5}
\begin{center}
\begin{tabular}{c | c | c | c}
 $\tilde{p}$ (ID) & $\tilde{p}$ (bitstring) & $\tilde{p}$ (name) & $\Delta w(\chi)$ \\
 \hline
 0 & $\DDDD$ & $\ALLD$ & $R-P + \BigO(\chi)$ \\
 1 & $\CDDD$ & & $R-P + \BigO(\chi)$ \\
 2 & $\DCDD$ & & $\frac{1}{2}(2R - T - S) + \BigO(\chi)$ \\
 3 & $\CCDD$ & & $\frac{1}{2}(2R - T - S) + \BigO(\chi)$ \\
 4 & $\DDCD$ & & $R-P + \BigO(\chi)$ \\
 5 & $\CDCD$ & & $R-P + \BigO(\chi)$ \\
 6 & $\DCCD$ & & $\frac{1}{2}(2R - T - S) + \BigO(\chi)$ \\
 7 & $\CCCD$ & & $\frac{1}{2}(2R - T - S) + \BigO(\chi)$ \\
 8 & $\DDDC$ & $\PDISC$ & $\frac{1}{3} (R-P) + \BigO(\chi)$ \\
 9 & $\CDDC$ & $\ODISC$ & $\frac{1}{3} (R-P) + \BigO(\chi)$ \\
 10 & $\DCDC$ & $\DISC$ & --- \\
 11 & $\CCDC$ & $\CDISC$ & $(P-S) \chi^2 + \BigO(\chi^3)$ \\
 12 & $\DDCC$ & & $\frac{1}{2}(R-P) + \BigO(\chi)$ \\
 13 & $\CDCC$ & & $\frac{1}{2}(R-P) + \BigO(\chi)$ \\
 14 & $\DCCC$ & & $(T-S) \chi + \BigO(\chi^2)$ \\
 15 & $\CCCC$ & $\ALLC$ & $(T-S) \chi + \BigO(\chi^2)$ \\
\end{tabular}
\end{center}
\end{table}

A direct inspection of Table~\ref{tab:stability-L5} reveals four conditions for L5 to resist invasion by a mutant strategy $\tilde{p} \neq p$:
\begin{enumerate}
 \item If $\tilde{p}=\wildcard\D\wildcard\wildcard$, then L5 resists invasion attempts if and only if $R>P$, which always holds across cooperative dilemmas~\eqref{R>P}.
 \item If $\tilde{p}=\wildcard\C\C\C$, then L5 resists invasion attempts if and only if $T>S$, which always holds across cooperative dilemmas~\eqref{T>S}.
 \item If $\tilde{p}=\CCDC$ (i.e., $\CDISC$), then L5 resists invasion attempts if and only if $P>S$, which holds for all prisoner's dilemmas~\eqref{PD} and stag hunts~\eqref{SH}, but for no snowdrift game~\eqref{SG}.
 \item For all remaining mutant strategies, L5 resists invasion attempts if and only if $2R>T+S$, which always holds across cooperative dilemmas~\eqref{2R>T+S}.
\end{enumerate}
All four conditions hold across prisoner's dilemmas and stag hunts, so the corresponding stability indices converge to one.
In contrast, $P>S$ fails for every snowdrift game, in which $\CDISC$ invades; the stability index for this class thus converges to zero.
\end{proof}

\subsubsection{L6 (DISC, Stern Judging)}

L6 belongs to $\QDisc$ (cf.\ Eqs.~\eqref{L6} and~\eqref{QDisc-set}), the set of consistent discriminating systems with $\DISC$ as action rule.
Its Stern Judging assessment rule has $\dGBC=\B$ and $\dBBC=\B$: like L5, it is protected against $\wildcard\C\C\C$ mutants; like L4, it can be invaded by $\CDISC$ only in the snowdrift games for which $R+P<2S$ holds.

\begin{proposition}[Stability of L6]\label{prop:stability-L6}
For $\epsilon=0$ and in the limit $\chi \to 0$, L6 is stable across prisoner's dilemmas and stag hunts, but it can be invaded by $\CDISC$ if the cooperative dilemma is a snowdrift game for which $R+P<2S$ holds.
Thus:
\begin{equation*}
 \lim_{\chi \to 0} \left( \sigmaPD_{\mathrm{L6}} \right)_{\epsilon=0} = 1, \ \lim_{\chi \to 0} \left( \sigmaSG_{\mathrm{L6}} \right)_{\epsilon=0} = 2/3, \  \lim_{\chi \to 0} \left( \sigmaSH_{\mathrm{L6}} \right)_{\epsilon=0} = 1 .
\end{equation*}
\end{proposition}

\begin{proof}
Immediate from Proposition~\ref{prop:stability-QDisc}, since L6 belongs to $\QDisc$.
\end{proof}

\subsubsection{L7 (DISC, Staying)}

L7 belongs to $\Rset'\setminus\Rset''$ (Proposition~\ref{prop:L-R''-membership}); its Staying assessment rule has $\dGBC=\G$ and $\dBBD=\dBBC=\B$.
The former exposes it, like L1, L3, and L4, to invasion by $\wildcard\C\C\C$ mutants (among them $\ALLC$) under stag-hunt payoffs.
Because of the latter, a $\CDISC$ mutant --- which deviates from $\DISC$ only by cooperating in interactions between bad individuals --- incurs no reputational penalty for the deviation, so that L7 shares L3's vulnerability to $\CDISC$ under snowdrift payoffs despite having $\dBBC=\B$ like L4 and L6.
The following proposition quantifies the resulting stability indices.

\begin{proposition}[Stability of L7] \label{prop:stability-L7}
For $\epsilon=0$ and in the limit $\chi \to 0$, L7 is stable across prisoner's dilemmas, but it can always be invaded by $\CDISC$ if the cooperative dilemma is a snowdrift game, and by $\tilde{p}=\wildcard\C\C\C$ if the cooperative dilemma is a stag hunt.
Thus:
\begin{equation*}
 \lim_{\chi \to 0} \left( \sigmaPD_{\mathrm{L7}} \right)_{\epsilon=0} = 1, \  \lim_{\chi \to 0} \left( \sigmaSG_{\mathrm{L7}} \right)_{\epsilon=0} = 0, \
 \lim_{\chi \to 0} \left( \sigmaSH_{\mathrm{L7}} \right)_{\epsilon=0} = 0 .
\end{equation*}
\end{proposition}

\begin{proof}
We compute $\Delta w(\chi)$ for each of the 15 mutant strategies $\tilde{p} \neq p$; the results are collected in Table~\ref{tab:stability-L7}.

\begin{table}
\caption{Taylor series expansion of the difference between resident and mutant expected payoffs for the leading-eight moral system L7 ($\DISC$, Staying) for different mutant strategies $\tilde{p} \neq \DISC$. Here, no execution errors are considered ($\epsilon=0$).}\label{tab:stability-L7}
\begin{center}
\begin{tabular}{c | c | c | c}
 $\tilde{p}$ (ID) & $\tilde{p}$ (bitstring) & $\tilde{p}$ (name) & $\Delta w(\chi)$ \\
 \hline
 0 & $\DDDD$ & $\ALLD$ & $R-P + \BigO(\chi)$ \\
 1 & $\CDDD$ & & $R-P + \BigO(\chi)$ \\
 2 & $\DCDD$ & & $\frac{1}{2}(2R - T - S) + \BigO(\chi)$ \\
 3 & $\CCDD$ & & $\frac{1}{2}(2R - T - S) + \BigO(\chi)$ \\
 4 & $\DDCD$ & & $R-P + \BigO(\chi)$ \\
 5 & $\CDCD$ & & $R-P + \BigO(\chi)$ \\
 6 & $\DCCD$ & & $\frac{1}{2}(2R - T - S) + \BigO(\chi)$ \\
 7 & $\CCCD$ & & $\frac{1}{2}(2R - T - S) + \BigO(\chi)$ \\
 8 & $\DDDC$ & $\PDISC$ & $\frac{1}{2} (R-P) + \BigO(\chi)$ \\
 9 & $\CDDC$ & $\ODISC$ & $\frac{1}{2} (R-P) + \BigO(\chi)$ \\
 10 & $\DCDC$ & $\DISC$ & --- \\
 11 & $\CCDC$ & $\CDISC$ & $(P-S) \chi^2 + \BigO(\chi^3)$ \\
 12 & $\DDCC$ & & $\frac{1}{2}(R-P) + \BigO(\chi)$ \\
 13 & $\CDCC$ & & $\frac{1}{2}(R-P) + \BigO(\chi)$ \\
 14 & $\DCCC$ & & $(T-R) \chi + \BigO(\chi^2)$ \\
 15 & $\CCCC$ & $\ALLC$ & $(T-R) \chi + \BigO(\chi^2)$ \\
\end{tabular}
\end{center}
\end{table}

A direct inspection of Table~\ref{tab:stability-L7} reveals four conditions for L7 to resist invasion by a mutant strategy $\tilde{p} \neq p$:
\begin{enumerate}
 \item If $\tilde{p}=\wildcard\D\wildcard\wildcard$, then L7 resists invasion attempts if and only if $R>P$, which always holds across cooperative dilemmas~\eqref{R>P}.
 \item If $\tilde{p}=\wildcard\C\C\C$, then L7 resists invasion attempts if and only if $T>R$, which holds for all prisoner's dilemmas~\eqref{PD} and snowdrift games~\eqref{SG}, but for no stag hunt~\eqref{SH}.
 \item If $\tilde{p}=\CCDC$ (i.e., $\CDISC$), then L7 resists invasion attempts if and only if $P>S$, which holds for all prisoner's dilemmas~\eqref{PD} and stag hunts~\eqref{SH}, but for no snowdrift game~\eqref{SG}.
 \item For all remaining mutant strategies, L7 resists invasion attempts if and only if $2R>T+S$, which always holds across cooperative dilemmas~\eqref{2R>T+S}.
\end{enumerate}
All four conditions hold across prisoner's dilemmas, so the corresponding stability index converges to one.
In contrast, $P>S$ fails for every snowdrift game and $T>R$ fails for every stag hunt; the stability indices for these two classes thus converge to zero.
\end{proof}

\subsubsection{L8 (DISC, Judging)}

L8 belongs to $\Rset''\setminus\Rset'''$ (Proposition~\ref{prop:L-R''-membership}); its Judging assessment rule has $\dGBC=\B$ and $\dBBD=\dBBC=\B$.
Like L5, it is protected against $\wildcard\C\C\C$ mutants by its first bit; like L7, its two remaining bits leave a $\CDISC$ mutant's deviation without reputational penalty, so that it can be invaded by $\CDISC$ under snowdrift payoffs.
The following proposition quantifies the resulting stability indices.

\begin{proposition}[Stability of L8] \label{prop:stability-L8}
For $\epsilon=0$ and in the limit $\chi \to 0$, L8 is stable across prisoner's dilemmas and stag hunts, but it can always be invaded by $\CDISC$ if the cooperative dilemma is a snowdrift game.
Thus:
\begin{equation*}
 \lim_{\chi \to 0} \left( \sigmaPD_{\mathrm{L8}} \right)_{\epsilon=0} = 1, \  \lim_{\chi \to 0} \left( \sigmaSG_{\mathrm{L8}} \right)_{\epsilon=0} = 0, \
 \lim_{\chi \to 0} \left( \sigmaSH_{\mathrm{L8}} \right)_{\epsilon=0} = 1 .
\end{equation*}
\end{proposition}

\begin{proof}
We compute $\Delta w(\chi)$ for each of the 15 mutant strategies $\tilde{p} \neq p$; the results are collected in Table~\ref{tab:stability-L8}.

\begin{table}
\caption{Taylor series expansion of the difference between resident and mutant expected payoffs for the leading-eight moral system L8 ($\DISC$, Judging) for different mutant strategies $\tilde{p} \neq \DISC$. Here, no execution errors are considered ($\epsilon=0$).}\label{tab:stability-L8}
\begin{center}
\begin{tabular}{c | c | c | c}
 $\tilde{p}$ (ID) & $\tilde{p}$ (bitstring) & $\tilde{p}$ (name) & $\Delta w(\chi)$ \\
 \hline
 0 & $\DDDD$ & $\ALLD$ & $R-P + \BigO(\chi)$ \\
 1 & $\CDDD$ & & $R-P + \BigO(\chi)$ \\
 2 & $\DCDD$ & & $\frac{1}{2}(2R - T - S) + \BigO(\chi)$ \\
 3 & $\CCDD$ & & $\frac{1}{2}(2R - T - S) + \BigO(\chi)$ \\
 4 & $\DDCD$ & & $R-P + \BigO(\chi)$ \\
 5 & $\CDCD$ & & $R-P + \BigO(\chi)$ \\
 6 & $\DCCD$ & & $\frac{1}{2}(2R - T - S) + \BigO(\chi)$ \\
 7 & $\CCCD$ & & $\frac{1}{2}(2R - T - S) + \BigO(\chi)$ \\
 8 & $\DDDC$ & $\PDISC$ & $\frac{1}{2} (R-P) + \BigO(\chi)$ \\
 9 & $\CDDC$ & $\ODISC$ & $\frac{1}{2} (R-P) + \BigO(\chi)$ \\
 10 & $\DCDC$ & $\DISC$ & --- \\
 11 & $\CCDC$ & $\CDISC$ & $(P-S) \chi^2 + \BigO(\chi^3)$ \\
 12 & $\DDCC$ & & $\frac{2}{3}(R-P) + \BigO(\chi)$ \\
 13 & $\CDCC$ & & $\frac{2}{3}(R-P) + \BigO(\chi)$ \\
 14 & $\DCCC$ & & $(T-S) \chi + \BigO(\chi^2)$ \\
 15 & $\CCCC$ & $\ALLC$ & $(T-S) \chi + \BigO(\chi^2)$ \\
\end{tabular}
\end{center}
\end{table}

A direct inspection of Table~\ref{tab:stability-L8} reveals four conditions for L8 to resist invasion by a mutant strategy $\tilde{p} \neq p$:
\begin{enumerate}
 \item If $\tilde{p}=\wildcard\D\wildcard\wildcard$, then L8 resists invasion attempts if and only if $R>P$, which always holds across cooperative dilemmas~\eqref{R>P}.
 \item If $\tilde{p}=\wildcard\C\C\C$, then L8 resists invasion attempts if and only if $T>S$, which always holds across cooperative dilemmas~\eqref{T>S}.
 \item If $\tilde{p}=\CCDC$ (i.e., $\CDISC$), then L8 resists invasion attempts if and only if $P>S$, which holds for all prisoner's dilemmas~\eqref{PD} and stag hunts~\eqref{SH}, but for no snowdrift game~\eqref{SG}.
 \item For all remaining mutant strategies, L8 resists invasion attempts if and only if $2R>T+S$, which always holds across cooperative dilemmas~\eqref{2R>T+S}.
\end{enumerate}
All four conditions hold across prisoner's dilemmas and stag hunts, so the corresponding stability indices converge to one.
In contrast, $P>S$ fails for every snowdrift game, in which $\CDISC$ invades; the stability index for this class thus converges to zero.
\end{proof}

\subsubsection{Summary: the leading eight}

Table~\ref{tab:leading-eight-summary} collects the stability results for the eight leading-eight moral systems (Propositions~\ref{prop:stability-L1}--\ref{prop:stability-L8}).
L2 ($\CDISC$, Consistent Standing) is the only universally stable system; every other member can be invaded in at least one class of cooperative dilemma.
All six systems with $\DISC$ as action rule can be invaded by $\CDISC$ in snowdrift games --- in every such game for L3, L5, L7, and L8, and in those satisfying $R+P<2S$ for L4 and L6.
The reason is that $S>P$ holds in snowdrift games: when meeting a co-player in bad standing while in bad standing itself, a contrite discriminator cooperates (earning $S$) where the resident defects (earning $P$), and thus profits from the deviation.
The four systems with $\dGBC=\G$ (L1, L3, L4, and L7) --- among them L1, which pairs the $\CDISC$ action rule with the inconsistent Standing norm --- fail in every stag hunt, where $R>T$ holds: mutants that also cooperate with co-players in bad standing retain their good standing and earn $R$ where residents earn only $T$.
These results reinforce the conclusion of the preceding sections: among the leading eight, universal stability requires both consistency and contrition, a combination found only in L2.

\begin{table}[t]
\caption{Stability of the leading eight across cooperative dilemmas, for $\epsilon = 0$ and in the limit $\chi \to 0$. Each entry gives the limiting value of the stability index for the indicated game class: a value of $1$ indicates stability across the entire class, a value of $2/3$ stability across a strict subset, and a value of $0$ instability throughout. The last column gives the average stability index $\sigmabar$~\eqref{avg-sigma}.}
\begin{center}
\begin{tabular}{c | c | c | c c c | c}
 Moral system & Action rule & Assessment rule & $\sigmaPD$ & $\sigmaSG$ & $\sigmaSH$ & $\sigmabar$ \\
 \hline
 L1 & $\CDISC$ & Standing & $1$ & $1$ & $0$ & $2/3$ \\
 L2 & $\CDISC$ & Consistent Standing & $1$ & $1$ & $1$ & $1$ \\
 L3 & $\DISC$ & Simple Standing & $1$ & $0$ & $0$ & $1/3$ \\
 L4 & $\DISC$ & SS-SJ & $1$ & $2/3$ & $0$ & $5/9$ \\
 L5 & $\DISC$ & SJ-SS & $1$ & $0$ & $1$ & $2/3$ \\
 L6 & $\DISC$ & Stern Judging & $1$ & $2/3$ & $1$ & $8/9$ \\
 L7 & $\DISC$ & Staying & $1$ & $0$ & $0$ & $1/3$ \\
 L8 & $\DISC$ & Judging & $1$ & $0$ & $1$ & $2/3$ \\
\end{tabular}
\end{center}
\label{tab:leading-eight-summary}
\end{table}

\section{Numerical results}
\label{sec:supp-numerical}

The analytical results of the preceding sections were derived in the idealized regime of zero execution errors ($\epsilon = 0$) and vanishing assignment errors ($\chi \to 0$), and the numerical results of the main text were computed for a single combination of small error rates ($\epsilon = \chi = 0.02$).
This section presents numerical results verifying that our findings are robust to these choices.
We first count moral systems up to the mirror symmetry of Section~\ref{subsec:mirror-systems} and specify the representatives our numerical explorations evaluate (Section~\ref{subsec:mirror-classes}), and then detail the computational pipeline behind these numerical results (Section~\ref{subsec:numerical-pipeline}).
We finally examine how the stability and cooperation indices behave across a grid of error-rate combinations (Section~\ref{subsec:alternative-errors}), how stability is distributed across the three classes of cooperative dilemma (Section~\ref{subsec:disaggregated-indices}), and how results change under the alternative, area-weighted construction of the average stability index (Section~\ref{subsec:area-weighted}).

\subsection{Mirror classes and representatives}
\label{subsec:mirror-classes}

Section~\ref{subsec:mirror-systems} established that a moral system and its mirror are behaviorally equivalent: they share the same cooperation and stability indices (Proposition~\ref{prop:mirror-invariance}), whereas label-dependent properties such as self-approval hold for at most one member of each mirror pair (Corollary~\ref{cor:mirror-classes}).
Here we extend these results to consistency, count moral systems up to the mirror symmetry, and specify the representatives evaluated in our numerical explorations.

Consistency is a label-dependent property: it requires conformity with the resident prescription to be rewarded with the particular label $\G$.
Unsurprisingly, it is not mirror-invariant.

\begin{proposition}[Mirrors of consistent systems] \label{prop:mirror-consistent}
Let $(p,d)$ be a consistent system.
Then its mirror $(\hat{p},\hat{d})$ satisfies
\begin{subequations} \label{mirror-consistent}
\begin{align}
 \hat{d}(\rf,\rc,\hat{p}(\rf,\rc),\hat{p}(\rc,\rf)) & = \B, \text{ and } \label{mirror-consistent-B} \\
 \hat{d}(\rf,\rc,\neg\hat{p}(\rf,\rc),\hat{p}(\rc,\rf)) & = \G \label{mirror-consistent-G}
\end{align}
\end{subequations}
for all $\rf,\rc \in \left\{\B,\G\right\}$.
Consequently, the mirror of a consistent system is not consistent, no consistent system is a self-mirror, and each mirror class contains at most one consistent system.
\end{proposition}

\begin{proof}
By~\eqref{mirror-def}, $\hat{d}(\rf,\rc,\hat{p}(\rf,\rc),\hat{p}(\rc,\rf)) = \neg\, d(\neg\rf,\neg\rc,p(\neg\rf,\neg\rc),p(\neg\rc,\neg\rf)) = \neg\,\G = \B$, using the consistency condition~\eqref{si-consistent-cond-G} at the reputation profile $(\neg\rf,\neg\rc)$; this establishes Eq.~\eqref{mirror-consistent-B}, and Eq.~\eqref{mirror-consistent-G} follows in the same way from~\eqref{si-consistent-cond-B}.
Since consistency~\eqref{si-consistent-cond-G} requires the left-hand side of Eq.~\eqref{mirror-consistent-B} to equal $\G$ for all reputation profiles, $(\hat{p},\hat{d})$ is not consistent.
If a consistent system were a self-mirror, its mirror -- the system itself -- would be consistent, contradicting the non-consistency of the mirror just established; the same argument shows that the two members of a mirror pair cannot both be consistent.
\end{proof}

In words, Eq.~\eqref{mirror-consistent} says that the mirror of a consistent system is consistent with the labels swapped: it rewards conformity with the label $\B$, which plays for the mirror system the role that $\G$ plays for the original one.
This is precisely the sense in which the two systems are relabelings of one another.

We now count self-mirrors, and with them the mirror classes.

\begin{proposition}[Number of self-mirrors] \label{prop:selfmirror-count}
A moral system $(p,d)$ is a self-mirror if and only if both
\begin{subequations} \label{selfmirror-cond}
\begin{align}
 \pBB & = \pGG \ \text{ and } \ \pBG = \pGB, \text{ and } \label{selfmirror-cond-p} \\
 \mathsf{d}_{\rf\rc\af\ac} & = \neg\, \mathsf{d}_{\neg\rf\,\neg\rc\,\af\ac} \ \text{ for all } \rf,\rc,\af,\ac \label{selfmirror-cond-d}
\end{align}
\end{subequations}
hold.
Exactly four strategies satisfy condition~\eqref{selfmirror-cond-p} (the palindromic bitstrings $\DDDD$, $\DCCD$, $\CDDC$, and $\CCCC$) and exactly $2^8 = 256$ assessment rules satisfy condition~\eqref{selfmirror-cond-d}, so there are $4 \times 256 = \numprint{1024}$ self-mirrors.
\end{proposition}

\begin{proof}
Conditions~\eqref{selfmirror-cond-p} and~\eqref{selfmirror-cond-d} restate $\hat{p} = p$ and $\hat{d} = d$, respectively, using~\eqref{mirror-def}.
For strategies, $\hat{p} = p$ says that the bitstring $\pBB\pBG\pGB\pGG$ is a palindrome, which leaves the two bits $\pBB$ and $\pBG$ free: the four solutions are $\DDDD$ ($\ALLD$), $\DCCD$, $\CDDC$ ($\ODISC$), and $\CCCC$ ($\ALLC$).
For assessment rules, the map $(\rf,\rc,\af,\ac) \mapsto (\neg\rf,\neg\rc,\af,\ac)$ pairs the sixteen argument profiles into eight pairs (it has no fixed point, since $(\rf,\rc) \neq (\neg\rf,\neg\rc)$ always holds), and on each pair the constraint in~\eqref{selfmirror-cond-d} admits exactly two of the four joint assignments; hence $2^8=256$ solutions.
\end{proof}

There are exactly as many self-mirrors as consistent discriminating systems -- $\numprint{1024}$ in both cases -- but the coincidence is accidental: the former number arises as (four palindromic strategies) $\times$ ($256$ self-mirror assessment rules), the latter as (four discriminating strategies) $\times$ ($256$ consistent assessment rules per strategy), and the two sets are disjoint by Proposition~\ref{prop:mirror-consistent}.
The two sets even share a strategy: $\ODISC$ ($p=\CDDC$) is the unique discriminating strategy equal to its own mirror -- yet, again by Proposition~\ref{prop:mirror-consistent}, no consistent $\ODISC$ system is a self-mirror.

\begin{proposition}[Number of mirror classes] \label{prop:mirror-count}
Let $\mathcal{M}$ be a set of moral systems that is closed under the mirror map, and let $m$ be the number of self-mirrors in $\mathcal{M}$.
Then $\mathcal{M}$ partitions into
\begin{equation} \label{mirror-class-count}
 \frac{|\mathcal{M}| - m}{2} + m = \frac{|\mathcal{M}| + m}{2}
\end{equation}
mirror classes.
In particular, the full space of $2^{20} = \numprint{1048576}$ moral systems, which contains $m = \numprint{1024}$ self-mirrors (Proposition~\ref{prop:selfmirror-count}), partitions into $\left(2^{20} + \numprint{1024}\right)/2 = \numprint{524800}$ mirror classes.
\end{proposition}

\begin{proof}
The mirror map restricts to an involution on $\mathcal{M}$, so it partitions $\mathcal{M}$ into classes of size one or two, and the classes of size one are exactly the self-mirrors in $\mathcal{M}$.
The number of classes is therefore $m$ plus half the number $|\mathcal{M}| - m$ of remaining systems.
\end{proof}

Equation~\eqref{mirror-class-count} is an instance of Burnside's lemma (ref.~\citenum{si-vanlint2001course}, Chapter~10), applied to the two-element group generated by the mirror map: the number of orbits (the mirror classes) equals the average number of elements of $\mathcal{M}$ fixed by the group elements -- $|\mathcal{M}|$ for the identity and $m$ for the mirror.

Beyond the full space, Eq.~\eqref{mirror-class-count} applies to any mirror-closed subset, such as those obtained by pairing the $2^4$ strategies with a type class $\mathcal{T}_{\mathcal{U}}$ of assessment rules (Section~\ref{subsec:norm-counting}): the mirror of a rule of type $\mathcal{U}$ is again of type $\mathcal{U}$, since the mirror negates reputation arguments and outputs but never adds a dependence on a new argument.
If $\mathcal{U}$ contains at least one reputation argument, the negation $(\rf,\rc) \mapsto (\neg\rf,\neg\rc)$ pairs the $2^{|\mathcal{U}|}$ argument profiles of type $\mathcal{U}$ into $2^{|\mathcal{U}|-1}$ pairs, so that, by the argument in the proof of Proposition~\ref{prop:selfmirror-count}, exactly $2^{2^{|\mathcal{U}|-1}}$ rules of type $\mathcal{U}$ are their own mirrors; if $\mathcal{U}$ contains no reputation argument, a self-mirror rule would have to equal its own negation, so there is none.
For instance, the space explored by Ohtsuki and Iwasa\citep{si-ohtsuki:JTB:2004a} pairs the $2^4$ strategies with the $2^{2^3} = 256$ third-order rules of type $\af\rf\rc$: it contains $4 \times 2^{2^2} = 2^6$ self-mirrors (four palindromic strategies times $2^{2^2} = 16$ self-mirror rules), so Eq.~\eqref{mirror-class-count} yields $\left(2^{12} + 2^6\right)/2 = \numprint{2080}$ mirror classes, recovering the count of `effective pairs' derived in their Appendix~A.

Because both members of a mirror pair share the same cooperation and stability indices (Proposition~\ref{prop:mirror-invariance}), numerical explorations of the space of moral systems need to evaluate only one member per mirror class -- a reduction already exploited by Ohtsuki and Iwasa\citep{si-ohtsuki:JTB:2004a}, who report a single member of each pair and omit its mirror image.
Sweeping the full space of moral systems, however, requires an explicit selection rule, which the following definition provides.

\begin{definition}[Majority-good representatives] \label{def:mirror-representative}
Fix error rates $\chi \in (0,1/2)$ and $\epsilon \in [0,1/2)$.
The \emph{representative} of a mirror class is:
its unique member, if the class is a self-mirror;
otherwise, the member with $g^*_{p,d} > 1/2$, if there is one;
otherwise (that is, if both members of the pair satisfy $g^*_{p,d} = 1/2$), the member with the larger strategy ID (Table~\ref{tab:strategies}) or, if the two members share the same strategy, the member whose assessment rule has the larger ID (the sixteen bits $\mathsf{d}_{\rf\rc\af\ac}$ read as a binary number, in the same little-endian convention as in Table~\ref{tab:strategies}).
\end{definition}

Definition~\ref{def:mirror-representative} is exhaustive and unambiguous: by Lemma~\ref{lem:mirror-reputation}, the two members of a mirror pair have resident reputation equilibria $g^*$ and $1-g^*$, so either exactly one member satisfies $g^* > 1/2$ or both sit at $g^* = 1/2$, in which case their strategy IDs differ unless $\hat{p}=p$, and then their assessment-rule IDs differ.
We call the representatives `majority-good' because, in the generic case, the representative is the member of the pair whose residents mostly hold good reputations at the working error rates.

All numerical explorations reported in the main text and in this section evaluate one representative per mirror class; by Proposition~\ref{prop:mirror-invariance}, which member is evaluated is immaterial.
Selecting majority-good representatives parallels, at finite error rates, the restriction to self-approving systems (Corollary~\ref{cor:mirror-classes}): a self-approving system satisfies $g^*_{p,d} > 1/2$ for sufficiently small error rates, and is then the representative of its mirror class.

\subsection{Computation of the cooperation and stability indices}
\label{subsec:numerical-pipeline}

We now describe the pipeline used to compute the cooperation and stability indices reported in the main text (Fig.~\ref{MT-fig:scatter}; $\chi = \epsilon = 0.02$) and in the robustness analyses of this section.
For each combination of error rates, we evaluate the $\numprint{524800}$ majority-good representatives of the mirror classes (Definition~\ref{def:mirror-representative}), obtained by enumerating all $2^{20}$ moral systems -- with conditional or unconditional action rules -- and retaining one representative per mirror class (Proposition~\ref{prop:mirror-count}).

For each representative $(p,d)$, we first compute the resident reputation equilibrium $g^*_{p,d}$, locating the unique root in $(0,1)$ of the quadratic with coefficients~\eqref{abc} by bisection, and from it the cooperation index $\kappa_{p,d}$~\eqref{coop-index}.

The stability indices are estimated by Monte Carlo integration.
For each game class $\Gamma \in \left\{\mathrm{PD},\mathrm{SG},\mathrm{SH}\right\}$, we draw once $\numprint{1000}$ payoff configurations uniformly at random from the normalized payoff region $\mathcal{G}_\Gamma$ by rejection sampling:
for prisoner's dilemmas, $T=1$ and $S=0$ are fixed and $(R,P)$ is drawn uniformly subject to~\eqref{normalized-PD};
for snowdrift games, $T=1$ and $P=0$ are fixed and $(R,S)$ is drawn uniformly subject to~\eqref{normalized-SG};
and for stag hunts, $R=1$ and $S=0$ are fixed and $(T,P)$ is drawn uniformly subject to~\eqref{normalized-SH}.
The resulting three samples are archived and reused for every moral system and every error-rate combination, so that comparisons across systems and error rates are based on the same payoff configurations.
For each system and each payoff configuration, we then check the stability condition~\eqref{evol-stab-cond} directly: we compute the resident payoff $w_{p,d}$~\eqref{w} and the mutant payoffs $\tilde{w}_{\tilde{p},p,d}$~\eqref{tildew} for all $15$ alternative strategies $\tilde{p} \neq p$, and record whether $w_{p,d} > \tilde{w}_{\tilde{p},p,d}$ holds for all of them.
The estimate of $\sigma^\Gamma_{p,d}(\epsilon,\chi)$ is the fraction of the $\numprint{1000}$ payoff configurations at which $(p,d)$ is stable -- a binomial proportion, whose standard error $\sqrt{\sigma^\Gamma_{p,d} \left(1-\sigma^\Gamma_{p,d}\right)/\numprint{1000}}$ is at most $\sqrt{1/(4 \times \numprint{1000})} \approx 0.016$.
The average stability index then follows from Eq.~\eqref{avg-sigma}.

All computations were implemented in a dedicated Julia package accompanying this paper, with the archived payoff samples distributed alongside the code; evaluating all representatives takes on the order of one hour per game class and error-rate combination on a standard desktop computer.
The code, the archived payoff samples, and the computed indices are all openly available; see the Data availability and Code availability statements in the main text.

\subsection{Alternative error rates}\label{subsec:alternative-errors}

Fig.~\ref{fig:alternative-errors} shows the stability and cooperation indices of all moral systems for a grid of error-rate combinations, with execution error rates $\epsilon \in \{0, 0.01, 0.02, 0.1\}$ and assignment error rates $\chi \in \{0.01, 0.02, 0.05, 0.1\}$.
Larger error rates lower both indices, but the cluster of consistent discriminating systems in the upper-right corner persists across all combinations examined, confirming that our results are not an artifact of the error rates chosen for the baseline analysis ($\epsilon = \chi = 0.02$).

\begin{figure}[tp]
    \centering
    \includegraphics[width=\textwidth]{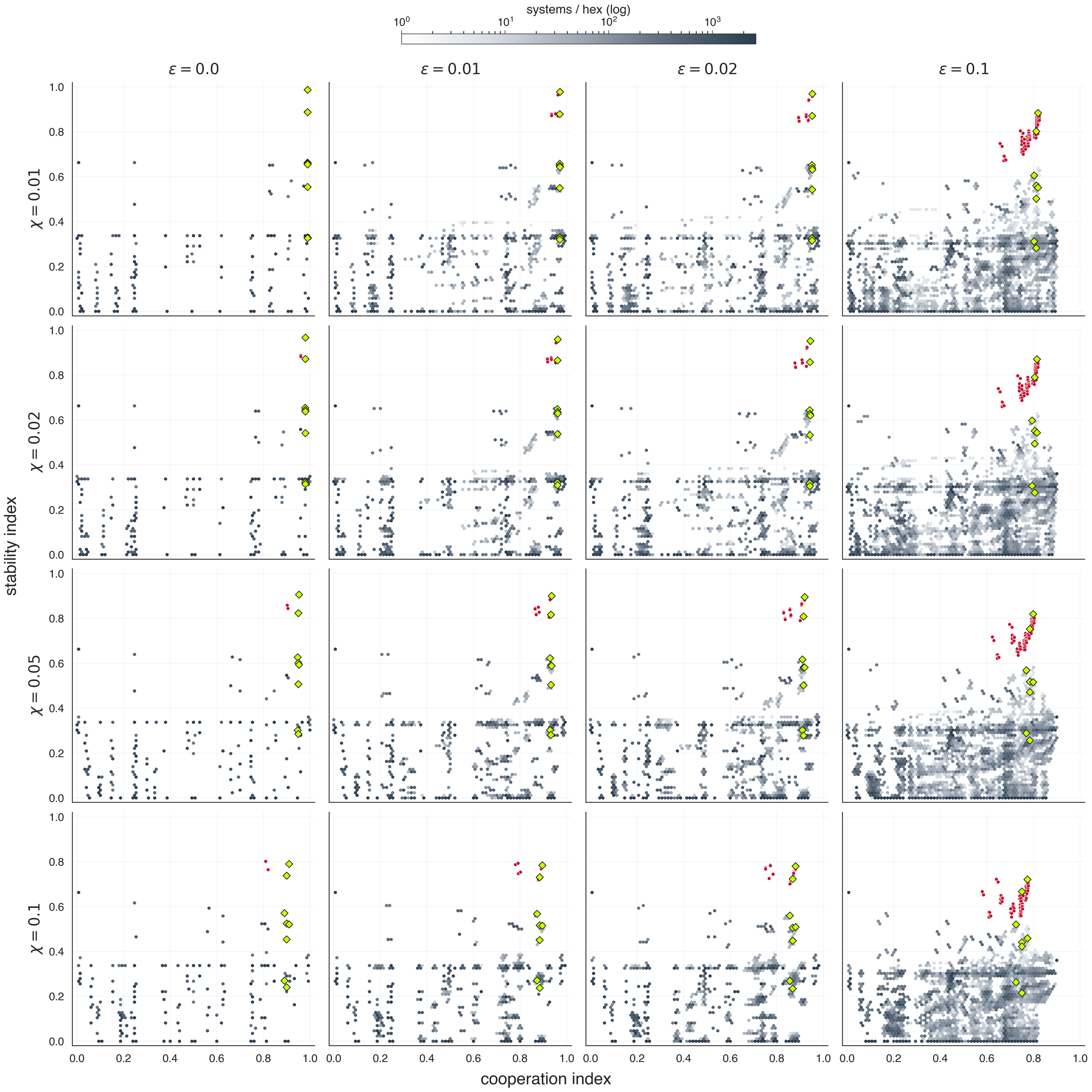}
    \caption{Stability and cooperation indices across a range of error rates. Each panel shows, for one combination of execution error rate $\epsilon \in \{0, 0.01, 0.02, 0.1\}$ (columns) and assignment error rate $\chi \in \{0.01, 0.02, 0.05, 0.1\}$ (rows), the cooperation index ($\kappa$) and average stability index ($\sigma$) of each of the $\numprint{524800}$ mirror-class representatives (Section~\ref{subsec:mirror-classes}). The consistent discriminating systems are shown as red circles, the leading eight as green diamonds, and the remaining systems as a hexbin density layer on a logarithmic color scale.}
    \label{fig:alternative-errors}
\end{figure}

Fig.~\ref{fig:alternative_errors_zoomed_scatter} zooms in on the consistent discriminating systems, showing their stability and cooperation indices over the same grid of error rates so that the variation within the cluster becomes visible.
L2 is among the most stable systems at every combination.
At the largest assignment error rate we examine ($\chi = 0.1$), however, the most stable consistent $\PDISC$ systems match or slightly exceed it in stability -- the differences lie at the resolution of our Monte Carlo estimates (Section~\ref{subsec:numerical-pipeline}) -- while cooperating markedly less.
Whatever the error rates, the consistent $\CDISC$ systems remain the most cooperative of the consistent discriminating systems.

\begin{figure}[tp]
  \centering
  \includegraphics[width=\textwidth]{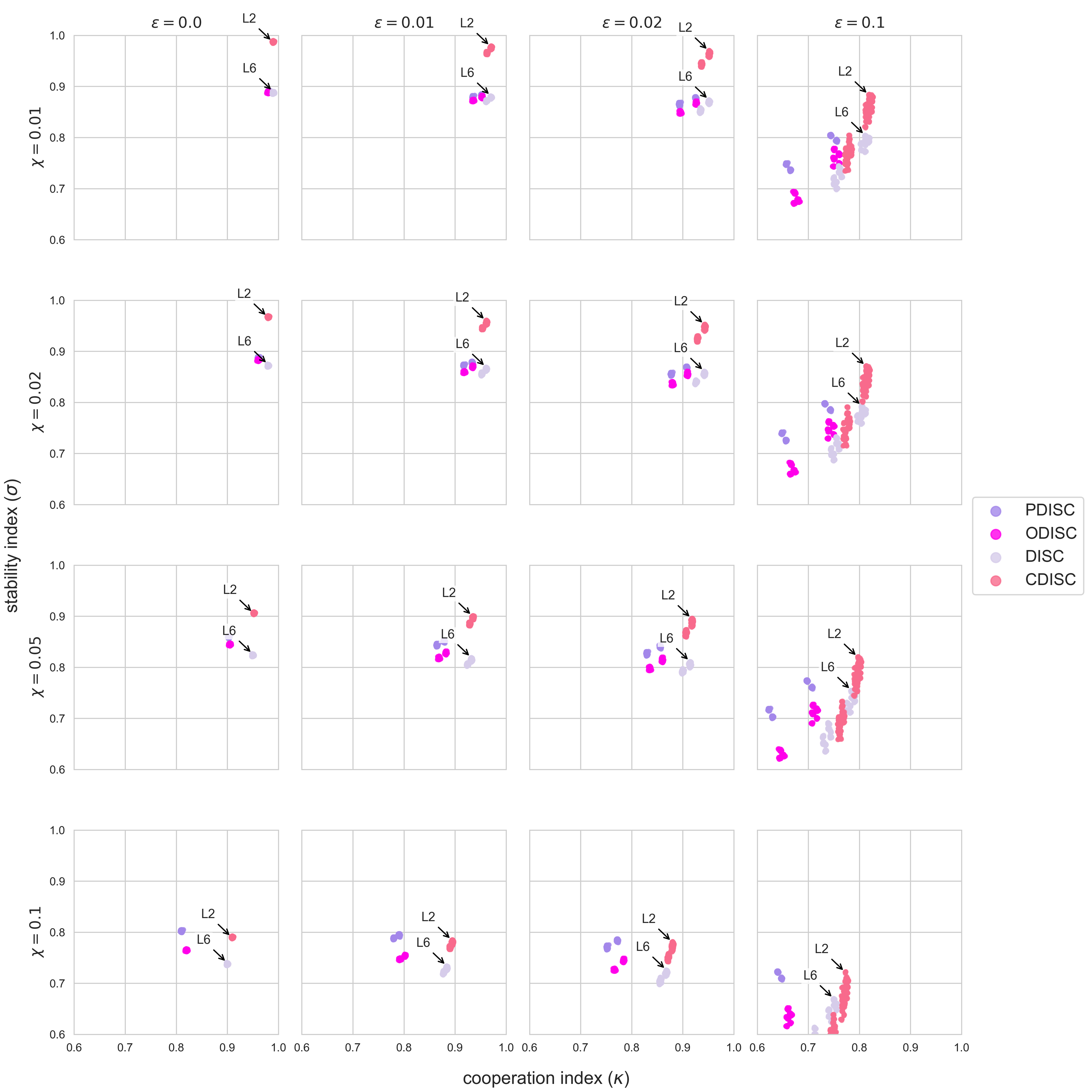}
  \caption{Zoomed scatter plots of the stability and cooperation indices of the consistent discriminating systems, for each combination of execution error rate $\epsilon \in \{0, 0.01, 0.02, 0.1\}$ (columns) and assignment error rate $\chi \in \{0.01, 0.02, 0.05, 0.1\}$ (rows), colored by strategy. L2 and L6, the only consistent discriminating systems among the leading eight, are marked with arrows. L2 is among the most stable systems at every combination. Points are jittered for clarity. \label{fig:alternative_errors_zoomed_scatter}
  }
\end{figure}

For completeness, Figs.~\ref{fig:nojitter-chi001-eps00}--\ref{fig:nojitter-chi01-eps1} reproduce Fig.~\ref{MT-fig:scatter} of the main text for each of the sixteen combinations of the error rates $\chi \in \{0.01, 0.02, 0.05, 0.1\}$ and $\epsilon \in \{0, 0.01, 0.02, 0.1\}$, with no jitter applied to the scatter points, so that small differences in stability and cooperation between individual systems can be read off reliably.

\begin{figure}[p]
  \centering
  \includegraphics[width=\textwidth]{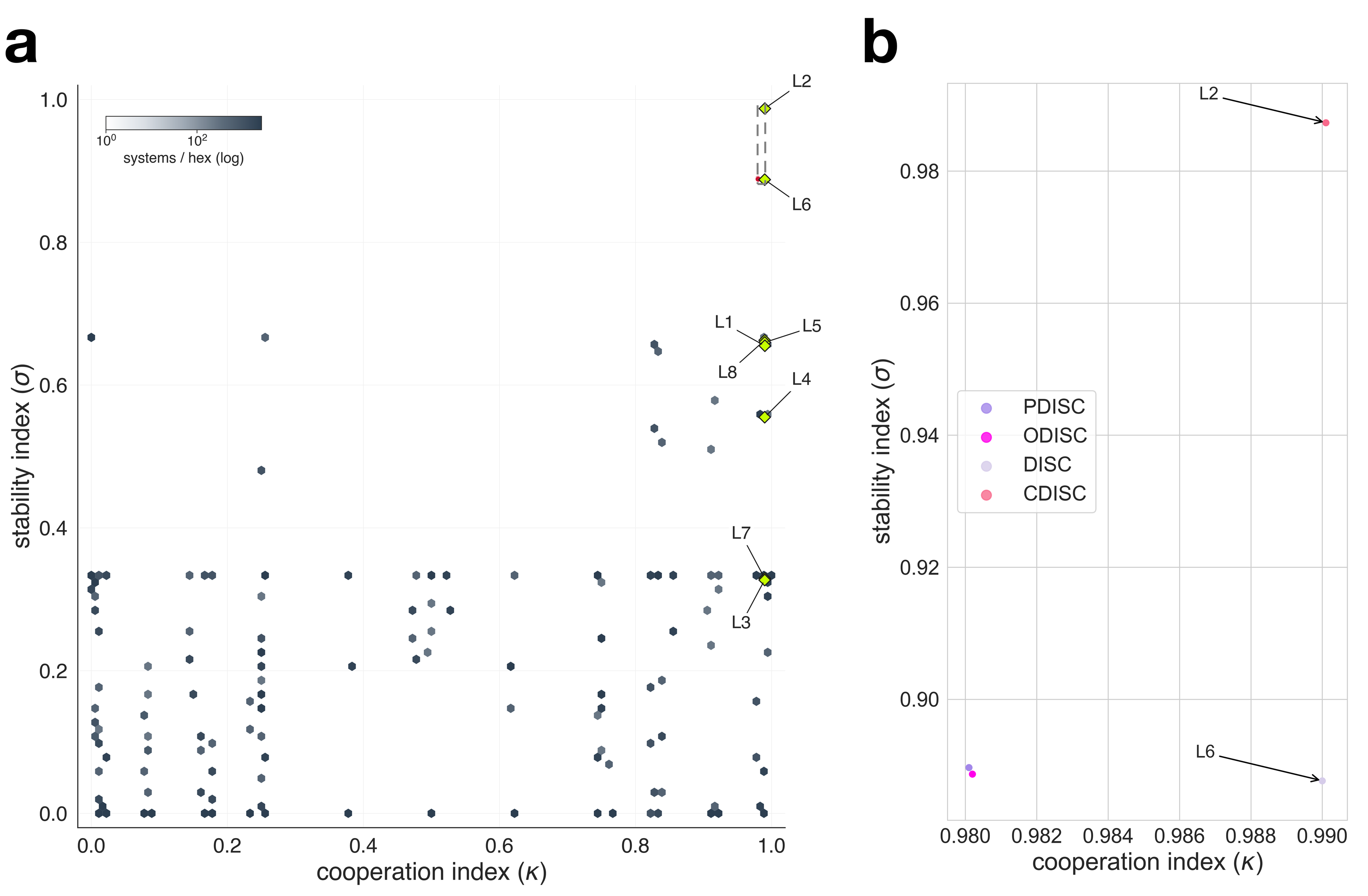}
  \caption{Stability and cooperation indices of all moral systems, without jitter, for $\chi = 0.01$ and $\epsilon = 0$. \textbf{a}, Each of the $\numprint{524800}$ moral systems that are distinct up to the relabeling of reputations (one representative per mirror class; Section~\ref{subsec:mirror-classes}) is positioned by its resident cooperation index ($\kappa$) and average stability index ($\sigma$); the consistent discriminating systems are shown as red circles, the leading eight as green diamonds, and the remaining systems as a hexbin density layer on a logarithmic color scale. \textbf{b}, Zoomed-in scatter plot of the consistent discriminating systems, colored by strategy ($\PDISC$, $\ODISC$, $\DISC$, $\CDISC$). No jitter is applied in either panel, so that small differences between individual systems can be read off reliably.\label{fig:nojitter-chi001-eps00}}
\end{figure}

\begin{figure}[p]
  \centering
  \includegraphics[width=\textwidth]{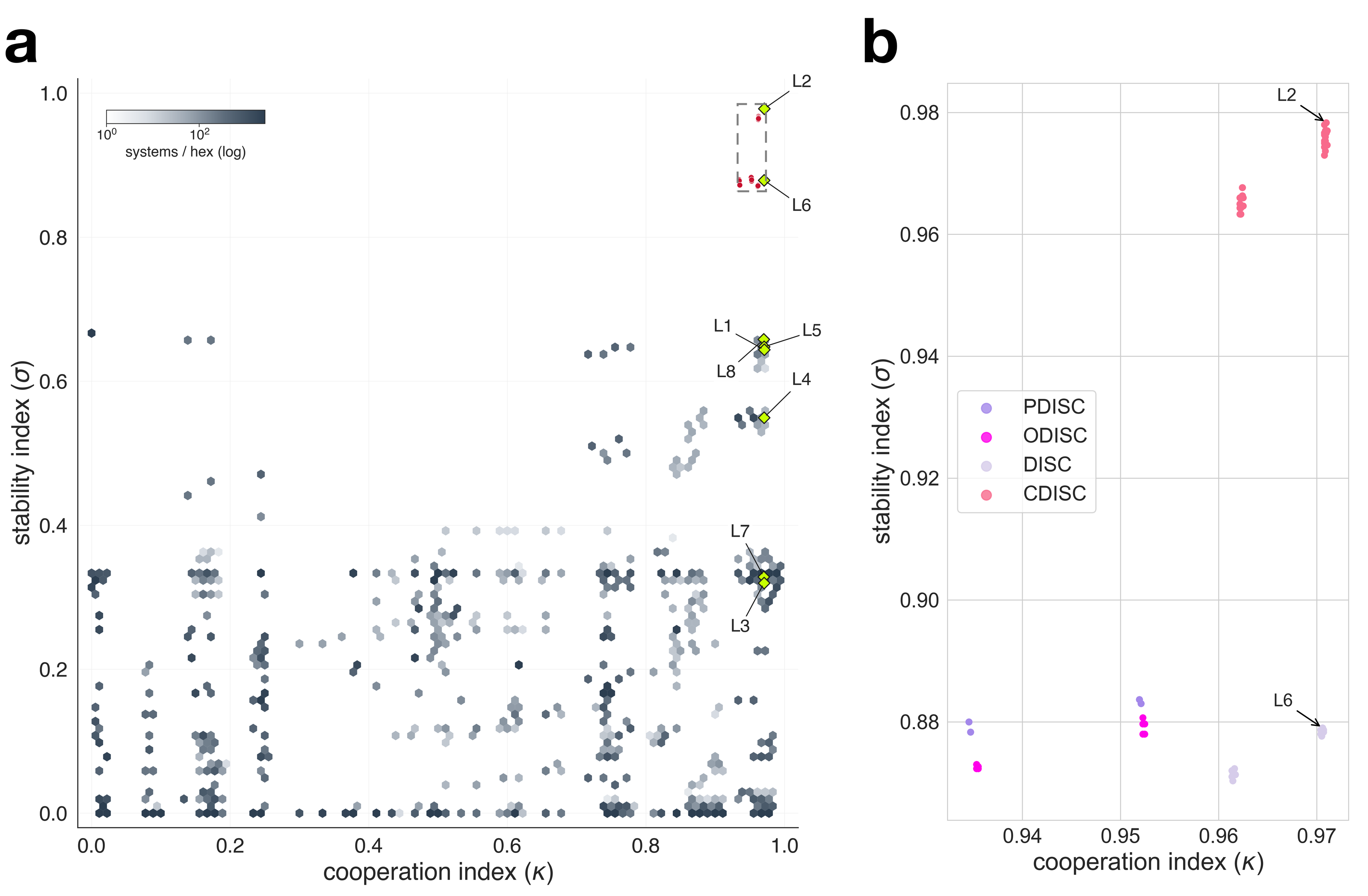}
  \caption{Stability and cooperation indices of all moral systems, without jitter, for $\chi = 0.01$ and $\epsilon = 0.01$. \textbf{a}, Each of the $\numprint{524800}$ moral systems that are distinct up to the relabeling of reputations (one representative per mirror class; Section~\ref{subsec:mirror-classes}) is positioned by its resident cooperation index ($\kappa$) and average stability index ($\sigma$); the consistent discriminating systems are shown as red circles, the leading eight as green diamonds, and the remaining systems as a hexbin density layer on a logarithmic color scale. \textbf{b}, Zoomed-in scatter plot of the consistent discriminating systems, colored by strategy ($\PDISC$, $\ODISC$, $\DISC$, $\CDISC$). No jitter is applied in either panel, so that small differences between individual systems can be read off reliably.\label{fig:nojitter-chi001-eps01}}
\end{figure}

\begin{figure}[p]
  \centering
  \includegraphics[width=\textwidth]{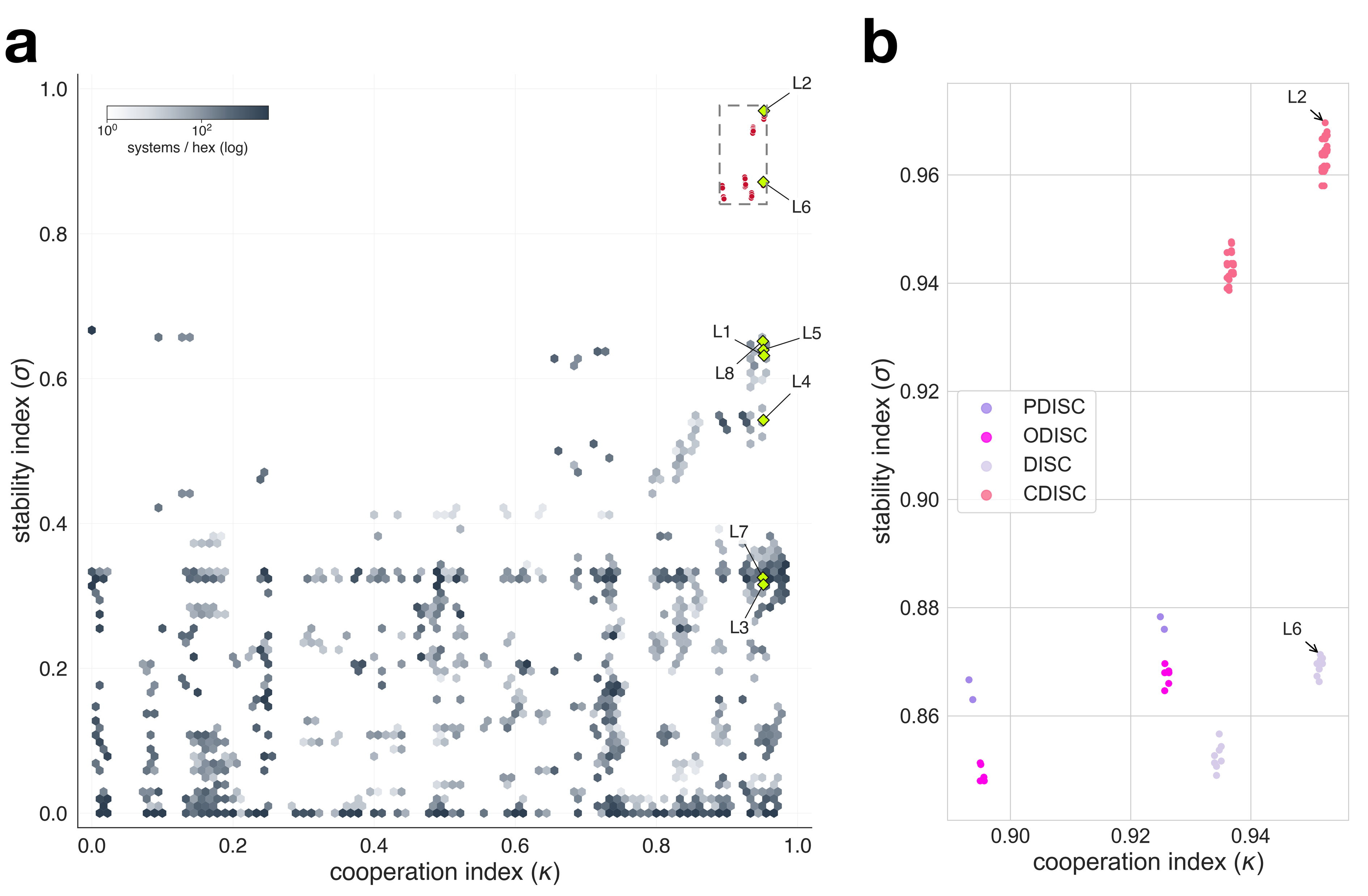}
  \caption{Stability and cooperation indices of all moral systems, without jitter, for $\chi = 0.01$ and $\epsilon = 0.02$. \textbf{a}, Each of the $\numprint{524800}$ moral systems that are distinct up to the relabeling of reputations (one representative per mirror class; Section~\ref{subsec:mirror-classes}) is positioned by its resident cooperation index ($\kappa$) and average stability index ($\sigma$); the consistent discriminating systems are shown as red circles, the leading eight as green diamonds, and the remaining systems as a hexbin density layer on a logarithmic color scale. \textbf{b}, Zoomed-in scatter plot of the consistent discriminating systems, colored by strategy ($\PDISC$, $\ODISC$, $\DISC$, $\CDISC$). No jitter is applied in either panel, so that small differences between individual systems can be read off reliably.\label{fig:nojitter-chi001-eps02}}
\end{figure}

\begin{figure}[p]
  \centering
  \includegraphics[width=\textwidth]{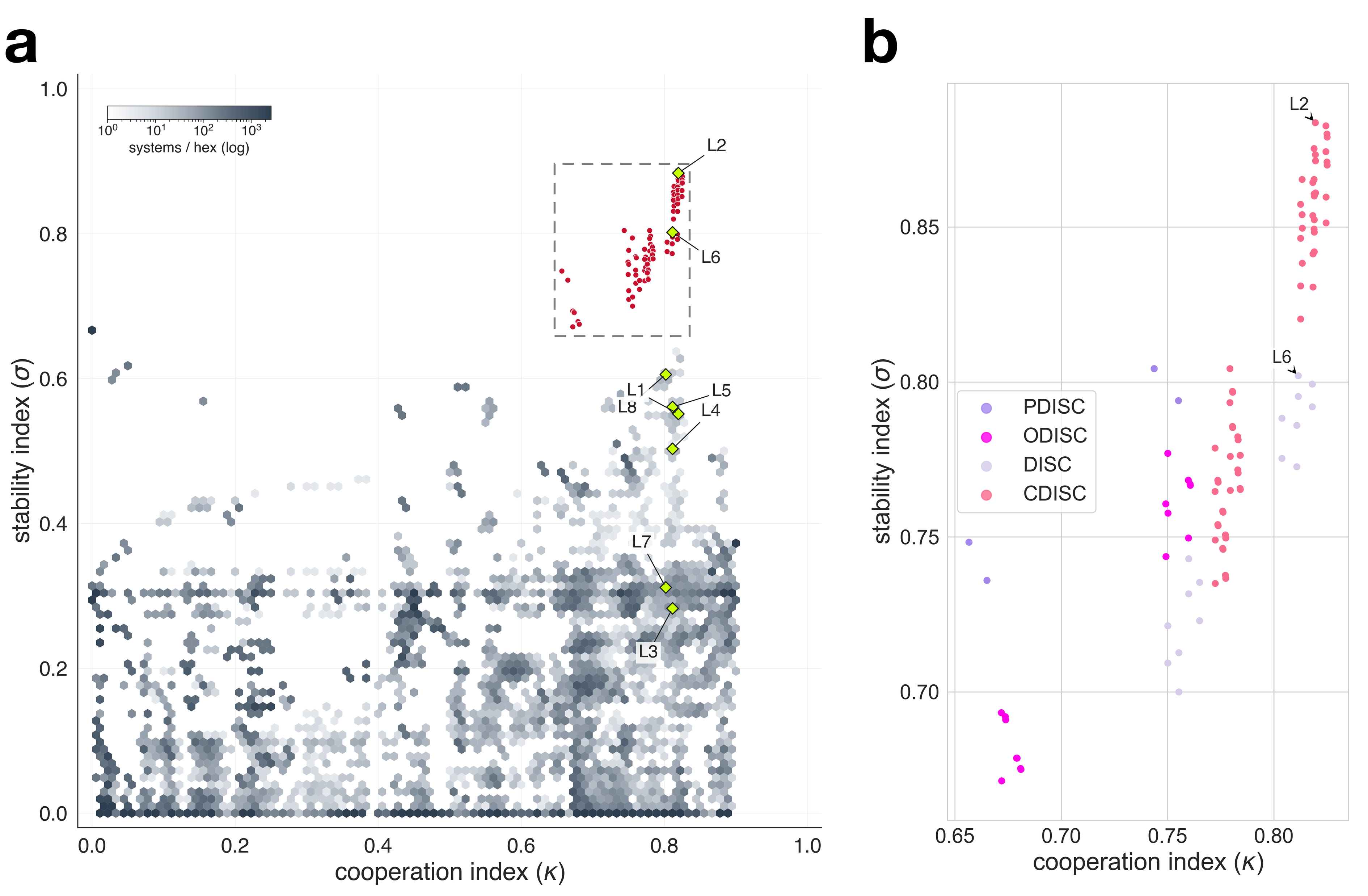}
  \caption{Stability and cooperation indices of all moral systems, without jitter, for $\chi = 0.01$ and $\epsilon = 0.1$. \textbf{a}, Each of the $\numprint{524800}$ moral systems that are distinct up to the relabeling of reputations (one representative per mirror class; Section~\ref{subsec:mirror-classes}) is positioned by its resident cooperation index ($\kappa$) and average stability index ($\sigma$); the consistent discriminating systems are shown as red circles, the leading eight as green diamonds, and the remaining systems as a hexbin density layer on a logarithmic color scale. \textbf{b}, Zoomed-in scatter plot of the consistent discriminating systems, colored by strategy ($\PDISC$, $\ODISC$, $\DISC$, $\CDISC$). No jitter is applied in either panel, so that small differences between individual systems can be read off reliably.\label{fig:nojitter-chi001-eps1}}
\end{figure}

\begin{figure}[p]
  \centering
  \includegraphics[width=\textwidth]{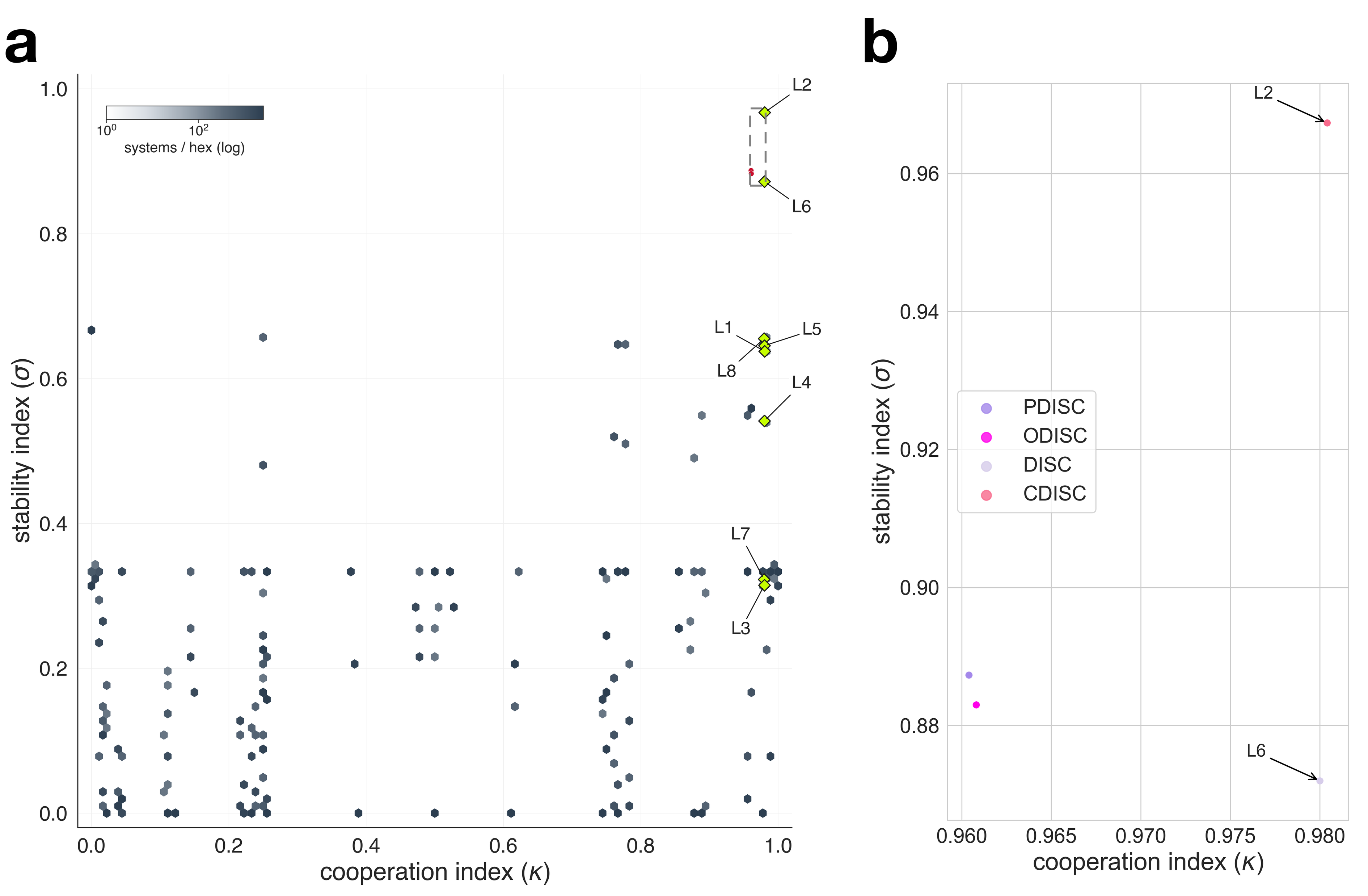}
  \caption{Stability and cooperation indices of all moral systems, without jitter, for $\chi = 0.02$ and $\epsilon = 0$. \textbf{a}, Each of the $\numprint{524800}$ moral systems that are distinct up to the relabeling of reputations (one representative per mirror class; Section~\ref{subsec:mirror-classes}) is positioned by its resident cooperation index ($\kappa$) and average stability index ($\sigma$); the consistent discriminating systems are shown as red circles, the leading eight as green diamonds, and the remaining systems as a hexbin density layer on a logarithmic color scale. \textbf{b}, Zoomed-in scatter plot of the consistent discriminating systems, colored by strategy ($\PDISC$, $\ODISC$, $\DISC$, $\CDISC$). No jitter is applied in either panel, so that small differences between individual systems can be read off reliably.\label{fig:nojitter-chi002-eps00}}
\end{figure}

\begin{figure}[p]
  \centering
  \includegraphics[width=\textwidth]{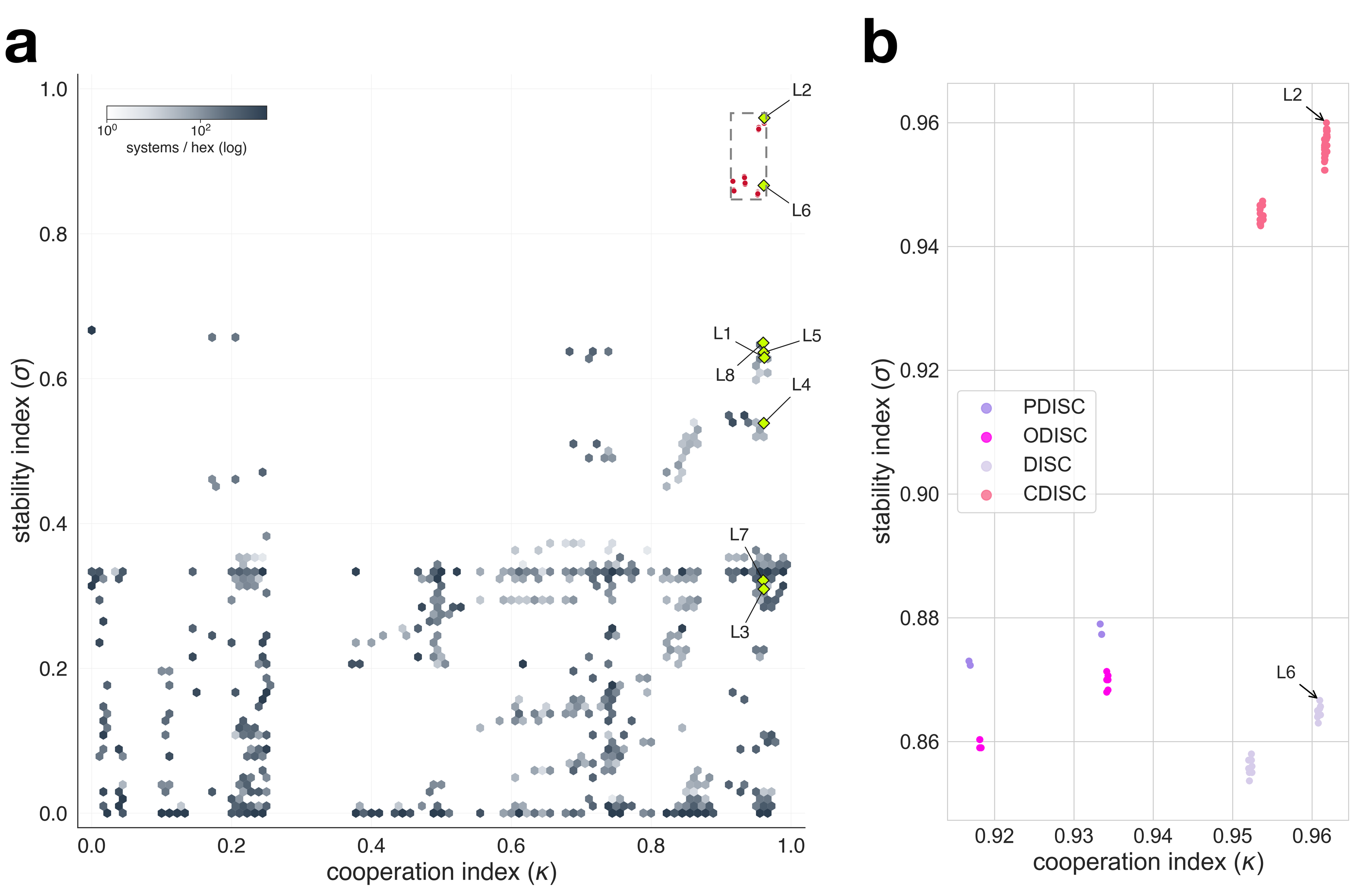}
  \caption{Stability and cooperation indices of all moral systems, without jitter, for $\chi = 0.02$ and $\epsilon = 0.01$. \textbf{a}, Each of the $\numprint{524800}$ moral systems that are distinct up to the relabeling of reputations (one representative per mirror class; Section~\ref{subsec:mirror-classes}) is positioned by its resident cooperation index ($\kappa$) and average stability index ($\sigma$); the consistent discriminating systems are shown as red circles, the leading eight as green diamonds, and the remaining systems as a hexbin density layer on a logarithmic color scale. \textbf{b}, Zoomed-in scatter plot of the consistent discriminating systems, colored by strategy ($\PDISC$, $\ODISC$, $\DISC$, $\CDISC$). No jitter is applied in either panel, so that small differences between individual systems can be read off reliably.\label{fig:nojitter-chi002-eps01}}
\end{figure}

\begin{figure}[p]
  \centering
  \includegraphics[width=\textwidth]{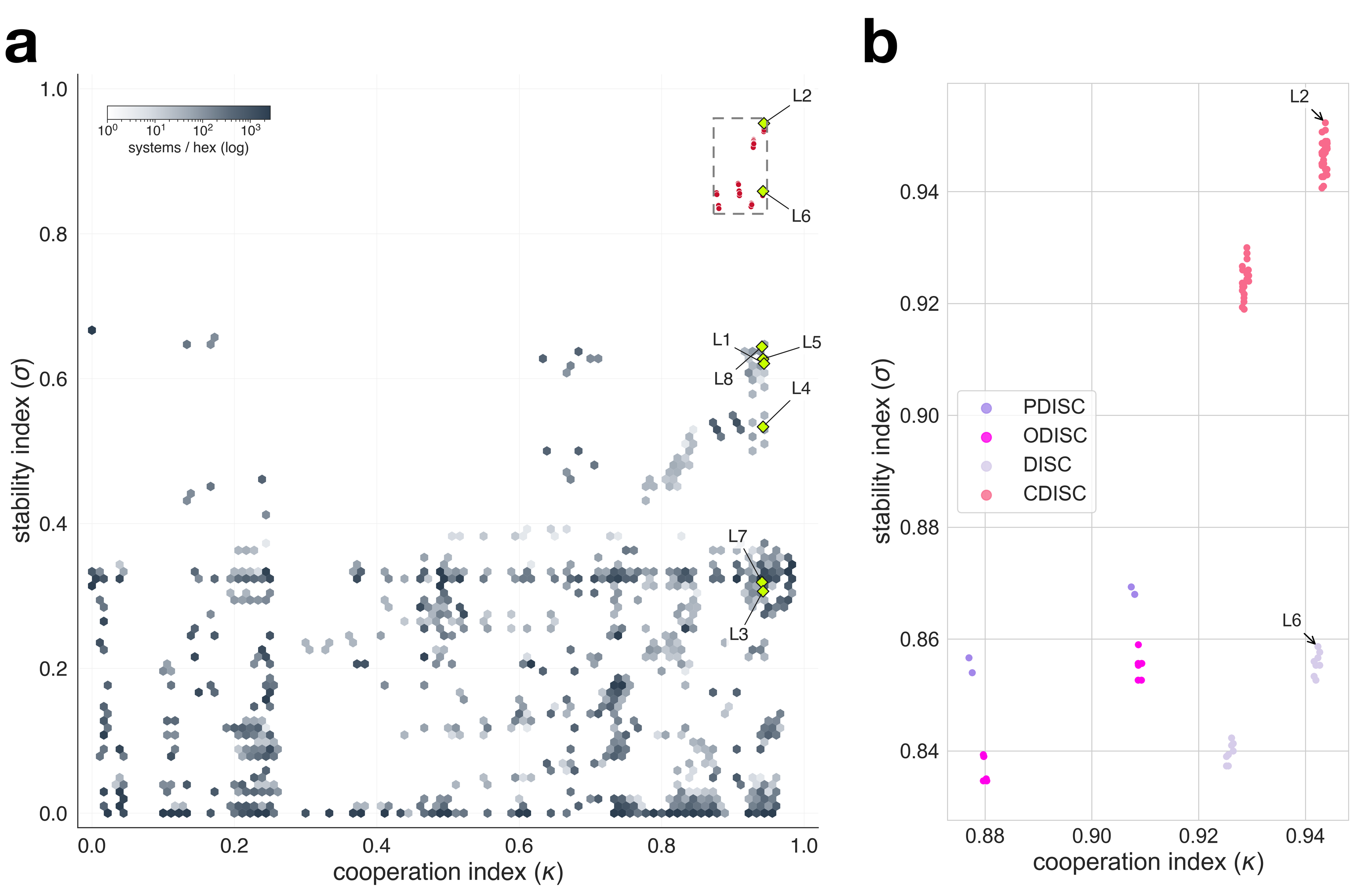}
  \caption{Stability and cooperation indices of all moral systems, without jitter, for $\chi = 0.02$ and $\epsilon = 0.02$. \textbf{a}, Each of the $\numprint{524800}$ moral systems that are distinct up to the relabeling of reputations (one representative per mirror class; Section~\ref{subsec:mirror-classes}) is positioned by its resident cooperation index ($\kappa$) and average stability index ($\sigma$); the consistent discriminating systems are shown as red circles, the leading eight as green diamonds, and the remaining systems as a hexbin density layer on a logarithmic color scale. \textbf{b}, Zoomed-in scatter plot of the consistent discriminating systems, colored by strategy ($\PDISC$, $\ODISC$, $\DISC$, $\CDISC$). No jitter is applied in either panel, so that small differences between individual systems can be read off reliably. This figure is the unjittered counterpart of Fig.~\ref{MT-fig:scatter} of the main text.\label{fig:nojitter-chi002-eps02}}
\end{figure}

\begin{figure}[p]
  \centering
  \includegraphics[width=\textwidth]{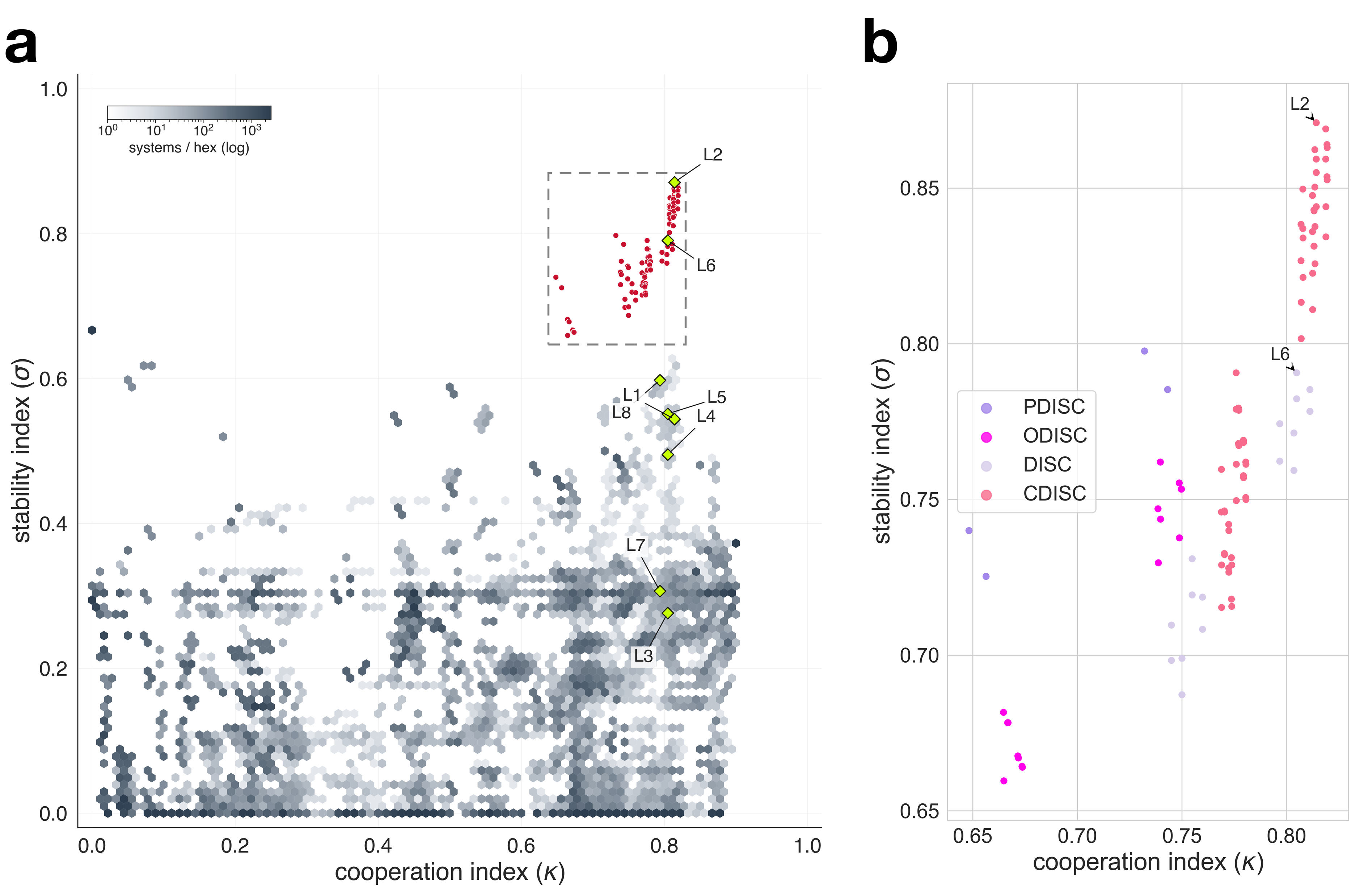}
  \caption{Stability and cooperation indices of all moral systems, without jitter, for $\chi = 0.02$ and $\epsilon = 0.1$. \textbf{a}, Each of the $\numprint{524800}$ moral systems that are distinct up to the relabeling of reputations (one representative per mirror class; Section~\ref{subsec:mirror-classes}) is positioned by its resident cooperation index ($\kappa$) and average stability index ($\sigma$); the consistent discriminating systems are shown as red circles, the leading eight as green diamonds, and the remaining systems as a hexbin density layer on a logarithmic color scale. \textbf{b}, Zoomed-in scatter plot of the consistent discriminating systems, colored by strategy ($\PDISC$, $\ODISC$, $\DISC$, $\CDISC$). No jitter is applied in either panel, so that small differences between individual systems can be read off reliably.\label{fig:nojitter-chi002-eps1}}
\end{figure}

\begin{figure}[p]
  \centering
  \includegraphics[width=\textwidth]{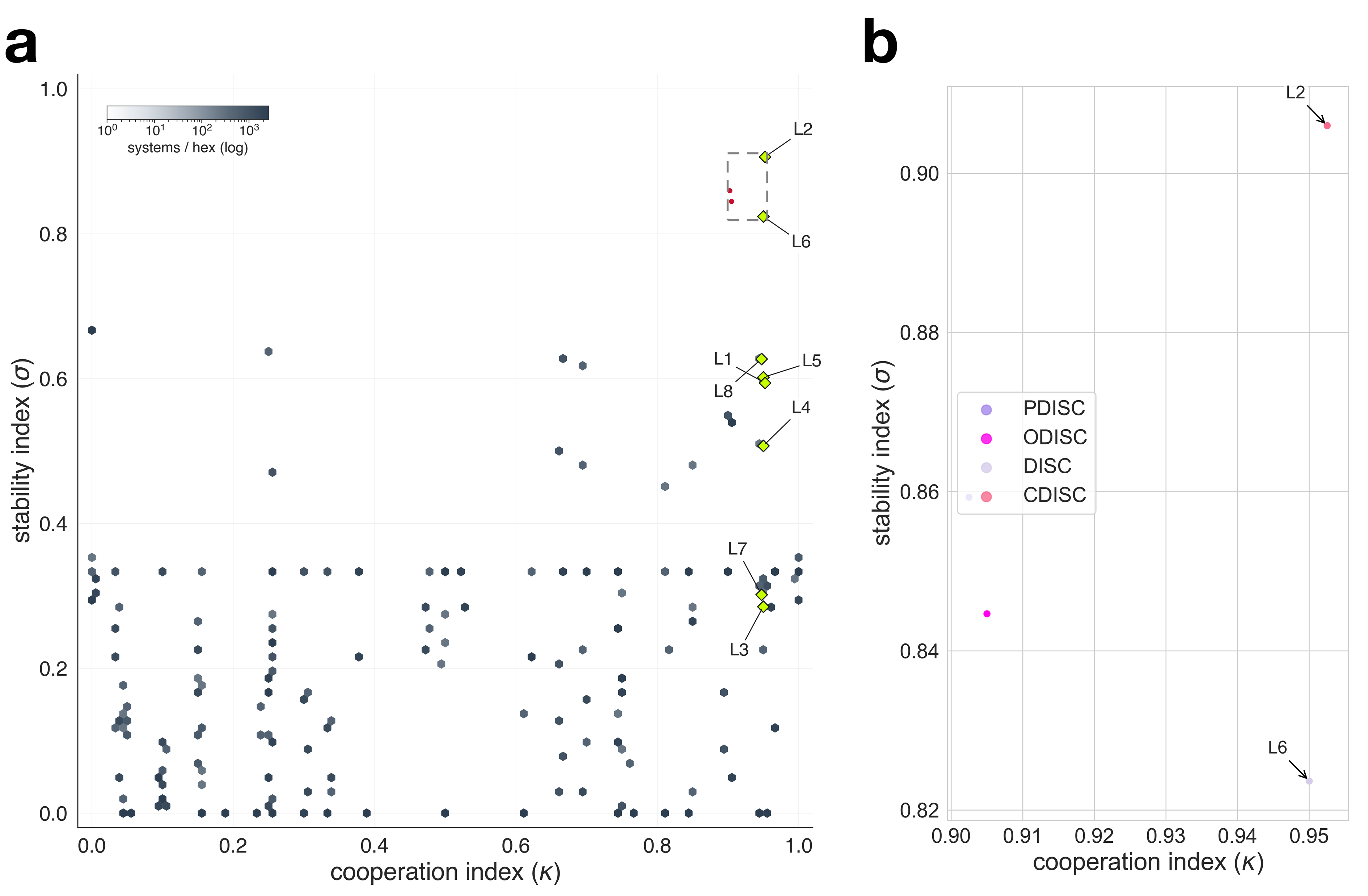}
  \caption{Stability and cooperation indices of all moral systems, without jitter, for $\chi = 0.05$ and $\epsilon = 0$. \textbf{a}, Each of the $\numprint{524800}$ moral systems that are distinct up to the relabeling of reputations (one representative per mirror class; Section~\ref{subsec:mirror-classes}) is positioned by its resident cooperation index ($\kappa$) and average stability index ($\sigma$); the consistent discriminating systems are shown as red circles, the leading eight as green diamonds, and the remaining systems as a hexbin density layer on a logarithmic color scale. \textbf{b}, Zoomed-in scatter plot of the consistent discriminating systems, colored by strategy ($\PDISC$, $\ODISC$, $\DISC$, $\CDISC$). No jitter is applied in either panel, so that small differences between individual systems can be read off reliably.\label{fig:nojitter-chi005-eps00}}
\end{figure}

\begin{figure}[p]
  \centering
  \includegraphics[width=\textwidth]{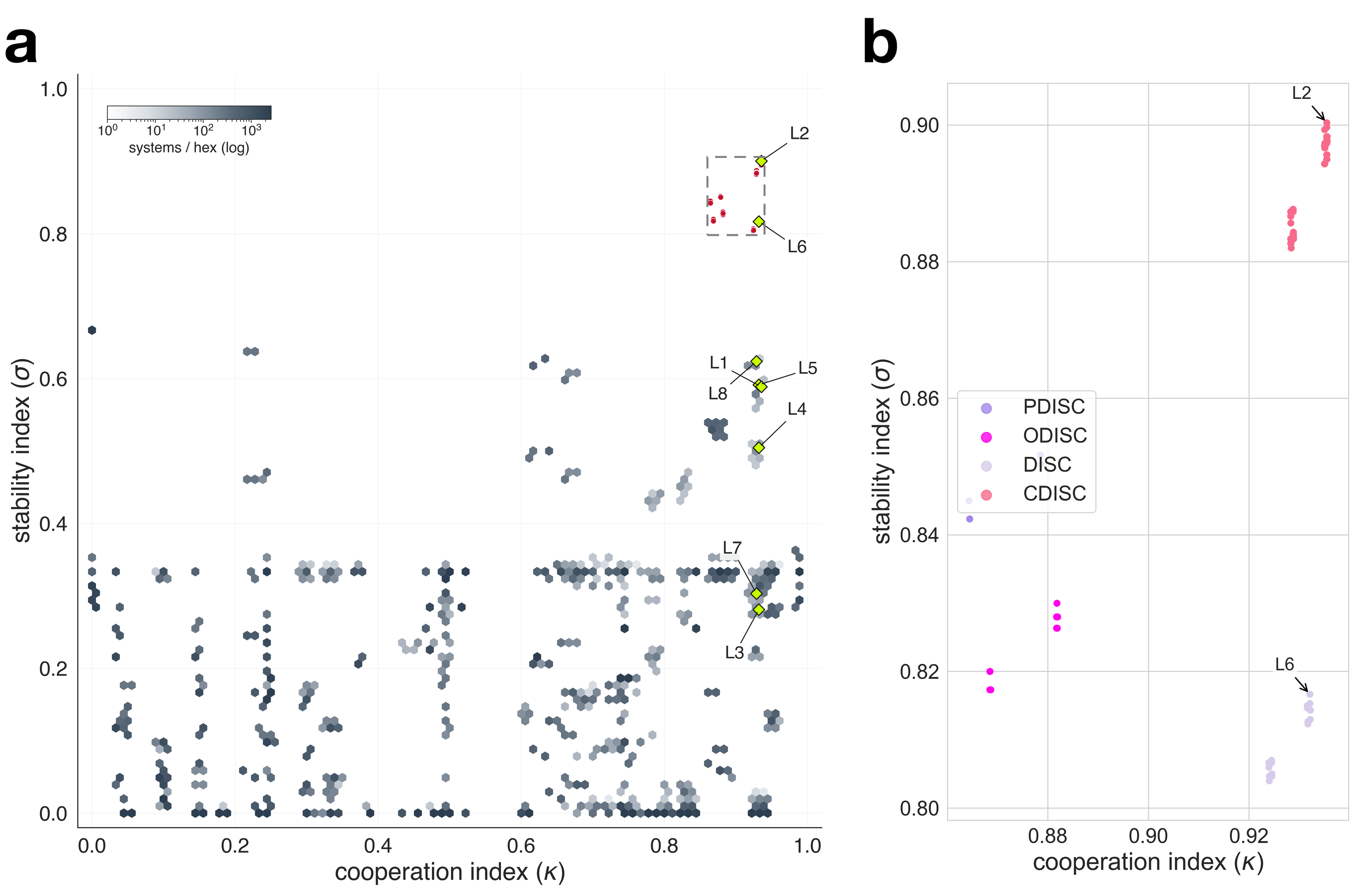}
  \caption{Stability and cooperation indices of all moral systems, without jitter, for $\chi = 0.05$ and $\epsilon = 0.01$. \textbf{a}, Each of the $\numprint{524800}$ moral systems that are distinct up to the relabeling of reputations (one representative per mirror class; Section~\ref{subsec:mirror-classes}) is positioned by its resident cooperation index ($\kappa$) and average stability index ($\sigma$); the consistent discriminating systems are shown as red circles, the leading eight as green diamonds, and the remaining systems as a hexbin density layer on a logarithmic color scale. \textbf{b}, Zoomed-in scatter plot of the consistent discriminating systems, colored by strategy ($\PDISC$, $\ODISC$, $\DISC$, $\CDISC$). No jitter is applied in either panel, so that small differences between individual systems can be read off reliably.\label{fig:nojitter-chi005-eps01}}
\end{figure}

\begin{figure}[p]
  \centering
  \includegraphics[width=\textwidth]{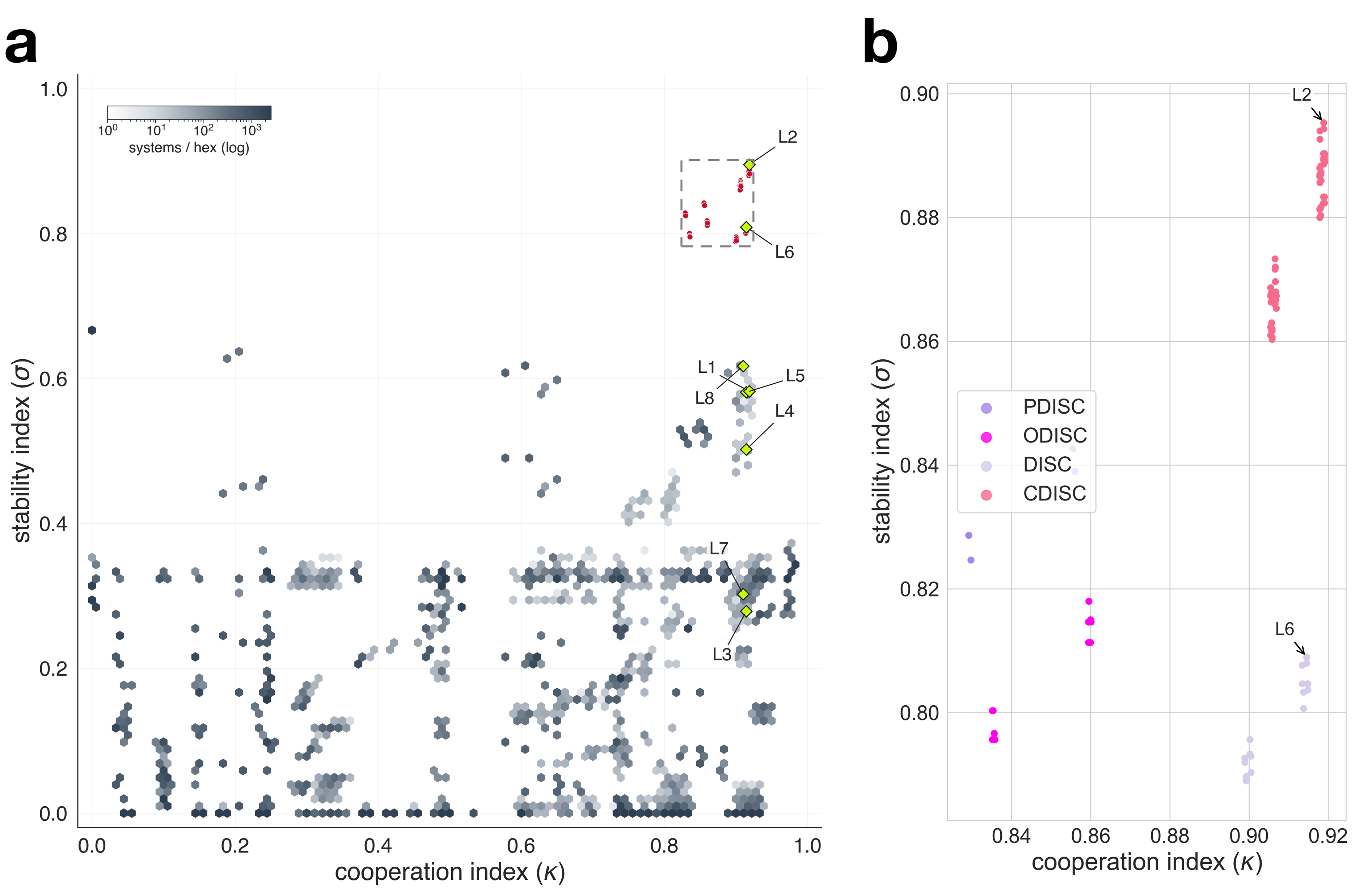}
  \caption{Stability and cooperation indices of all moral systems, without jitter, for $\chi = 0.05$ and $\epsilon = 0.02$. \textbf{a}, Each of the $\numprint{524800}$ moral systems that are distinct up to the relabeling of reputations (one representative per mirror class; Section~\ref{subsec:mirror-classes}) is positioned by its resident cooperation index ($\kappa$) and average stability index ($\sigma$); the consistent discriminating systems are shown as red circles, the leading eight as green diamonds, and the remaining systems as a hexbin density layer on a logarithmic color scale. \textbf{b}, Zoomed-in scatter plot of the consistent discriminating systems, colored by strategy ($\PDISC$, $\ODISC$, $\DISC$, $\CDISC$). No jitter is applied in either panel, so that small differences between individual systems can be read off reliably.\label{fig:nojitter-chi005-eps02}}
\end{figure}

\begin{figure}[p]
  \centering
  \includegraphics[width=\textwidth]{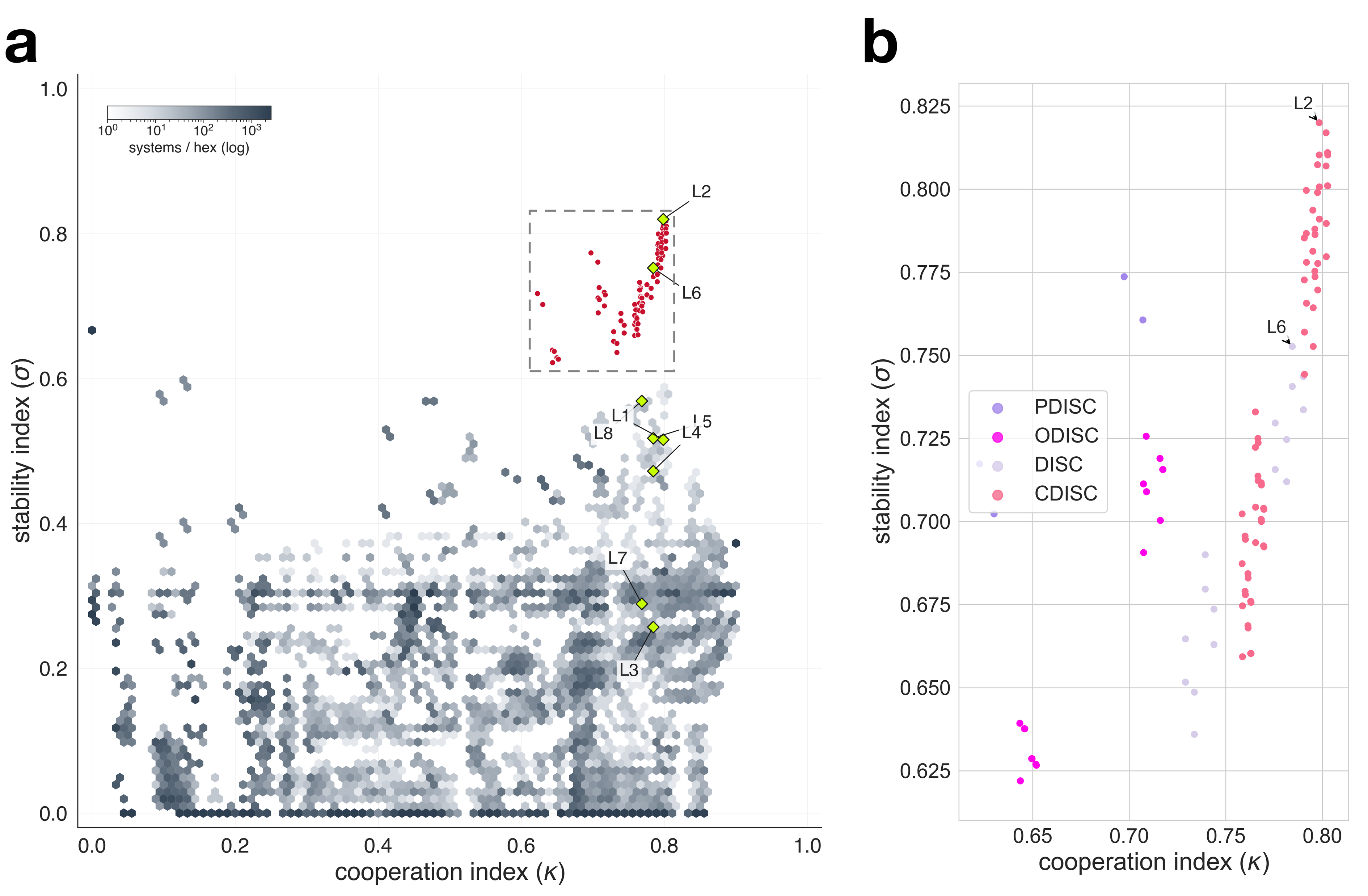}
  \caption{Stability and cooperation indices of all moral systems, without jitter, for $\chi = 0.05$ and $\epsilon = 0.1$. \textbf{a}, Each of the $\numprint{524800}$ moral systems that are distinct up to the relabeling of reputations (one representative per mirror class; Section~\ref{subsec:mirror-classes}) is positioned by its resident cooperation index ($\kappa$) and average stability index ($\sigma$); the consistent discriminating systems are shown as red circles, the leading eight as green diamonds, and the remaining systems as a hexbin density layer on a logarithmic color scale. \textbf{b}, Zoomed-in scatter plot of the consistent discriminating systems, colored by strategy ($\PDISC$, $\ODISC$, $\DISC$, $\CDISC$). No jitter is applied in either panel, so that small differences between individual systems can be read off reliably.\label{fig:nojitter-chi005-eps1}}
\end{figure}

\begin{figure}[p]
  \centering
  \includegraphics[width=\textwidth]{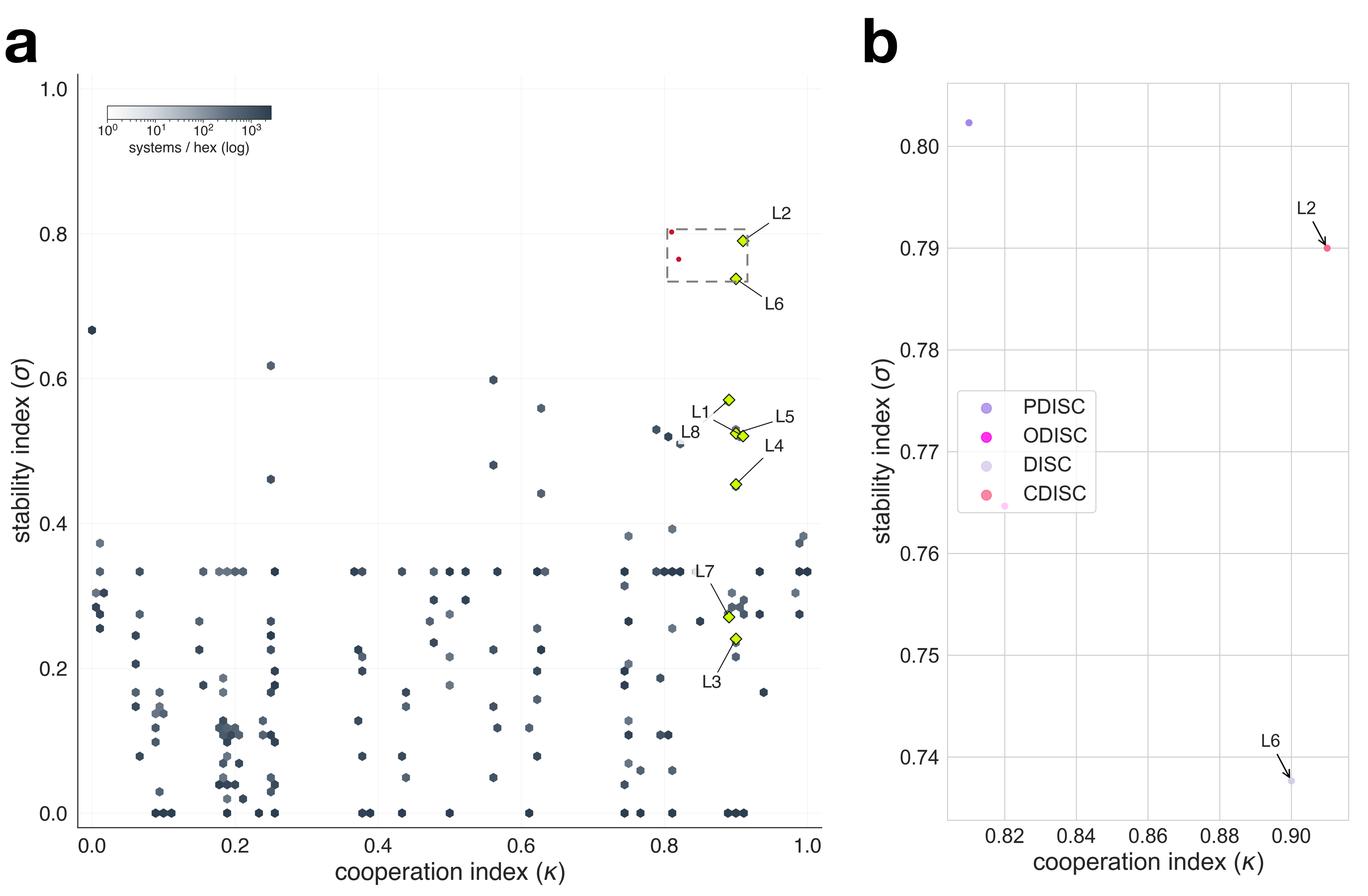}
  \caption{Stability and cooperation indices of all moral systems, without jitter, for $\chi = 0.1$ and $\epsilon = 0$. \textbf{a}, Each of the $\numprint{524800}$ moral systems that are distinct up to the relabeling of reputations (one representative per mirror class; Section~\ref{subsec:mirror-classes}) is positioned by its resident cooperation index ($\kappa$) and average stability index ($\sigma$); the consistent discriminating systems are shown as red circles, the leading eight as green diamonds, and the remaining systems as a hexbin density layer on a logarithmic color scale. \textbf{b}, Zoomed-in scatter plot of the consistent discriminating systems, colored by strategy ($\PDISC$, $\ODISC$, $\DISC$, $\CDISC$). No jitter is applied in either panel, so that small differences between individual systems can be read off reliably.\label{fig:nojitter-chi01-eps00}}
\end{figure}

\begin{figure}[p]
  \centering
  \includegraphics[width=\textwidth]{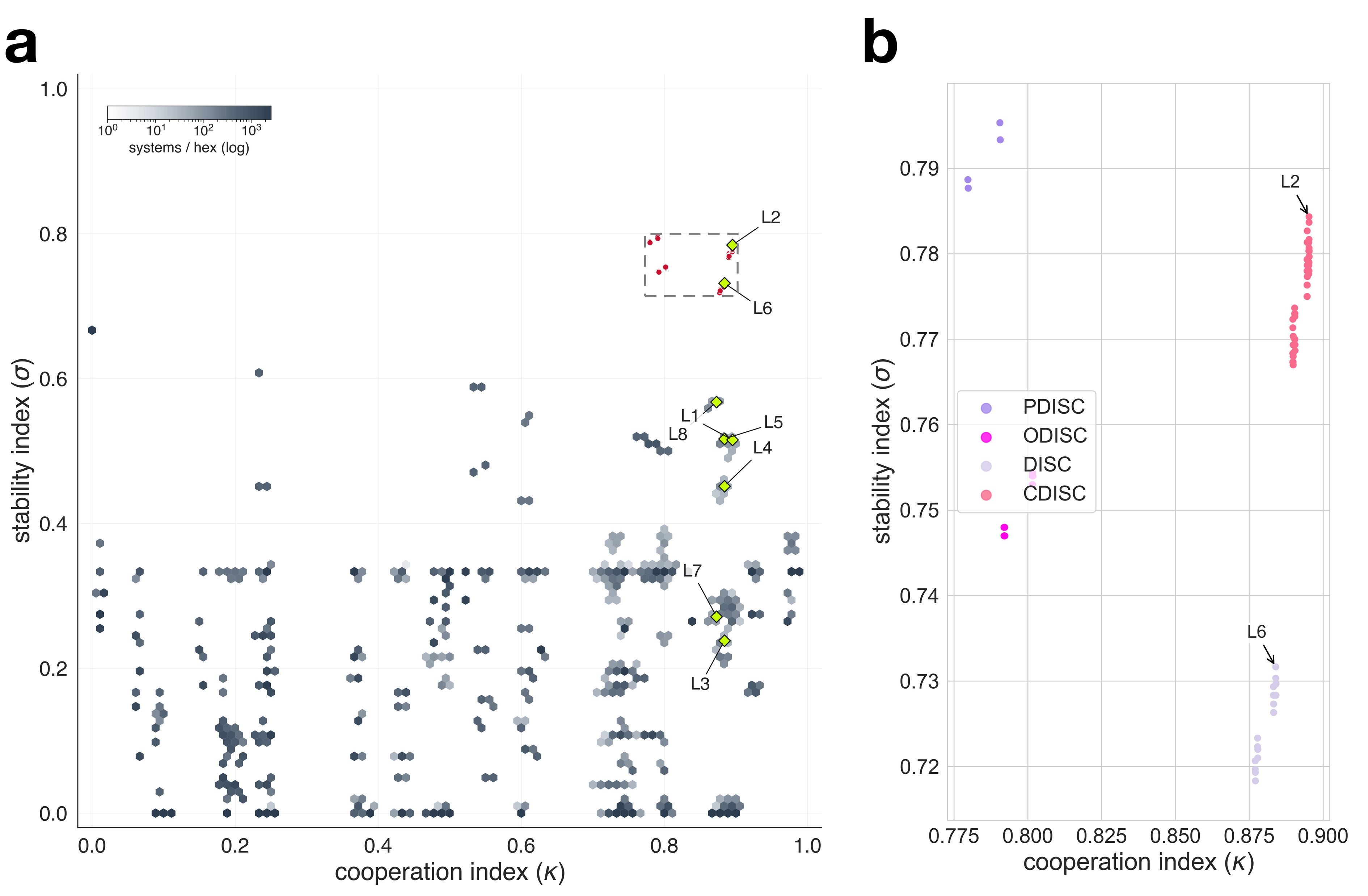}
  \caption{Stability and cooperation indices of all moral systems, without jitter, for $\chi = 0.1$ and $\epsilon = 0.01$. \textbf{a}, Each of the $\numprint{524800}$ moral systems that are distinct up to the relabeling of reputations (one representative per mirror class; Section~\ref{subsec:mirror-classes}) is positioned by its resident cooperation index ($\kappa$) and average stability index ($\sigma$); the consistent discriminating systems are shown as red circles, the leading eight as green diamonds, and the remaining systems as a hexbin density layer on a logarithmic color scale. \textbf{b}, Zoomed-in scatter plot of the consistent discriminating systems, colored by strategy ($\PDISC$, $\ODISC$, $\DISC$, $\CDISC$). No jitter is applied in either panel, so that small differences between individual systems can be read off reliably.\label{fig:nojitter-chi01-eps01}}
\end{figure}

\begin{figure}[p]
  \centering
  \includegraphics[width=\textwidth]{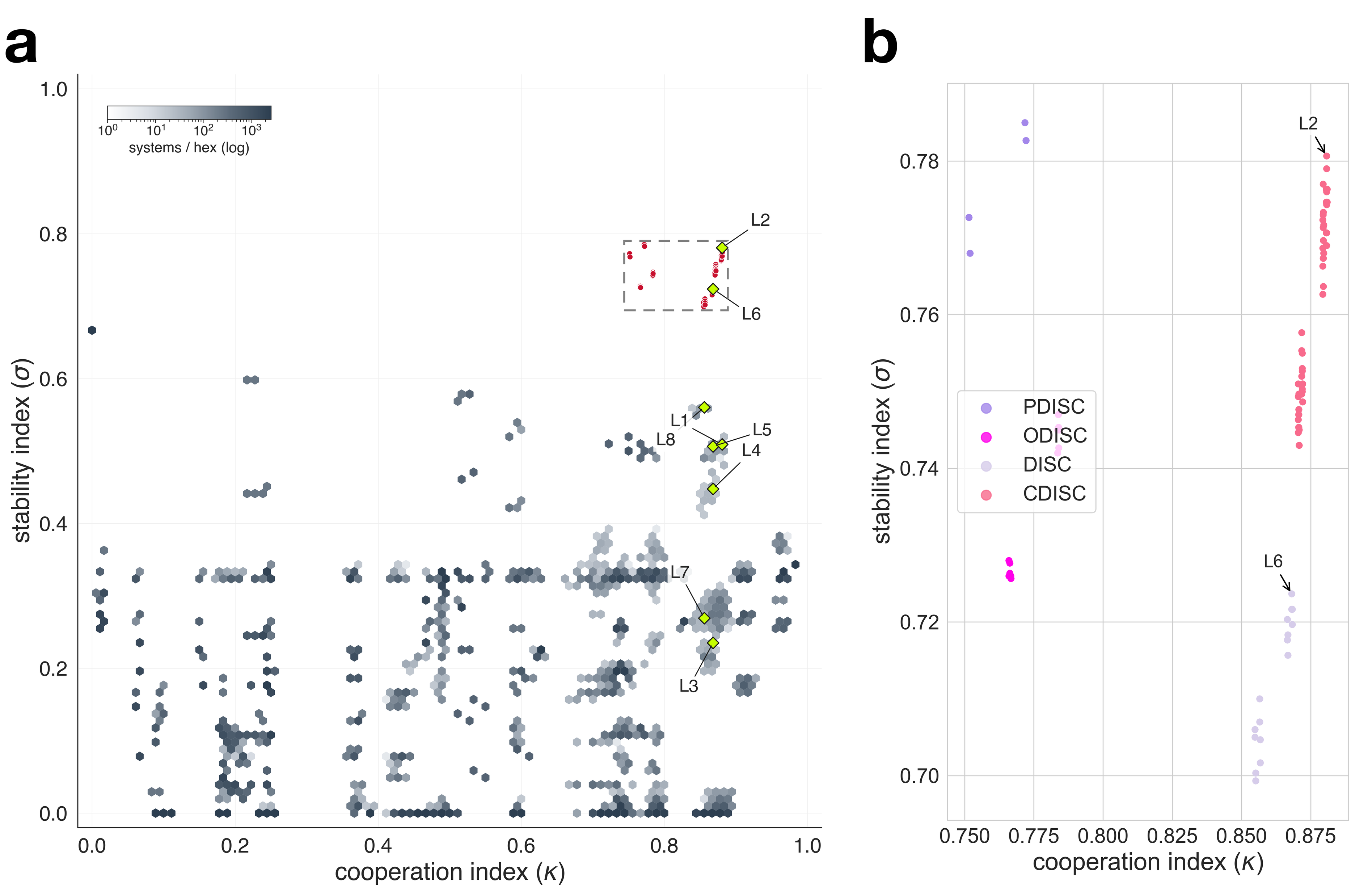}
  \caption{Stability and cooperation indices of all moral systems, without jitter, for $\chi = 0.1$ and $\epsilon = 0.02$. \textbf{a}, Each of the $\numprint{524800}$ moral systems that are distinct up to the relabeling of reputations (one representative per mirror class; Section~\ref{subsec:mirror-classes}) is positioned by its resident cooperation index ($\kappa$) and average stability index ($\sigma$); the consistent discriminating systems are shown as red circles, the leading eight as green diamonds, and the remaining systems as a hexbin density layer on a logarithmic color scale. \textbf{b}, Zoomed-in scatter plot of the consistent discriminating systems, colored by strategy ($\PDISC$, $\ODISC$, $\DISC$, $\CDISC$). No jitter is applied in either panel, so that small differences between individual systems can be read off reliably.\label{fig:nojitter-chi01-eps02}}
\end{figure}

\begin{figure}[p]
  \centering
  \includegraphics[width=\textwidth]{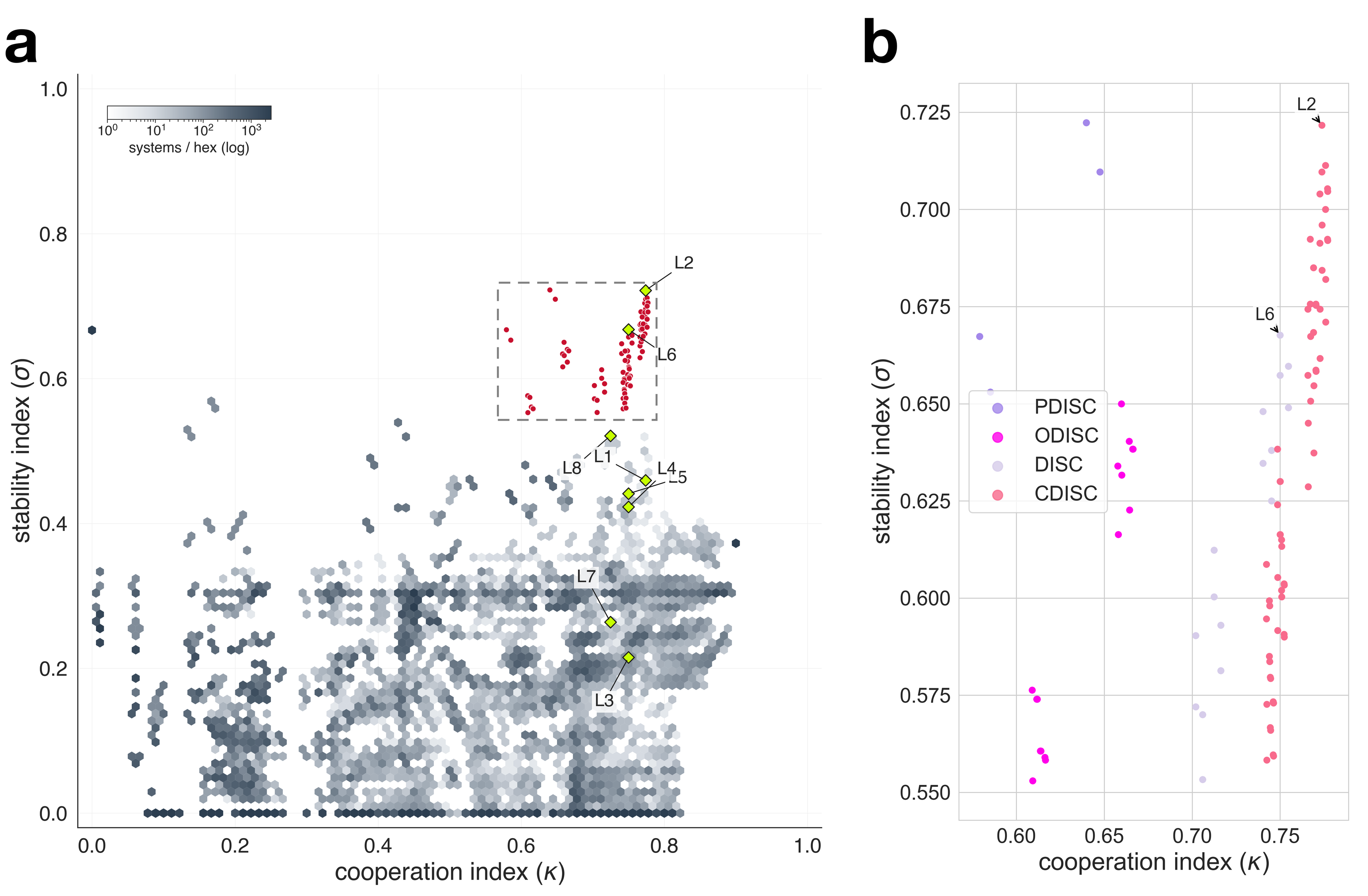}
  \caption{Stability and cooperation indices of all moral systems, without jitter, for $\chi = 0.1$ and $\epsilon = 0.1$. \textbf{a}, Each of the $\numprint{524800}$ moral systems that are distinct up to the relabeling of reputations (one representative per mirror class; Section~\ref{subsec:mirror-classes}) is positioned by its resident cooperation index ($\kappa$) and average stability index ($\sigma$); the consistent discriminating systems are shown as red circles, the leading eight as green diamonds, and the remaining systems as a hexbin density layer on a logarithmic color scale. \textbf{b}, Zoomed-in scatter plot of the consistent discriminating systems, colored by strategy ($\PDISC$, $\ODISC$, $\DISC$, $\CDISC$). No jitter is applied in either panel, so that small differences between individual systems can be read off reliably.\label{fig:nojitter-chi01-eps1}}
\end{figure}

\clearpage

\subsection{Disaggregated stability indices}\label{subsec:disaggregated-indices}

The average stability index $\sigmabar$~\eqref{avg-sigma} used in Fig.~\ref{MT-fig:scatter} of the main text compresses the three classes of cooperative dilemma into a single number, so that systems with similar averages may owe their stability to very different classes.
Figs.~\ref{fig:disaggregated_pd}--\ref{fig:disaggregated_sh} therefore plot the stability index of each class separately.
At the small error rates considered ($\chi = \epsilon = 0.02$), the Monte Carlo estimates lie close to the limiting indices derived analytically in Sections~\ref{subsec:stab-anal-disc} and~\ref{subsec:stab-anal-leading-eight} (Tables~\ref{tab:consistent-disc-summary} and~\ref{tab:leading-eight-summary}), which they thus corroborate.

The zoomed-in panels b show how the stability of the consistent discriminating systems is distributed across classes.
Except for those using $\CDISC$, all consistent discriminating systems have one weak class, in which they fall well below the rest and cluster near the limiting index of $2/3$: the systems using $\ODISC$ or $\PDISC$ in prisoner's dilemmas ($\sigmaPD \approx 0.6$; Fig.~\ref{fig:disaggregated_pd}b), and the systems using $\DISC$ in snowdrift games ($\sigmaSG \approx 0.7$; Fig.~\ref{fig:disaggregated_sg}b).
No system is weak in stag hunts, where all consistent discriminating systems are highly stable ($\sigmaSH > 0.92$; Fig.~\ref{fig:disaggregated_sh}b).
Within any single class, the systems using $\CDISC$ need not be the most stable: those using $\DISC$ attain comparable values in prisoner's dilemmas, and those using $\ODISC$ or $\PDISC$ are stable in essentially every stag hunt ($\sigmaSH \approx 1$).
What distinguishes the systems using $\CDISC$ is instead the absence of a weak class, in line with their universal stability (Theorem~\ref{thm:universal-stability}).

The disaggregated indices also clarify the position of the leading eight (green diamonds in panels a).
All eight are highly stable in prisoner's dilemmas ($\sigmaPD \approx 0.9$; Fig.~\ref{fig:disaggregated_pd}a), the class that contains the donation game for which they were originally derived; outside it, they scatter in line with Table~\ref{tab:leading-eight-summary}.
In snowdrift games (Fig.~\ref{fig:disaggregated_sg}a), only L1 and L2 remain highly stable; L4 and L6 sit near their limiting index of $2/3$, alongside the consistent systems using $\DISC$, while L3, L5, L7, and L8, which can be invaded by $\CDISC$ in every snowdrift game, drop to $\sigmaSG \approx 0$.
In stag hunts (Fig.~\ref{fig:disaggregated_sh}a), L2, L5, L6, and L8 remain highly stable, while L1, L3, L4, and L7, which can be invaded by unconditional cooperators, fall to $\sigmaSH \approx 0$.
L2 is thus the only leading-eight system that is highly stable in all three classes, and the high average stability of L6 in Fig.~\ref{MT-fig:scatter} of the main text rests on prisoner's dilemmas and stag hunts alone.

\begin{figure}[p]
  \centering
  \includegraphics[width=\textwidth]{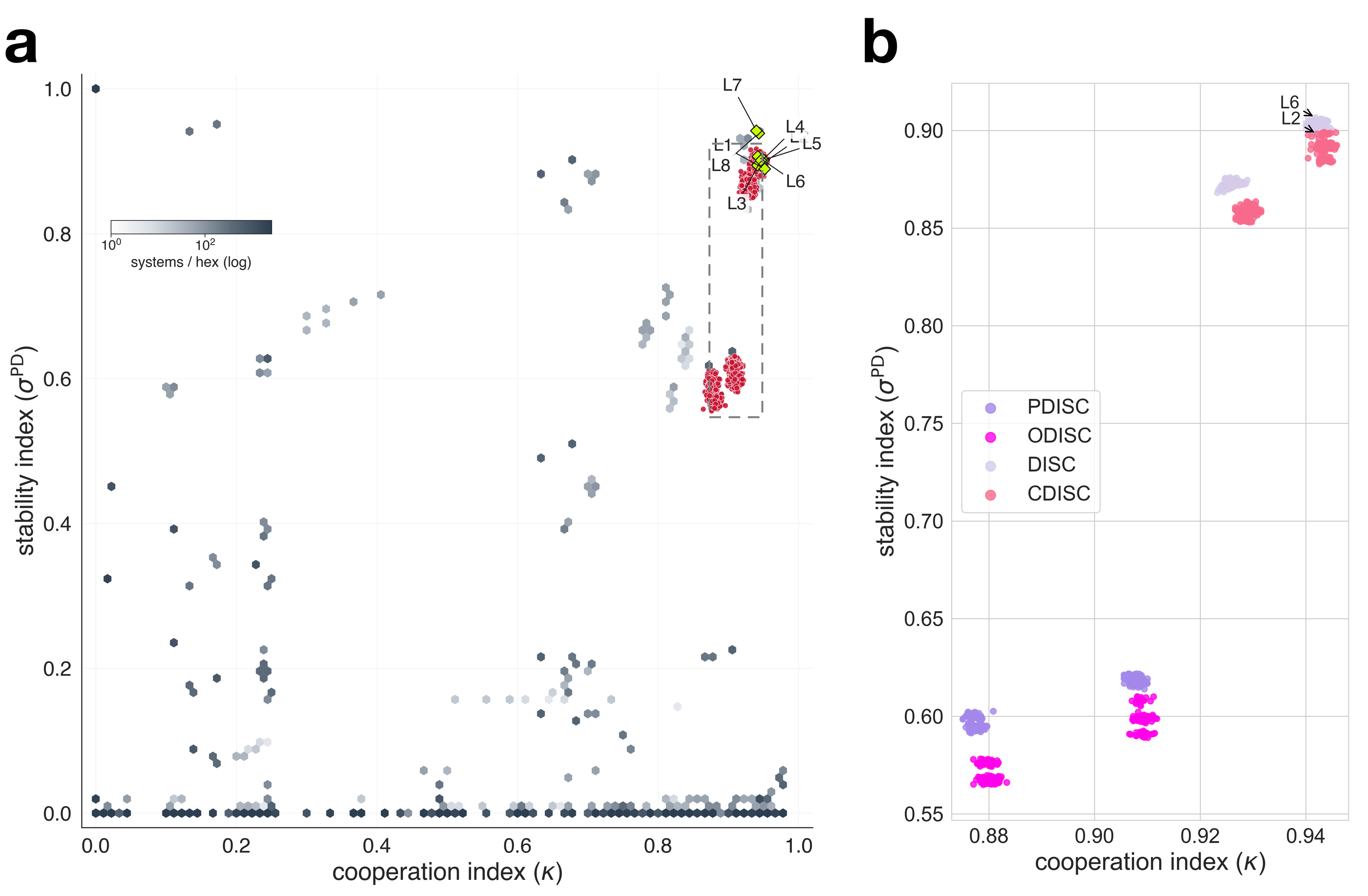}
  \caption{Stability and cooperation indices of all moral systems in the space of prisoner's dilemmas only, for $\chi = \epsilon = 0.02$. Unlike Fig.~\ref{MT-fig:scatter} of the main text, the vertical axis is the stability index of a single class of dilemma, $\sigma^{\mathrm{PD}}$, rather than the average $\sigmabar$ across the three classes. \textbf{a}, Each of the $\numprint{524800}$ moral systems that are distinct up to the relabeling of reputations (one representative per mirror class; Section~\ref{subsec:mirror-classes}) is positioned by its resident cooperation index ($\kappa$) and its prisoner's dilemma stability index ($\sigma^{\mathrm{PD}}$); the consistent discriminating systems are shown as red circles, the leading eight as green diamonds, and the remaining systems as a hexbin density layer on a logarithmic color scale. \textbf{b}, Zoomed-in scatter plot of the consistent discriminating systems, colored by strategy ($\PDISC$, $\ODISC$, $\DISC$, $\CDISC$). Points are jittered for clarity in both panels.\label{fig:disaggregated_pd}}
\end{figure}

\begin{figure}[p]
  \centering
  \includegraphics[width=\textwidth]{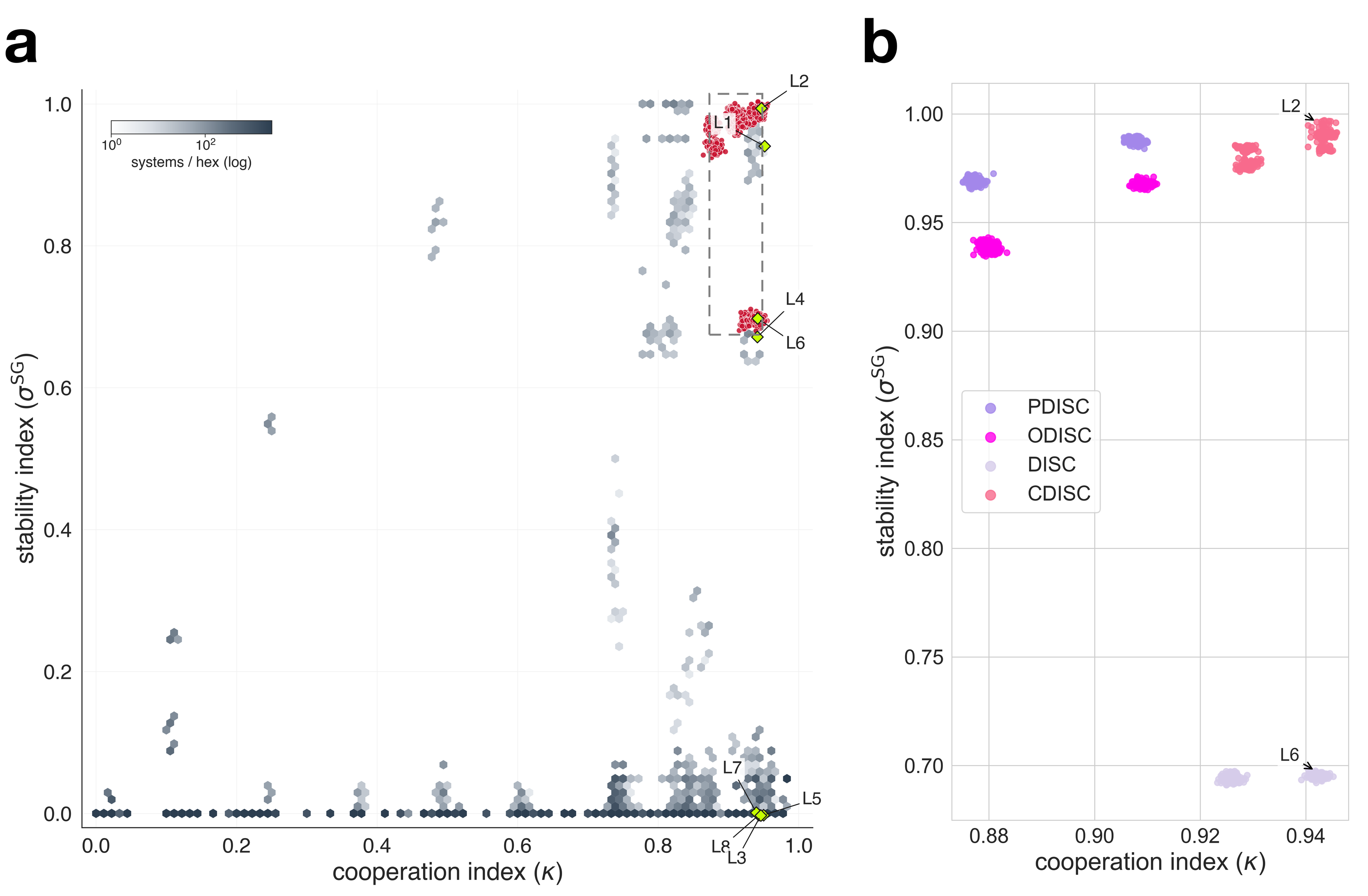}
  \caption{Stability and cooperation indices of all moral systems in the space of snowdrift games only, for $\chi = \epsilon = 0.02$. Unlike Fig.~\ref{MT-fig:scatter} of the main text, the vertical axis is the stability index of a single class of dilemma, $\sigma^{\mathrm{SG}}$, rather than the average $\sigmabar$ across the three classes. \textbf{a}, Each of the $\numprint{524800}$ moral systems that are distinct up to the relabeling of reputations (one representative per mirror class; Section~\ref{subsec:mirror-classes}) is positioned by its resident cooperation index ($\kappa$) and its snowdrift game stability index ($\sigma^{\mathrm{SG}}$); the consistent discriminating systems are shown as red circles, the leading eight as green diamonds, and the remaining systems as a hexbin density layer on a logarithmic color scale. \textbf{b}, Zoomed-in scatter plot of the consistent discriminating systems, colored by strategy ($\PDISC$, $\ODISC$, $\DISC$, $\CDISC$). Points are jittered for clarity in both panels.\label{fig:disaggregated_sg}}
\end{figure}

\begin{figure}[p]
  \centering
  \includegraphics[width=\textwidth]{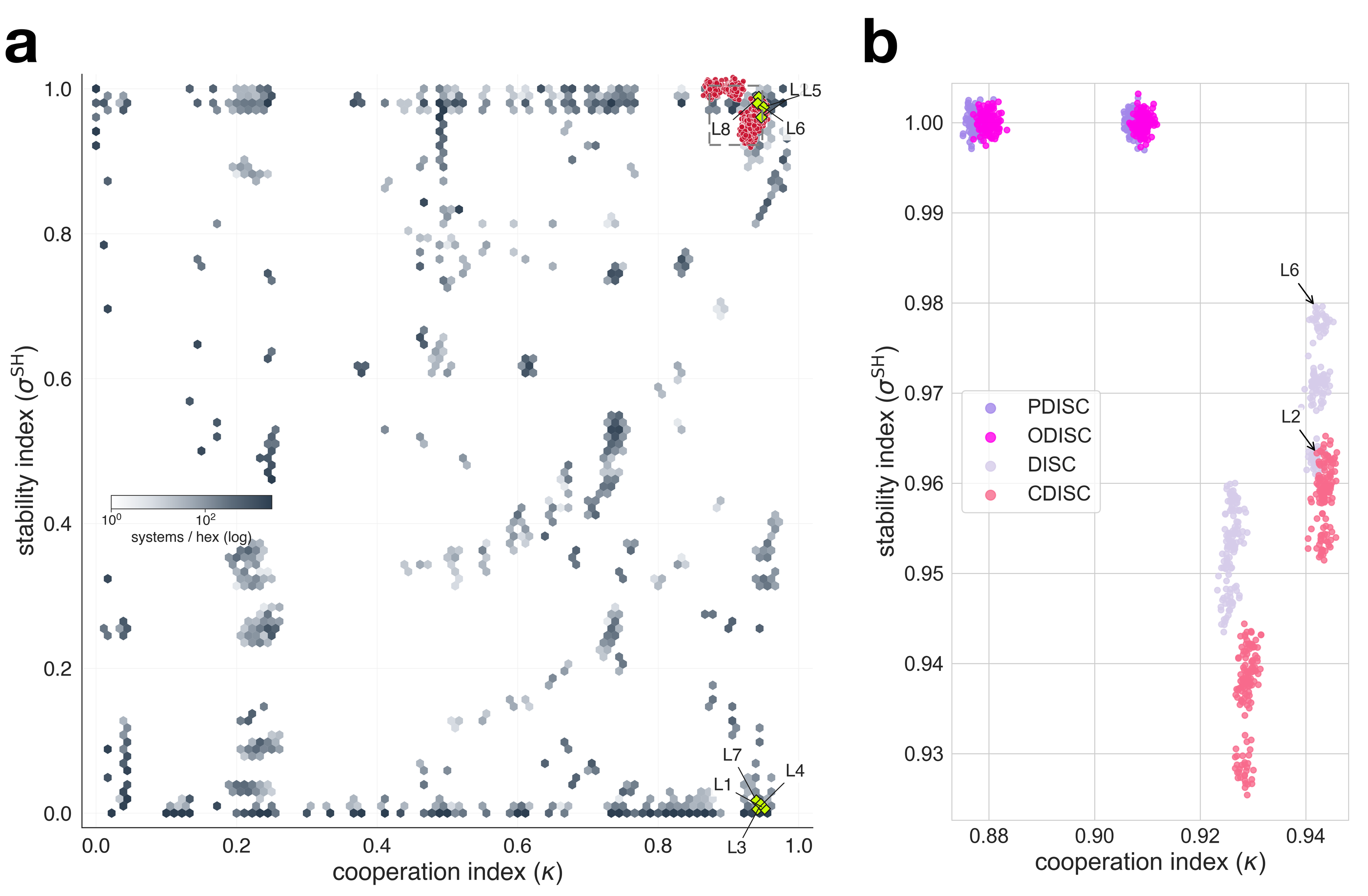}
  \caption{Stability and cooperation indices of all moral systems in the space of stag hunts only, for $\chi = \epsilon = 0.02$. Unlike Fig.~\ref{MT-fig:scatter} of the main text, the vertical axis is the stability index of a single class of dilemma, $\sigma^{\mathrm{SH}}$, rather than the average $\sigmabar$ across the three classes. \textbf{a}, Each of the $\numprint{524800}$ moral systems that are distinct up to the relabeling of reputations (one representative per mirror class; Section~\ref{subsec:mirror-classes}) is positioned by its resident cooperation index ($\kappa$) and its stag hunt stability index ($\sigma^{\mathrm{SH}}$); the consistent discriminating systems are shown as red circles, the leading eight as green diamonds, and the remaining systems as a hexbin density layer on a logarithmic color scale. \textbf{b}, Zoomed-in scatter plot of the consistent discriminating systems, colored by strategy ($\PDISC$, $\ODISC$, $\DISC$, $\CDISC$); note the compressed vertical scale, spanning less than $0.08$. Points are jittered for clarity in both panels.\label{fig:disaggregated_sh}}
\end{figure}

\clearpage

\subsection{Area-weighted stability index}\label{subsec:area-weighted}

The average stability index $\sigmabar$~\eqref{avg-sigma} gives the three classes of cooperative dilemma equal weight.
The area-weighted index $\sigmabarw$ of Section~\ref{subsubsec:weighted-sigma} instead weights each class in proportion to the area of its normalized payoff region.
Figs.~\ref{fig:scatter_zoom_area} and~\ref{fig:scatter_panel_area} show that this choice makes no qualitative difference.
The consistent discriminating systems remain the only cluster to combine high stability with high cooperation, both at the baseline error rates -- where the systems using $\CDISC$ are again the most stable of all -- and across the error-rate grid of Fig.~\ref{fig:alternative-errors}.
The main visible change is a tighter gap between the cluster and the rest, as the ceiling for systems outside the cluster rises from $\sigmabar = 2/3$ to $\sigmabarw = 7/9$ (exact values, attained at every error-rate combination we examine); both are attained by fully uncooperative systems ($\kappa = 0$), which are stable in every prisoner's dilemma and stag hunt but in no snowdrift game.

\begin{figure}[tp]
    \centering
    \includegraphics[width=\textwidth]{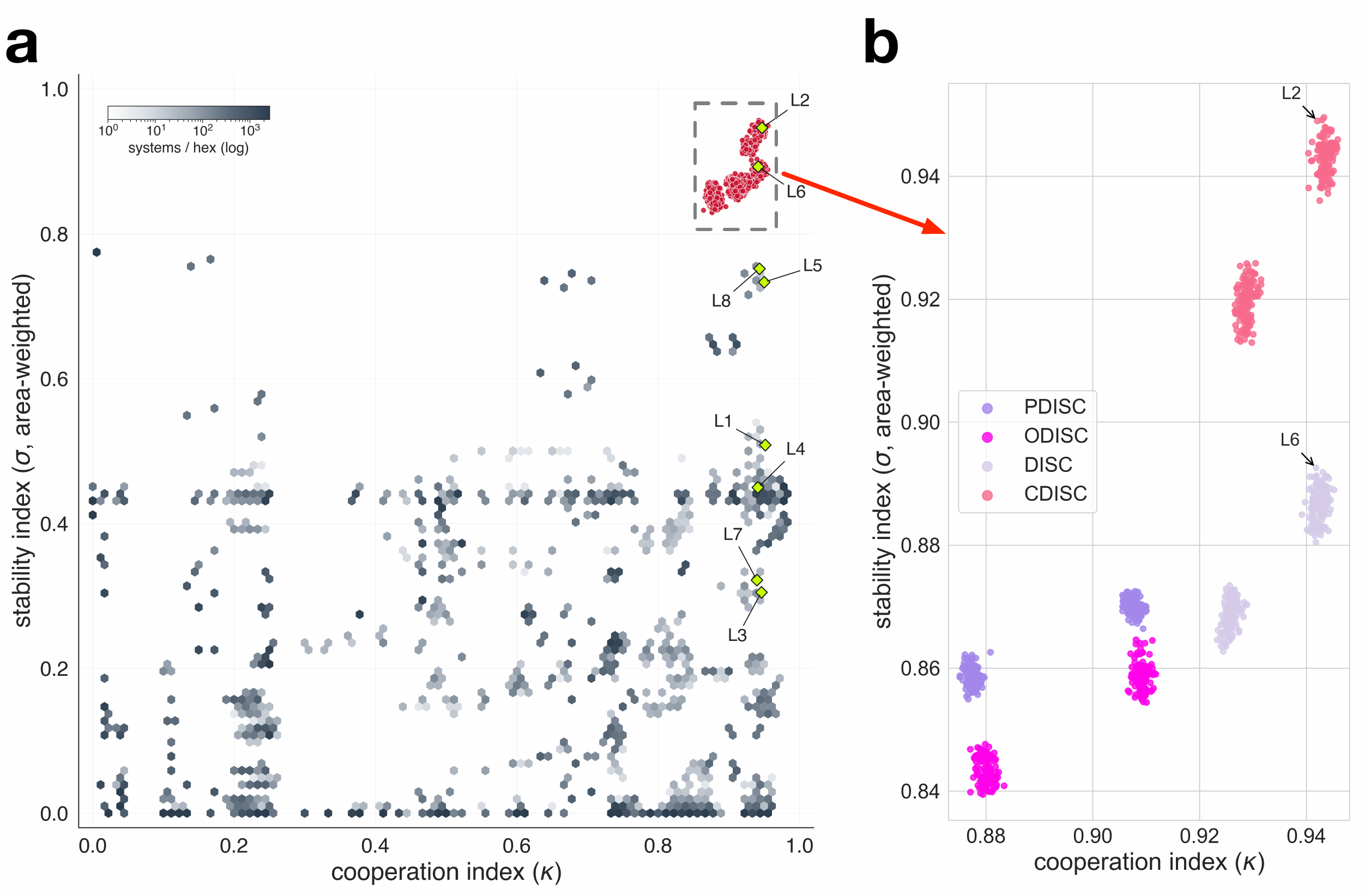}
    \caption{Stability and cooperation indices of all moral systems under the area-weighted index $\sigmabarw$ of Section~\ref{subsubsec:weighted-sigma}, for $\chi = \epsilon = 0.02$; compare Fig.~\ref{MT-fig:scatter} of the main text, which uses the equally weighted index $\sigmabar$. \textbf{a}, Each of the $\numprint{524800}$ moral systems that are distinct up to the relabeling of reputations (one representative per mirror class; Section~\ref{subsec:mirror-classes}) is positioned by its resident cooperation index ($\kappa$) and area-weighted stability index ($\sigmabarw$); the consistent discriminating systems are shown as red circles, the leading eight as green diamonds, and the remaining systems as a hexbin density layer on a logarithmic color scale. \textbf{b}, Zoomed-in scatter plot of the consistent discriminating systems (jittered for clarity), colored by strategy ($\PDISC$, $\ODISC$, $\DISC$, $\CDISC$); as in Fig.~\ref{MT-fig:scatter}b, the systems using $\CDISC$ achieve the highest stability.\label{fig:scatter_zoom_area}}
\end{figure}

\begin{figure}[tp]
    \centering
    \includegraphics[width=\textwidth]{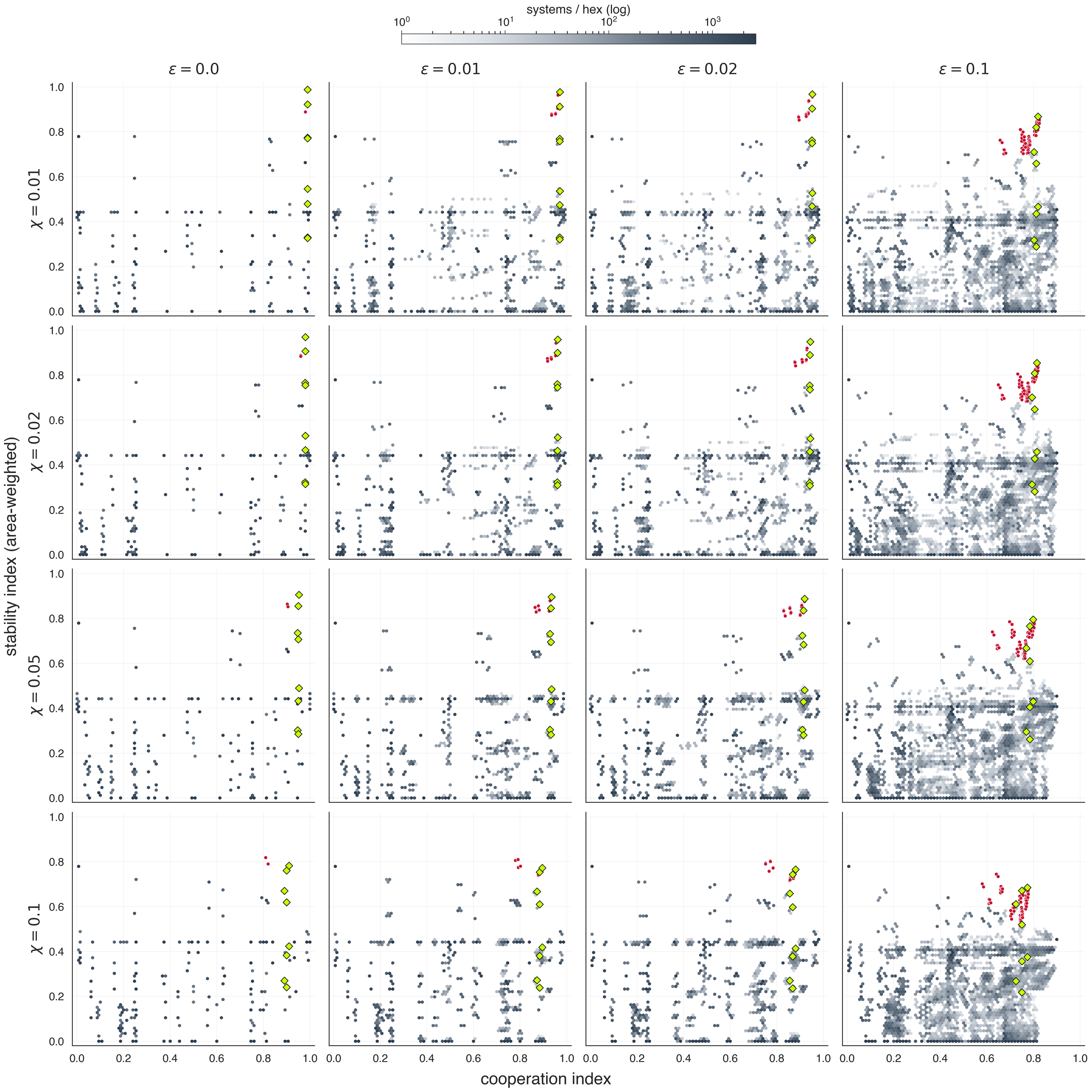}
    \caption{Stability and cooperation indices across a range of error rates, under the area-weighted index $\sigmabarw$ of Section~\ref{subsubsec:weighted-sigma}; compare Fig.~\ref{fig:alternative-errors}, which uses the equally weighted index $\sigmabar$. Each panel shows, for one combination of execution error rate $\epsilon \in \{0, 0.01, 0.02, 0.1\}$ (columns) and assignment error rate $\chi \in \{0.01, 0.02, 0.05, 0.1\}$ (rows), the cooperation index ($\kappa$) and area-weighted stability index ($\sigmabarw$) of each of the $\numprint{524800}$ mirror-class representatives (Section~\ref{subsec:mirror-classes}). The consistent discriminating systems are shown as red circles, the leading eight as green diamonds, and the remaining systems as a hexbin density layer on a logarithmic color scale.}
    \label{fig:scatter_panel_area}
\end{figure}

\end{document}